\documentclass[11pt,reqno]{amsart}

\newcommand{\RemoveComments}{}
\usepackage{amsmath}
\usepackage{amsthm}
\usepackage{amsfonts}
\usepackage{amssymb}

\usepackage{dsfont}
\usepackage{stmaryrd}
\usepackage{bm}

\usepackage{geometry}
\usepackage{fixmath}

\usepackage{upgreek}

\usepackage{graphicx}
\usepackage{fullpage}

\usepackage{color, tikz} 

\usepackage{times}

\usepackage{float}

\usepackage[shortlabels]{enumitem}

\usepackage{algorithm}
\usepackage{algpseudocode}

\usepackage{thmtools}
\usepackage{thm-restate}

\usepackage{hyperref}
\usepackage{cleveref}

\newtheorem{theorem}{Theorem}[section]
\newtheorem{lemma}[theorem]{Lemma}
\newtheorem{proposition}[theorem]{Proposition}
\newtheorem{claim}[theorem]{Claim}
\newtheorem{fact}[theorem]{Fact}
\newtheorem{corollary}[theorem]{Corollary}

\newtheorem{definition}[theorem]{Definition}
\newtheorem{condition}[theorem]{Condition}

\numberwithin{equation}{section}

\newcommand{\cB}{{\mathcal{B}}}
\newcommand{\cC}{{\mathcal{C}}}
\newcommand{\cD}{{\mathcal{D}}}
\newcommand{\cE}{{\mathcal{E}}}

\newcommand{\cK}{{\mathcal{K}}}

\newcommand{\cN}{{\mathcal{N}}}

\newcommand{\cP}{{\mathcal{P}}}
\newcommand{\cQ}{{\mathcal{Q}}}

\newcommand{\cS}{{\mathcal{S}}}

\newcommand{\G}{\mathbold{G}}

\newcommand{\bK}{{\mathbold{K}}}

\newcommand{\bgamma}{\mathbold{\gamma}}

\newcommand{\bsigma}{{\mathbold{\sigma}}}

\DeclareSymbolFont{EuclidLetter}{U}{eur}{m}{n}

\DeclareSymbolFont{EuclidLetter}{U}{eur}{m}{n}
\DeclareMathSymbol{\UpA}{\mathord}{EuclidLetter}{65}
\DeclareMathSymbol{\Upa}{\mathord}{EuclidLetter}{97}

\DeclareMathSymbol{\UpB}{\mathord}{EuclidLetter}{66}
\DeclareMathSymbol{\Upb}{\mathord}{EuclidLetter}{98}

\DeclareMathSymbol{\UpC}{\mathord}{EuclidLetter}{67}
\DeclareMathSymbol{\upc}{\mathord}{EuclidLetter}{99}

\DeclareMathSymbol{\UpD}{\mathord}{EuclidLetter}{68}
\DeclareMathSymbol{\upd}{\mathord}{EuclidLetter}{100}

\DeclareMathSymbol{\UpE}{\mathord}{EuclidLetter}{69}
\DeclareMathSymbol{\upe}{\mathord}{EuclidLetter}{101}

\DeclareMathSymbol{\UpF}{\mathord}{EuclidLetter}{70}
\DeclareMathSymbol{\upf}{\mathord}{EuclidLetter}{102}

\DeclareMathSymbol{\UpG}{\mathord}{EuclidLetter}{71}
\DeclareMathSymbol{\upg}{\mathord}{EuclidLetter}{103}

\DeclareMathSymbol{\UpH}{\mathord}{EuclidLetter}{72}
\DeclareMathSymbol{\uph}{\mathord}{EuclidLetter}{104}

\DeclareMathSymbol{\UpI}{\mathord}{EuclidLetter}{73}

\DeclareMathSymbol{\UpJ}{\mathord}{EuclidLetter}{74}

\DeclareMathSymbol{\UpK}{\mathord}{EuclidLetter}{75}
\DeclareMathSymbol{\upk}{\mathord}{EuclidLetter}{107}

\DeclareMathSymbol{\UpL}{\mathord}{EuclidLetter}{76}

\DeclareMathSymbol{\UpM}{\mathord}{EuclidLetter}{77}

\DeclareMathSymbol{\UpN}{\mathord}{EuclidLetter}{78}
\DeclareMathSymbol{\UpP}{\mathord}{EuclidLetter}{80}
\DeclareMathSymbol{\UpQ}{\mathord}{EuclidLetter}{81}
\DeclareMathSymbol{\Upq}{\mathord}{EuclidLetter}{113}
\DeclareMathSymbol{\UpR}{\mathord}{EuclidLetter}{82}

\DeclareMathSymbol{\UpS}{\mathord}{EuclidLetter}{83}
\DeclareMathSymbol{\ups}{\mathord}{EuclidLetter}{115}

\DeclareMathSymbol{\UpT}{\mathord}{EuclidLetter}{84}
\DeclareMathSymbol{\upt}{\mathord}{EuclidLetter}{116}
\DeclareMathSymbol{\UpW}{\mathord}{EuclidLetter}{87}
\DeclareMathSymbol{\upw}{\mathord}{EuclidLetter}{119}
\DeclareMathSymbol{\UpX}{\mathord}{EuclidLetter}{88}
\DeclareMathSymbol{\UpY}{\mathord}{EuclidLetter}{89}
\DeclareMathSymbol{\UpZ}{\mathord}{EuclidLetter}{90}
\DeclareMathSymbol{\upz}{\mathord}{EuclidLetter}{122}

\newcommand{\dist}{{\rm dist}}

\newcommand{\degr}{\mathrm{deg}}

\newcommand{\Tsaw}{{T_{\rm SAW}}}
\newcommand{\cp}{{\tt A}}

\newcommand{\Exp}{\mathbb{E}}

\newcommand{\Ind}{{\mathds{1}}}
\newcommand{\Ent}{\mathrm{Ent}}
\newcommand{\Dirc}{\cE}
\newcommand{\Cov}{\mathrm{Cov}}

\newcommand{\Gaussian}{\cN}
\newcommand{\Poisson}{\mathtt{Poisson}}

\newcommand{\opnorm}[1]{\| #1 \|_{2}}

\newcommand{\spradius}[1]{ {\rho}(#1)}

\newcommand{\pdleq}{\preceq}
\newcommand{\pdle}{\prec}
\newcommand{\pdgeq}{\succeq}
\newcommand{\pdge}{\succ}

\newcommand{\Id}{\UpI}

\newcommand{\MatrixDW}{{\UpD}}

\newcommand{\Adjacency}{\UpA}

\newcommand{\NB}{\UpB}

\newcommand{\MontA}{\UpA}
\newcommand{\MontD}{\UpD}
\newcommand{\MontC}{\UpK}

\newcommand{\InAct}{\UpJ}

\newcommand{\bBdInAct}{{\Upsigma}}

\definecolor{ao(english)}{rgb}{0.0, 0.5, 0.0}
\definecolor{airforceblue}{rgb}{0.36, 0.54, 0.66}
\definecolor{amber}{rgb}{1.0, 0.75, 0.0}
\definecolor{amber(sae/ece)}{rgb}{1.0, 0.49, 0.0}
\definecolor{amethyst}{rgb}{0.6, 0.4, 0.8}
\definecolor{applegreen}{rgb}{0.55, 0.71, 0.0}
\definecolor{azure(colorwheel)}{rgb}{0.0, 0.5, 1.0}
\definecolor{arylideyellow}{rgb}{0.91, 0.84, 0.42}

\definecolor{bostonuniversityred}{rgb}{0.8, 0.0, 0.0}
\definecolor{bittersweet}{rgb}{1.0, 0.44, 0.37}
\definecolor{bleudefrance}{rgb}{0.19, 0.55, 0.91}
\definecolor{blue(pigment)}{rgb}{0.2, 0.2, 0.6}
\definecolor{blue-violet}{rgb}{0.54, 0.17, 0.89}
\definecolor{britishracinggreen}{rgb}{0.0, 0.26, 0.15}
\definecolor{brilliantrose}{rgb}{1.0, 0.33, 0.64}
\definecolor{byzantine}{rgb}{0.74, 0.2, 0.64}
\definecolor{byzantium}{rgb}{0.44, 0.16, 0.39}

\definecolor{charcoal}{rgb}{0.21, 0.27, 0.31}
\definecolor{cadmiumgreen}{rgb}{0.0, 0.42, 0.24}
\definecolor{cadmiumorange}{rgb}{0.93, 0.53, 0.18}
\definecolor{coquelicot}{rgb}{1.0, 0.22, 0.0}
\definecolor{capri}{rgb}{0.0, 0.75, 1.0}

\definecolor{deeppink}{rgb}{1.0, 0.08, 0.58}
\definecolor{dollarbill}{rgb}{0.52, 0.73, 0.4}

\definecolor{darkmagenta}{rgb}{0.55, 0.0, 0.55}
\definecolor{darkmidnightblue}{rgb}{0.0, 0.2, 0.4}
\definecolor{darkpastelpurple}{rgb}{0.59, 0.44, 0.84}
\definecolor{darkspringgreen}{rgb}{0.09, 0.45, 0.27}
\definecolor{darktangerine}{rgb}{1.0, 0.66, 0.07}
\definecolor{darkgoldenrod}{rgb}{0.72, 0.53, 0.04}

\definecolor{eggplant}{rgb}{0.38, 0.25, 0.32}
\definecolor{electricviolet}{rgb}{0.56, 0.0, 1.0}
\definecolor{ferrarired}{rgb}{1.0, 0.11, 0.0}
\definecolor{forestgreen(traditional)}{rgb}{0.0, 0.27, 0.13}

\definecolor{green(pigment)}{rgb}{0.0, 0.65, 0.31}
\definecolor{glaucous}{rgb}{0.38, 0.51, 0.71} 
\definecolor{goldenbrown}{rgb}{0.6, 0.4, 0.08}
\definecolor{gold(metallic)}{rgb}{0.83, 0.69, 0.22}
\definecolor{gold(web)(golden)}{rgb}{1.0, 0.84, 0.0}
\definecolor{goldenpoppy}{rgb}{0.99, 0.76, 0.0}
\definecolor{goldenyellow}{rgb}{1.0, 0.87, 0.0}
\definecolor{goldenrod}{rgb}{0.85, 0.65, 0.13}

\definecolor{harvestgold}{rgb}{0.85, 0.57, 0.0}
\definecolor{heartgold}{rgb}{0.5, 0.5, 0.0}

\definecolor{heliotrope}{rgb}{0.87, 0.45, 1.0}

\definecolor{iris}{rgb}{0.35, 0.31, 0.81}
\definecolor{oldgold}{rgb}{0.81, 0.71, 0.23}

\definecolor{palegold}{rgb}{0.9, 0.75, 0.54}

\definecolor{rosegold}{rgb}{0.72, 0.43, 0.47}
\definecolor{red(pigment)}{rgb}{0.93, 0.11, 0.14}

\definecolor{sapphire}{rgb}{0.03, 0.15, 0.4}
\definecolor{satinsheengold}{rgb}{0.8, 0.63, 0.21}

\definecolor{uclagold}{rgb}{1.0, 0.7, 0.0}

\definecolor{vegasgold}{rgb}{0.77, 0.7, 0.35}

\newcommand{\Inf}{ {\rm \Upgamma}}

\newcommand{\Field}{\uph}

\newcommand{\norwsphere}{\Psi}
\newcommand{\sqwsphere}{\red{\tt SQR}}

\newcommand{\fun}{f}
\newcommand{\itfun}{g}

\newcommand{\WA}{{\UpA}}

\newcommand{\WB}{{\tt W}}

\newcommand{\pure}{{locally simple}}

\newcommand{\cappedWA}{\UpM}
\newcommand{\cappedWB}{{\WB}}

\newcommand{\gauss}{\bgamma}

\newcommand{\Tmix}{{\rm T}_{\rm mix}}

\newcommand{\In}{\mathrm{In}}
\newcommand{\Out}{\mathrm{Out}}

\newcommand{\hatC}{\widehat{C}}

\newcommand{\singleBlockV}{S}

\newcommand{\maxDelta}{L}
\newcommand{\maxDeltaVal}{\frac{\log n}{\sqrt{d}}}

\newcommand{\badmormw}{\upchi}

\newcommand{\GDmlogSob}{{C}_{\rm mLSI}}
\newcommand{\BMotion}{\UpB}

\newcounter{proptyno}

\newcounter{kcomcount}
\usepackage{pifont}

\newcommand{\red}[1]
{{\color{red}#1}}
\newcommand{\blue}[1]
{{\color{blue}#1}}

\newcommand{\iris}[1]
{{\color{iris}#1}}
\newcommand{\rosegold}[1]
{{\color{rosegold}#1}}
\newcommand{\bluepigment}[1]{{\color{blue(pigment)}#1}}

\newcommand{\greenpigment}[1]
{{\color{green(pigment)}#1}}
\newcommand{\deeppink}[1]
{{\color{deeppink}#1}}

\newcommand{\goldenbrown}[1]
{{\color{goldenbrown}#1}}

\newcommand{\maxEll}{\frac{\log n}{(\log d)^2}}
\newcommand{\mincycdist}{(150/\varepsilon)\cdot\log\log n}
\newcommand{\mincycdistBD}{\red{(40/\varepsilon)\cdot\log\log n}}
\newcommand{\rootB}{\alpha}

\newcommand{\lBlockBound}{\upz}
\newcommand{\lDenom}{ d^{(1/2)+10^{-2}}}
\newcommand{\maxPathL}{\frac{\log n}{\lDenom}}

\newcommand{\kostas}[1]{
--\greenpigment{\goldenbrown{Kostas:}#1}
}

\newcommand{\charis}[1]{
--\bluepigment{\rosegold{Charis:} #1}
}

\newcommand{\spreadpoint}{\newpage}
\newcommand{\TT}{\mathbf{T}}

\newcommand{\CtrlMT}{\UpC_t}
\newcommand{\CtrlM}{\UpC}

\ifdefined \RemoveComments

\renewcommand{\charis}[1]{}
\renewcommand{\kostas}[1]{}
\renewcommand{\spreadpoint}{}

\renewcommand{\blue}[1]{#1}
\renewcommand{\bluepigment}[1]{#1}

\renewcommand{\deeppink}[1]{#1}
\renewcommand{\goldenbrown}[1]{#1}

\renewcommand{\greenpigment}[1]{#1}
\renewcommand{\iris}[1]{#1}
\renewcommand{\red}[1]{#1}

\renewcommand{\rosegold}[1]{#1}

\fi

\title{On sampling Diluted Spin-Glasses with unbounded interactions.}

 \author{Charilaos Efthymiou and Kostas Zampetakis}

\address{Charilaos Efthymiou, {\tt charilaos.efthymiou@warwick.ac.uk}, University of Warwick, Coventry, CV4 7AL, UK.}
\address{Kostas Zampetakis, {\tt konstantinos.zampetakis@tu-dortmund.de}, TU Dortmund, Faculty of Computer Science, 12 Otto-Hahn St, Dortmund 44227, Germany.}

\date{\today}

\begin{document}

\maketitle

\begin{abstract}

Spin-glasses are natural Gibbs distributions that have been studied in theoretical computer science for many decades. Recently, they have been gaining renewed attention from the community as they emerge naturally in neural computation 
and learning, network inference, optimisation, and many other areas.

We study the problem of efficiently sampling from spin-glass distributions when the underlying graph is a typical instance 
of $\G(n,d/n)$, 
%i.e., the random graph on $n$ vertices such that each edge appears independently with probability $d/n$, 
where the expected degree $d=\Theta(1)$.

Our main focus is on the 2-spin model at inverse temperature $\beta$, which serves as a prototypical example of a disordered 
system. We consider this distribution to be one of the most interesting cases of spin-glasses for $\G(n,d/n)$, and one of the most challenging 
to analyse, since its Gaussian couplings give rise to {\em unbounded} interaction between the sites of the system. 

We employ the well-known {\em Glauber dynamics} to sample from the aforementioned distribution. 
For the  typical instances of the 2-spin model on $\G(n,d/n)$, we show that the mixing time of  Glauber dynamics is 
${\textstyle n^{\left(1+\Theta\left(\frac{1}{\sqrt{d} }\right)\right)} }$, 
for any $\beta \leq \frac{1}{4\sqrt{d}}$.

Our results can also be adapted for the case of spin-glass distributions where the interactions are bounded. In that respect, 
we obtain rapid mixing of Glauber dynamics for the Viana-Bray model on $\G(n,d/n)$ when $\beta \leq \frac{1}{4\sqrt{d}}$. This 
improves on the current best bound which is $\beta<\frac{0.18}{\sqrt{d}}$.

We utilise the powerful {\em stochastic localisation} in our analysis and, in particular, we build and improve on the scheme 
introduced in [Liu, Mohanty, Rajaraman and Wu: FOCS 2024]. 
Our work provides, for the first time, an analysis with stochastic localisation in which both the degrees and the interactions of 
the diluted spin-glass can be unbounded. 
This is achieved by introducing an elaborate block partition of the set of vertices to deal with the effect of high-degree vertices. 
For the unbounded interactions, we  introduce appropriate matrix norms to control the conservation of the entropy  along the localisation path. 

\end{abstract}

\ifdefined \KeepComments
\newpage
\section*{\deeppink{General Comments}}

\kostas{In the theorems talking about the existence of block partition with certain properties make sure we can satisfy them all simultenuously}

\newpage
\fi

 \setcounter{tocdepth}{1}

 \tableofcontents

%\newpage
%\setcounter{page}{1}

\section{Introduction}
\label{sec:Intro}

{\em Spin-glasses} are natural, high-dimensional Gibbs distributions that have been studied in theoretical computer science for many decades. 
Recently, they have been gaining renewed attention from the community as they emerge naturally in {\em neural computation} and {\em learning},
 e.g., in the Hopfield model, as models of {\em network inference}, 
e.g., in the stochastic block model, in {\em optimisation}, {\em counting-sampling} and many other areas see, e.g.,
\cite{SteinNewmanSpinGlassBook, GamarnikJWFOCS20, CoEfJKKCMI, OptElAlaoui, KoehLRColt22, eldan2022spectral, EfthICALP22}. 

Spin-glasses are widely considered canonical models of {\em extremely} disordered 
systems \cite{mezard1990spin,SteinNewmanSpinGlassBook} and, as such, have been 
studied extensively in mathematics and mathematical physics see, e.g., 
\cite{TalagrandAnnals,panchenko2013parisi,franz2001exact,guerra2004high}. They have also 
been studied in statistical physics  since the early 1980s see, e.g.,  \cite{mezard1990spin,SteinNewmanSpinGlassBook,RSBParisi}. The seminal, groundbreaking work of Giorgio Parisi on spin-glasses 
earned him the Nobel Prize in Physics in 2021.

We focus on a natural, well-known case of spin glasses, the {\em 2-spin model}.
For fixed graph $G=(V,E)$ and a set of independent identically distributed (i.i.d)
standard Gaussians $\{ \gauss_{e} \}_{e \in E(G)}$, we let the {\em interaction matrix} 
$\InAct = \InAct(G,\{ \gauss_{e} \}) \in \mathbb{R}^{V\times V}$ be defined by
\begin{align}\label{eq:DefofIntMatGauss}
\InAct(u,w) &=\Ind\{u\sim w\}\times \gauss_{\{ u,w\}} & \forall \; u,w\in V \enspace.
\end{align}
The 2-spin model with {\em inverse temperature} $\beta \in \mathbb{R}_{\geq 0}$ and {\em external field} 
$\Field\in \mathbb{R}^V$, corresponds to the Gibbs distribution $\mu_{\InAct,\beta, \Field}$ such that 
 \begin{align}\label{eq:DefOfGibbs}
\mu_{\InAct,\beta, \Field}(\sigma)
&\propto 
\exp\left(\frac{\beta}{2} \cdot \langle \sigma, \InAct \cdot \sigma\rangle + \langle \sigma, \Field \rangle \right), & 
\forall \; \sigma\in \{\pm 1\}^V 
\enspace, 
\end{align}
where $\langle {\bf x}, {\bf y} \rangle$ stands for the standard inner product of vectors ${\bf x}$ and ${\bf y}$, 
while the symbol $\propto$ stands for ``proportional to". 
The 2-spin model on the complete graph corresponds to the well-known {\em Sherrington-Kirkpatrick} model 
(SK model) \cite{SKModel}.

Due to the Gaussian couplings $\{ \gauss_{e} \}$, we note that $\mu_{\InAct,\beta, \Field}$ is a {\em random} 
Gibbs distribution on $\{\pm 1\}^{V}$. A more tame version of the above distribution is when, instead of the Gaussian 
variables, we use random $\{\pm 1\}$. This corresponds to the {\em Viana-Bray} model. 

In this work, we focus on the distribution in \eqref{eq:DefOfGibbs} with the underlying graph being an 
instance of the sparse random graph $\G(n,d/n)$. Recall that $\G(n, d/n)$ is the random graph on the vertex set 
$V_n=\{1,\ldots, n\}$ where each edge appears independently with probability $d/n$. Here, we assume that the 
expected degree $d\in \Theta(1)$, i.e., it is a fixed number. We obtain an instance of the 2-spin model by first drawing 
the underlying graph from the distribution $\G(n,d/n)$ and then, given the graph, we generate the random Gibbs 
distribution as in \eqref{eq:DefofIntMatGauss} and \eqref{eq:DefOfGibbs}.

On a first account, the {2-spin model} on $\G(n,d/n)$ may appear innocent, i.e., it resembles the standard Ising model with just 
the addition of Gaussian couplings. It turns out, though, that it is a {\em fascinating} distribution with a lot of intricacies, 
while the configuration space has an extremely rich structure; e.g.,  it is conjectured to exhibit ``{\em infinite} 
Replica Symmetry Breaking" \cite{mezard1990spin}.

In this work, we consider the natural problem of efficiently sampling from the 2-spin model on $\G(n,d/n)$.
We employ the powerful Markov Chain Monte Carlo (MCMC) method, and in particular, we use the 
{\em Glauber dynamics}. This is a popular Markov chain utilised for sampling from high-dimensional 
Gibbs distributions such as the Hard-core model, the Colouring model, spin-glasses, etc. In this setting, 
the measure of efficiency is the {\em mixing time} of the Markov chain.

Sampling from Gibbs distributions induced by instances of $\G(n, d/n)$, or, more generally, instances of the 
so-called {\em random Constraint Satisfaction Problems}, is at the heart of the endeavours to investigate 
connections between {\em phase transitions} and the efficiency of algorithms, e.g. 
\cite{OptasOghlan08,alaoui2020algorithmic,COghlanEfth11,GalStefVigJACM15,GamSudan17,SlySun12, mossel2013exact}. 
The MCMC sampling problem on the sparse random graph has garnered considerable attention, e.g., see 
\cite{dyer2006randomly,BezakovaGGS22,EftFeng23,mossel2010gibbs, mossel2013exact, EfthymiouHSV18,DyerFriez10,ChManMo23},
and is considered an intriguing case to study. Until recently, the focus has been on sampling 
standard Gibbs distributions, which is already  a very challenging problem. Spin-glasses take
us a step further. Working with the {2-spin model}, we introduce an {\em extra level} of disorder due to the 
Gaussian couplings at the edges of the graph. Hence, in our analysis, we need to deal with the disorder of 
both $\G(n,d/n)$ and the random Gibbs distribution on this graph.

Given a typical instance of $\G(n,d/n)$,  the objective of this work is to investigate the range of the inverse temperature $\beta$ such that the corresponding 
Glauber dynamics exhibits {\em fast mixing}. It is a folklore conjecture that we have fast mixing time for any 
$\beta<\beta_{\rm rec}$, where $\beta_{\rm rec}=\beta_{\rm rec}(d)$ specifies the onset of the so-called 
{\em reconstruction region} of the 2-spin model on $\G(n,d/n)$. In particular, $\beta_{\rm rec}$ satisfies 
\begin{align} 
d \cdot \textstyle \mathbb{E}\left [ \left (\tanh (\gauss\cdot \beta_{\rm rec}) \right)^2 \right] &= 1 \enspace, 
\end{align}
where the expectation is with respect to the standard normal random variable $\gauss$. 
An upper bound for $\beta_{\rm rec}$ was found in \cite{guerra2004high}, while in 
\cite{EfthZampICALP23} it was shown that this bound is tight, yielding the above relation.
It is elementary to verify that $\lim_{d\to\infty}\beta_{\rm rec}=\frac{1}{\sqrt{d}}$.

For our results, we also define $\beta_c$ such that, 
\begin{align}\label{eq:defOfBc}
d\cdot \Exp \left[ \left( \tanh( \gauss\cdot \beta_c) \right)^2\right] = \upkappa=1/4 \enspace,
\end{align}
where again the expectation is with respect to the standard Gaussian random variable $\gauss$.

The following is the main result of our paper. 

\begin{theorem}\label{thrm:MainResult}
There exists sufficiently large constants $C>0$ and $d_0>0$ such that for any $d \ge d_0$, 
%for any field $\Field\in \mathbb{R}^{V_n}$  
and any  $0<\beta \le \beta_c(d)$ the following is true:

Let $\mu$ be the 2-spin model on $\G(n, d/n)$, with inverse temperature $\beta$ and external field 
\red{$\Field \in \mathbb{R}^{V_n}$}. 
Then, with probability $1-o(1)$ over the instances of $\mu$, 
Glauber dynamics exhibits mixing time $\Tmix \leq n^{\left(1+\frac{C}{\sqrt{d}} \right)}$.
\end{theorem}

Note that the constant $C>0$ in \Cref{thrm:MainResult} does not depend on $d$. 
%Also, we allow the external field $\Field$ to vary with $n$. 

We did not try to optimise the bound on the mixing time in \Cref{thrm:MainResult}. 
There is a natural lower bound  %on the mixing time of Glauber dynamics for  the 2-spin model on $\G(n,d/n)$,
which is $n \exp(\Omega(\sqrt{\log n}))$.
This is due to  the presence of isolated stars in $\G(n,d/n)$ with rather large couplings at their edges.

 \Cref{thrm:MainResult} still holds if we replace the Gaussians with couplings drawn from an 
arbitrary {\em sub-Gaussian} distribution. 
Our results can be easily adapted to apply to the  Viana-Bray model, too, 
i.e., the couplings are random $\{\pm 1\}$. 
In this case,  we obtain rapid mixing of Glauber dynamics for any $\beta \leq 0.25/\sqrt{d}$, 
improving on  the bound $\beta<0.18/\sqrt{d}$ obtained 
by Liu, Mohanty, Rajaraman and Wu in \cite{KuiKuiSpinGlass2024}.

We obtain  \Cref{thrm:MainResult} by building and improving on the {\em stochastic localisation scheme}
 proposed   in \cite{KuiKuiSpinGlass2024},  for the  Viana-Bray model on $\G(n,d/n)$. 
Our approach   obtains the best possible  rapid mixing bounds that  the specific 
   scheme can give.
%the  optimal way to  put together the two main components of the analysis
%introduced in \cite{KuiKuiSpinGlass2024}, i.e.  stochastic localisation for the mixing bounds of Glauber dynamics and 
%the block-partition for dealing with the effect of high degree vertices. 
Furthermore, our  analysis is novel and   allows us to
handle the untypically  strong couplings of the 2-spin model.  
Our bounds are still  a factor 4 away from $\beta_{\rm rec}$.  
Getting rid of the factor 4 is a major  open problem in the area of MCMC sampling of spin-glasses {\em in general}.
\subsubsection*{Related Work:}
In this part, we provide a high-level overview of the related work on the problem. For the economy
of space, we focus on the results for ``diluted systems", i.e., with bounded (expected) degree.

Related to our paper is the recent work of Anari, Koehler and Vuong in \cite{AnariSTOC24}, which studies 
the Viana-Bray spin-glass model on the random $d$-regular graph. They establish rapid mixing of Glauber 
dynamics for $0<\beta < \frac{0.295}{\sqrt{d-1}}$. This model differs from what we consider here because 
both the degrees of the vertices and the edge couplings  are {\em bounded}.

The subsequent work of Liu, Mohanty, Rajaraman and Wu in \cite{KuiKuiSpinGlass2024} considers the
Viana-Bray model on $\G(n,d/n)$. That is, {\em bounded} couplings, i.e., random $\{\pm 1\}$, but the degrees of the graph can 
grow {\em unbounded}. 
Note that for typical instances of $\G(n,d/n)$, % the vast majority of the vertices  are of degree $\approx d$, while there are vertices 
the maximum degree is  as huge as $\Theta(\frac{\log n}{\log\log n})$. 
In \cite{KuiKuiSpinGlass2024}, they establish rapid mixing for Glauber dynamics for any $0<\beta < \frac{0.18}{\sqrt{d}}$. 
The bound they obtain for $\beta$ is smaller than that in \cite{AnariSTOC24}, which is for the random 
$d$-regular graph. On the positive side, the approach in \cite{KuiKuiSpinGlass2024} allows for unbounded degrees.

Interestingly, both of the aforementioned works rely on the {\em stochastic localisation} technique to 
establish their mixing bounds. This technique was introduced by Eldan in \cite{eldan2013thin} and 
has been pivotal in proving functional inequalities in various settings, e.g., it gives the state-of-the-art 
bounds for the Kannan-Lov\'asz-Simonovich (KLS) conjecture in \cite{chen2021almost}. Recently, 
it has been used to analyse the mixing time of Glauber dynamics, e.g., see the influential works of 
Eldan, Koehler and Zeitouni in \cite{eldan2022spectral} and Chen and Eldan in \cite{chen2022localization}.
Stochastic localisation, naturally, imposes spectral conditions for the rapid mixing of Glauber dynamics 
for distributions such as the Ising model and spin-glasses. This turns out to be quite useful in our case, because 
it yields bounds for $\beta$ which are of the form $c/\sqrt{d}$, for a constant 
$c>0$ independent of $d$.

The work of Efthymiou and Zampetakis in \cite{EfthZamoCLT24} considers the same distribution as 
the one we consider here. They obtain a $\beta< \sqrt{\frac{\pi}{2}}\cdot \frac{1}{d}$ bound for the rapid 
mixing of Glauber dynamics on the 2-spin model on $\G(n,d/n)$. The bound on $\beta$ they obtain is 
even worse than that in \cite{KuiKuiSpinGlass2024}, but their analysis allows for unbounded couplings 
and degrees. The work in \cite{EfthZamoCLT24} does not use stochastic localisation, it rather uses the 
{\em path coupling} technique introduced by Bubley and Dyer in \cite{bubley1997path}.

Considering the previous works, we conclude that the approaches utilising stochastic localisation
seem to give better bounds for $\beta$, but, on the other hand, they do not seem to handle well
the unbounded couplings. Indeed, incorporating such couplings into the analysis turns out to be a
challenging task. Already getting from the bounded-degree case in \cite{AnariSTOC24} to the 
unbounded one in \cite{KuiKuiSpinGlass2024} is a non-trivial technical achievement.
In that respect, an  important contribution of this work  is to show how we can utilise 
the stochastic localisation technique to obtain rapid mixing bounds
in the presence of both unbounded couplings and degrees, i.e., allowing us to deal with 
the 2-spin model on $\G(n,d/n)$.

\subsubsection*{Block Partitions with Stochastic Localisation:}
To establish rapid mixing for Glauber dynamics on $\G(n,d/n)$, not necessarily for spin-glasses 
but also for general Gibbs distributions, the main challenge is to deal with the  {\em effect of high-degree 
vertices}. That is, to show that the relatively rare, high-degree vertices we encounter in the 
typical instances of $\G(n,d/n)$ have a negligible effect. Hence, it is only the vertices of 
degree $\approx d$ that are important in the analysis. 
%In turn, this allows us to obtain rapid  mixing bounds in terms of $d$. 

The natural  strategy to circumvent the effect of high-degree vertices 
is to introduce a partition on the set of vertices of the underlying graph into two parts: the ``good" 
and the ``bad" part. The bad part contains the high-degree vertices 
and their nearby small-degree neighbours. The components induced by the vertices in the bad part are typically small and simply structured,
i.e., each component is a tree or a unicyclic graph. In each such component, the high-degree vertices 
are hidden deep inside, while there is a  ``buffer" of low-degree vertices that separates them 
from the good part of the graph. This buffer of low-degree vertices diminishes the effect of the 
high-degree ones.

The above approach was introduced by Dyer, Flaxman, Frieze and Vigoda in \cite{dyer2006randomly}
to establish rapid mixing of Glauber dynamics on the colourings of $\G(n,d/n)$, using path coupling.
It turns out that there is more than one way someone can choose how to create the good-bad partitions. 
Further improvements on \cite{dyer2006randomly}   typically amount to introducing a better partition of the set of vertices.
The work in \cite{EfthZamoCLT24}  exploits these ideas to establish mixing bounds for the 2-spin model.

Typically, the aforementioned approaches  utilise path-coupling to establish rapid mixing of Glauber dynamics. 
The authors in \cite{KuiKuiSpinGlass2024}  also rely on a similar partitioning of the vertex set.
However, rather than  using path coupling, they utilise stochastic localisation. That is,  they combine 
stochastic localisation for the mixing analysis and the block-partition for dealing with  high-degree
vertices. 
Their approach, novel as it is,  has certain {\em limitations}.  E.g., for the Viana-Bray model, for which  this approach 
was introduced,  it does not provide  the natural rapid mixing bound  $\beta \leq \frac{0.25}{\sqrt{d}}$, one might have expected.
It rather gives $\beta < \frac{0.18}{\sqrt{d}}$. 
The discrepancy between these two  bounds %, i.e.,  the one they obtain and the natural one, 
indicates that  the optimal way to  put together the two components of the analysis, i.e., stochastic localisation  and the block-partition, is still missing. 
Furthermore, the analytic tools  they provide are not always  strong enough to be used for  other interesting cases of spin-glasses.
%, i.e, they rather seem  to be tailored only to the Viana-Bray model.

In light of all the above, our objective-motivation  is to address the aforementioned limitations. To this end, 
we reconcile any discrepancies between the stochastic localisation analysis and the block-partition by introducing a
new,  more elaborate way of partitioning the set of vertices into good and bad parts. This block-partition
 allows us to obtain the desirable bound  $\beta\leq \frac{0.25}{\sqrt{d}}$.  We keep the same set-up for 
 the stochastic localisation as before and we use it as a basis for the new block-partition. 
Our approach further allows us to get a better handle on the bad part of the graph and apply a more abstract analysis
that utilises novel matrix norms. In turn, this yields stronger results that allow us to obtain the rapid mixing bounds for 
the more general 2-spin model.

\subsection*{Notation}
Let $G= (V, E)$ be a graph and $B \subseteq V$. We write $\partial_{\rm in}B$ for inner boundary of the set $B$, i.e., the set of vertices in $B$ that have a neighbor outside $B$. We also denote with $\partial _{\rm out }B$, the outer boundary of $B$, i.e., the set of vertices that do not belong to $B$ and have a neighbor inside $B$.

%\charis{Is $\langle \cdot, \cdot \rangle$ defined?}

\newcommand{\WSA}{{\Uptheta}}
\newcommand{\cappedWSA}{\UpM}

\spreadpoint

\section{Approach}

In this section, we provide a high-level description of our approach and summarise our contributions. We assume a minimal prior knowledge on the reader's part.
 The presentation of the material is  informal.

\subsection*{The modified log-Sobolev constant:}
We obtain our  results by utilising functional inequalities, and in 
particular, by establishing a bound on the so-called {\em modified} log-Sobolev constant of Glauber dynamics.

For concreteness, let $G=(V_n,E)$ be a fixed graph. Consider the Glauber dynamics on a Gibbs distribution 
$\mu :\{\pm 1\}^{V_n}\to [0,1]$. 

For a function $f\in \mathbb{R}_{>0}$, we let  the {\em entropy} of $f$ with respect to $\mu$ be such that
\begin{align}\nonumber 
\mathrm{Ent}_{\mu }(f) &= \textstyle 
\mu\left(f\log f\right) -\mu\left(f\right ) \log \mu(f)
\enspace, 
\end{align} 
%%%
where $\mu(f) = \sum\nolimits_{\sigma \in \{\pm 1\}^{V_n}}\mu(\sigma)\cdot f(\sigma) $.

The associated {\em Dirichlet form} of Glauber dynamics for function $f\in \mathbb{R}_{>0}$ is given by 
\begin{align}\nonumber 
\Dirc_{\mu}(f, \log f)&=\sum\nolimits_{\sigma,\tau\in \{\pm 1\}^{V_n}}
(f(\sigma)-f(\tau))\cdot \log \frac{ f(\sigma)}{f(\tau) } 
\cdot \mu(\sigma)\cdot {P}(\sigma,\tau)\enspace, 
\end{align}
where ${P}(\sigma,\tau)$ is the transition probability from state $\sigma$ to state $\tau$ for 
Glauber dynamics.

We say that Glauber dynamics satisfies the modified log-Sobolev inequality with constant $\GDmlogSob>0$, 
if for any $f\in \mathbb{R}_{>0}$, we have
\begin{align}\nonumber 
 \Dirc(f, \log f) & \geq \GDmlogSob \cdot \Ent_\mu(f) \enspace. 
\end{align}
We  use $\GDmlogSob$ to bound the mixing time of Glauber dynamics. Using results from 
\cite{bobkov2006modified},  we have 
\begin{equation}\nonumber 
{\rm T}_{\rm mix} \leq (\GDmlogSob)^{-1} \cdot \left( \log(\log(\mu^{-1}_{\min}))+\log 2+1 \right) \enspace ,
\end{equation} 
where $\mu_{\min}=\min\{\mu(\sigma)\ |\ \sigma\in \{\pm 1\}^V, \ \mu(\sigma)>0\}$.

Our  aim is to  show that $\GDmlogSob \geq \frac{1}{n^c}$, for a small constant $c>0$.

\subsection*{General Stochastic Localisation:}
We  bound  $\GDmlogSob$ by employing  {\em stochastic localisation}.

With stochastic localisation, we introduce a stochastic process $\{\mu_t\}_{0\leq t\leq 1}$ over the probability 
measures on $\mathbb{R}^{V_{n}}$.  Typically, we set  $\mu_{0}=\mu$, where $\mu$ is the distribution of interest, 
e.g,  the 2-spin model.  This process specifies a  {\em localisation path} from $\mu_0$ to $\mu_1$. 
We  use this process to  obtain a lower bound on the modified log-Sobolev constant of $\mu_0$ 
(and hence $\mu$) in terms of  the modified log-Sobolev of  $\mu_1$ and the ``conservation of the entropy"  along 
the localisation path.

 Let us be a bit more concrete. 
We parametrise the localisation process $\{\mu_t\}_{0\leq t\leq 1}$  with respect to the \red{positive-definite} 
$V_n\times V_n$ control matrices $\{\CtrlMT\}_{ t\in [0,1]}$. We set 
$\mu_0=\mu$, while $\mu_t$ satisfies the following system of stochastic differential equations: for 
every $x\in \mathbb{R}^{V_n}$, we have
\begin{align}\nonumber
d \left( \frac{\mu_t}{\mu_{0}} \right)(x) &= 
\frac{\mu_t}{\mu_{0}} (x) \cdot \langle x - \Upb_t, \CtrlMT \cdot d \BMotion_t \rangle 
& 0\leq t\leq 1\enspace, 
\end{align}
where $\Upb_t=\mathbb{E}_{\bsigma\sim \mu_t}[\bsigma]$, while $\{ \BMotion_t\}_{0\leq t\leq 1}$ is 
the standard Brownian motion in $\mathbb{R}^{V_n}$.

It is well-known, e.g., see \cite{chen2022localization}, that for any $0\leq t \leq 1$, almost surely we have
\begin{align}\nonumber
\mu_{t}(\sigma) &
\propto 
\exp\left(\frac{\beta}{2} \cdot \langle \sigma, \InAct_t \cdot \sigma\rangle + \langle \sigma, \Field_t \rangle \right), 
 &\quad \forall \sigma\in \{\pm 1\}^{V_n} 
\enspace,
\end{align}
where $\InAct_t=\InAct-\int^t_0\CtrlM^2_s ds$ and $\Field_t=\int^t_0[\CtrlM_s \cdot d\BMotion_s+\CtrlM^2_s \cdot \Upb(\mu_s)]ds$.

Intuitively, the above implies that at time $t$, we can have a ``simpler" interaction matrix $\InAct_t$ obtained 
by removing matrix $\int^t_0\CtrlM^2_s ds$ from the initial interaction matrix $\InAct$. The process also 
changes the external field with $t$, but for our discussion here, this is not a so important phenomenon.

Suppose that the choice of $\{\CtrlMT\}_{ t\in [0,1]}$ is such that for each function 
$f\in \mathbb{R}_{>0}$, we have
\begin{align}\label{eq:ExpEntrM0VsEntrM1Intro}
\frac{\Exp[\mathrm{Ent}_{\mu_1}(f)]}{\mathrm{Ent}_{\mu_0 }(f) }\geq \upxi \enspace,
\end{align}
where the expectation is for the random measure $\mu_1$. 
Usually \eqref{eq:ExpEntrM0VsEntrM1Intro} is called the {\em conservation of entropy inequality} along the localisation path.

Furthermore, suppose that almost surely $\mu_1$ is such that
\begin{align}\label{eq:ExpEntrM0VsEntrM2Intro}
\GDmlogSob(\mu_1) >\uppsi \enspace.
\end{align}
Then, results in  \cite{chen2022localization}, imply that
\begin{align}
\GDmlogSob(\mu_0) >\uppsi \cdot \upxi \enspace. 
\end{align}

In light of the above, the idea is to choose $\CtrlMT$'s appropriately and restrict the localisation path so that $\mu_1$ is somehow ``simple", ``easy" to analyse. 
E.g., one might aim for  $\mu_1$  to be  a product measure.  At the same time, we would like the entropy along the localisation path to be conserved, i.e.,
$\upxi$ to be large.

\subsection*{Control Matrices Vs Entropy Conservation:}
A natural question at this point is how to choose the control matrices $\{\CtrlMT\}_{t\in [0,1]}$
for the distribution $\mu_{\UpJ,\beta,\Field}$ on $\{\pm 1\}^{V_n}$ defined in \eqref{eq:DefOfGibbs}.
The choice of $\CtrlMT$ reflects directly on the conservation of  entropy.

Recall that the covariance matrix of $\mu_{\UpJ,\beta,\Field}$, denoted as $\Cov (\InAct,\beta,\Field )$, is defined by
\begin{align}\nonumber
\Cov ( \InAct,\beta, \Field ) =\Exp\left[\bsigma \cdot \bsigma^T \right] -\Exp\left[\bsigma \right]\cdot \Exp\left[\bsigma^T\right]\enspace,
\end{align}
where $\mathbold{\sigma}\in \{\pm 1\}^{V_n}$ is distributed as in $\mu_{\UpJ,\beta,\Field}$.

Suppose that $\{\CtrlMT\}_{t\in [0,1]}$ is such  that for all $t\in [0,1]$ and any external field $\Field'\in \mathbb{R}^{V_n}$, we have
\begin{align}\label{eq:CovarianceBoundWithControlMatrixIntro}
\beta\cdot \opnorm{ \CtrlMT \cdot \Cov (\InAct_t,\beta, \Field' ) \cdot \CtrlMT } &\leq \upalpha_t\enspace, 
\end{align}
for $\upalpha_t>0$. Then, it is shown in \cite{chen2022localization} that,  for any $f\in \mathbb{R}_{> 0}$, we have
\begin{align}\nonumber
\frac{\Exp[\mathrm{Ent}_{\mu_1}(f)]}{\mathrm{Ent}_{\mu_0 }(f) }\geq\exp\left( -\int^1_0\upalpha_t \cdot dt \right)\enspace. 
\end{align}
The above relation provides the connection between $\{ \CtrlMT\}_{t\in [0,1]}$ and the conservation of entropy we  use here. 

\subsection*{Stochastic Localisation for $\G(n,d/n)$ - A Toy Example} 
%To cast more light on the aforementioned connection, 
Consider an example in which $\mu$ is the standard  Ising model with inverse temperature $\beta$ and external field $\Field$ on a fixed graph $G=(V,E)$. 
The use of the Ising model in this example is to keep the discussion simple.

With the Ising model, the interaction matrix $\InAct$ can be equal to the {\em adjacency 
matrix} $\Adjacency_G$ of graph $G$. Recall that $\Adjacency_G$ is such that for any two vertices 
$u,w$, we have $$\Adjacency_G( u,w)= \Ind\{\textrm{$u, w$ are adjacent in $G$}\} \enspace.$$

In this case, stochastic localisation provides naturally rapid mixing bounds for inverse temperature $\beta$ 
expressed in terms of the {\em spectral radius} $\varrho$ of the adjacency matrix $\Adjacency_G$, i.e., 
$\varrho=\opnorm{\Adjacency_G}$. 

For technical reasons we deviate slightly from the above and  specify that 
\begin{enumerate}[(i)]
\item $\InAct=\Adjacency_G+(1+\upzeta)\cdot \varrho \cdot \Id$, for small $\upzeta>0$, 
\item for each $t\in [0,1]$, set $\CtrlMT=\InAct^{\frac{1}{2}}$.
\end{enumerate}
Item (i) implies that the interaction matrix $\InAct$ we use is not exactly the standard $\Adjacency_G$, 
but a ``shifted version" of this matrix. That is, we add $(1+\upzeta)\cdot \varrho$ to each diagonal entry of the matrix, 
for small $\upzeta>0$. 
This shifting does not alter the Gibbs distribution. 
It guarantees that $\InAct$ and $\CtrlMT$ are positive definite matrices, i.e.,  
$\InAct, \CtrlMT \succ 0$.

The above choice of $\InAct$ and $\CtrlMT$ implies that \red{almost surely} $\InAct_1$ gives rise to a product measure on $\mathbb{R}^{V_n}$, 
and hence, computing the modified log-Sobolev 
for $\mu_1$ is trivial. The mixing bounds in terms of $\beta$ arise from the entropy conservation and, 
in particular, our desire to control the quantity
\begin{align}\nonumber 
\opnorm{ \CtrlMT \cdot \Cov (\InAct_t,\beta, \Field' ) \cdot \CtrlMT } \enspace.
\end{align}
As mentioned above, the natural bound we obtain for $\beta$ is in terms of the spectral radius  of
$\Adjacency_G$.

A vanilla application of the above to the Ising model on $\G(n, d/n)$ makes it challenging to control 
$\opnorm{ \CtrlMT \cdot \Cov (\InAct_t,\beta, \Field' ) \cdot \CtrlMT }$ when  $\beta\in \Theta(1)$.
The natural bound one gets from the vanilla approach requires $\beta\in o(1)$ as $\|\Adjacency_{\G(n,d/n)}\|_2\approx \sqrt{\frac{\log n}{\log\log n}}$.
Liu, Mohanty, Rajaraman and Wu in \cite{KuiKuiSpinGlass2024} propose a natural scheme that allows us 
to circumvent this problem.  Even though they consider the Viana-Bray model, the scheme also applies to 
the standard Ising model. To make their contributions clearer in this exposition, we stick to the Ising model on $\G(n,d/n)$.

They propose a partition of the set of vertices into two parts, the ``good" and the ``bad" parts. The ``bad" part contains 
the high-degree vertices along with the nearby low-degree vertices. By low-degree here, we mean vertices of degree 
$\leq (1+\epsilon)d$ for a small $\epsilon>0$. The vertices which are not in the bad part belong to the good one.

Let $\Adjacency_{\G}$ be the adjacency matrix of $\G(n,d/n)$. Also, let $\Adjacency_{\rm good}$ be the adjacency 
matrix supported only at vertices in the good part of $\G(n,d/n)$, while let $\varrho_{\rm good}=\|\Adjacency_{\rm good}\|_2$. Then, the set-up is (roughly) as follows:
\begin{enumerate}[(i)]
\item $\InAct=\Adjacency_{\G}+(1+\upzeta) \cdot \varrho_{\rm good} \cdot \Id$, for small $\upzeta>0$, 
\item for each $t\in [0,1]$ set $\CtrlMT=(\Adjacency_{\rm good}+(1+\upzeta)\cdot \varrho_{\rm good}\cdot \Id)^{\frac{1}{2}}$.
\end{enumerate}
The above choice for $\CtrlMT$ implies that,  
at time $t=1$,  the off-diagonal entries of the interaction matrix $\InAct_1$ are only supported by 
vertices at the bad part of $\G(n,d/n)$. %That is, only the good part of $\Adjacency_{\G}$ disappears.

They further show that one can control the quantity $\opnorm{\CtrlMT\cdot \Cov (\InAct_t,\beta, \Field' ) \cdot \CtrlMT}$ 
and subsequently the entropy conservation  by taking $\beta\in O(1/\varrho_{\rm good})$. 
This is a significant improvement as %$\|\Adjacency_{\G}\|_2=\sqrt{\frac{\log n}{\log\log n}}$, while
$\varrho_{\rm good}\approx d$. 

What remains to be done is to obtain a lower bound on the modified log-Sobolev constant for $\mu_1$, which, now, is 
not a product measure. There are off-diagonal entries in the interaction matrix $\InAct_1$ , which are non-zero. These 
non-zero entries correspond to edges between vertices in the bad part of the graph. But still, it is relatively easy to obtain 
the desired bound for the modified log-Sobolev constant of $\mu_1$. This is because the vertices of the bad part of the 
graph induce a simple, structured subgraph of the initial graph. Specifically, the connected components in this subgraph are trees with at most one extra edge.

%What remains to be done is to obtain a lower bound on the modified log-Sobolev constant for $\mu_1$, which, now,  is 
% {\em not} a product measure.  There are off diagonal entries in the interaction matrix $\InAct_1$ 
% which are non-zero. These non-zero entries correspond to edges between vertices in the bad part of the graph. But still, it is 
% relatively easy to obtain the desired bound for the modified log-Sobolev constant of  $\mu_1$.
%%
%This is because the vertices of the bad part of the graph induce  a simple structured subgraph of the initial graph. Specifically, the connected 
%components in this subgraph  are trees with at most one extra edge. 

\subsection*{From Ising to Viana-Bray and Beyond}

The above setup can also be used to obtain  rapid mixing of  Glauber dynamics for the Viana-Bray spin-glass on $\G(n,d/n)$,
for any $\beta<\frac{0.18}{\sqrt{d}}$. 
The only essential difference is that,  for the Viana Bray model,  one needs to consider the good part of the randomly signed 
adjacency matrix $\Adjacency_{\G}$ of $\G(n,d/n)$, i.e., each non-zero entry 
of the matrix is independently and randomly set to $\{\pm 1\}$. For such a matrix,  they show that  $\varrho_{\rm good}\approx 2\sqrt{d}$.

The novelty of \cite{KuiKuiSpinGlass2024} in obtaining this bound lies in showing how we can combine  stochastic 
localisation for the mixing analysis  and the block-partition for dealing with  high-degree vertices. 
As mentioned above, there are  certain  limitations with this approach.
Indicative of these limitations is the 
discrepancy between the rapid mixing  bound on $\beta$ they obtain  and the  natural bound $\beta<0.25/\sqrt{d}$,
which someone expects to get from stochastic localisation. 
Furthermore,  the analysis in  \cite{KuiKuiSpinGlass2024} relies heavily on the fact that the entries of $\InAct$ are bounded numbers, 
i.e., they are random  signs of the entries of $\Adjacency_{\G}$. This is quite limiting as, e.g.,
in the  2-spin model, we should be able to handle  edge couplings which are as huge as $\Theta(\sqrt{\log n})$.

In light of all the above, we revisit the stochastic localisation framework proposed in
\cite{KuiKuiSpinGlass2024}. We introduce  a {\em new} approach that  allows us to obtain improved 
bounds on $\beta$, while at the same time, we are able to handle unbounded interactions.

 \subsection*{Contribution 1: Block Partition}
Our approach does not provide a new  analysis for the stochastic localisation. It rather uses the
one from \cite{KuiKuiSpinGlass2024} as a basis  to introduce a new block partition.
The block partition we propose is induced by the 2-spin model $\mu_{\InAct, \beta,\Field}$ on $\G(n,d/n)$.

In our setting with the Gaussians, we still have a partition of the set of vertices into a ``good" and a ``bad" 
part. Rather than focusing only on the degree of the vertices, we consider the {\em combined} effect of 
the degree and the magnitude of the couplings at the edges.
For every vertex $w\in V_n$, we let 
\begin{align}\nonumber
	\WSA(w) &= \sum\nolimits_{{z \sim w }} (\tanh(\beta \cdot \InAct_{z,w}) )^2\enspace.
\end{align}
We, then, introduce  a weighting scheme for the vertices and the paths of $\G(n,d/n)$. 
Given a parameter $\varepsilon>0$, for any $u\in V_n$, define the $\varepsilon$-\emph{weight} by
\begin{align} \nonumber %\label{def:VertexWeights}
\cappedWSA(u) &= \begin{cases}
 {1-\frac{\varepsilon}{4}} & \text{ if } \WSA(u) \le {1-\frac{\varepsilon}{2},} \\
 d\cdot\WSA(u) & \text{ otherwise}.
 \end{cases}
\end{align}
Furthermore, each path $P=(v_0, \ldots, v_\ell)$ in $\G(n,d/n)$ is assigned weight $\cappedWA(P)$ such that
\begin{equation*}
\cappedWA(P) = \prod\nolimits_{0\leq i \leq \ell} \cappedWA(v_{i})\enspace.
\end{equation*}
We say that  $w$ is an $\varepsilon$-block vertex if every path $P$ that emanates from $w$ is ``light", i.e., it has weight 
$\cappedWA(P) < 1$. Intuitively, an %the heavier a vertex $u$ is, the further apart it is from the 
$\varepsilon$-block vertex is ``far away" from every heavy vertex $w$ in the graph, i.e.,  i.e., having $\cappedWSA(w) \gg 1$.

We utilise the above weighting scheme and introduce an {\em elaborate} block construction for the graph. We call it 
 $\varepsilon$-block partition $\cB$, and it consists of single-vertex and multi-vertex blocks.  The vertices which belong 
 to the multi-vertex block $\cB$  are considered to be the bad part  of the graph, while every single-vertex block belongs to the good part.

We think of each multi-vertex block $B\in \cB$ as containing heavy vertices $w$ deep inside them. There is a sufficiently large ``buffer" of light vertices between a heavy vertex in $B$ and 
the boundary of the block. The inner boundary $\partial_{\rm in}B$ consists of $\varepsilon$-block vertices.

Furthermore, each multi-vertex block $B$ is small and simply structured, i.e., it is a tree with at most one extra 
edge. 

Apart from the somewhat generic description we provide above, let us remark that the multi-vertex blocks have 
certain  {\em secondary} properties that we exploit  in the analysis using stochastic localisation.
E.g., it turns out to be very useful in our analysis that each vertex  $u\in \partial_{\rm in}B$, i.e., the inner boundary 
of the multi-vertex block $B$,   is weakly connected  to the rest of its block. That is, it has only one neighbour $z\in B$, while 
$\InAct_{z,u} \leq 1$. 
Ensuring  these more ``subtle" properties  makes the whole construction of the $\varepsilon$-block partition
quite intricate. 
We refer the interested reader to \Cref{sec:WeightBlockPartition} for the technical specifications of the $\varepsilon$-block partition.
For  the  description of how this block partition is constructed,  see
 \Cref{sec:theorem:HBlockExistence,sec:lemma:DEpsilonBlockLemma}.

Perhaps, the scheme with the $\varepsilon$-weights we introduce here is reminiscent of the one considered in 
\cite{EfthZamoCLT24}. The weight of vertex $w$, there, is given by
\begin{align}\nonumber
\WA(w)=\sum\nolimits_{{z \sim w }}  |\tanh(\beta \cdot \InAct_{z,w})|  \enspace.
\end{align}
The weight $\WA(w)$ seems relevant to the path-coupling analysis, whereas the weight 
$\WSA(w)$ we use here is more relevant to dealing with the quantities that arise with stochastic localisation.

At this point, it is worth mentioning that  a slight modification of our block-partition, i.e., substitute the gaussian couplings with $\{\pm 1\}$,   is sufficient to yield a rapid-mixing bound
$\beta\leq 0.25/\sqrt{d}$ for the Viana-Bray model on $\G(n,d/n)$.  
Our follow-up contributions regarding the use of matrix norms allow us to address the additional complication arising from the unbounded couplings in the 2-spin model.

\subsection*{Contribution 2: Matrix Norms for Entropy Conservation}
Consider  the 2-spin model $\mu_{\InAct, \beta,\Field}$ on $\G(n,d/n)$, 
while suppose that we have the aforementioned $\varepsilon$-block partition $\cB$. 

For the conservation of entropy along  the localisation path,  recall that we need to get a handle on 
$\opnorm{\CtrlMT\cdot \Cov (\InAct_t,\beta, \Field' ) \cdot \CtrlMT}$. 
With our choice of the parameters, this problem reduces to somehow controlling  the effect we have in
the analysis from the bad part of the graph. 

Let us be more specific.  Let set  $H$ consist of the vertices of the multi-vertex blocks in the block-partition $\cB$. 
Recall that the vertices in $H$ are the bad part of the graph. 
Also, let $\partial H\subset H$ be the set of bad vertices which are at the boundary between the
good and the bad parts.  In   \cite{KuiKuiSpinGlass2024}, they show that one controls the 
 entropy conservation 
by  bounding appropriately the quantities
\begin{align}\label{eq:Exposition2Norms}
\opnorm{\Cov_H(\mu_1)} &&\textrm {and} &&
 \opnorm{\Cov_{\partial H}(\mu_1)} 
 \enspace,
\end{align}
where $\Cov_H(\mu_1)$ and $\Cov_{\partial H}(\mu_1)$ denote the covariance matrices of distribution 
$\mu_1$ restricted to the set of vertices $H$ and $\partial H$, respectively. 
Note that $\mu_1$ is the distribution we have at the end of the localisation process.

The challenge in dealing with the quantities in \eqref{eq:Exposition2Norms} arises from the fact 
that the connected components induced by $H$ contain vertices with very large degrees and very 
strong couplings, i.e., both can become unbounded.
In particular, the presence of unbounded couplings precludes the use of techniques from  \cite{KuiKuiSpinGlass2024}. 
In this high-level exposition we focus  on  $\opnorm{\Cov_H(\mu_1)}$;   for $\opnorm{\Cov_{\partial H}(\mu_1)}$, the  approach 
is quite similar.

For  the bound on $\opnorm{\Cov_H(\mu_1)}$, we introduce a matrix norm which, in hindsight, seems to be the natural one to work with.
The use of this norm gives rise to  algebraic quantities that are standard to  bound using standard branching process arguments. 

Let us be more concrete. Note that  $\Cov_H(\mu_1)$ is a block matrix supported on the multi-vertex blocks in $\cB$.
We have
\begin{align}\nonumber
\opnorm{\Cov_H(\mu_1)} =\max\nolimits_{\widehat{B}}\{\opnorm{\Cov_{\widehat{B}}(\mu_1)} \}\enspace,
\end{align}
where the variable $\widehat{B}$ runs over the multi-vertex blocks in $\cB$.
Suppose the maximum is achieved for some block $B\in \cB$. Then, the focus is on bounding $\opnorm{\Cov_{B}(\mu_1)}$. 
To avoid being too technical at this early stage in the discussion, assume that $B$ is a tree rooted at vertex $r$.

Since $\Cov_{B}(\mu_1)$ is a symmetric matrix, for a non-singular, diagonal, $V_n\times V_n$ matrix $\UpD$, we have 
\begin{align}\label{ex:CovBMu12NormD}
\opnorm{\Cov_{B}(\mu_1)} \leq \| \UpD^{-1} \cdot \Cov_{B}(\mu_1) \cdot \UpD \|_{\infty}\enspace.
\end{align}
The matrix $\UpD$  is defined as follows: For an appropriately chosen  parameter $\delta>0$, we specify the weight function 
$\badmormw\in \mathbb{R}_{>0}$ on the set of vertices in $B$. 
For the root $r$, we have $\badmormw(r) = 1$. Then, inductively, for every vertex $u$ 
whose parent $z$ has $\badmormw(z)$ specified, we let
\begin{align}\nonumber
\badmormw(u) &= (1+\delta) \cdot |\tanh(\beta \cdot \InAct_{z,u})| \cdot \badmormw(z) \enspace. 
\end{align}
Then, for each vertex $w\in B$, we set $\UpD(w,w)=\badmormw(w)$, while for each $w\notin B$, we set
$\UpD(w,w)=1$. 

The choice of the norm $\| \UpD^{-1} \cdot \Cov_{B}(\mu_1) \cdot \UpD \|_{\infty}$  allows us  to
get a clean bound on $\opnorm{\Cov_{B}(\mu_1)}$.
Specifically,  calibrating the quantity $\delta$ appropriately, yields
\begin{align}\label{ex:NormDBOund}
\| \UpD^{-1} \cdot \Cov_{B}(\mu_1) \cdot \UpD \|_{\infty} &\leq C(\delta,\beta)\cdot \max\nolimits_{v\in B} 
\sum\nolimits_{P}  \prod\nolimits_{e\in P}\left(\tanh(\beta \cdot \InAct_{e})\right)^2\enspace,
\end{align}
where $P$ varies over all paths in block $B$ that emanate from vertex $v$, while $C(\delta,\beta)>0$ is a fixed quantity
that depends only of $\delta$ and $\beta$.

The summation inside the max, on the  r.h.s. of the inequality above, is a natural quantity to work with in $\G(n,d/n)$. 
%One can bound
Using   standard arguments for branching processes one can show that the r.h.s.  $O( \frac{\log n}{\sqrt{d}})$.  
Then,   the bound on $\opnorm{\Cov_H(\mu_1)} $ is then obtained from \eqref{ex:CovBMu12NormD} and \eqref{ex:NormDBOund}.

\subsection*{The modified log-Sobolev constant for $\mu_1$}

For bounding the modified log-Sobolev constant for $\mu_1$, 
we introduce a new stochastic localisation scheme $\{\upnu_t\}_{t\in [0,1]}$, with the process now starting at $\mu_1$, i.e., $\upnu_0=\mu_1$.
 The aim of this new process is for   $\upnu_1$ to be a  product measure. 

The analysis for this new process is quite different from the previous one.   Overall, there are two new  challenges we have to deal with.

(A) Recall that we need to ensure that the interaction matrix of $\upnu_0=\mu_1$ is  positive definite.  
The shifting of the interaction matrix for $\mu_1$ is more delicate than in the previous case.
Following the standard approach of adding a sufficiently large multiple of the identity matrix to $\InAct_1$  is too crude.   
To this end, we instead rely on the results in \cite{fan2017well} for the generalised {\em Ihara-Bass formula} to accomplish the shifting.

(B) For the entropy conservation, we work directly with $\opnorm{\CtrlMT\cdot \Cov (\InAct_t,\beta, \Field) \cdot \CtrlMT}$. 
At this point, the ideas about the matrix norm we discussed earlier remain very useful. 

We refer the interested reader to \Cref{sec:thm:BoundMLSIonH} for the full details on how we obtain the bound on 
the modified log-Sobolev constant for $\mu_1$.

\spreadpoint

\section{Rapid Mixing using the Modified log-Sobolev Constant}\label{sec:DefsForTmixMLSI}

%\subsection{Modified log-Sobolev Constant}
Let $G=(V,E)$ be a graph and let $\Omega=\{\pm 1\}^V$.
Consider a discrete-time, ergodic \footnote{A Markov chain is \emph{ergodic} when it is aperiodic and irreducible.} 
Markov chain $\{X_t\}$ with state space $\Omega$, transition matrix $P$ and equilibrium distribution~$\mu$. Assume also that $\{X_t\}$ is {\em time-reversible} \footnote{ A Markov chain with transition matrix $P$ and stationary measure 
$\mu$ is time-reversible if it satisfies the {\em detailed balance equation}, i.e., for any $x, y\in\Omega$ we have 
$\mu(x)P(x,y)=P(y,x)\mu(y). $}. 

The {\em mixing time} of this Markov chain is defined by 
\begin{align}\label{def:MixingTime}
 \Tmix &= 
 \max_{\sigma\in \Omega} 
 \min { \left\{ t > 0 \mid \Vert P^{t}(\sigma, \cdot) 
 - 
 \mu \Vert_{\rm TV} \leq {1}/{(2 \mathrm{e})} \right \}}\enspace ,
\end{align}
where, for $t\geq 0$ and $\sigma\in \Omega$, we let $P^{t}(\sigma, \cdot)$ denote the distribution of $X_t$ when the 
initial state of the chain satisfies $X_0=\sigma$.

For function $f:\{\pm 1\}^V\to\mathbb{R}_{\geq 0}$, let $\mu(f)$ be the expected value of $f$ with respect to $\mu$, i.e., 
\begin{align*}
\mu(f) &= \sum\nolimits_{\sigma \in \{\pm 1\}^V}\mu(\sigma)\cdot f(\sigma) \enspace.
\end{align*}
Similarly, we define the entropy of $f$ with respect to $\mu$ by 
\begin{align}
\mathrm{Ent}_{\mu }(f) &= \textstyle 
\mu\left(f\log f\right) -\mu\left(f\right ) \log \mu(f)
\enspace,
\end{align} 
with the convention that $0 \log 0 = 0$. 

%Let $P$ be the transition matrix of a time reversible
%Markov chain $\{X_t\}$ with state space $\{\pm 1\}^V$ and equilibrium distribution~$\mu$. 
% 
For the Markov chain $\{X_t\}$ defined above and for any two functions $f,g:\{\pm 1\}^V\to \mathbb{R}_{\geq 0}$, we define
the associated {\em Dirichlet form} $\Dirc(f, g)$ by
\begin{align}
\Dirc(f, g)&=\sum\nolimits_{\sigma,\tau\in \{\pm 1\}^V}
(f(\sigma)-f(\tau))\cdot (g(\sigma)-g(\tau))
\cdot \mu(\sigma)\cdot P(\sigma,\tau)\enspace.
\end{align}
We say that $\{X_t\}$ satisfies the {\em modified log-Sobolev inequality} with constant $\GDmlogSob>0$, if for any
function $f:\{\pm 1\}^V\to \mathbb{R}_{\ge0}$ we have 
\begin{align}\label{eq:ModLogSobInequality}
 \Dirc(f, \log f)\geq \GDmlogSob \cdot \Ent_\mu(f) \enspace. 
\end{align}
The following well-known result provides an upper bound on $\Tmix$ in terms of the modified log-Sobolev constant~$\GDmlogSob$.

\begin{fact}[\cite{bobkov2006modified}]\label{fact:TmixVsMLSI}
For an ergodic, time-reversible Markov chain with transition matrix 
$P$ and stationary distribution $\mu$ that satisfies \eqref{eq:ModLogSobInequality}, we have
\begin{equation}\label{eq:mix-ET}
 \Tmix \leq (\GDmlogSob)^{-1} \cdot 
 \left(
 \log(\log(\mu^{-1}_{\min}))+\log 2+1 
 \right) \enspace ,
\end{equation} 
where $\mu_{\min}=\min\{\mu(\sigma)\ |\ \sigma\in \{\pm 1\}^V, \ \mu(\sigma)>0\}$.
\end{fact}

In light of all the above, we prove \Cref{thrm:MainResult} by showing the following result.

\begin{theorem}[Bound for the mLSI constant]\label{thm:MainModLogSobBound}
For any $\varepsilon \in (0,1)$,  there exist constants $C>0$, $d_0 = d_0(\varepsilon) > 0 $, such that for any $d \ge d_0$, 
  and any  $\beta \le \beta_c(d)$ the following is true:

Let $\mu$ be the 2-spin model on $\G(n, d/n)$, with inverse temperature $\beta$ and external field $\Field\in \mathbb{R}^{V_n}$. Then, with probability at 
least $1-o(1)$ over the instances of $\mu$, Glauber dynamics satisfies the modified log-Sobolev inequality 
with constant $ \GDmlogSob \ge n^{-\left(1+\frac{C}{\sqrt{d}} \right) }$.
\end{theorem}

Note that the constant $C>0$ in \Cref{thm:MainModLogSobBound} does not depend on $d$. 

\Cref{thrm:MainResult} follows as a corollory from \Cref{thm:MainModLogSobBound} and \Cref{fact:TmixVsMLSI}. The proof of \Cref{thm:MainModLogSobBound} appears in \Cref{sec:thm:MainModLogSobBound}.

\spreadpoint

\newcommand{\exteps}{\upzeta}
\newcommand{\hbpartition}{\cB_{H}}
\newcommand{\boundGammaH}{\blue{1/D}}

\section{Vertex Weights \& Block Partitions of $\G(n,d/n)$}\label{sec:WeightBlockPartition}

Let $\G=\G(n,d/n)$, and let $\{ \gauss_{e} \}_{e\in E(\G)}$ be a collection of independent and identically distributed (i.i.d.) standard Gaussian random variables. We define the $V_n\times V_n$ 
interaction matrix $\InAct$ by
\begin{align}\label{eq:InteractionMatrixAtBlockPartition}
\InAct(u,w) &=\Ind\{u\sim w\}\times \gauss_{\{u,w\}} & \forall u,w\in V_n \enspace.
\end{align}

For $\beta\in \mathbb{R}_{\geq 0}$ and $\Field \in \mathbb{R}^{V_n}$, let the Gibbs distribution 
$\mu_{\InAct, \beta, \Field}$ 
be such that
 \begin{align}\label{eq:DefOfGibbsii}
\mu_{\InAct,\beta, \Field}(\sigma)
&\propto 
\exp\left(\frac{\beta}{2} \cdot \langle \sigma, \InAct \cdot \sigma\rangle + \langle \sigma, \Field \rangle \right), &\quad 
\forall \sigma\in \{\pm 1\}^{V_n} 
\enspace. 
\end{align}
Furthermore, for each $e\in E(\G)$ we define the \emph{edge influence} of $e$ as
\begin{align}\label{eq:DefOFInf}
\Inf_e&= %\frac{\left|\exp\left( \beta \cdot \InAct_e \right)-\exp\left( -\beta \cdot \InAct_e \right) \right|}{\exp\left( \beta \cdot\InAct_e \right)+\exp\left( -\beta \cdot \InAct_e\right) } =
|\tanh(\beta \cdot \InAct_e)|\enspace.
\end{align}
For every vertex $u\in V_n$, we also define the {\em aggregate squared influence} by
\begin{align}
	\WSA(u) &= \sum\nolimits_{{z \sim u }} (\Inf_{\{u,z\}})^2\enspace.
\end{align}

In \Cref{sec:theorem:TailBound4WA} we prove the following tail bound for $\WSA(u)$, establishing it is well
concentrated.

\begin{theorem}\label{theorem:TailBound4WA}
For $\updelta \in (0,1)$, $d>20$, and $0<\beta \le \beta_c(d)$ the following is true:

Consider $\G=\G(n,d/n)$ and the 2-spin model on this graph at inverse temperature $\beta$. For a fixed vertex $u$ in $\G$, we have that
\begin{align}\label{eq:WeightAVertexTailBound}
\Pr\left[\WSA(u) \geq {\Exp[\WSA(u)]+\frac{\updelta}{2}} \right]
\le 
2\exp\left(-\frac{\updelta^2 }{8+2\updelta}\cdot d\right)
\enspace.
\end{align}
\end{theorem}

\subsubsection*{$\varepsilon$-Block Partition:} 
We introduce a weighting-scheme for the vertices of the graph $\G(n,d/n)$ that uses $\WSA(u)$'s.
Specifically, given $\varepsilon>0$, for any $u\in V_n$ define the $\varepsilon$-\emph{weight} by
\begin{align}\label{def:VertexWeights}
\cappedWSA(u) &= \begin{cases}
 {1-\frac{\varepsilon}{4}} & \text{ if } \WSA(u) \le {1-\frac{\varepsilon}{2},} \\
 d\cdot\WSA(u) & \text{ otherwise}.
 \end{cases}
\end{align}
Similarly, we define weights over the paths of the graph. For path $P=(w_0, \ldots, w_\ell)$ in $\G$, 
we define the $\varepsilon$-{\em weight} by 
\begin{align}\label{def:WightOfPath}
\cappedWSA(P) &= \prod\nolimits_{w\in P}\cappedWSA(w)\enspace.
\end{align}
Finally, we introduce the notion of $\varepsilon$-{\em block vertex}. A vertex $u\in V_n$ is called $\varepsilon$-{block vertex} if the following holds
\begin{enumerate}[(a)]
\item every path $P$ of length $\leq \log n$ emanating from $u$, satisfies $\cappedWA(P) < 1$\label{itm:blockA}, 
\item for every neighbour $z$ of $u$, we have $\Inf_{\{u,z\}}\leq d^{-1/10}$,
\item $\sum_{w\sim u} (\InAct_{u,w})^2 \le (1+\varepsilon) d$.
\end{enumerate}

We use the above $\varepsilon$-weights to introduce the notion of the $\varepsilon$-block partition. 

\begin{definition}[$\varepsilon$-Block Partition]\label{def:BlockPartition}
For  $\varepsilon>0$, consider the $\varepsilon$-weights defined in \eqref{def:VertexWeights} and \eqref{def:WightOfPath}. 
The partition $\mathcal{B}=\{B_1, \ldots, B_{N}\}$ of vertex set $V_n$ 
is called $\varepsilon$-\emph{block partition} if for every $B\in \cB$ the following is true: 
 \begin{enumerate}
\item the subgraph of $\G(n,d/n)$ induced by $B$ is a tree with at most one extra edge, \label{itm:BPunicyclic}
\item if $B$ is a multi-vertex block, then we have the following:
\begin{enumerate}
\item every $w\in \partial_{\rm in}B$ is an $\varepsilon$-block vertex and has exactly one neighbour in $B$, 
\label{itm:BPSingleNeighBoundary} \label{itm:BPblockBoundary}
\item if $B$ contains a cycle, this is a short one, i.e., its length is at most $4\frac{\log n}{(\log d)^4}$, \label{itm:BPShortCycle}
\item the distance of the short cycle from the $\partial_{\rm out}B$ is $\geq \mincycdist$.
\label{itm:BuffCond}
\end{enumerate}
\item if $B=\{u\}$ is a single-vertex block, then $u$ is an $\varepsilon$-block vertex.
\end{enumerate}
\end{definition}

For the rest of this work, we view the $\varepsilon$-block partition of $V_n$ as being induced by the Gibbs distribution $\mu_{\InAct, \beta,\Field}$. 

\begin{theorem}\label{theorem:HBlockExistence}
For any $\varepsilon\in (0,1)$, $\Field\in \mathbb{R}^{V_n}$, there exists $d_0=d_0(\varepsilon)$ such that for any $d>d_0$ and $\beta< \beta_c(d)$, the following is true:

Let the Gibbs distribution $\mu_{\InAct,\beta, \Field}$, specified by $\G(n,d/n)$ and interaction matrix $\InAct$ defined as
in \eqref{eq:InteractionMatrixAtBlockPartition}. Then, with probability $1-o(1)$ over the instance of 
$\mu_{\InAct,\beta, \Field}$, the vertex set $V_n$ admits an $\varepsilon$-block partition.
\end{theorem}
The proof of \Cref{theorem:HBlockExistence} appears in \Cref{sec:theorem:HBlockExistence}.

Suppose $\cB$ is an $\varepsilon$-block partition of $V_n$. Let us write $H\subseteq V$ for the set of vertices that belong to multi-vertex blocks of $\cB$, and $\partial H=\cup_{B\in \cB}\{u\in \partial_{\rm in } B\}$. We also write $S= (V\setminus H) \cup \partial H$. Then, we let the $V_n\times V_n$ ``partition matrices" $(\InAct_S, \InAct_H)$ 
be such that
\begin{align}
\InAct_S(u,w)& =\InAct(u,w)\cdot \Ind\{u,w\in S\}& \textrm{and} &&
\InAct_H(u,w)& =\InAct(u,w)\cdot \Ind\{u,w\in H\} \enspace.
\end{align}

It is easy to verify that 
\begin{align}\label{eq:InActVsPartMatr}
\InAct=\InAct_S+ \InAct_H\enspace. 
\end{align}
For the sake of definiteness, when there is no $\varepsilon$-block partition for $V_n$, 
set $S=\emptyset$ and $H=\InAct$. This also implies that $\InAct_S$ is the all-zeros 
$V_n\times V_n$ matrix.

\spreadpoint

\section{Stochastic Localisation and Block Partition}

In the section we show how we combine the {\em stochastic localisation} method and the 
block partition construction we describe in \Cref{sec:WeightBlockPartition} to prove 
\Cref{thm:MainModLogSobBound}.

We start by giving a general overview of the use of stochastic localisation to bound 
the mixing time of Glauber dynamics.

\subsection{Stochastic Localisation}
For the $V_n\times V_n$ interaction matrix $\InAct$, for $\beta>0$ and 
$\Field\in \mathbb{R}^n$, let the Gibbs distribution$\mu_{\InAct, \beta, \Field}$ 
on $\{\pm 1\}^{V_n}$.

Stochastic localisation introduces a stochastic process $\{\mu_t\}_{t\in [0,1]}$ 
on the probability measures on $\mathbb{R}^{V_n}$. This process is parametrised 
with respect to the {\em positive-definite} $V_n\times V_n$ matrices 
$\{\CtrlMT\}_{ t\in [0,1]}$, which are allowed to vary with $t$.

We set $\mu_0=\mu_{\InAct, \beta, \Field}$, while $\mu_t$ satisfies the following 
system of stochastic differential equations: for each $x\in \mathbb{R}^{V_n}$, 
we have
\begin{align}
d \left( \frac{\mu_t}{\mu_{0}} \right) (x) &= 
\frac{\mu_t}{\mu_{0}} (x)\cdot \langle x - \Upb_t, \CtrlMT \cdot d \BMotion_t \rangle 
&\forall \; 0\leq t\leq 1\enspace, 
\end{align}
where $\Upb_t\in [-1,+1]^{V_n}$ is the {\em barycenter} of $\mu_t$, i.e., 
$\Upb_t=\mathbb{E}_{\bsigma\sim \mu_t}[\bsigma]$. Also, $\{ \BMotion_t\}_{0\leq t\leq 1}$ 
is the standard Brownian motion in $\mathbb{R}^{V_n}$.

\begin{fact}[\cite{chen2022localization}]\label{fact:MutIsing}
For every $t\in[0,1]$, almost surely $\mu_t$ attains the form
\begin{align}
\mu_{t}(\sigma) & \propto 
\exp\left(\frac{\beta}{2} \cdot \langle \sigma, \InAct_t \cdot \sigma\rangle + \langle \sigma, \Field_t \rangle \right), 
 &\quad \forall \; \sigma\in \{\pm 1\}^{V_n} 
\enspace,
\end{align}
such that $\InAct_t=\InAct-\int^t_0\CtrlM^2_s ds$ and $\Field_t=\int^t_0[\CtrlM_s \cdot d\BMotion_s+\CtrlM^2_s \cdot \Upb(\mu_s)]ds$.
\end{fact}
We use stochastic localisation to bound the modified log-Sobolev constant 
for Glauber dynamics on $\mu_0$ by exploiting results from \cite{chen2022localization}. 
 
If the process $\{\mu_t\}_{t\in [0,1]}$ is such that for each function 
$f:\mathbb{R}^{V_n}\to \mathbb{R}_{>0}$, we have
\begin{align}\label{eq:ExpEntrM0VsEntrM1}
\frac{\Exp[\mathrm{Ent}_{\mu_1}(f)]}{\mathrm{Ent}_{\mu_0 }(f) }\geq \upxi \enspace,
\end{align}
while, almost surely $\mu_1$ is such that
\begin{align}\label{eq:ExpEntrM0VsEntrM2}
\GDmlogSob(\mu_1) >\uppsi \enspace,
\end{align}
then, \cite[Theorem 47]{chen2022localization} implies that 
\begin{align}
\GDmlogSob(\mu_0) >\uppsi \cdot \upxi\enspace. 
\end{align}

The general idea is to introduce appropriate control matrices $\{\CtrlMT\}_{t\in [0,1]}$ 
that would guarantee that the quantities $\upxi$ and $\uppsi$ in \eqref{eq:ExpEntrM0VsEntrM1} 
and \eqref{eq:ExpEntrM0VsEntrM2}, respectively, are not too small. A standard case to 
have is that $\mu_1$ is a product measure on $\{\pm 1\}^{V_n}$. For the distributions 
we consider here, we do not expect to have such a simple measure for $\mu_1$.

We bound the parameter $\upxi$ in \eqref{eq:ExpEntrM0VsEntrM1} by means of the
control matrices $\CtrlMT$ and the {\em covariance matrix} of $\mu_t$. Recall that 
the $V_n \times V_n$ covariance matrix of $\mu_t=\mu_{J_t,\beta, \Field_t}$, 
denoted as $\Cov (\InAct_t,\beta,\Field_t)$, is defined by
\begin{align}\label{eq:DefOfCovMatrix}
\Cov ( \InAct_t,\beta, \Field_t) &=
\Exp\left[\bsigma \cdot \bsigma^T \right] -\Exp\left[\bsigma \right]\cdot \Exp\left[\bsigma^T\right]
\enspace, 
\end{align}
where the random variable $\bsigma\in \{\pm 1\}^n$ is distributed according to 
$\mu_{\InAct_t,\beta, \Field}$. Also, recall from \Cref{fact:MutIsing} that 
$\InAct_t=\InAct-\int^t_0\CtrlM^2_s ds$ and 
$\Field_t=\int^t_0[\CtrlM_s \cdot d\BMotion_s+\CtrlM^2_s \cdot \Upb(\mu_s)]ds$.

Suppose that almost surely, for all $t\in [0,1]$ and any 
$\Field'\in \mathbb{R}^n$, there exists $\alpha_t>0$ such that 
\begin{align}\label{eq:CovarianceBoundWithControlMatrix}
\beta \cdot \opnorm{ \CtrlMT \cdot \Cov (\InAct_t,\beta, \Field' ) \cdot \CtrlMT} 
&\leq \upalpha_t\enspace.
\end{align}
Then, for any function $f:\{\pm 1\}^{V_n}\to \mathbb{R}_{> 0}$, we have
\begin{align}\nonumber 
\Exp[\mathrm{Ent}_{\mu_1}(f) ]&\geq 
\exp\left( -\int^1_0\upalpha_t \cdot dt \right) \cdot \mathrm{Ent}_{\mu_0 }(f) \enspace.
\end{align}
The above approach is formally presented in the following results which is proved 
in \cite{KuiKuiSpinGlass2024}. It is worth mentioning that \Cref{thrm:LocBounMLSI} 
relies heavily on results in \cite{chen2022localization}.

\begin{theorem}[\cite{KuiKuiSpinGlass2024}] \label{thrm:LocBounMLSI}
Let $\beta>0$, also let $\UpM_1$ and $\UpM_2$ be matrices $\mathbb{R}^{V_n\times V_n}$ 
with $\UpM_1$ being positive definite and suppose that there are $\Upq, \updelta>0$
such that for every external field $\Field'\in \mathbb{R}^{V_n}$ and $t \in [0,1]$ we have
\begin{enumerate}
\item $\beta \cdot 
\opnorm{ \UpM^{1/2}_1 \cdot \Cov((1-t)\UpM_1+\UpM_2, \beta, \Field') \cdot \UpM^{1/2}_1 } 
\leq \Upq$,
\item Glauber dynamics for $\mu_{\UpM_2,\Field'}$ satisfies 
$\GDmlogSob(\mu_{\UpM_2, \beta, \Field'})\geq \updelta$.
\end{enumerate}
Then, for any external field $\Field\in \mathbb{R}^{V_n}$, we have 
\begin{align}
\GDmlogSob(\mu_{\UpM_1+\UpM_2, \beta, \Field})\geq \updelta \cdot e^{-\Upq}\enspace. 
\end{align}
\end{theorem}

We use a slightly different notation from \cite{KuiKuiSpinGlass2024}. For this reason 
the above theorem may look different than what is stated in the paper, e.g. in the 
first condition we have an extra factor $\beta$. 

\newcommand{\PosInAct}{\overline{\UpJ}}

\subsection{Stochastic Localisation with Block Partition}\label{sec:SLGeneral}

For a $V_n \times V_n $ interaction matrix $\InAct$, for $\beta>0$ and 
$\Field\in \mathbb{R}$, let the Gibbs distribution $\mu_{\InAct, \beta, \Field}$. 
For appropriately chosen $\varepsilon>0$, suppose that $(\InAct, \beta)$ 
admits a $\varepsilon$-block partition with partition matrices $(\InAct_S,\InAct_H)$.
Consider the $V_n\times V_n$ matrix 
\begin{align}\label{def:ShiftedInteraction}
\PosInAct=\InAct_S+\InAct_H+(1+\updelta)\cdot \opnorm{\InAct_S} \cdot \Id_S+
\frac{\upzeta \cdot \sqrt{d}}{\log n}\cdot \Id_{H\setminus \partial H}
\end{align}
where $\updelta>0$ and $\upzeta>0$ are parameters which we specify later in 
the analysis. We think of those as very small numbers. Also, consider the Gibbs 
distribution $\mu_{\PosInAct,\beta,\Field}$. Note that, since 
$\PosInAct-\InAct$ is a diagonal matrix, $\mu_{\PosInAct, \beta,\Field}$ and 
$\mu_{\InAct,\beta,\Field}$ correspond to exactly the same distribution.

We exploit the above set-up to introduce a stochastic localisation process 
$\{\mu_t\}_{t\in [0,1]}$ such that $\mu_0=\mu_{\PosInAct,\beta,\Field}$, 
while for each $t\in [0,1]$, we have 
\begin{align}\label{eq:DefOfControlMatrix}
\CtrlMT &= \left(
\InAct_S+ \left(1+{\textstyle\frac{\updelta}{2}}\right)\cdot \opnorm{\InAct_S} \cdot \Id_S
+{\textstyle \frac{\upzeta}{\log n}} \cdot \Id_{H\setminus \partial H}
\right)^{\frac{1}{2}}\enspace, 
\end{align}
where the quantities $\updelta$ and $\upzeta$ above are the same as those in 
\eqref{def:ShiftedInteraction}. Note that the process is well-defined. It is 
straightforward to check that $\CtrlMT \succ 0$.

We use the above scheme to bound the modified log-Sobolev constant in 
\Cref{thm:MainModLogSobBound} using \Cref{{thrm:LocBounMLSI}}. Specifically, 
we apply \Cref{thrm:LocBounMLSI} for $\UpM_1\approx (1-t)\CtrlM^2_t$ 
and $\UpM_2 \approx \InAct_H$.

 \spreadpoint

\section{The bound on the m-LSI  of the  2-spin model - Proof of \Cref{thm:MainModLogSobBound}}\label{sec:thm:MainModLogSobBound}

Consider $\G=\G(n,d/n)$, and the set i.i.d standard normal random 
variables $\{ \gauss_{e} \}_{e \in E(\G)}$, while we let the $V_n\times V_n$ 
interaction matrix $\InAct $ be such that 
\begin{align}\label{eq:InteractG(np)}
\InAct(u,w) &=\Ind\{\{u,w\} \in E(\G)\}\times \gauss_{\{u,w\}} & \forall u,w\in V_n \enspace.
\end{align}

For $\beta\in \mathbb{R}_{>0}$ and $\Field \in \mathbb{R}^{V_n}$, we let the Gibbs distribution $\mu_{\beta, \Field, \InAct }$ 
be such that
 \begin{align}\label{eq:DefOfGibbProof22}
\mu_{\InAct,\beta, \Field}(\sigma)
&\propto 
\exp\left(\frac{\beta}{2} \cdot \langle \sigma, \InAct \cdot \sigma\rangle + \langle \sigma, \Field \rangle \right), &\quad 
\forall \sigma\in \{\pm 1\}^{V_n} 
\enspace. 
\end{align}

\charis{``Admits an $\varepsilon$-block partition" in the statement of the following theorems}

\begin{theorem}\label{thm:SmallNormAdjac}
For any $\varepsilon \in (0,1)$ there exists $d_0 = d_0(\varepsilon) > 0 $, for any $d \ge d_0$ and $0<\beta \le \beta_c(d)$ 
the following is true:

Consider $\InAct=\InAct(\G(n,d/n), \{ \gauss_{e} \})$ defined as in \eqref{eq:InteractG(np)}.
With probability $1-o(1)$, there is an $\varepsilon$-block partition of $(\InAct, \beta)$ with partition matrices $(\InAct_S,\InAct_H)$ such that 
\begin{align}\label{eq:thm:SmallNormAdjac}
\opnorm{\InAct_S} \le 2 \sqrt{d} (1+o(1)) \enspace.
\end{align}
\end{theorem}

The proof of \Cref{thm:SmallNormAdjac} appears in \Cref{sec:thm:SmallNormAdjac}. 

\begin{theorem} \label{thm:CovPathBound}
For any $\varepsilon \in (0,1)$ satisfying $1-\varepsilon/8\geq \upkappa$, there exist constants
$\upzeta>0$, $d_0 = d_0(\varepsilon) > 0 $ 
and $C=C(\zeta)>0$ such that for any $d \ge d_0$,% for any field $\Field\in \mathbb{R}^{V_n}$ 
and  $0<\beta \le \beta_c(d)$ the following is true:

Consider $\InAct=\InAct(\G(n,d/n), \{ \gauss_{e} \})$ defined as in \eqref{eq:InteractG(np)}. 
With probability $1-o(1)$, there is an $\varepsilon$-block partition of $(\InAct, \beta)$ with partition matrices $(\InAct_S,\InAct_H)$, while for matrices
\begin{align}%\label{def:ShiftedInteraction}
\PosInAct &=\InAct_S+\InAct_H+\left(1+{\textstyle \frac{\varepsilon}{100}}\right)\cdot \opnorm{\InAct_S} \cdot \Id_S+ \frac{\upzeta}{\maxDelta}\cdot \Id_{H\setminus \partial H} \enspace, \nonumber \\
\CtrlMT &= \left(\InAct_S+ \left(1+{\textstyle\frac{\varepsilon}{200}}\right)\cdot \opnorm{\InAct_S} \cdot \Id_S
+{\textstyle \frac{\upzeta}{\sqrt{d}\cdot \maxDelta}} \cdot \Id_{H\setminus \partial H}
\right)^{\frac{1}{2}}\nonumber \enspace, 
\end{align}
and   $\InAct_t=\PosInAct-\int^t_0\CtrlM^2_s\cdot ds$, for  any external field $\Field\in \mathbb{R}^{V_n}$,  we have 
\begin{align}\label{eq:CovPathBound}
 \Cov[\InAct_t, \beta, \Field] & \pdleq C \cdot \maxDelta \cdot \Id
 & \forall t\in [0,1] \enspace,
\end{align}
where $\maxDelta = \maxDeltaVal$.
\end{theorem}
In \Cref{thm:CovPathBound}, the quantification  implies that  $C>0$ is independent of $d$.
We prove \Cref{thm:CovPathBound} in \Cref{sec:thm:CovPathBound}.

\begin{theorem}\label{thm:BoundMLSIonH}
For any $\varepsilon \in (0,1)$ there exists $d_0 = d_0(\varepsilon) > 0 $ such that for any $d>d_0$, 
and $0<\beta \le \beta_c(d)$  
the following is true:

Consider $\InAct=\InAct(\G(n,d/n), \{ \gauss_{e} \})$ defined as in \eqref{eq:InteractG(np)}. 
With probability $1-o(1)$, there is an $\varepsilon$-block partition of $(\InAct, \beta)$ with partition matrices $(\InAct_S,\InAct_H)$ such that 
for any $\Field\in \mathbb{R}^{V_n}$, we have 
\begin{align}\label{eq:MLSIonH}
\GDmlogSob(\mu_{\InAct_H, \beta,\Field} )& \ge \red{n^{-\left(1+\frac{120}{\sqrt{d}} \right)} } \enspace, 
\end{align}
where $\GDmlogSob(\mu_{\InAct_H, \beta,\Field} )$ is the modified log-Sobolev constant for
Glauber dynamics on $\mu_{\InAct_H, \beta,\Field}$.
\end{theorem}
The proof of \Cref{thm:BoundMLSIonH} appears in \Cref{sec:thm:BoundMLSIonH}. 

\begin{proof}[Proof of \Cref{thm:MainModLogSobBound}]
Consider the interaction matrix $\InAct=\InAct(\G(n,d/n), \{ \gauss_e\})$ defined as in \eqref{eq:InteractG(np)}. 
Let $\cE$ be the event that $(\InAct, \beta)$ admits an $\varepsilon$-partition with partition matrices $(\InAct_S, \InAct_H)$ such that 
for any $t\in [0,1]$, for $0<\beta\leq \beta_c$, for  constants $\widehat{C}>0$, \red{independent of $d$},  and $\maxDelta$ as specified in \Cref{thm:CovPathBound} and for 
any $\Field\in \mathbb{R}^{V_n}$,  we have 
\begin{enumerate}
\item $\opnorm{\InAct_S} \le 2 \sqrt{d} (1+o(1))$,
\item $\Cov[ \InAct_t, \beta, \Field] \pdleq \widehat{C}\cdot \maxDelta 
\cdot \Id$, 
\item $\GDmlogSob(\mu_{\InAct_H, \beta,\Field}) \ge \red{n^{-\left(1+\frac{120}{\sqrt{d}} \right)} }$.
\end{enumerate}
where matrix $\InAct_t$ is defined with respect to $\InAct_S, \InAct_H$, as specified
in the statement of \Cref{thm:CovPathBound}. 
From \Cref{thm:SmallNormAdjac,thm:CovPathBound,thm:BoundMLSIonH} we have  $\Pr[\cE]=1-o(1)$. 

\Cref{thm:MainModLogSobBound} follows by showing that on the event $\cE$, 
Glauber dynamics on $\mu_{\InAct,\beta, \Field}$ satisfies the modified log-Sobolev inequality with constant 
$ \GDmlogSob$ such that 
\begin{align}\label{eq:Target4thm:MainModLogSobBound}
 \GDmlogSob \ge \red{n^{-\left(1+\frac{120+\widehat{C}}{\sqrt{d}} \right)} }\enspace, 
\end{align}
where $\widehat{C}>0$ is the constant from the definition of event $\cE$.
Specifically, we have $C=120+\widehat{C}$. 

In light of \Cref{thrm:LocBounMLSI},  \eqref{eq:Target4thm:MainModLogSobBound} follows by showing that,  on the event $\cE$, we have
\begin{align}
\beta \cdot \opnorm{ \CtrlMT \cdot \Cov[ \InAct_j, \beta, \Field ] \cdot \CtrlMT } & \leq \widehat{C}\cdot \maxDelta = {\textstyle \widehat{C}\cdot \frac{\log n}{\sqrt{d}} }\enspace, \label{eq:NormBound4thm:MainModLogSobBound}\\
\GDmlogSob(\mu_{\InAct_H, \beta,\Field} ) &\ge  \red{n^{-\left(1+\frac{120}{\sqrt{d}} \right)} }
\label{eq:mLSIBound4thm:MainModLogSobBound}\enspace, 
\end{align}
where matrix $\CtrlMT$ is defined in the statement of \Cref{thm:CovPathBound}.

Event $\cE$ implies that \eqref{eq:mLSIBound4thm:MainModLogSobBound} is true, i.e., due to item 3. 
It remain to show that  \eqref{eq:NormBound4thm:MainModLogSobBound} is true. 
Standard calculations, e.g., see \Cref{claim:FromPDOrder2NormBound4Cov}, imply 
\begin{align}\label{eq:Base4NormBound4thm:MainModLogSobBound}
\beta\cdot \opnorm{ \CtrlMT \cdot \Cov[\InAct_t,\beta,\Field] \cdot \CtrlMT} &
\leq \beta\cdot \opnorm{ \Cov[\InAct_t,\beta,\Field] } \cdot \opnorm{ \CtrlM^2_t}\enspace. 
\end{align}
The definition of matrix $\CtrlMT$ in \Cref{thm:CovPathBound}, item 1, and 
our assumption that $\beta\leq \beta_c$, imply 
\begin{align}\label{eq:CtNormBound4mLSBTheorem}
\beta\cdot \opnorm{\CtrlM^2_t} \leq 1 \enspace. 
\end{align}
Item 2 in the definition of the event $\cE$ above, implies
\begin{align}\label{eq:CovNormBound4mLSBTheorem}
\opnorm{\Cov[\InAct_t,\beta,\Field]} & \leq \widehat{C}\cdot \maxDelta \enspace. 
\end{align}
Plugging \eqref{eq:CovNormBound4mLSBTheorem} and 
\eqref{eq:CtNormBound4mLSBTheorem} into \eqref{eq:Base4NormBound4thm:MainModLogSobBound}
we obtain \eqref{eq:NormBound4thm:MainModLogSobBound}.

All the above conclude the proof of \Cref{thm:MainModLogSobBound}. 
\end{proof}

\spreadpoint

\newcommand{\infmatrix}{\mathcal{I}}
\newcommand{\infweight}{\upphi}
\newcommand{\dlogtrecur}{h}
\newcommand{\gratio}{{R}}
\newcommand{\APTree}{\blue{\mathcal{T}_{\rm AP}}}

\newcommand{\Tbase}{T_{\rm base}}
\newcommand{\Tright}{T_{\rm high}}
\newcommand{\Tleft}{T_{\rm low}}
\newcommand{\Cleft}{u_{\rm low}}
\newcommand{\Cright}{u_{\rm high}}

\section{Preliminaries for the analysis}

\subsection{Tree of self-avoiding walks construction}
Consider a (fixed) graph $G=(V,E)$ and the Gibbs distribution $\mu:\{\pm 1\}^V\to [0,1]$.

For $\Lambda\subset V$ and a configuration $\tau$ at $\Lambda$, the {\em pairwise 
influence matrix} $\infmatrix^{\Lambda,\tau}_{G}$ is indexed by the vertices in 
$V\setminus \Lambda $ such that, for $ u, w\in V\setminus \Lambda$, we have
\begin{align}\label{def:InfluenceMatrix}
\infmatrix^{\Lambda,\tau}_{G}( w, u) &
= \mu_{ u}(+1 \ |\ (\Lambda, \tau), ( w, +1))- \mu_{ u}(+1 \ |\ (\Lambda, \tau), ( w, -1)) \enspace, 
\end{align} 
where $\mu_{ u}(+1 \ |\ (\Lambda, \tau), ( w, +1))$ is the probability of the event that vertex 
$u$ has configuration $+1$, conditional on the configuration at set $\Lambda$ is $\tau$ and
the configuration at $ w$ is $+1$. We define $\mu_{ u}(+1 \ |\ (\Lambda, \tau), ( w, -1))$, analogously.

When set $\Lambda=\emptyset$, then we denote the influence matrix by $\infmatrix_{G}$.

\begin{lemma}\label{lemma:CovMVsInflM}
For the $V\times V$ interaction matrix $\InAct$ and $\beta>0$, let the Gibbs distribution 
$\mu_{\InAct, \beta, 0}$, i.e. the external field $\Field=0$.
We have $\Cov(\UpJ,\beta,0) = 4^{-1} \cdot \infmatrix_G$. 
\end{lemma}

\begin{proof}
We abbreviate $\Cov(\UpJ,\beta,0)$ to $\Cov$.

The definitions of $\Cov$ and $\infmatrix_G$ imply that for any $u,v\in V$ we have
$\infmatrix_G( u, v) = \frac{\Cov( u, v)}{\Cov( u, u)}$. Then, \Cref{lemma:CovMVsInflM} 
follows from the observation that $\Cov( u, u)=4^{-1}$, for all $u\in V$. 
\end{proof}

\subsubsection*{The $T_{\rm SAW}$ construction:}
Let $G=(V,E)$ be a fixed graph and let the Gibbs distribution $\mu_G:\{\pm 1\}^V\to [0,1]$. 
Assume w.l.o.g. that there is a {\em total ordering} of the vertices in $G$.

For each vertex $w\in V$, we define $T=\Tsaw(G,w)$, the {\em tree of self-avoiding walks} 
starting from $w$, as follows:
Consider the set consisting of every walk $ v_0, \ldots, v_{ r}$ in graph $G$ that emanates 
from vertex $ w$, i.e., $ v_0= w$, while one of the following two holds
\begin{enumerate}
\item $ v_0, \ldots, v_{ r}$ is a self-avoiding walk,
\item $ v_0, \ldots, v_{ r-1}$ is a self-avoiding walk, while there is $j\leq r-3$ such that $ v_{ r}= v_{j}$.
\end{enumerate}
Each of these walks corresponds to a vertex in $T$. Two vertices in $T$ are adjacent 
if the corresponding walks are adjacent, i.e., one walk extends the other by one vertex.

The vertex in $T$ that corresponds to the walk $ v_0, \ldots, v_{ r}$, 
is called ``copy of vertex $ v_{ r}$", i.e., $ v_{ r}$ is the last vertex in the path. 
For each vertex $ v$ in $G$, we let $\cp( v)$ be the set of its copies in $T$. 
Note that vertex $|\cp( v)|\geq 1$.

In tree $T$ we follow the convention to consider as its root the vertex that corresponds to 
the self-avoiding walk of length $\ell=0$. Clearly, the root is a copy of vertex $w$.

For $\Lambda\subset V$ and $\tau\in \{\pm \}^{\Lambda}$, let the influence matrix 
$\infmatrix^{\Lambda,\tau}_{G}$ induced by $\mu^{\Lambda,\tau}_G$. We describe how 
the entry $\infmatrix^{\Lambda,\tau}_{G}( w, u)$ can be expressed using an appropriately 
defined spin-system on $T=\Tsaw(G, w)$.

Let $\mu_T$ be a Gibbs distribution on $T$ which has the same specification as $\mu_G$.
Each $ z \in \cp( x)$ in tree $T$, such that $ x \in \Lambda$, is pinned to configuration 
$\tau( x)$. Furthermore, if vertex $ z$ in $T$ corresponds to walk $ v_0, \ldots, v_{ \ell}$ 
in $G$ such that $ v_{\ell}= v_j$, for $0\leq j \leq \ell-3$, then we set a pinning at vertex $ z$, 
as well. This pinning depends on the total ordering of the vertices. Particularly, we set at $ z$
\begin{enumerate}[(a)]	
\item $-1$ if $ v_{ j+1}> v_{\ell-1}$, 
\item $+1$ otherwise.
\end{enumerate}
Let $\Gamma=\Gamma(G,\Lambda)$ be the set of vertices in $T$ which have been pinned 
in the above construction, while let $\sigma=\sigma(G,\tau)$ be the pinning at $\Gamma$.

Let $\infmatrix^{\Gamma,\sigma}_{T}$ be the influence matrix that is induced by the Gibbs distribution
$\mu^{\Gamma,\sigma}_T$, i.e., in tree $T$. The following results gives a relation between the influence 
matrices $\infmatrix^{\Lambda,\tau}_{G}$ and $\infmatrix^{\Gamma,\sigma}_{T}$.

\begin{proposition}[\cite{OptMCMCIS,VigodaSpectralInd}]\label{prop:Inf2TreeRedaux}
For every $ u, w\in V\setminus \Lambda$ we have 
\begin{align}%\label{eq:InfluenceEntryAsWeightedSum}
\infmatrix^{\Lambda,\tau}_{G}( w, u) &=
\sum\nolimits_{v\in \cp(u)}\infmatrix^{\Gamma,\sigma}_{T}(r,v) \enspace,
\nonumber
\end{align}
where $r$ is the root of $\Tsaw(G, w)$, while recall that $\cp(u)$ is the set of copies of vertex $u$ in this tree. 
\end{proposition}

\subsubsection*{Bounding Influences on trees:}
Consider tree $T=(V,E)$ rooted at vertex $r$. Let the Gibbs distribution $\mu_T$ be specified as 
in \eqref{def:InfluenceMatrix}. For $ K \subseteq V\setminus\{ r\}$ and $\tau\in \{\pm 1\}^{ K}$, let 
the {\em ratio of marginals} 
\begin{align}\label{eq:DefOfR}
\gratio^{ K, \tau}_{ r}=\frac{\mu^{ K,\tau}_{ r}(+1)}{\mu^{ K,\tau}_{ r}(-1)} \enspace.
\end{align}
Recall that $\mu^{ K,\tau}_{ r}(\cdot)$ denotes the marginal of $\mu^{ K,\tau}(\cdot)$ 
at the root $ r$. The above allows for $\gratio^{ K, \tau}_{ r}=\infty$ when 
$\mu^{ K,\tau}_{ r}(-1)=0$.

For vertex $ u\in T$, we let $T_{ u}$ be the subtree of $T$ that includes $ u$ and all its descendants.
We always assume that the root of $T_{ u}$ is vertex $ u$. With a slight abuse of notation, we let 
$\gratio^{ K, \tau}_{ u}$ denote the ratio of marginals at the root for the subtree $T_{ u}$, where the 
Gibbs distribution is with respect to $T_{ u}$, while we impose pinning $\tau( K\cap T_{ u})$.

For any edge $e\in E$, we let the function 
\begin{align}\label{eq:DerivOfLogRatio}
\dlogtrecur_e:[-\infty,+\infty]\to \mathbb{R}&& \textrm{s.t. }&& 
 x \mapsto -\frac{(1-\exp(4 \beta \UpJ_{e}))\cdot \exp( x)}{(\exp( x+2\beta\UpJ_e)+1)(\exp( x)+\exp(2\beta\UpJ_e))} \enspace.
\end{align}

For each edge $ e$ in $T$, we specify weight $\infweight(e)$ as follows: letting $ e=\{ x, z\}$ be 
such that $ x$ is the parent of $ z$ in $T$, we set 
\begin{align}\label{def:OfInfluenceWeights}
\infweight( e)=\left \{ 
\begin{array}{lcl}
0 & \quad&\textrm{if there is a pinning at either $ x$ or $ z$},\\
\textstyle \dlogtrecur_e\left(\log \gratio^{\Gamma, \sigma}_{ z}\right) & & \textrm{otherwise}.
\end{array}
\right .
\end{align}

We have the following result from \cite{VigodaSpectralInd}
\begin{lemma}[\cite{VigodaSpectralInd}]\label{lem:InfToTanh}
Let vertex $u\in T$, $K\subset V\setminus\{v,r\}$ and $\tau\in \{\pm 1\}^K$.
For the influence matrix $\infmatrix^{K,\tau}_{T}$, induced by $\mu^{K,\tau}_T$,
we have 
\begin{align*}
\infmatrix^{K,\tau}_{T}( r, u) &=\prod\nolimits_{e\in P}\infweight( e) \enspace,
\end{align*}
where $P$ is the path from the root $r$ to vertex $u$ in $T$.
\end{lemma}

\subsection{All Paths Tree}\label{sec:AllPathTreeDefinition}
We introduce a tree construction for the multi-vertex blocks of the $\varepsilon$-partition 
$\cB$ that we call the {\em all paths tree}. That is, this tree construction applies only to 
graphs that are trees or unicyclic.

Consider first the unicyclic graph $G=(V,E)$. That is, suppose for some integer $\ell>2$, 
$G$ contains the cycle $C=x_1, \ldots, x_{\ell}$. We assume a total ordering of the vertices 
of $G$.

Consider a vertex $w\in V$, while w.l.o.g. assume that vertex $x_1$ in cycle $C$ is the vertex 
that is closest to $w$, while $\{x_1,z\}$ is the first edge of the shortest path from $x_1$ to $w$. 
Also, assume that $x_2>x_{\ell}$ in the total ordering.

Let $G_1$ be the connected component containing the vertices in cycle $C$ 
after we delete the edges $\{x_1,z\}$ and $\{x_1,x_{\ell}\}$ from $G$.
Similarly, let $G_2$ be the connected component containing the vertices in $C$ after 
deleting the edges $\{x_1,z\}$ and $\{x_1,x_{2}\}$ from $G$.

We define the all paths tree $\APTree(G,w)$ rooted at $w$ as follows:
We let the trees $\Tbase=\Tsaw(G,w)$, $\Tright=\Tsaw(G_1,x_1)$ and $\Tleft=\Tsaw(G_2,x_1)$. 
The tree $\APTree(G,w)$ is constructed by appending $\Tright$ and $\Tleft$ to $\Tbase$. 
Specifically, since $\Tbase$ is a tree of self-avoiding walks that start from $w$, there are exactly 
two leaves in $\Tbase$ which are copies of vertex $x_1$. The first copy, vertex $\Cright$ corresponds 
to the  path from $w$ to $x_1$ and then to $x_2, \ldots x_1$. Similarly, the second copy, vertex 
$\Cleft$ corresponds to the  path from $w$ to $x_1$ and then to $x_{\ell}, x_{\ell-1}, \ldots x_1$. 
We obtain $\APTree(G,w)$ by hanging the two trees $\Tright$ and $\Tleft$, from the leaves $\Cright$ 
and $\Cleft$, of the tree $\Tbase$, respectively.

If graph $G$ is a tree, then we have the above construction with 
$\Tbase=\Tsaw(G,w)$ while $\Tright$ and $\Tleft$ being empty. That is, $\APTree(G,w)$ and $G$ 
are identical.

For what follows, we extend the notion of the copy of a vertex to $\APTree(G,w)$ in the natural way,
i.e., as in the case of self-avoiding walk trees.

Also, when we define the all-path tree $\APTree(G,w)$ we refer to the subtrees $\Tbase$, $\Tright$
and $\Tleft$ in a standard way, i.e., without redefining them. Similarly,  for $\Cleft$ and $\Cright$, 
the copies of vertex $x_1$.

\begin{lemma}[All paths property]\label{lemma:PathPropertyAPTree}
For $G=(V,E)$ which is a tree with at most one extra edge and for $w\in V$ the following 
is true: For every self-avoiding walk $P=z_1, \ldots, z_{\ell}$ in $G$, there exists a self-avoiding 
walk $Q=y_1, \ldots, y_{\ell}$ in $\APTree(G,w)$ such that $y_i$ is a copy of vertex $z_i$, 
for all $i=1,\ldots, \ell$. 
\end{lemma}

\begin{proof}
When $G$ is a tree, \Cref{lemma:PathPropertyAPTree} is true because $\APTree(G,w)$ is identical 
to $G$. We now focus on the case where $G$ is a unicyclic graph, with cycle $C=x_{1}, \ldots, x_{\ell}$ 
as we describe above.

It is not hard to verify that each self-avoiding walk $P=z_1, \ldots, z_{\ell}$ that does {\em not} 
include the subpath $x_{\ell}, x_1, x_2$ or subpath $x_2, x_1, x_{\ell}$, there exists walk 
$Q=y_1, \ldots, y_{\ell}$ in $\APTree(G,w)$ such that $y_i$ is a copy of vertex $z_i$, for all 
$i=1,\ldots, \ell$. In particular, such walk $Q$ appears in tree $\Tbase$. Hence, since 
$\Tbase$ is a subtree of $\APTree(G,w)$, $Q$ is also included in $\APTree(G,w)$.

Consider now a self-avoiding walk $P$ that includes the subpath $x_{\ell}, x_1, x_2$. For 
such $P$ let the subpath $P_1=z_{1}, \ldots, z_{k-1}, z_{k}$ be such that $z_{k-1}=x_{\ell}$ 
and $z_{k}=x_{1}$. Also, let the subpath $P_2=z_{k}, \ldots, z_{\ell}$. We have $z_{k}=x_{1}$ 
and $z_{k+1}=x_2$.

It is not hard to verify that in $\Tbase$ there is walk $Q_1=y_1, \ldots, y_k$ such that $y_i$ 
is a copy of $z_i$, for all $1\leq i\leq k$. Moreover, vertex $y_k$ is the leaf $\Cright$ of 
$\Tbase$, i.e., the leaf from which we hang $\Tright$. In the subtree $\Tright$, there is walk 
$Q_2=y_k, \ldots, y_{\ell}$ such that $y_i$ is a copy of $z_i$, for all $k \leq i\leq \ell$, while 
$y_k$ is the root of $\Tright$. Then, the walk $Q$ in $\APTree(G,w)$ we are looking for 
corresponds to the concatenation of walks $Q_1$ and $Q_2$.

With a very similar argument we treat the case of self-avoiding walk $P=z_1, \ldots, z_{\ell}$ 
that includes the subpath $x_{2}, x_1, x_{\ell}$.

All the above conclude the proof of \Cref{lemma:PathPropertyAPTree}. 
\end{proof}

From the above construction it is not hard to obtain the following corollary

\begin{corollary}\label{cor:NoOfCopiesAllTree}
For $G=(V,E)$ which is a tree with at most one extra edge and for $w\in V$ the following 
is true: Each each vertex $v\in V$ has at most 4 copies in $\APTree(G,w)$. 
\end{corollary}

\charis{Maybe put some figures?}

\subsection{Properties of the random graph $\G(n,d/n)$}

In a few places in the analysis,  we make use of the following, standard to prove,  
lemma about the random graph $\G(n,d/n)$.

\begin{lemma}\label{lem:SamllUnicyclicGnp}
There exist $d_0 \ge 1$ such that for every $d\ge d_0$, the following is true:

Let $\cE_S$ be the event that ``every set of vertices $S$ in $\G(n,d/n)$ with 
$|S| \le 2\frac{\log n}{(\log d)^2}$, spans at most $|S|$ edges." Then, we have 
 $\Pr[\cE_S]\geq 1-n^{-3/4}$.
\end{lemma}

For a proof, see \Cref{sec:lem:SamllUnicyclicGnp}.

\spreadpoint

\section{Proof of \Cref{thm:SmallNormAdjac}}\label{sec:thm:SmallNormAdjac}

For $\varepsilon, \beta,d$ as specified in the statement of \Cref{thm:SmallNormAdjac}, 
consider the $V_n\times V_n$ matrix $\InAct$ specified as in \eqref{eq:InteractG(np)}, 
with respect to the random graph $\G(n,d/n)$. 

We let $\cP$ be the event that $(\InAct, \beta)$ admits an $\varepsilon$-block partition $\cB$.
On the event $\cP$, recall from \Cref{sec:WeightBlockPartition} the partition sets $S$ and $H$, as
well as the partition matrices $(\InAct_S, \InAct_H)$. 

We utilise the weighted non-backtracking spectrum of $\G(n,d/n)$ and a 
generalisation of the Ihara-Bass type formula from \cite{fan2017well} to obtain 
the desired bound on $\|\InAct_S\|_2$. 
A similar approach was followed in \cite{KuiKuiSpinGlass2024}, and it 
utilises results from the work of Stephan \& Massouli{\'e} \cite{stephan2022non}.
A technical challenge we need to circumvent is here
that the typical instances of $\InAct$ have unbounded entries.

Let the set of i.i.d. random variables $\{ \widetilde{\gauss}_{e} : e \in E(\G) \}$, where 
each $\widetilde{\gauss}_e$ is a standard Gaussian truncated on the interval 
$[-\varepsilon\sqrt{d}, \varepsilon\sqrt{d}]$. Then, matrix $\bBdInAct$ is such that 
\begin{align}\label{eq:ImMatDef}
\bBdInAct(u,w)&= \Ind\left\{|\InAct(u,w)|>\varepsilon \sqrt{d} \right\}\cdot \widetilde{\gauss}_{\{u,w\}} +
\Ind\left\{|\InAct(u,w)|\leq \varepsilon \sqrt{d} \right \}\cdot \InAct(u,w) & \forall u,w\in V_n \enspace.
%\bBdInAct(u,w)&= \Ind\{u\sim w\}\times \widetilde{\gauss}_{u,w} & \forall u,w\in V(\G) \enspace.
\end{align}
From the above definition, it is elementary to verify that for each $\{ u,w\}\in E(\G)$, the entry 
$\bBdInAct(u,w)$ is distributed as in $\cN(0,1)$, conditional on $|\bBdInAct(u,w)| < \varepsilon \sqrt{d}$.

On the event $\cP$, we define $\bBdInAct_S$, i.e., similarly to the partition matrix $\InAct_S$.
That is, 
\begin{align*}
\bBdInAct_S(u,w) &=\Ind\{u,w\in S\} \cdot \bBdInAct(u,w) & \forall u,w\in V_n\enspace. 
\end{align*}
Recall that if the event $\cP$ does not hold, then the partition set $S=\emptyset$
and hence $\bBdInAct_S$ is the all-zeros matrix.

The definition of $\bBdInAct$ implies that for edges $\{u,w\}\in E(\G)$ such that $|\InAct(u,w)|\leq \varepsilon\sqrt{d}$,
we also have $\InAct(u,w)=\bBdInAct(u,w)$. Furthermore, the definition of the $\varepsilon$-block partition, implies 
that all the edges $e=\{u,w\}$ with $u,w\in S$ satisfy $|\InAct(u,w)|\leq \varepsilon\sqrt{d}$. We have
 \begin{align}\label{eq:InActSVsBInActS}
\Pr[\InAct_S&=\bBdInAct_{S}]=1\enspace. 
\end{align}

\begin{proposition}\label{thm:BoundOnImMatrixNorm}
For $d>1$, $\beta>0$ and $\varepsilon \in (0,1)$ as in the statement of \Cref{thm:SmallNormAdjac},
we have 
\begin{align}\label{eq:thm:BoundOnImMatrixNorm}
\Pr[\Ind\{\cP\}\cdot \opnorm{\bBdInAct_\singleBlockV} \le 2 \sqrt{d} \cdot (1+o(1))]=1-o(1)
\enspace.
\end{align}
\end{proposition}

\begin{proof}[Proof of \Cref{thm:SmallNormAdjac}]
We have
\begin{align}
\Pr[ \opnorm{\InAct_\singleBlockV} \le 2 \sqrt{d} \cdot (1+o(1))\ |\ \cP] &=
\Pr[\Ind\{\cP\}\cdot \opnorm{\InAct_\singleBlockV} \le 2 \sqrt{d} \cdot (1+o(1)) \ |\ \cP] \nonumber\\
&=\frac{\Pr[\Ind\{\cP\}\cdot \opnorm{\InAct_\singleBlockV} \le 2 \sqrt{d} \cdot (1+o(1)) ]}
{\Pr[\cP]} =1-o(1)\enspace. \nonumber 
\end{align}
For the last equality we bound from above the nominator by using \eqref{eq:InActSVsBInActS} and \Cref{thm:BoundOnImMatrixNorm}
while we bound from below the denominator using \Cref{theorem:HBlockExistence}. 
\end{proof}

\subsection{Proof of \Cref{thm:BoundOnImMatrixNorm}}\label{sec:thm:BoundOnImMatrixNorm}

In what follows, %for any $n\times n$ matrix $\UpW$, 
we let $\spradius{W}$ denote the {\em spectral radius} of matrix $W$, i.e.,
the maximum in magnitude eigenvalue. 

%\subsection{Weighted non-backtracking matrix}
Furthermore, for each {\em symmetric} matrix $W\in \mathbb{R}^{n\times n}$, we 
let $\NB_W \in \mathbb{R}^{n(n-1)\times n(n-1)}$ be such that 
\begin{align}\label{eq:nonBaclDef}
\NB_W\left( (i,j), (k,\ell)\right) = \begin{cases}
 W_{k,\ell} &\text{ if } i\neq \ell \text{ and } j=k,\\
 0 &\text{otherwise.}
 \end{cases}
\end{align}
We call $\NB_W$ the {\em non-backtracking matrix induced by} $W$.

\begin{proposition}\label{thm:BoundW}
For any $\delta\in (0,1)$, for any $L>0$, for any 
symmetric $\UpQ\in \mathbb{R}^{n \times n}$, such that $\max_{i,j}|\UpQ(i,j)| \le L$ the following is true:

For $\lambda = \max\{\spradius{\NB_\UpQ},L\}$, where $\NB_\UpQ$ is
defined as in \eqref{eq:nonBaclDef}, 
% and $\widehat{\UpQ}$ any principal submatrix of $\UpQ$, 
we have
\begin{align}\label{eq:BoundW}
\opnorm{\UpQ} 
<
\frac{\lambda}{1-\delta}+ 
\frac{1-\delta}{\lambda-(1-\delta)L}
\cdot 
\left\|\UpQ^{\circ 2}\right\|_{\infty}
\enspace,
\end{align}
where $\UpQ^{\circ 2} $ denotes the $2$nd Hadamard power of $\UpQ$.

Furthermore, \eqref{eq:BoundW} holds for $\widehat{\UpQ}$, any principal submatrix of $\UpQ$. 
\end{proposition}

\begin{lemma}\label{lem:BoundRadiusBW}
For matrix $\bBdInAct$ defined in \eqref{eq:ImMatDef} and $\NB_{\bBdInAct}$
defined as in \eqref{eq:nonBaclDef}, we have
% With probability $1 - o(1)$, we have 
\begin{align} \nonumber %\label{eq:lem:BoundRadiusBW}
\Pr[\spradius{\NB_{\bBdInAct}} \le (1+o(1)) \cdot \sqrt{d}]=1-o(1)\enspace.
\end{align}
\end{lemma}
The above lemma is a consequence of \cite[Theorem 2]{stephan2022non}

\begin{proof}[Proof of \Cref{thm:BoundOnImMatrixNorm}]
Let $\cE_1$ be the event that $\spradius{\NB_{\bBdInAct}} \le (1+o(1)) \cdot \sqrt{d}$ is true.

On the event $\cE_1$,
we apply \Cref{thm:BoundW}, with $\UpQ=\bBdInAct_{\singleBlockV}$, $\lambda=\sqrt{d}\cdot(1+o(1))$, and $L=\varepsilon \sqrt{d}$, yielding
\begin{align*}
\Ind\{\cP\}\cdot \opnorm{\bBdInAct_{\singleBlockV}} \leq \opnorm{\bBdInAct_{\singleBlockV}} 
&<
\frac{\lambda}{1-\delta}+ 
\frac{1-\delta}{\lambda-(1-\delta)\varepsilon\sqrt{d}}
\cdot 
\left\|{\bBdInAct_{\singleBlockV}}^{\circ 2}\right\|_{\infty} \\
&\le (1+o(1))\sqrt{d}\cdot \left(\frac{1}{1-\delta}+\frac{(1-\delta)(1-\varepsilon)}{1-(1-\delta)\varepsilon}\right) \le 2\sqrt{d}\cdot (1+o(1))
\enspace.
\end{align*}
In the second inequality we use the fact that $\left\|{\bBdInAct_{\singleBlockV}}^{\circ 2}\right\|_{\infty} \le (1+\varepsilon) d$.

From the all above, we conclude that on the event $\cE_1$ we have $\Ind\{\cP\}\cdot \opnorm{\bBdInAct_{\singleBlockV}} \le 2\sqrt{d} \cdot (1+o(1))$. 
 \Cref{thm:BoundOnImMatrixNorm} follows from \Cref{lem:BoundRadiusBW} which implies that $\Pr[\cE_1]=1-o(1)$.

\end{proof}

\subsection{Proof of \Cref{thm:BoundW}}\label{sec:thm:BoundW}
Let us first introduce a few notions. 
For a symmetric matrix $\UpQ \in \mathbb{R}^{n \times n}$ and
 $t \in \mathbb{R}$, we define the $n \times n$ matrices $\MontA_t=\MontA(\UpQ,t)$ and 
 $\MontD_{t}=\MontD(\UpQ,t)$ by
\begin{align}\label{eq:MontAdjDegDef}
\MontA_t(i,j) = \frac{t\UpQ(i,j)}{1-\left(t\UpQ(i,j)\right)^2}, 
&& \text{ and } &&
\MontD_t(i,j) = \Ind\{i=j\} \cdot 
\sum_{k=1}^{n}\frac{\left(t\UpQ(i,k)\right)^2}{1-\left(t\UpQ(i,k)\right)^2}
\enspace.
\end{align}

Note in particular that $\MontD_t$ is a diagonal matrix. The following result is proven in \cite{fan2017well}.

\begin{theorem}[Generalised Ihara-Bass formula \cite{fan2017well}] \label{thm:GenIharaBass}
Let $\UpQ \in \mathbb{R}^{n \times n}$ be symmetric, and let $\NB_{\UpQ}$, and $\MontA_t, \MontD_t$ be defined as in \eqref{eq:nonBaclDef}, and \eqref{eq:MontAdjDegDef}, respectively. Assume that $t$ is such that 
$|t\UpQ(i,j)| \neq 1$ for all $i, j$. Then, we have 
\begin{align}\nonumber %\label{eq:IharaBassMont}
\det\left(\Id - t\NB_{\UpQ}\right) = \det(\Id + \MontD_t - \MontA_t ) \prod\nolimits_{1\le i \le j\le n} \left(1 - \left(t\UpQ(i,j)\right)^2\right)
\enspace.
\end{align}
\end{theorem}

The above theorem, yields the following inequality.

\begin{lemma}\label{cor:PDBetheHessian}
For any symmetric $\UpQ \in \mathbb{R}^{n \times n}$, let $\lambda = \max\{ \spradius{\NB_{\UpQ}}, \max_{i,j} {|\UpQ(i,j)|}\}$, where $\NB_{\UpQ}$ is defined in \eqref{eq:nonBaclDef}. Then for every $t \in \left(-\lambda^{-1}, \lambda^{-1} \right)$, we have
\begin{align}\label{eq:PDBHessian}
\Id + \MontD_t -\MontA_t \pdge 0 \enspace,
\end{align}
where $\MontA_t=\MontA(\UpQ,t), \MontD_t=\MontD(\UpQ,t)$ are defined in \eqref{eq:MontAdjDegDef}.
\end{lemma}

\begin{proof}[Proof of \Cref{cor:PDBetheHessian}]
From \Cref{thm:GenIharaBass}, we have 
\begin{align}\label{eq:Base4cor:PDBetheHessian}
\det\left(\Id - t\NB_{\UpQ}\right) = \det(\Id + \MontD_t - \MontA_t ) \prod\nolimits_{1\le i \le j\le n} \left(1 - \left(t\UpQ(i,j)\right)^2\right)
\enspace.
\end{align}
Choosing $|t| < \lambda^{-1}$ has two consequences. Firstly, it implies that $|t\UpQ(i,j)| < 1$ for all $i, j$, and hence, 
the product in the r.h.s. of \eqref{eq:Base4cor:PDBetheHessian} is strictly positive. Secondly, it implies that 
$|t\cdot \spradius{\NB_\UpQ}| < 1$, yielding $\det\left(\Id - t\NB_\UpQ\right) > 0$. Therefore, have that 
$\det(\Id + \MontD_t - \MontA_t ) > 0$, for all $t \in \left(-\lambda^{-1}, \lambda^{-1}\right)$.

When $t = 0$ we observe that $\Id + \MontD_0 -\MontA_0 = \Id \pdge 0$. Since the eigenvalues of 
$\Id + \MontD_t -\MontA_t$ vary continuously with respect to $t$, we conclude that none of them changes sign 
in $\left(-\lambda^{-1}, \lambda^{-1}\right)$, i.e., $\Id + \MontD_t -\MontA_t \pdge 0$, as desired.
\end{proof}

In light of \Cref{cor:PDBetheHessian}, we establish \Cref{thm:BoundW} as follows. 

\begin{proof}[Proof of \Cref{thm:BoundW}]
The arguments that prove \eqref{eq:BoundW} for a principal submatrix of
$\UpQ$, are similar to those we use for $\UpQ$. They rely on the observation
that for any $t\in (-\lambda^{-1}, \lambda^{-1})$ it is not only that 
matrix $\UpM=\Id + \MontD_t -\MontA_t\pdge 0$, i.e., as shown in 
\Cref{cor:PDBetheHessian}, but also any principal submatrix of $\UpM$ is positive definite. 
To avoid the tedious notations, in this proof, we focus on the case of $\UpQ$. 

For $t \in \mathbb{R}$, consider the matrices $\MontA_t=\MontA(\UpQ,t)$ and 
$\MontD_{t}=\MontD(\UpQ,t)$ as defined in \eqref{eq:MontAdjDegDef}. 
Also, let us define the matrix $\MontC_{t} = \MontC(\UpQ, t) \in \mathbb{R}^{n\times n}$ with
\begin{align} \nonumber
\MontC_{t}(i,j) = \frac{\left(t\UpQ(i,j)\right)^3}{1-\left(t\UpQ(i,j)\right)^2}
\enspace.
\end{align}
Note that $\MontA_t = t\UpQ + \MontC_t$. Let $\theta = \frac{1-\delta}{\lambda}$. Applying 
\Cref{cor:PDBetheHessian} on $\UpQ$ with $t=\theta$, and rearranging gives
\begin{align*}
\UpQ \pdle \theta^{-1}\cdot \left(\Id +\MontD_{\theta} -\MontC_{\theta}\right)
\enspace,
\end{align*}
so that 
\begin{align}\label{eq:lmaxBoundW}
\lambda_{\max}(\UpQ) <
\theta^{-1}\cdot \lambda_{\max}\left(\Id +\MontD_{\theta} -\MontC_{\theta}\right) 
\le
\frac{1}{{\theta}} \cdot \left\|\Id +\MontD_{\theta} -\MontC_{\theta}\right\|_{\infty} \enspace.
\end{align}
Similarly, applying \eqref{eq:PDBHessian} with $t = -{\theta}$, and rearranging gives
\begin{align*}
\UpQ \pdge- \theta^{-1}\cdot \left(\Id +\MontD_{-\theta} -\MontC_{-\theta}\right)
\enspace,
\end{align*}
so that
\begin{align}\label{eq:lminBoundW}
\lambda_{\min}(\UpQ) >
-\theta^{-1}\cdot \lambda_{\max}\left(\Id +\MontD_{-\theta} -\MontC_{-\theta}\right) 
\ge
-\theta^{-1}\cdot \left\|\Id +\MontD_{-{\theta}} -\MontC_{-{\theta}}\right\|_{\infty} \enspace.
\end{align}
Noticing that $\MontD_t = \MontD_{-t}$, and $\MontC_t = -\MontC_{-t}$, for every $t$, and using the triangle inequality, we observe 
\begin{align}\nonumber 
\max\{\left\|\Id +\MontD_{\theta} -\MontC_{\theta}\right\|_{\infty}, \left\|\Id +\MontD_{-\theta} -\MontC_{-\theta}\right\|_{\infty}\}
\le 
\left\|\,\Id + \left|\MontD_{\theta}\right|+
\left|\MontC_{\theta}\right|\,\right\|_{\infty}
\enspace,
\end{align}
where for matrix $\UpD $, we let $|\UpD|$ denote the matrix having entries $|\UpD_{i,j}|$.
 Eq. \eqref{eq:lmaxBoundW} and \eqref{eq:lminBoundW} imply
\begin{align}\label{eq:rhoToInftyNorm}
\|\UpQ\|_2=\spradius{\UpQ} & < \theta^{-1}\cdot
\left\|\,\Id + \left|\MontD_{\theta}\right|+
\left|\MontC_{\theta}\right|\,\right\|_{\infty}
\enspace.
\end{align}
The first equality follows from the fact that $\UpQ$ is symmetric. 
The definitions of $\MontD_{\theta}$, and $\MontC_{\theta}$, imply
\begin{align} %\label{eq:InftyNormExp}
\left\|\,\Id + \left|\MontD_{\theta}\right|+
\left|\MontC_{\theta}\right|\,\right\|_{\infty}
&=
\max_{i \in \{1, \ldots, n\}} \left\{
1+ 
\left|
\sum_{k=1}^{n}
\frac{\left({\theta}\UpQ(i,k)\right)^2}
 {1-\left({\theta}\UpQ(i,k)\right)^2}
\right|
+
\sum_{q=1}^{n}\left|
\frac{\left({\theta}\UpQ(i,q)\right)^3}
 {1-\left({\theta}\UpQ(i,q)\right)^2}
 \right|
\right\} \nonumber \\
&= 1 + \max_{i \in \{1, \ldots, n\}}\left\{
\sum_{k=1}^{n}
\frac{\left({\theta}\UpQ(i,k)\right)^2 + \left|{\theta}\UpQ(i,k)\right|^3}
 {1-\left({\theta}\UpQ(i,k)\right)^2}
\right\} \nonumber \\
&=
1+ 
\max_{i \in \{1, \ldots, n\}} 
\left\{
\sum_{k=1}^{n}
\frac{\left({\theta}\UpQ(i,k)\right)^2}
 {1-\left|{\theta}\UpQ(i,k)\right|} 
\right\} \nonumber 
\enspace.
\end{align}
Substituting $\theta=\frac{1-\delta}{\lambda}$ into the r.h.s. above, and recalling that $\max_{i,j}|\UpQ(i,j)|\le L$, we further get
\begin{align}
\nonumber
\frac{1}{\theta}\cdot
\left\|\,\Id + \left|\MontD_{\theta}\right|+
\left|\MontC_{\theta}\right|\,\right\|_{\infty}
&\le
\frac{\lambda}{1-\delta}
\cdot
\left( 1
+ \frac{(1-\delta)^2}{\lambda^2 - \lambda (1-\delta)L} \cdot 
\max_{i \in \{1, \ldots, n\}} 
\left\{
\sum_{k=1}^{n}
{\left(\UpQ(i,k)\right)^2} 
\right\}\right)\\ \label{eq:InfNormDL}
&=
\frac{\lambda}{1-\delta}+ \frac{1-\delta}{\lambda - (1-\delta)L} \cdot 
\left\|\UpQ^{\circ 2}\right\|_{\infty}
\enspace.
\end{align}
Combining \eqref{eq:rhoToInftyNorm} and \eqref{eq:InfNormDL}, we get
\eqref{eq:BoundW}. 
\Cref{thm:BoundW} follows.
\end{proof}

\spreadpoint

\section{Entropy Decay - Proof of \Cref{thm:CovPathBound}}
\label{sec:thm:CovPathBound}

Let $\varepsilon>0$ be as specified in the statement of \Cref{thm:CovPathBound}. 
As per standard notation, we consider the interaction matrix $\InAct=\InAct(\G(n,d/n), \{ \gauss_e \})$ defined 
as in \eqref{eq:InteractG(np)}. For $(\UpJ, \beta)$ that admits an $\varepsilon$-block partition $\cB$
with partition matrices $(\InAct_S, \InAct_H)$, we define set 
\begin{align}\label{eq:SetParialInH}
\partial H=\bigcup\nolimits_{B\in \cB: |B|\neq 1}\partial_{\rm in } B\enspace.
\end{align}
We also let $\Cov_{\partial H}(\InAct_{H}, \beta, \Field)$ be the restriction of 
$\Cov(\InAct_{H}, \beta, \Field)$ to set ${\partial H}$, i.e., we have
\begin{align*}
[\Cov_{\partial H}(\InAct_{H}, \beta, \Field)] (u,v) & = [\Cov(\InAct_{H}, \beta, \Field)] (u,v) \cdot 
\Ind\left\{u,v \in \partial H\right\} & \forall u,v\in V_n
\enspace.
\end{align*}

\begin{theorem}\label{thm:CovBoundForH}
For $\varepsilon \in (0,1)$ satisfying $1-\varepsilon/8\geq \upkappa$, there is $d_0 = d_0(\varepsilon)>0$
 such that for a sufficiently large constant $c>0$, for any $d \ge d_0$, $\beta \le \beta_c(d)$ 
 the following is true:

With probability $1-o(1)$ over the instances of matrix $\InAct = \InAct(\G(n,d/n), \{\gauss_e\})$ defined as in \eqref{eq:InteractG(np)}, 
$(\InAct,\beta)$ admits an $\varepsilon$-block partition $\cB$ with partition matrices $(\UpJ_S, \UpJ_H)$ such that 
for  any $\Field\in \mathbb{R}^{V_n}$  we have
\begin{align}
\opnorm{\Cov(\InAct_H, \beta, \Field)} &\leq c \cdot\maxDelta \enspace, \label{eq:thm:CovBoundForHA}\\
 \opnorm{\Cov_{\partial H}(\InAct_H, \beta, \Field)} &\leq \frac{c}{\varepsilon\cdot d^{1/10}}
 \enspace, \label{eq:thm:CovBoundForHB} 
\end{align}
where $\maxDelta=\maxDeltaVal$, while $\upkappa$ is from the definition of $\beta_c$ in \eqref{eq:defOfBc}. 
\end{theorem}

Note that the quantification in \Cref{thm:CovBoundForH}
implies that the constant $c>0$ is independent of the expected degree $d$. 
The proof of \Cref{thm:CovBoundForH} appears in \Cref{sec:thm:CovBoundForH}. 

To make use of \Cref{thm:CovBoundForH} and complete the proof of \Cref{thm:CovPathBound},
we also need the following results proved in \cite{KuiKuiSpinGlass2024}.

\begin{lemma}[\cite{KuiKuiSpinGlass2024}]\label{lemma:314Kuikui}
Let $\UpM_1$ and $\UpM_2$ be $V_{n}\times V_{n}$ matrices on $\mathbb{R}$, with $\UpM_1$ being positive definite, while let $\Field\in \mathbb{R}^{V_n}$. 
For $\beta>0$, suppose there exists positive definite matrix $\Upgamma$ such that for any $z\in \mathbb{R}^{V_n}$, we have
\begin{align}\nonumber 
(\beta\cdot \UpM_1)- (\beta\cdot \UpM_1) \cdot \Cov(\UpM_2, \beta, \UpM_1\cdot z+\Field)\cdot (\beta\cdot \UpM_1) \pdgeq \Upgamma\enspace. 
\end{align}
Then, we have $\Cov(\UpM_1+\UpM_2, \beta, \Field) \pdleq \Upgamma^{-1}$.
\end{lemma}

\begin{lemma}[\cite{KuiKuiSpinGlass2024}]\label{fact:315Kuikui}
Let $\UpK$ and $\UpL$ be $V_n\times V_n$ symmetric matrices. Suppose there exist $\alpha,\eta_1,\eta_2>0$ such that 
$\opnorm{\UpK}\leq \alpha$ and
$\eta_1\Id \pdleq \UpL \pdleq \frac{\eta_2}{\alpha}\Id$. Then, we have
\begin{align}\nonumber 
\UpL-\UpL\cdot \UpK\cdot \UpL \pdgeq \eta_1(1-\eta_2)\cdot \Id\enspace,
\end{align}
where $\Id$ is the $V_n\times V_n$ identity matrix. 
\end{lemma}

We use a slightly different notation from \cite{KuiKuiSpinGlass2024}. For this reason, the above results may look different
from what is stated in that paper.

\begin{proof}[Proof of \Cref{thm:CovPathBound}] 
\Cref{thm:CovPathBound} follows by using derivations which are similar to those in the proof of Theorem 4.6 in \cite{KuiKuiSpinGlass2024}.
We provide the basic calculations here.

Let $\mathcal{G}$ be the event: for $\varepsilon>0$ as specified in the statement of \Cref{thm:CovPathBound}, 
for any $\Field\in \mathbb{R}^{V_n}$ and for constant $c>0$ as specified in \Cref{thm:CovBoundForH}, we have 
\begin{enumerate}[(i)]
\item an $\varepsilon$-block partition $\cB$ with partition matrices $(\InAct_S, \InAct_H)$,
\item $\opnorm{\InAct_S} \le 2 \sqrt{d} (1+o(1))$, 
\item $\opnorm{ \Cov_{\partial H}(\InAct_H, \beta, \Field)} \leq c \cdot \left (\varepsilon \cdot d^{1/10} \right)^{-1}$, \label{item:eq:thm:CovBoundForHB}
\item $\opnorm{\Cov(\InAct_H, \beta, \Field)} \leq c \cdot\maxDelta$. \label{item:eq:thm:CovBoundForHA}
\end{enumerate}
Recall that $\maxDelta=\frac{\log n}{\sqrt{d}}$. 
From \Cref{theorem:HBlockExistence,thm:CovBoundForH,thm:SmallNormAdjac},
we have $\Pr[\mathcal{G}]=1-o(1)$.

From now on, assume that event $\mathcal{G}$ holds. We use \Cref{lemma:314Kuikui} for matrices
\begin{align}\label{eq:DefOfM1M2for}
\UpM_1& =(1-t) \cdot \UpK +t\cdot \frac{\varepsilon}{200}\cdot \Id_{S} +\frac{\upzeta}{\maxDelta}\cdot 
\left(1-\textstyle{\frac{t}{\sqrt{d}}} \right)\cdot \Id_{H\setminus \partial H}&\textrm{and} &&
\UpM_2 &= \InAct_{H}\enspace,
\end{align}
where $\UpK=\InAct_{S}+\left(1+{\textstyle \frac{\varepsilon}{100}}\right)\cdot\|\InAct_S\|_2\cdot \Id_S$.
We specify parameter $\upzeta$ later in this proof.

W.l.o.g. we choose $t=0$ to prove our result. 
\Cref{thm:CovPathBound} follows by showing that on the event $\mathcal{G}$, for any $z\in \mathbb{R}^{V_n}$, 
we have 
\begin{align}\label{eq:Target4thm:CovPathBound}
\beta\cdot \UpM_1- \beta^2\cdot \UpM_1\cdot \Cov(\UpM_2, \beta, \UpM_1\cdot z+\Field)\cdot \UpM_1 &\pdgeq \Upgamma\enspace, 
\end{align}
where %$\Upgamma$ is an $V_n\times V_n$ matrix such that 
$\Upgamma \pdgeq \frac{1}{d^{1/15}} \cdot \Id_S+ \frac{\upzeta}{2\maxDelta}\Id_{H\setminus \partial H}$. 
In this case, the constant $C$ in the statement of \Cref{thm:CovPathBound} satisfies $C=\frac{3}{\upzeta}$. 

Fix any $z\in \mathbb{R}^{V_n}$.
When there is no danger of confusion, we abbreviate $\Cov(\UpM_2, \beta, \UpM_1\cdot z+\Field)$
to $\Cov$. 
Furthermore, for any $\Lambda\subseteq V_n$, we let the $V_n\times V_n$ matrix $\Cov_{\Lambda}$ be 
 defined by 
\begin{align}\label{eq:RestrictionCov}
\Cov_{\Lambda}(u,w) &=\Ind\{u,w\in \Lambda\}\cdot \Cov(u,w) & \forall u,w\in V_n\enspace. 
\end{align}
Standard calculations which can be found in \cite{KuiKuiSpinGlass2024} (proof of Theorem 4.6) imply that 
for any $\delta'>0$, we have
\begin{align}
\lefteqn{
\beta\cdot \UpM_1- \beta^2\cdot \UpM_1\cdot \Cov\cdot \UpM_1 
} \hspace{.4cm} \label{eq:M1M2VSJSJHKantor}\\
& \pdgeq (\beta\cdot \UpK-(1+\delta')\cdot\beta^2 \cdot \UpK \cdot \Cov_S \cdot \UpK)
+\frac{\upzeta}{\maxDelta}\left(
\Id_{H\setminus \partial H}- \left(1+\frac{1}{\delta'} \right)\frac{\upzeta}{\maxDelta} \cdot 
\Cov_{H\setminus \partial H} 
\right) \nonumber \enspace.
\end{align}
From the definition of sets $S$, $H$ and \eqref{eq:RestrictionCov} it is not hard to verify that 
$\Cov_S=\Cov_{\partial H}$, since $H\cap S=\partial H$. 
From \cref{item:eq:thm:CovBoundForHB} in the definition of event $\mathcal{G}$ and for large $d>0$, we have 
\begin{align}
\opnorm{\Cov_S} =\opnorm{\Cov_{\partial H}} \leq d^{-1/12} \enspace.
\end{align}

\noindent
Our assumption $0<\beta\leq \beta_c$ and the event $\mathcal{G}$ imply that there are numbers $0<\eta, \theta <1 $
such that
\begin{align}
\eta \cdot \Id_S \pdleq \beta\cdot \UpK \pdleq \theta \cdot \Id_S\enspace. 
\end{align}
Combining the above inequalities with \Cref{fact:315Kuikui}, for any $\frac{1}{2}\leq \delta \leq 1$, we obtain 
\begin{align}\label{eq:Bound4JSCovParHJS}
(\beta\cdot \UpK -(1+\delta')\cdot \beta^2 \cdot \UpK \cdot \Cov_S \cdot \UpK) \pdgeq d^{-1/15} \cdot \Id_S\enspace. 
\end{align}
Similarly, using \cref{item:eq:thm:CovBoundForHA} of the definition of the event $\mathcal{G}$, for any $\Field'$, we have
\begin{align}
\opnorm{\Cov_{H\setminus \partial H}(\InAct_H, \beta, \Field')} \leq \opnorm{ \Cov(\InAct_H, \beta, \Field')} \leq 
c \cdot\maxDelta \enspace. 
\end{align}
For $\delta'=\frac{1}{2} $ and $\upzeta=\frac{1}{10c(1+\frac{1}{\delta'})}$, 
the above inequality implies 
\begin{align}\label{eq:Bound4JHCovHJH}
\Id_{H\setminus \partial H}- \left(1+\frac{1}{\delta'} \right)\frac{\upzeta}{\maxDelta} 
\Cov_{H\setminus \partial H} \pdgeq \frac{1}{2}\Id_{H\setminus \partial H}\enspace. 
\end{align}

Plugging \eqref{eq:Bound4JSCovParHJS} and \eqref{eq:Bound4JHCovHJH} into \eqref{eq:M1M2VSJSJHKantor}, we get
\begin{align}\nonumber
\beta\cdot \UpM_1- \beta^2\cdot \UpM_1\cdot\Cov\cdot \UpM_1 \pdgeq 
\frac{1}{d^{1/15}} \cdot \Id_S+ \frac{\upzeta}{2\maxDelta}\Id_{H\setminus \partial H}\enspace. 
\end{align}
The above proves \eqref{eq:Target4thm:CovPathBound} and concludes the proof of 
\Cref{thm:CovPathBound}. 
\end{proof}

\spreadpoint

\section{Bounds for the Covariance Matrix  - Proof of \Cref{thm:CovBoundForH}}
\label{sec:thm:CovBoundForH}

\subsection{Weighted-Spheres}
Consider $\G=\G(n,d/n)$, and the set i.i.d standard normal random 
variables $\{ \gauss_{e} \}_{e \in E(\G)}$, while we let the $V_n\times V_n$ 
interaction matrix $\InAct $ be such that 
\begin{align}\nonumber
\InAct(u,w) &=\Ind\{\{u,w\} \in E(\G)\} \cdot \gauss_{\{u,w\}} & \forall u,w\in V_n \enspace.
\end{align}

For $\beta\in \mathbb{R}_{> 0}$ and $\Field \in \mathbb{R}^{V_n}$, we let the Gibbs distribution $\mu_{\beta, \Field, \InAct }$ 
be such that
 \begin{align} \nonumber
\mu_{\InAct,\beta, \Field}(\sigma)
&\propto 
\exp\left(\frac{\beta}{2} \cdot \langle \sigma, \InAct \cdot \sigma\rangle + \langle \sigma, \Field \rangle \right), &\quad 
\forall \sigma\in \{\pm 1\}^{V_n} 
\enspace. 
\end{align}
For each path $P$ in $\G$ define the quantity 
\begin{align}
{\tt W}(P)=\prod\nolimits_{e\in P}\Inf_{e}\enspace,
\end{align}
where $\Inf_{e}$ is defined in \eqref{eq:DefOFInf}.

For each vertex $v\in V_n$ and each integer $\ell>0$, we let $\sqwsphere(v,\ell)$ be defined by
\begin{align}\label{eq:DefOfSQW}
\sqwsphere(v,\ell)=\sum\nolimits_{P} {\tt W}^2(P)\enspace,
\end{align}
where variable $P$ in the summation varies over all paths of length $\ell$ emanating from $v$.

\begin{theorem}\label{lemma:BoundOnWSphereR}
For any $\varepsilon \in (0,1)$ satisfying $1-\varepsilon/8\geq \upkappa$, there exists $d_0 = d_0(\varepsilon) > 0 $, such that for any 
$d \ge d_0$, $\beta \le \beta_c(d)$, and $0 \le \ell \le \maxEll$ the following is true:

With probability $1-o(1)$ over the instances of matrix $\InAct = \InAct(\G(n,d/n), \{\gauss_e\})$ defined as in \eqref{eq:InteractG(np)}, 
there exists an $\varepsilon$-block partition $\cB$ such that %for every block $B\in \cB$ the following is true:
\begin{enumerate}[(a)]
\item for every multi-vertex block $B\in \cB$, for any $v\in V(B)$ we have 
$\sqwsphere(v,\ell)\leq 3\cdot \maxDelta\cdot\upkappa^\ell$,
\item for every $\varepsilon$-block vertex $v$ 
we have $\sqwsphere(v,\ell)\leq (1-\varepsilon/8)^{\ell}$.
\end{enumerate}
The quantity $\upkappa$ is from \eqref{eq:defOfBc}, while $\maxDelta=\maxDeltaVal$.
\end{theorem}
The proof of \Cref{lemma:BoundOnWSphereR} appears \Cref{sec:lemma:BoundOnWSphereR}.

For $\varepsilon>0$, suppose that $(\InAct, \beta)$ admits the $\varepsilon$-block partition $\cB$. 
For each multivertex block $B\in \cB$ and vertex $\rootB \in \partial_{\rm in}B$, let the all paths 
tree $T=\APTree(B,\rootB)$. Consider the 2-spin model on $T$ with inverse temperature $\beta$ such 
that the coupling at each edge $e=\{\overline{x}, \overline{z}\}$ is equal to 
$\gauss_{\{x,z\}}$, where $\overline{x},\overline{z} $ are
copies of $x,z\in B$, respectively.

For the above spin model on $T$, vertex $w\in V(T)$ and integer $\ell>0$, 
 consider the quantity $\sqwsphere_T(w,\ell)$ as in \eqref{eq:DefOfSQW}, in the natural way. 
Furthermore, we let $\sqwsphere_{\partial T}(w,\ell)$ be such that 
\begin{align} \label{eq:PartialTSQR}
\sqwsphere_{\partial T}(w,\ell)=\sum\nolimits_{P} {\tt W}^2(P)\enspace,
\end{align}
where $P$ varies over all paths of length $\ell$ from $w$ to the set of copies of $\partial_{\rm in}B$ in $T$.

In the following result, we use \Cref{lemma:BoundOnWSphereR} to obtain bounds on $\sqwsphere_T(v,\ell)$ and $\sqwsphere_{\partial T}(w,\ell)$.

\begin{proposition}\label{prop:SQSphereAllPathTree}
Let $\varepsilon \in (0,1)$ satisfying $1-\varepsilon/8\geq \upkappa$, 
let sufficiently large $d_0 = d_0(\varepsilon) > 0 $, let $d \ge d_0$ and 
$\beta \le \beta_c(d)$ also let integer $0\leq \ell \leq \maxEll$. 

With probability $1-o(1)$ over the instances of matrix $\InAct = \InAct(\G(n,d/n), \{\gauss_e\})$ defined as in \eqref{eq:InteractG(np)}, there exists an $\varepsilon$-block partition $\cB$ such that for every multi-vertex block $B\in \cB$, for $\rootB \in \partial_{\rm in} B$ and $T=\APTree(B,\rootB)$ the following is true:
for any $w\in V(T)$, we have
\begin{enumerate}[(a)]
\item %for any $w\in V(T)$, we have 
$\sqwsphere_T(w,\ell)\leq 10 \cdot \maxDelta \cdot{\upkappa}^\ell$
\item if additionally $w$ is a copy of a vertex in $\partial_{\rm in} B$, we have 
$\sqwsphere_T(w,\ell)\leq \frac{10}{d^{1/10}} \cdot { (1-\varepsilon/8)}^{\ell}$
\item $\sqwsphere_{\partial T}(w,\ell)\leq \frac{2}{d^{1/10}}\cdot (1-\varepsilon/16)^{\ell-1} $.
\end{enumerate}
The quantity $\upkappa$ is from \eqref{eq:defOfBc}, $\sqwsphere_{\partial T}(w,\ell)$
is defined in \eqref{eq:PartialTSQR},
while $\maxDelta=\maxDeltaVal$. 
\end{proposition}
The proof of \Cref{prop:SQSphereAllPathTree} appears in \Cref{sec:prop:SQSphereAllPathTree}.

\subsection{Proof of eq.\eqref{eq:thm:CovBoundForHA} in \Cref{thm:CovBoundForH}}\label{sec:DefOfPotential}

Let $\mathcal{G}$ be the event that there is an $\varepsilon$-block partition $\cB$
of the set of vertices $V_n$ with partition matrices $\InAct_H$ and $\InAct_S$.

\begin{claim}\label{claim:FromSpinGlass2IsingBlock}
On the event $\mathcal{G}$, for any multi-vertex block $\hat{B}\in \cB$ we have 
\begin{align}\label{eq:claim:FromSpinGlass2IsingBlock}
 \opnorm{\Cov(\InAct_{\hat{B}}, \beta, \Field)} \leq 
\opnorm{\Cov(|\InAct_{\hat{B}}|, \beta, 0)}
\enspace. 
\end{align}
\end{claim}
Recall that for a matrix $\UpD\in \mathbb{R}^{ N \times N}$, we let 
$|\UpD|$ denote the matrix having entries $|\UpD_{i,j}|$.

The proof of \Cref{claim:FromSpinGlass2IsingBlock} is provided in \Cref{sec:claim:FromSpinGlass2IsingBlock}.

Assume from now on that the event $\mathcal{G}$ holds. 
 $\InAct_H$ is a block matrix, and therefore we have 
\begin{align} \label{eq:CovToBlockaN}
\opnorm{\Cov (\InAct_H, \beta, \Field)} & = 
\max\nolimits_{\hat{B}} \opnorm{\Cov(\InAct_{\hat{B}}, \beta, \Field)} \leq 
\max\nolimits_{\hat{B}} \opnorm{\Cov(|\InAct_{\hat{B}}|, \beta, 0) }
\enspace, 
\end{align}
where $\hat{B}$ varies over all blocks in $\cB$ with more than one vertex. 
The last inequality follows from \Cref{claim:FromSpinGlass2IsingBlock}.

Let block $B$ maximise the r.h.s. of \eqref{eq:CovToBlockaN}. 
Recall that $B$ is a tree with at most one extra edge. 
 
For a vertex $w\in \partial_{\rm in}B$, let the all paths tree $T=\APTree(B,w)$ (see definition in \Cref{sec:AllPathTreeDefinition}). 
Consider the $V_n \times V(T)$ matrix $\UpK$ such that
for any $x\in V_n$ and any $z\in V(T)$, we have
\begin{align}\label{eq:defOFK}
\UpK(x, z) &= \Ind\{\textrm{$z$ is a copy of $x$ in $\APTree(B,w)$}\} \enspace.
\end{align}
Consider also the $V(T)\times V(T)$ matrix $\UpY$, such that for any $x,z \in V(T)$, we have
\begin{align}\label{eq:DefOfY}
\UpY(x,z) &= \prod\nolimits_{e\in P}\Inf_{e}\enspace,
\end{align}
where $P$ is the unique self-avoiding walk from vertex $x$ to vertex $z$ in $T$, while 
for each edge $e=\{\overline{x},\overline{z}\}$ in $P$ such that $\overline{x}$ is a copy of $x$ and $\overline{z}$ is a copy of 
$z$ in $B$, we let 
\begin{align}
\Inf_{e} &=|\tanh(\beta \cdot \InAct_{\{ x, z \} } )|\enspace. 
\end{align}
Furthermore, let the $V_n\times V_n$ matrix $\UpR$ be such that 
\begin{align}
\UpR &= \UpK \cdot \UpY \cdot \overline{\UpK} \enspace,
\end{align}
where $\overline{\UpK}$ is the matrix transpose of $\UpK$.

\Cref{prop:Inf2TreeRedaux}, \Cref{lemma:PathPropertyAPTree} and the definition of $\Inf_{e}$ imply 
the entry-wise inequality 
$$\UpR \geq \Cov ( |\InAct_{B}|, \beta, 0 ).$$
\charis{The above is the covariance not the influence matrix}

Then, since the entries of $\Cov ( |\InAct_{B}|, \beta, 0 )$ are non-negative, it standard 
\begin{align}
\opnorm{\Cov (| \InAct_{B}|, \beta, 0 )} \le \opnorm{\UpR} \enspace.
\label{eq:finCovBoundforYell}
\end{align}
 \Cref{eq:thm:CovBoundForHA} follows by showing 
\begin{align}\label{eq:Target4thm:CovBoundForH}
\opnorm{\UpR} &\leq 40 \cdot \left(1-\upkappa^{1/4}\right)^{-2} \cdot \maxDelta\enspace,
\end{align}
and setting $c=\frac{40}{\left(1-\upkappa^{1/4}\right)^2}$. To this end, we use the following claim.

\begin{claim}\label{claim:2NormYVs2NormR}
We have $\opnorm{\UpR} \leq 4 \opnorm{\UpY}$. 
\end{claim}

\begin{proof}[Proof of \Cref{claim:2NormYVs2NormR}]
From the definition of $\UpR$, we have
\begin{align}\label{eq:claim:2NormYVs2NormR}
\opnorm{\UpR} &= \opnorm{\UpK\cdot \UpY \cdot \overline{\UpK}} 
\leq \opnorm{\UpK} \cdot \opnorm{\UpY} \cdot \opnorm{\overline{\UpK}} \enspace.
\end{align}
Furthermore, for any $x,z\in V_n $, we have
\begin{align}
(\UpK\cdot \overline{\UpK}) (x,z) &=\sum\nolimits_{y\in V(T) } \UpK(x,y)\cdot \UpK(z,y) \nonumber \\
&=\sum\nolimits_{y\in V(T) } \Ind\{ \textrm{$y$ is a copy of $x$} \} \cdot \Ind\{ \textrm{$y$ is a copy of $z$} \} \nonumber \\
&=\sum\nolimits_{y\in V(T) } \Ind\{ \textrm{$y$ is a copy of $x$} \} \cdot \Ind\{ z=x \} \enspace. \nonumber 
\end{align}
The second equality is from the definition of matrix $\UpK$, while the third one is from the fact that vertex $y\in V(T)$ can be a copy of a single vertex $z\in V_n$. 
From the above, we conclude that $\UpK\cdot \overline{\UpK}$ is a $V_n\times V_n$ diagonal matrix such that the entry $(\UpK\cdot \overline{\UpK})(v,v)$ is equal to the number of copies of vertex $v$ in $\APTree(B,w)$. 

From the definition of $\APTree$, and in particular \Cref{cor:NoOfCopiesAllTree}, the maximum entry in the diagonal matrix $(\UpK\cdot \overline{\UpK})$ is $4$. This implies that $\opnorm{\UpK\cdot \overline{\UpK}}\leq 4$. 
We conclude that $\opnorm{\UpK}=\opnorm{\overline{\UpK}}= \opnorm{\UpK\cdot \overline{\UpK}}^{1/2}\leq 2$.

Then, the claim follows from \eqref{eq:claim:2NormYVs2NormR}. 
\end{proof}

For what follows, recall $T=\APTree(B, w)$. For integer $\ell\geq 0$, define the $V(T)\times V(T)$ matrix $\UpY_{\ell}$ by
\begin{align}
\UpY_{\ell}(z,x) &= \Ind\{\dist(z,x)=\ell\}\cdot \UpY (z,x)\enspace, 
\end{align}
where the distance is on tree $T$. 
Note that matrices $\UpY$ and $\UpY_{\ell}$, for all $\ell\geq 0$, are symmetric. 

The triangle inequality implies 
\begin{align}\label{eq:YVsYellTriangle}
\opnorm{\UpY} &\leq \sum\nolimits_{\ell\geq 0}\opnorm{\UpY_{\ell}} \enspace.
\end{align}
We proceed to bound $\opnorm{\UpY_{\ell}}$ for each $\ell>0$.

For a parameter $\delta>0$ which we specify later, we define the weight function $\badmormw: V(T) \to [0, +\infty)$ for each vertex of $T=\APTree(B, w)$ as follows:

For the root $r$ of $T$, we define $\badmormw(r) = 1$. For every vertex $z$ with $\badmormw(z)$ being determined, we define $\badmormw(u) = \badmormw(z)(1+\delta) \cdot \Inf_{\{z,u\}}$ for all children $u$ of $z$.

A useful observation is that for every $u \in V(T)$ we have
\begin{align}\label{eq:BadBlockNormWeightProperties}
\badmormw(u) &= (1+\delta)^{\dist(r,u)} \cdot \UpY(r, u) \enspace. 
\end{align}
% where $\dist(r,u)$ is the distance of the vertex $u$ from the root $r$. 

We also define $\MatrixDW$ to be the diagonal $V(T) \times V(T)$ matrix such that 
\begin{align}
\MatrixDW(u,u)&=\badmormw(u), &\forall u\in V(T)\enspace.
\end{align}
Since $\UpY_{\ell}$ is symmetric, we have
\begin{align}\label{eq:BreakNormEll} 
\opnorm{\UpY_{\ell}} \le \ \|\MatrixDW^{-1}\cdot \UpY_\ell \cdot \MatrixDW\|_{\infty} 
 = \max_{w}\left\{\sum\nolimits_{ u} 
\left( \badmormw(w)\right )^{-1} \cdot \UpY_{\ell} (w,u) \cdot \badmormw(u)\right\}
\enspace,
\end{align}
where the last equality follows from the definition of the norm and
the observation that all entries of $\UpY_{\ell}$ are nonnegative. 
For every $w \in V(T)$, which is at distance $q $ from the root $r$ of $T$, we show
\begin{align}\label{eq:Target4CovNormTree}
\sum\nolimits_{ u} \left( \badmormw(w) \right)^{-1} \cdot \UpY_{\ell} (w,u) \cdot \badmormw(u) &\leq 
10 \cdot \maxDelta\cdot \sum\nolimits_{0\leq i \leq q} (1+\delta)^{\ell-2i} 
\cdot \upkappa^{\ell-i}
% \nonumber\\
% &\le 
% 10 \cdot \maxDelta\cdot \sum\nolimits_{0\leq i \leq \ell} (1+\delta)^{\ell-2i} 
% \cdot \upkappa^{\ell-i}
\enspace. 
\end{align}
Before proving that \eqref{eq:Target4CovNormTree} is true, let us see how it implies \eqref{eq:Target4thm:CovBoundForH}.

We choose $\delta$ such that it satisfies both $(1+\delta)\upkappa < 1$ 
and $(1+\delta)^2\upkappa > 1$. Specifically, we set $\delta = \upkappa^{-3/4} -1$.
Let $w_0 \in V(T)$ be the vertex attaining the maximum value of the rhs in \eqref{eq:Target4CovNormTree}, and let $q_0$ be the distance of $w_0$ from the root of $T$. For our choice of $\delta$, \eqref{eq:Target4CovNormTree} implies
\begin{align}
\|\MatrixDW^{-1} \cdot \UpY_\ell \cdot \MatrixDW\|_{\infty} 
\le 10 \cdot \maxDelta \cdot (1+\delta)^{\ell}\cdot \upkappa^\ell 
\cdot
\sum\nolimits_{0\le i \le {q_0}} (1+\delta)^{-2i}
	\cdot \upkappa^{-i}
\le 10 \cdot
\maxDelta \cdot \frac{\upkappa^{\ell/4} }{1-\sqrt{\upkappa}}
\enspace .
\label{eq:YellboundInfNorm}
\end{align}
Plunging the above into \eqref{eq:YVsYellTriangle} yields
\begin{align}
\opnorm{\UpY} &\le \maxDelta\cdot
\frac{10}{1-\sqrt{\upkappa}} \sum\nolimits_{\ell \ge 0} \upkappa^{\ell/4} \le
\maxDelta
\cdot \frac{10}{\left(1-\upkappa^{1/4}\right)^2} \enspace.
\label{eq:finalTreeCase}
\end{align}
 \Cref{eq:Target4thm:CovBoundForH} follows from \eqref{eq:finalTreeCase}
and \Cref{claim:2NormYVs2NormR}.
It remains to show that \eqref{eq:Target4CovNormTree} is true.

Fix vertex $z \in V(T)$ and let $(z_0=z,z_1, \ldots, z_{q-1}, z_q = r)$ be the path from $z$ to~$r$. Recall that $q = \dist(z, r)$.

For $0\leq i \leq q$, let $T_i$ be the subtree of $T$ rooted at $z_i$ and 
also containing its progeny. For $\ell > q$, we have 
\begin{align}
\sum\nolimits_{u} \left(\badmormw(z)\right)^{-1} \cdot \UpY_{\ell}(z,u) \cdot \badmormw(u) 
&\leq \sum\nolimits_{0 \le i \le q} \left( \badmormw(z) \right)^{-1}\cdot 
\prod\nolimits_{0< j \leq i} \Inf_{\{z_{j}, z_{j-1}\}}
\sum\nolimits_{ u \in V(T_{i})}
\UpY_{\ell-i}(z_i,u) \cdot \badmormw(u) \nonumber \\
&\leq 
\sum\nolimits_{0 \le i \le q} \left( \red{\badmormw(z) }\right)^{-1}\cdot (1+\delta)^{-i}
\sum\nolimits_{u \in V(T_{i})}
\UpY_{\ell-i}(z_i,u) \cdot \badmormw(u)\enspace, \label{eq:Chi(w)VsChi(zi)}
\end{align}
where \eqref{eq:Chi(w)VsChi(zi)} follows from the definition of 
$\badmormw(z)$ and $\badmormw(z_i)$. Specifically, due to 
\eqref{eq:BadBlockNormWeightProperties}, we have
\begin{align} \nonumber
\badmormw(z_i)=\left((1+\delta)^{i} \cdot \prod\nolimits_{0< j\leq i} \Inf_{\{z_{j}, z_{j-1}\}} \right) \cdot \badmormw(z) \enspace. 
\end{align}
Note that, for $u$ in $V(T_{i})$ with which vertex $z_i$ is connected through the path $P$, we have 
\begin{align}\nonumber
\badmormw(u)&=\left((1+\delta)^{|P|} \cdot \prod\nolimits_{e\in P} \Inf_{e} \right) \cdot \badmormw(z_i) = (1+\delta)^{|P|} \cdot \UpY_{\ell-|P|}(z_i,u) \cdot \badmormw(z_i) \enspace.
\end{align}
Plugging the above into \eqref{eq:Chi(w)VsChi(zi)}, we obtain
\begin{align}
\sum\nolimits_{u} \left( \badmormw(z)\right)^{-1} \cdot \UpY_{\ell}(z,u) \cdot \badmormw(u)
&\leq 
\sum\nolimits_{0\le i \le q} (1+\delta)^{\ell-2i}
\sum\nolimits_{u \in V(T_{i}) } \left(\UpY_{\ell-i}(z_i,u)\right)^2
 \nonumber \\
& \le \sum\nolimits_{0\le i \le q} (1+\delta)^{\ell-2i} \cdot \sqwsphere_{T}(z_i, \ell-i)
\label{eq:64} \\
& \le 10 \cdot \maxDelta\cdot \sum\nolimits_{0\le i \le q} (1+\delta)^{\ell-2i}\cdot \upkappa^{\ell-i}
\enspace.
\label{eq:BreakNormSums}
\end{align}
where \eqref{eq:64} follows from the definition of $\sqwsphere_{T}(z_i, \ell-i)$, 
while \eqref{eq:BreakNormSums} follows from part (a) of \Cref{prop:SQSphereAllPathTree}. 
The above implies that \eqref{eq:Target4CovNormTree} is true for $\ell>q$.

With a similar line of arguments, for $\ell \le q$, we obtain
\begin{align}
\sum\nolimits_{u} \frac{1}{\badmormw(z)} \cdot \UpY_{\ell}(z,u) \cdot \badmormw(u) 
& \le 10 \cdot \maxDelta\cdot \sum\nolimits_{0\le i \le \ell-1} (1+\delta)^{\ell-2i}
	\cdot \upkappa^{\ell-i} \nonumber \enspace,
\end{align}
which proves \eqref{eq:Target4CovNormTree}. 

All the above conclude the proof of eq. \eqref{eq:thm:CovBoundForHA} of \Cref{thm:CovBoundForH}.

\subsection{Proof of eq.\eqref{eq:thm:CovBoundForHB} in \Cref{thm:CovBoundForH}}
In this proof, we use some basic steps which are similar to those we use for the proof of \eqref{eq:thm:CovBoundForHA}. 

We note that $\InAct_{\partial H}$ is a block matrix. Then, arguing as in 
\Cref{claim:FromSpinGlass2IsingBlock}, we have
\begin{align}\nonumber % \label{eq:CovToBlockaNPAr}
\opnorm{\Cov_{\partial H}(\InAct_{H}, \beta, \Field)} \leq 
%\max_{\hat{B}} \opnorm{\Cov_{\partial H}(\InAct_{\partial_{\rm in} \hat{B}}, \beta, 0)} =
\max\nolimits_{\hat{B}} \opnorm{\Cov_{\partial \hat{B}}(| \InAct_{\partial_{\rm in} \hat{B}}|, \beta, 0)}\enspace, 
\end{align}
where $\hat{B}$ varies over all blocks in $\cB$ with more than one vertex. 

Let block $B$ maximise the r.h.s. of \eqref{eq:CovToBlockaN}. 
Recall that $B$ is a tree with at most one extra edge, while
 each vertex of $\partial_{\rm in} B$ is an $\varepsilon$-block vertex.

For vertex $w\in \partial_{\rm in}B$, let the all paths tree $T=\APTree(B,w)$ (see definition in 
\Cref{sec:AllPathTreeDefinition}). 

Consider the $V(T)\times V(T)$ matrices $\UpY$ and $\tilde\UpY$. $\UpY$ is the matrix defined in \eqref{eq:DefOfY}. 
Matrix $\tilde\UpY$ is such that for any $x,z \in V(T)$, we have
\begin{align}\nonumber
\tilde\UpY(x,z) &= \Ind\{x, z \in \partial_{\rm in}B\}\cdot\UpY(x,z)\enspace.
\end{align}
Furthermore, we define the $V_n\times V_n$ matrix $\tilde\UpR$ be such that 
\begin{align}\nonumber
\tilde\UpR &= \UpK \cdot \tilde\UpY \cdot \overline{\UpK} \enspace,
\end{align}
where matrix $\UpK$ is defined in \eqref{eq:defOFK}, and $\overline{\UpK}$ is the transpose of $\UpK$. 
Arguing as in \eqref{eq:finCovBoundforYell}, we have 
% 
% With arguments similar to those used above to derive \eqref{eq:finCovBoundforYell}, we see that
\begin{align}\nonumber
\opnorm{\Cov_{\partial B} ( | \InAct_{\partial_{\rm in}B}|, \beta, 0 )} \le \opnorm{\tilde\UpR} \enspace.
%\label{eq:finCovBoundforTildeYell}
\end{align}

We have the following claim which is obtained by using exactly the same steps as in \Cref{claim:2NormYVs2NormR}.

\begin{claim}\label{claim:2NormYVs2NormRParB}
We have $\opnorm{\tilde\UpR} \leq 4 \opnorm{\tilde\UpY}$. 
\end{claim}

For integer $\ell\geq 0$, let the $V(T)\times V(T)$ matrix $\tilde\UpY_{\ell}$ be such that
\begin{align*}
\tilde\UpY_{\ell}(z,x) &= \Ind\{\dist(z,x)=\ell\}\cdot \tilde\UpY (z,x)\enspace. 
\end{align*}
Note that matrices $\tilde\UpY$ and $\tilde\UpY_{\ell}$, for all $\ell\geq 0$, are symmetric. 
The triangle inequality implies 
\begin{align}\label{eq:YVsYellTriangleBVer}
\opnorm{\tilde\UpY} &\leq \sum\nolimits_{\ell\geq 0}\opnorm{\tilde\UpY_{\ell}} \enspace.
\end{align}
We proceed to bound $\opnorm{\tilde\UpY_{\ell}}$ for each $\ell>0$.

For a parameter $\delta>0$ which we specify later, we define the same weight function $\badmormw: V(T) \to [0, +\infty)$ 
on the vertices of $T=\APTree(B, w)$ as in the proof of \eqref{eq:thm:CovBoundForHA}, in this proposition.

Recall that for each vertex $u\in V(T)$, we have
\begin{align}\nonumber %\label{eq:BadBlockNormWeightPropertiesBVar}
\badmormw(u) &= (1+\delta)^{\dist(r,u)} \cdot \tilde\UpY(r, u) \enspace,
\end{align}
while $r$ is the root of $T$.

We also define the $V(T) \times V(T)$ diagonal matrix $\tilde\MatrixDW$ 
such that $\tilde\MatrixDW(u,u)=\badmormw(u)$ for all $u\in V(T)$.
Since $\tilde\UpY_{\ell}$ is symmetric with non-negative entries, we have
\begin{align}\nonumber %\label{eq:BreakNormEllBVer}
\opnorm{\tilde\UpY_{\ell}} 
\le
\|\tilde\MatrixDW^{-1} \cdot 
\tilde\UpY_\ell \cdot \tilde\MatrixDW\|_{\infty}
\le \ \max\nolimits_{w}
\left\{\sum\nolimits_{ u} 
\left( \badmormw(w) \right)^{-1} \cdot \tilde\UpY_{\ell} (w,u) \cdot \badmormw(u)\right\} \enspace.
\end{align}
For every $w \in V(T)$, which is at distance $q$ from the root of $T$, we have
\begin{align}\label{eq:Target4CovNormTreeBVer}
\sum\nolimits_{ u} \left( \badmormw(w) \right)^{-1} \cdot \tilde\UpY_{\ell} (w,u) \cdot \badmormw(u) &\leq 
{2} \cdot {d^{-1/10}}\cdot \sum\nolimits_{0 \le i \le q} (1+\delta)^{\ell-2i} \cdot (1-\varepsilon/16)^{\ell-1-i}
\enspace. 
\end{align}
We omit the calculation that establish the above inequality since they are almost identical to
those we have for \eqref{eq:Chi(w)VsChi(zi)} and \eqref{eq:BreakNormSums}. 

In light of \eqref{eq:Target4CovNormTreeBVer}, we get \eqref{eq:thm:CovBoundForHB} as follows:
We choose $\delta$ such that it satisfies both $(1+\delta)\cdot (1-\varepsilon/16) < 1$ 
and $(1+\delta)^2(1-\varepsilon/16) > 1$. Specifically, we set $\delta = (1-\varepsilon/16)^{-3/4} -1$.
For such $\delta$, \eqref{eq:Target4CovNormTreeBVer} and standard derivations which are similar to
those we have for \eqref{eq:YellboundInfNorm}, we have
\begin{align}\nonumber
\|\tilde\MatrixDW^{-1} \cdot 
\tilde\UpY_\ell \cdot \tilde\MatrixDW\|_{\infty}
%&\le 2 \cdot d^{-1/10} \cdot (1-\varepsilon/16)^{-1} \cdot \frac{(1-\varepsilon/16)^{\ell/4} }{1-\sqrt{1-\varepsilon/16}}
&\le K \cdot (1-\varepsilon/16)^{\ell/4} \enspace .
%\label{eq:YellboundInfNormBVer}
\end{align}
where $K^{-1}=\frac{1}{2}\cdot d^{1/10} \cdot (1-\varepsilon/16)\cdot (1-\sqrt{1-\varepsilon/16})$.

Plunging the above into \eqref{eq:YVsYellTriangleBVer} and using standard derivations, we obtain %yields
\begin{align}\nonumber
\opnorm{\tilde\UpY} &\le K\cdot
\sum\nolimits_{\ell \ge 0}(1-\varepsilon/16)^{\ell/4} 
\le d^{-1/9}
\enspace,
%\label{eq:finalTreeCaseBVer}
\end{align}
as desired. 
All the above conclude the proof of eq.\eqref{eq:thm:CovBoundForHB}.
\hfill $\Box$

\spreadpoint

\section{Proof of \Cref{lemma:BoundOnWSphereR}}
\label{sec:lemma:BoundOnWSphereR}

Recall that for $r\ge 1$, and a vertex $v$, we define
\begin{align}\nonumber
\sqwsphere(v,r) & = \sum\nolimits_{P} \prod\nolimits_{e \in P} |\tanh(\beta \cdot \gauss_e)|^2
\enspace,
\end{align}
where the sum is over all paths $P$ of length $r$ emanating from the vertex $v$.

We write $\sqwsphere_G$ to indicate that 
$\sqwsphere$ is considered with respect to the underlying graph $G$.

\newcommand{\Tbfs}{T_{\rm BFS}}

\begin{proof}[Proof of \Cref{lemma:BoundOnWSphereR} part (a)]
For integer $\ell>0$ and vertex $w$ in $\G=\G(n,d/n)$, let $T=\Tbfs(w,\ell)$ be the breadth-first-search tree 
from vertex $w$ restricted to its first $\ell$ levels. 

In what follows, we consider the quantity $\sqwsphere_T(w, \ell)$. Note that this quantity might be
 different from $\sqwsphere(w, \ell)$ which we need to bound.  Since the weights that induce the two quantities
 are non-negative, we always have 
 \begin{align}\nonumber
 \sqwsphere_T(w, \ell)\leq \sqwsphere(w, \ell)\enspace.
 \end{align}
We have equality when the ball of radius $\ell$ around $w$ induces a graph which is a tree.

It is standard to show that in $T=\Tbfs(w,\ell)$,  each vertex at level $h<\ell$ has a number of
children which is dominated by $\Poisson(d)$. Hence, there is a coupling between
$\Tbfs(w,\ell)$ and the standard Galton-Watson tree $\TT$ with $\Poisson(d)$ offspring
distribution, such that with probability 1, $\TT$ contains a copy of $\Tbfs(w,\ell)$. 

Let $\widehat{\TT}$ be a subtree of $\TT$ which is isomorphic to $T$. Let $F:V(\widehat{\TT})\to V(T)$ 
be an isomorphism map between the two trees.

Consider an instance of the 2-spin model on $\TT$ with inverse temperature $\beta$. Assume that
the couplings at the edges of $\TT$ are obtained as follows: for each edge 
$e=\{x,z\}\in \widehat{\TT}$, the coupling is equal to $\gauss_{\{ F(x), F(z)\}}$. 
For every edge $e\notin \widehat{\TT}$, its coupling is an independent standard
normal $\gauss_e$. 

Consider the quantity $\sqwsphere_{\TT}(u, \ell)$ for the system in $\TT$, where $u=F^{-1}(w)$. 
Since the weights that specify $\sqwsphere_T(w, \ell)$ and $\sqwsphere_{\TT}(u, \ell)$ are non-negative 
numbers,  with probability $1$, we have
\begin{align}\label{eq:BounDTreeB}
\sqwsphere_T(w, \ell)\leq \sqwsphere_{\TT}(u, \ell)\enspace. 
\end{align}

Recalling that $\maxDelta=\frac{\log n}{\sqrt{d}}$, an application of \Cref{thm:GWSphereTail} 
for $\frac{10}{\sqrt{d}}<t<\frac{d}{2\upkappa}$ implies 
\begin{align*}
\Pr \left[\sqwsphere_{\TT}(u, \ell) \le \maxDelta \cdot {\upkappa}^\ell\right] \ge 1 - n^{5/2}
	\enspace.
\end{align*}
Combining the above with \eqref{eq:BounDTreeB}, and a simple union bound, we get 
\begin{align}\label{eq:BoundSQSSE1}
\Pr \left[\forall w\in V_n, \quad \sqwsphere_T(w, \ell) \le \maxDelta \cdot \upkappa^\ell\right] \ge 1 - n^{3/2}
	\enspace.
\end{align}

Let $\mathcal{E}_1$ be the event that $\forall w\in V_n$ we have 
$\sqwsphere_T(w, \ell) \le \maxDelta \cdot \upkappa^\ell $. Also, let $\cE_2$ be the event that 
for all $\ell < \maxEll$ and for all vertices $w\in V_n$, 
the graph that is induced by the vertices within distance $\ell$ from $w$ is 
a tree with at most one extra edge.

Note that the desirable event $\cE_S$ of \Cref{lem:SamllUnicyclicGnp} implies $\cE_2$. 
Recall that $\cE_S$ is the event that ``every set of vertices $S$ in $\G(n,d/n)$ with  
$|S| \le 2\frac{\log n}{(\log d)^2}$, spans at most $|S|$ edges." To see why $\cE_S$ implies 
$\cE_2$, assume towards  a contradiction that $\cE_S$ occurs while we have $\overline{\cE_2}$.

Fix $\ell < \maxEll$.  Recall that $N_\ell(w)$ denotes the vertices at distance $\le \ell$ from 
vertex $w$.  Consider a BFS tree  $T_{w,\ell}$  emanating from vertex $w$  and truncated at 
depth $\ell$. If the vertices of $N_\ell(w)$ span two or more cycles,  then $N_\ell(w)$ spans  
two edges $e_1, e_2$ that do not appear in $T_{w, \ell}$. Let $M$ be the set of vertices containing 
the endpoints of $e_1$ and  $e_2$, along with their shortest paths to $w$. Then, we have a set 
of vertices $M$ such that  $|M| \le 2\ell < 2\frac{\log n}{(\log d)^2}$ which spans at least 
$|M|+1$ edges. Clearly, this contradicts our assumption that event $\cE_S$ holds.  

The above argument and \Cref{lem:SamllUnicyclicGnp} imply that $\Pr[\cE_2]\geq 1-n^{-3/4}$. 
Combining this observation with   \eqref{eq:BoundSQSSE1}, 
a simple union bound yields
\begin{align}\label{eq:SQSphe101AA}
\Pr[\cE_1, \cE_2]\geq 1-2n^{-2/3}\enspace. 
\end{align}

For what follows, assume that events $\cE_1$ and $\cE_2$ hold. Then,  for a vertex  $w\in V_n$, 
such that the ball of radius  $\ell<\frac{\log n}{(\log d)^2}$ centred at $w$ is   a tree, we have 
\begin{align}\label{eq:SQSphe101A}
\sqwsphere(w, \ell) \le \maxDelta \cdot 
\upkappa^\ell \enspace.
\end{align}
Suppose now that for vertex $w\in V_n$, the ball of radius $\ell$ induces a unicylic graph. 
Recall the BFS tree $T=\Tbfs(w,\ell)$. Let $\widehat{e}=\{x,z\}$ be the crossing edge.

For $\sqwsphere_T(w, \ell)$, we consider only paths  in  $T=\Tbfs(w,\ell)$. The extra paths that 
we consider with $\sqwsphere(w, \ell)$  correspond to those that either contain the subpath $(x,z)$ 
or the subpath $(z,x)$.  Let $\sqwsphere_A$ is the weight of the paths of the first category, i.e., 
those that use the subpath $(x,z)$. Similarly, $\sqwsphere_B$ be  the weight of the paths of the 
second one.  It is not hard to verify that 
\begin{align}\nonumber
\sqwsphere_A,\ \sqwsphere_B \leq \max\nolimits _{v\in V_n} \{ \sqwsphere_T(v, \ell) \}\enspace. 
\end{align}
Therefore, for such vertex $w$ we have 
\begin{align}\label{eq:SQSphe101B}
\sqwsphere(w, \ell) \le 3\cdot \max\nolimits_{v\in V_n}\{ \sqwsphere_T(v, \ell)\}
\enspace.
\end{align}
Clearly, \eqref{eq:SQSphe101A} and \eqref{eq:SQSphe101B} imply that, on the events 
$\cE_1, \cE_2$, we have
\begin{align}\nonumber 
\sqwsphere(w, \ell) & \le 3\cdot \maxDelta \cdot 
\upkappa^\ell & \forall w\in V_n\enspace. 
\end{align}
\Cref{lemma:BoundOnWSphereR} part (a) follows from the above and \eqref{eq:SQSphe101AA}.
\end{proof}

\begin{proof}[Proof of \Cref{lemma:BoundOnWSphereR} part (b)]
Let $\cE_1$ be the event that $(\InAct,\beta)$ admits an $\varepsilon$-block partition $\cB$.
Also, let $\cE_2$ be the event that for all positive integers $\ell < \maxEll$ and for all vertices 
$w\in V_n$, the subgraph of $\G=\G(n,d/n)$ that is induced by the vertices within distance 
$\ell$ from $w$ is a tree with at most one extra edge. 

From \Cref{lemma:DEpsilonBlockLemma} and \Cref{lem:SamllUnicyclicGnp}, we have
\begin{align}
\Pr[\cE_1, \cE_2] &= \red{1-o(1)} \enspace. 
\end{align}

Assume that the events $\cE_1$ and $\cE_2$ hold. Let $v\in V_n$ be an $\varepsilon$-block vertex. 
For the range of $\ell$ we consider in this proof, the vertices along paths of length $\ell$ 
that emanate from vertex $v$ induce the subgraph $B_{v,\ell}$ of $\G$, which is a tree with 
at most one extra edge. 
We consider two cases. The first corresponds to having $B_{v,\ell}$ tree, while
the second one corresponds to having $B_{v,\ell}$ unicyclic. 

Suppose that $B_{v,\ell}$ is a tree. The assumption that $v$ is an $\varepsilon$-block vertex implies  
the following result.

\begin{corollary}\label{cor:HeavyWABound}
Let $Q$ be any path of length $\ell \le \maxEll$ emanating from an $\varepsilon$-block vertex $u$.
Let also $M$ be the set of heavy vertices in $Q$, i.e., for all $x\in M$ we have $\WSA(x)>1-\varepsilon/2$. 
Property \ref{itm:blockA} in the definition of $\varepsilon$-block vertices implies
\begin{align*}
\prod\nolimits_{x\in M}\WSA(x) \leq d^{-|M|} \cdot (1-\varepsilon/4)^{|M|-\ell}\enspace.
\end{align*}
\end{corollary}
Unfolding $\sqwsphere(v,\ell)$ we get 
\begin{align}\label{eq:OneStep4SQRSpherRedaux}
\sqwsphere(v,\ell) 
&= \sum\nolimits_{v_i \sim v} \left|\tanh\left(\beta \cdot \gauss_{\{v,v_i\}}\right)\right|^2 
\cdot \sqwsphere(v_i, \ell-1) \enspace,
\end{align}
where for each $v_i$, we have
\begin{align}\nonumber 
\sqwsphere(v_i, \ell-1) = \sum\nolimits_{Q} \prod\nolimits_{e \in Q} |\tanh 
({ \beta }\cdot \gauss_e)|^2
\enspace,
\end{align}
while $Q$ varies over all paths of length $\ell-1$ that emanate from $v_i$ and do not
use the edge $\{v,v_i\}$. 

From \eqref{eq:OneStep4SQRSpherRedaux}, we have
\begin{align}
\sqwsphere(v,\ell) 
&\leq \max\nolimits_{v_i}\left\{ \sqwsphere(v_i, \ell-1)\right\} \cdot 
\sum\nolimits_{v_i \sim v} \left|\tanh\left(\beta \cdot \gauss_{\{v,v_i\}}\right)\right|^2 \nonumber\\
&\leq \max\nolimits_{v_i}\left\{ \sqwsphere(v_i, \ell-1)\right\} \cdot \WSA(v) \enspace. \label{eq:OneStepIndTH}
\end{align}
Repeating the above step $\ell$ times in total, we obtain 
\begin{align}\label{eq:SQRSphereVsWorstPath}
\sqwsphere(v,\ell) &\leq \max\nolimits_{Q} \left\{ \prod\nolimits_{w\in Q} 
\WSA(w) \right\} \enspace,
\end{align}
where $Q$ varies over all paths of length $\ell$ that emanate from vertex $v$.

Let $P$ be a path that maximizes the r.h.s. in \eqref{eq:SQRSphereVsWorstPath}.
Let vertex set $M$ consist of each $x\in P$ such that $\WSA(x) > 1 - \varepsilon/2$.
Also, let $m=|M|$. 
Using \eqref{eq:SQRSphereVsWorstPath} and \Cref{cor:HeavyWABound}, we have
\begin{align}
\sqwsphere(v,\ell) 
&\le \left(1-{\varepsilon}/{2} \right)^{\ell-m} \prod\nolimits_{x \in M} \WSA(x) \label{eq:SQRstepRedaux}\\
&\le \left(1-{\varepsilon}/{2}\right)^{\ell-m} \cdot d^{-m} 
\cdot \left(1-{\varepsilon}/{4}\right)^{m-\ell} \nonumber \\
&\le \left(\frac{1-\varepsilon/{2}}{1-\varepsilon/{4}}\right)^{\ell-m}
\cdot d^{-m} \le \left(1-{\varepsilon}/{4}\right)^{\ell}
\enspace, \label{eq:SQRfinStepTr}
\end{align}
where the last inequality is valid for large enough $d$.

We proceed with the case where $B_{v,\ell}$ is a unicyclic graph. Let $C=x_1, \ldots, x_k$ be 
the unique cycle of $B_{v,\ell}$.  W.l.o.g., assume that $x_1$ is the closest vertex of $C$ to 
vertex $v$.  Let $P$ be the path from $x_1$ to $v$, while let $\{z, x_1\}$ be the first edge of 
the path. Letting $|P| =r $, recall that $r \ge \mincycdist$.

For $\ell \le r$, $B_{v,\ell}$ is a tree. In this case,  the bound on $\sqwsphere(v,\ell)$ follows 
by repeating the previous arguments.

Suppose now that $\ell > r$.
Repeating the step in \eqref{eq:OneStepIndTH} $r$ times in total, we obtain 
\begin{align}\label{eq:SQRSphereVsWorstPathCy}
\sqwsphere(v,\ell) &\leq \max\nolimits_{Q=y_1,\ldots, y_r} 
\prod\nolimits_{i\leq r-1} \WSA(y_i) \cdot \sqwsphere(y_{r},\ell-r) \enspace,
\end{align}
where $Q$ varies over all paths of length $r$ that emanate from vertex $v$. If the maximum in 
the r.h.s. of \eqref{eq:SQRSphereVsWorstPathCy} is  attained at a path different from $P$, then 
all extensions of this path avoid the cycle $C$ and therefore we use the same argument as in the tree
 case above, and obtain the same bound on $\sqwsphere(v,\ell) $ as in \eqref{eq:SQRfinStepTr}.

Suppose now that path $P$ maximises the  r.h.s. of \eqref{eq:SQRSphereVsWorstPathCy}. From 
\Cref{lemma:BoundOnWSphereR} part (a), we have 
\begin{align}\label{eq:KimbarCYclSQR}
\sqwsphere(y_{r},\ell-r) % 
\leq (\log n)^3\cdot {\upkappa}^{\ell-r}\enspace. 
\end{align}
Plugging \eqref{eq:KimbarCYclSQR} into \eqref{eq:SQRSphereVsWorstPathCy} yields 
\begin{align}\label{eq:SQRSphereVsWorstPathCyC}
\sqwsphere(v,\ell) &\leq (\log n)^3\cdot {\upkappa}^{\ell-r}\cdot 
\prod\nolimits_{x\in P} \WSA(x)
 \enspace. 
\end{align}
Working as in \eqref{eq:SQRfinStepTr} above, we further obtain
\begin{align}%\label{eq:SQRSphereVsWorstPathCyCC}
\prod\nolimits_{x\in P} \WSA(x) &\leq \left(1-{\varepsilon}/{4}\right)^{r}
\enspace.\nonumber
\end{align}
Since $\varepsilon$ is such that $1-\varepsilon/8 \ge \upkappa$,
plugging the above into \eqref{eq:SQRSphereVsWorstPathCyC}, 
we obtain
\begin{align}\label{eq:UseOfR2BoundSQR1963}
\sqwsphere(v,\ell) &\leq (\log n)^3\cdot \left(1-{\varepsilon}/{4}\right)^{r} \cdot {\upkappa}^{\ell-r} % 
\le (\log n)^3\cdot \left(1-{\varepsilon}/{8}\right)^{2r} 
\cdot {\upkappa}^{\ell-r} 
\le \left(1-{\varepsilon}/{8}\right)^{\ell}
\enspace,
\end{align}
where in the last inequality we use that $r \ge \mincycdist$ and hence 
$(\log n)^3\left(1-{\varepsilon}/{8}\right)^{r}<1$. 

All the above conclude the proof of \Cref{lemma:BoundOnWSphereR} part (b).
\end{proof}

%\spreadpoint

\section{Proof of \Cref{prop:SQSphereAllPathTree}}
\label{sec:prop:SQSphereAllPathTree}

\subsection{Proof of \Cref{prop:SQSphereAllPathTree}, part (a)}

If block $B$ is a tree,  $B$ and $\APTree(B,\rootB)$ are identical. Then,  for any $w\in V(T)$, the 
bound on $\sqwsphere_T(w,\ell)$ can be obtained using part (a) of \Cref{lemma:BoundOnWSphereR}.

Consider now the case where block $B$ is unicyclic. For such block $B$ and vertex 
$\rootB\in \partial_{\rm in}B$,  consider the all-path tree $T=\APTree(B,\rootB)$. Also 
consider the  subtrees   $\Tbase$, $\Tright$ and $\Tleft$ in $T$, i.e., as they are defined 
in \Cref{sec:AllPathTreeDefinition}.  Similarly, consider the vertices $\Cleft$ and $\Cright$ in $\Tbase$.

We  focus on bounding $\sqwsphere_T(w,\ell)$, for $w\in V(T)$ in $T$. W.l.o.g.,  assume that  
$w \in \Tbase$.   Let $\sqwsphere_0(w,\ell)$ be the contribution to $\sqwsphere_T(w,\ell)$ coming 
from paths  that are restricted to tree $\Tbase$. The paths in $\Tbase$ that are not considered for 
$\sqwsphere_0(w,\ell)$ can be divided into two groups:
\begin{description}
\item[Group 1] This group consists of each path $P=y_{1},\ldots, y_{\ell}$ in $\APTree(B,\rootB)$ such that 
$y_1=w$, and for appropriate $j$, $P'=y_{1},\ldots, y_{j}$ corresponds to the path from the root $w$ to vertex $\Cright$, while the subpath 
$P''=y_{j}, \ldots, y_{\ell}$ is in  $\Tright$. 

\item[Group 2] The group consists of each path $P=y_{1},\ldots, y_{\ell}$ in $\APTree(B,\rootB)$
such that $y_1=w$, and for appropriate $j$, $P'=y_{1},\ldots, y_{j}$ corresponds 
to the path from the root $w$ to vertex $\Cleft$, while the subpath 
$P''=y_{j}, \ldots, y_{\ell}$ is in subtree $\Tleft$. 
\end{description}
Let $\sqwsphere_1(w,\ell)$ be the contribution to $\sqwsphere_T(w,\ell)$ from the paths in Group 1. Similarly, we have $\sqwsphere_2(w,\ell)$ for the paths
in Group 2. 

Arguing as in the proof of part (a) of \Cref{lemma:BoundOnWSphereR}, we have 
\begin{align}\nonumber
\sqwsphere_0(w,\ell) &\leq 3\cdot \maxDelta \cdot {\upkappa}^\ell\enspace. 
\end{align} 
Also, with similar arguments we obtain
\begin{align}\nonumber
\sqwsphere_1(w,\ell) &\leq 3\cdot \maxDelta \cdot \upkappa^\ell & \textrm{and} && \sqwsphere_2(w,\ell) &\leq 3\cdot \maxDelta \cdot \upkappa^\ell
\enspace .
\end{align}
Then, our result follows by noting that 
\begin{align}\nonumber
\sqwsphere_T(w,\ell) &= \sqwsphere_0(w,\ell)+\sqwsphere_1(w,\ell)+\sqwsphere_2(w,\ell)\enspace.
\end{align}
All the above conclude the proof of part (a) of \Cref{prop:SQSphereAllPathTree}.
\hfill $\Box$

\subsection{Proof of \Cref{prop:SQSphereAllPathTree}, part (b)}
Since both $w, \rootB \in \partial_{\rm in} B$ are arbitrary, we assume w.l.o.g. that $\rootB = w$, to simplify the 
notation that follows.
We consider two cases for $B$. The first corresponds to 
$B$ being a tree, while the second one corresponds to  $B$ being a unicyclic graph.

When $B$ is tree, note that for $w\in \partial_{\rm in }B$, we have that 
block $B$ and $T=\APTree(B,w)$ are identical. Such  $w$, is also an $\varepsilon$-block vertex. 
Furthermore, we have

\begin{align}\label{eq:OneStep4SQRSpherRedauxB}
\sqwsphere_T(w,\ell) 
&= \sum\nolimits_{w_i \sim w} \left|\tanh\left(\beta \cdot \gauss_{\{w,w_i\}}\right)\right|^2 
\cdot \sqwsphere_{T_i}(w_i, \ell-1) \enspace,
\end{align}
where for each $w_i$, neighbour of vertex $w$,
we let $T_i$ be the subtree of $T$ that contains $w_i$ and its descendants. Hence, 
\begin{align}\nonumber 
\sqwsphere_{T_i}(w_i, \ell-1) = \sum\nolimits_{Q} \prod\nolimits_{e \in Q} |\tanh(\beta \cdot \gauss_e)|^2
\enspace,
\end{align}
where $Q$ varies over all paths of length $\ell-1$ that emanate from $w_i$ and do not
use the edge $\{w,w_i\}$.

Note that the specification of block $B$ implies that vertex $w$ has only one
neighbour $z$ in $T$, while $\Inf_{w,z}\leq \deeppink{d^{-1/10}}$.
Combining this observation with \eqref{eq:OneStep4SQRSpherRedauxB}, we have
\begin{align}
\sqwsphere_T(w,\ell) 
&\leq \deeppink{d^{-1/10}}\cdot \sqwsphere_{T_z}(z, \ell-1) \label{eq:OneStepIndTHBoundary} \enspace. 
\end{align}
For $\sqwsphere_{T_z}(z, \ell-1)$, repeating the same argument as those we use to obtain \eqref{eq:SQRSphereVsWorstPath}  
in the proof \Cref{lemma:BoundOnWSphereR} part (b),  we get
%\begin{align}
%\sqwsphere_{T_z}(z,\ell-1) 
%&\leq \max\nolimits_{z_i}\left\{ \sqwsphere(z_i, \ell-2)\right\} \cdot 
%\sum\nolimits_{z_i \sim z} \left|\tanh\left(\beta \cdot \gauss_{\{z,z_i\}}\right)\right|^2 \nonumber\\ 
%&\leq \max\nolimits_{z_i}\left\{ \sqwsphere_{T_{z_i}}(z_i, \ell-2)\right\} \cdot \WSA(z) \enspace. \label{eq:OneStepIndTHForChildren}
%\end{align}
%Repeating the above step $\ell-1$ times in total, we obtain 
\begin{align}\label{eq:SQRSphereVsWorstPathZz}
\sqwsphere_{T_z}(z,\ell-1) &\leq \max\nolimits_{Q} \left\{ \prod\nolimits_{v\in Q} 
\WSA(v) \right\} \enspace,
\end{align}
where $Q$ in varies over all paths of length $\ell-1$ that emanate from vertex $z$ in $T_z$. Recall that 
$z$ is the single neighbour of $w$ in $T$.

Let $P$ be the path in the subtree $T_z$ that maximizes the r.h.s. in \eqref{eq:SQRSphereVsWorstPathZz}.
 Let vertex set $M$ consist of each $x\in P$ such that $\WSA(x) > (1 - \varepsilon/2)$. Also, let $m=|M|$.

Using \eqref{eq:SQRSphereVsWorstPathZz}, the fact that $z$ is next to the $\varepsilon$-block vertex $w$, 
and \Cref{cor:HeavyWABound}, we have
\begin{align}
\sqwsphere_{T_z}(z,\ell-1) 
&\le \left(1-{\varepsilon}/{2} \right)^{\ell-m-1} \prod\nolimits_{x \in M} \WSA(x) \label{eq:SQRstepRedauxForBoundary}\\
&\le \left(1-{\varepsilon}/{2}\right)^{\ell-m-1} \cdot d^{-m} 
\cdot \left(1-{\varepsilon}/{4}\right)^{m-\ell} \nonumber \\
&\le \frac{1}{d^m(1-{\varepsilon}/{4})} \cdot \left(\frac{1-\varepsilon/{2}}{1-\varepsilon/{4}}\right)^{\ell-m-1} \le \left(1-\frac{\varepsilon}{4}\right)^{\ell-1}
\enspace, \label{eq:SQRfinStepTrForBoundary}
\end{align}
where the last inequality is valid for large enough $d$.
Plugging \eqref{eq:SQRfinStepTrForBoundary} into \eqref{eq:OneStepIndTHBoundary} gives
\begin{align}
\sqwsphere_T(w,\ell) 
&\leq \deeppink{d^{-1/10}}\cdot\left(1-{\varepsilon}/{4}\right)^{\ell-1} 
\enspace.
\label{eq:OneStepIndTHBoundaryFin}
\end{align}
The above proves  \Cref{prop:SQSphereAllPathTree}, part (b),  for the case where 
$B$ is a tree.

We proceed with the case where $B$ is a unicyclic graph and $w\in \partial_{\rm in} B$.  Let 
$C=x_1, \ldots, x_k$  be the unique cycle of $B$. W.l.o.g.,  assume that $x_1$ is the closest 
vertex of $C$ to vertex $w$. Let $\{v, x_1\}$ be the first edge in the shortest 
path $P$ connecting $w$ to $x_1$. Letting the length of $P$ be $|P| =r $, we 
recall that $r \ge \mincycdist$.

For such block $B$ and vertex $w \in \partial_{\rm in}B$,  consider the all-path 
tree $T=\APTree(B,w)$. Also consider the standard subtrees of $T$, i.e., the subtrees $\Tbase$, $\Tright$ 
and $\Tleft$. Similarly,  consider the vertices $\Cleft$ and $\Cright$ in $\Tbase$.

The assumption that $w \in \partial_{\rm in} B$ implies immediately the following result.

\begin{corollary}\label{cor:HeavyWABoundNear}
Let $Q$ be any path in $T=\APTree(B,w)$ of length $\ell$, which emanates from the root and does 
not contain copies of the vertices in $C$.  Let also $M$ be the set of heavy vertices in $Q$, i.e., for 
all $x\in M$ we have $\WSA(x)>1-\varepsilon/2$. The definition of $\varepsilon$-block vertices implies
\begin{align*}
\prod\nolimits_{x\in M}\WSA(x) \leq d^{-|M|} \cdot (1-\varepsilon/4)^{|M|-\ell-1}\enspace.
\end{align*}
\end{corollary}
Furthermore, working as in \eqref{eq:OneStepIndTHBoundary}, we have 
\begin{align}
\sqwsphere_T(w,\ell) 
&\leq \deeppink{d^{-1/10}}\cdot \sqwsphere_{T_z}(z, \ell-1) \label{eq:OneStepIndTHBoundaryUnCycl}
\end{align}
where $z$ is the unique neighbour of $w$ in $T$.  We also have
\begin{align} 
\label{eq:OneStep4SQRSpherRedauxBB}
\sqwsphere_{T_z}(z,\ell-1) 
&= \sum\nolimits_{z_i \sim z} \left|\tanh\left(\beta \cdot \gauss_{\{z,z_i\}}\right)\right|^2  \cdot \sqwsphere_{T_{z_i}}(z_i, \ell-2) \enspace,
\end{align}
where for each $z_i$, we have
\begin{align}\nonumber 
\sqwsphere_{T_{z_i}}(z_i, \ell-2) = \sum\nolimits_{Q} \prod\nolimits_{e \in Q} |\tanh(\beta \cdot \gauss_e)|^2 \enspace,
\end{align}
while $Q$ varies over all paths $Q$ in $T_{z_i}$ of length $\ell-2$ that emanate from $z_i$ and do 
not use the edge $\{z,z_i\}$. 
Using standard argument from previous proofs, \eqref{eq:OneStep4SQRSpherRedauxBB} implies
\begin{align}
\sqwsphere_{T_z}(z,\ell-1) 
% &\leq \max\nolimits_{z_i}\left\{ \sqwsphere_{T_{z_i}}(z_i, \ell-2)\right\} \cdot 
% \sum\nolimits_{z_i \sim z} \left|\tanh\left(\beta \cdot \gauss_{\{z,z_i\}}\right)\right|^2 \nonumber\\ 
&\leq \max\nolimits_{z_i}\left\{ \sqwsphere_{T_{z_i}}(z_i, \ell-2)\right\} \cdot \WSA(z) \enspace. \label{eq:OneStepIndTHBB}
\end{align}
For $\ell \le r$, we obtain the same bound on $\sqwsphere_{T_{z}}(z,\ell-1)$ and
$\sqwsphere_{T}(w,\ell)$ as in the tree case, i.e., as in \eqref{eq:OneStepIndTHBoundaryFin}. 

For $\ell >r$, repeating the above step $r$ times,  we obtain 
\begin{align}\label{eq:SQRSphereVsWorstPathCyBB}
\sqwsphere_{T_z}(z,\ell-1) &\leq \max\nolimits_{Q=y_1,\ldots, y_r} 
\prod\nolimits_{i\leq r-1} \WSA(y_i) \cdot \sqwsphere_{T_{y_r}}(y_{r},\ell-r) \enspace,
\end{align}
where $Q$ in varies over all paths of length $r-1$ that emanate from vertex $z$.

If the maximum in the rhs of \eqref{eq:SQRSphereVsWorstPathCyBB} is attained at a path different 
from $P$,  then all extensions of this path avoid the copies of the vertices in cycle $C$, and therefore 
with an argument  similar to those we use above, we obtain the same bound on $\sqwsphere_T(w,\ell) $ 
as in \eqref{eq:OneStepIndTHBoundaryFin}.

If the maximum in rhs of \eqref{eq:SQRSphereVsWorstPathCyBB} is 
attained by path $P$, then, using \Cref{lemma:BoundOnWSphereR} part (a), we have 
\begin{align}\label{eq:KimbarCYclSQRBB}
\sqwsphere_{T_{y_r}}(y_{r},\ell-r) 
% \leq \sqwsphere(C, \ell-r) 
\leq (\log n)^3\cdot {\upkappa}^{\ell-r}\enspace. 
\end{align}
Plugging \eqref{eq:KimbarCYclSQRBB} into \eqref{eq:SQRSphereVsWorstPathCyBB} yields 
\begin{align}\label{eq:SQRSphereVsWorstPathCyCBB}
\sqwsphere_{T_z}(z,\ell-1) &\leq (\log n)^3\cdot {\upkappa}^{\ell-r}\cdot 
\prod\nolimits_{x\in P} \WSA(x)
 \enspace. 
\end{align}
Using \Cref{cor:HeavyWABoundNear}, we further obtain
\begin{align}
\prod\nolimits_{x\in P} \WSA(x) &\leq \left(1-{\varepsilon}/{4}\right)^{r}
\enspace.\nonumber
\end{align}
Since $1-\varepsilon/8 \ge \upkappa$,
plugging the above into \eqref{eq:SQRSphereVsWorstPathCyCBB} and working 
as in \eqref{eq:UseOfR2BoundSQR1963} we obtain
\begin{align}
\sqwsphere_{T_z}(z,\ell-1) 
% &\leq (\log n)^3\cdot \left(1-\frac{\varepsilon}{4}\right)^{r} \cdot {\upkappa}^{\ell-r} \nonumber \\
&\le (\log n)^3\cdot \left(1-\varepsilon/8\right)^{2r} 
\cdot {\upkappa}^{\ell-r} \leq (1-\varepsilon/8)^{\ell} 
\enspace, \nonumber
\end{align}
where in the last inequality we use that $r \ge \mincycdist$ and hence 
$(\log n)^3\left(1-\frac{\varepsilon}{8}\right)^{r}<1$. 
% Also we use the assumption that $1-\varepsilon/8 \ge \upkappa$.
% 
We obtain the desirable bound on $\sqwsphere_T(w,\ell) $ by plugging the above inequality
into \eqref{eq:OneStepIndTHBoundaryUnCycl}.

All the above conclude the proof of \Cref{prop:SQSphereAllPathTree} part (b).
% \end{proof}
\hfill $\Box$

\subsection{Proof of \Cref{prop:SQSphereAllPathTree}, part (c)}
We consider two cases for $B$. The first corresponds to  $B$ being a tree, while the second one 
corresponds to $B$ being a unicyclic graph.

For the case of $B$ being a tree, we obtain the desired bound 
\begin{align}
\sqwsphere_{\partial T}(w,\ell) 
&\leq \deeppink{d^{-1/10}}\cdot\left(1-{\varepsilon}/{4}\right)^{\ell-1} 
\enspace.
\label{eq:OneStepIndTHBoundaryFinBound}
\end{align}
by following similar arguments to those we use to obtain \eqref{eq:OneStepIndTHBoundaryFin}.

Let us now focus on the case $B$ being a unicyclic graph. Consider $T=\APTree(B,\rootB)$ 
for some $\rootB \in \partial_{\rm in} B$, while let $w\in V(T)$. 

Letting $\Lambda$ be the set of leaves of $T$, we distinguish the following cases: 
\begin{description}
\item[Case 1.a] $\dist_{T}(w,\Lambda) > \mincycdistBD $ and $\ell < \mincycdistBD$,
\item[Case 1.b] $\dist_{T}(w,\Lambda) > \mincycdistBD $ and $\ell \geq \mincycdistBD$,
\item[Case 2.a]$\dist_{T}(w,\Lambda) \leq \mincycdistBD $ and $\ell < (60/\varepsilon) \cdot \log\log n$,
\item[Case 2.b]$\dist_{T}(w,\Lambda) \leq \mincycdistBD $ and $\ell \geq (60/\varepsilon) \cdot \log\log n$.
\end{description}\label{eq:Case1aWVsBoundarySQR}

Since  the copies of $\partial_{\rm in }B$ in $T$ belong to set $\Lambda$,  for Case 1.a, we have 
\begin{align}\label{eq:SQRTPartTCas1A}
\sqwsphere_{\partial T}(w,\ell) &=0\enspace.
\end{align}

For Case 1.b, we have 
\begin{align}\label{eq:SQRTPartTCas1b}
\sqwsphere_{\partial T}(w,\ell) &\leq d^{-1/10}\cdot \sqwsphere_T(w,\ell-1)\leq 10 \cdot d^{-1/10}\cdot \maxDelta \cdot{\upkappa}^{\ell-1}
\end{align}
where for the first inequality,  we recall that each copy of $\partial_{\rm in}B$ in $T$ is of
degree 1, while the coupling of the edge that connects it to the rest of the tree is at most $d^{-1/10}$. 
The second inequality is due to \Cref{prop:SQSphereAllPathTree}, part (a). 

Since $\varepsilon$ is such that $1-\varepsilon/8 \ge \upkappa$, for the rightmost quantity
in  \eqref{eq:SQRTPartTCas1b}, we have 
\begin{align}
 10 \cdot d^{-1/10}\cdot \maxDelta \cdot{\upkappa}^{\ell-1} 
&\leq d^{-1/10}\cdot \log n \cdot \left(1- {\varepsilon}/{8}\right)^{\ell-1} 
\le d^{-1/10}\cdot 
\log n\cdot \left(1-{\varepsilon}/{16}\right)^{2(\ell-1)} \nonumber \\
 & \leq d^{-1/10}\cdot (1-\varepsilon/16)^{\ell-1} 
\enspace, \nonumber 
%\label{eq:SQRTPartialTEatLogs}
\end{align}
where in the last inequality we use that $\ell \geq \mincycdistBD$, and hence 
$\log n \cdot \left(1-\frac{\varepsilon}{16}\right)^{\ell-1}<1$. 
Plugging the above inequality into \eqref{eq:SQRTPartTCas1b} we obtain
\begin{align} \label{eq:Case1bWVsBoundarySQR}
\sqwsphere_{\partial T}(w,\ell) &\leq d^{-1/10}\cdot (1-\varepsilon/16)^{\ell-1} \enspace. 
\end{align}

For Case 2.a, since $\ell \le (60/\varepsilon) \cdot \log\log n$, no path of length $\ell$ intersects 
with the copies of the vertices in the cycle of $B$. Recall that the cycle is at distance $r>\mincycdist$ 
from $w$. 
Employing the same arguments as those used for Case 1.b, we get
\begin{align}\label{eq:Case2aWVsBoundarySQR}
\sqwsphere_{\partial T}(w,\ell) &\leq d^{-1/10}\cdot (1-\varepsilon/16)^{\ell-1} \enspace. 
\end{align}
Consider now Case 2.b. Let $\sqwsphere_{1}(w,\ell)$ be the contribution to $\sqwsphere_{\partial T}(w,\ell)$ 
from the paths of length $\ell$  that do not intersect copies of vertices in the cycle of block $B$. Similarly, let 
$\sqwsphere_{2}(w,\ell)$ be the contribution to $\sqwsphere_{\partial T}(w,\ell)$ from paths of length $\ell$ that 
contain copies of vertices in the cycle of block $B$. We have
\begin{align}\label{eq:SQRTPartTCas2BBase}
\sqwsphere_{\partial T}(w,\ell) \le \sqwsphere_{1}(w,\ell)
+ \sqwsphere_{2}(w,\ell) \enspace.
\end{align}
Working as in Case 2.a., we have 
\begin{align}\nonumber
\sqwsphere_{1}(w,\ell) \le d^{-1/10}\cdot (1-\varepsilon/16)^{\ell-1} \enspace.
\end{align}
Furthermore,  working as in Case 1.b., we get 
\begin{align}\nonumber
\sqwsphere_{2}(w,\ell) \le d^{-1/10}\cdot (1-\varepsilon/16)^{\ell-1} \enspace.
\end{align}
Plugging the above two inequalities into \eqref{eq:SQRTPartTCas2BBase}, implies
\begin{align}\label{eq:Case2bWVsBoundarySQR}
\sqwsphere_{\partial T}(w,\ell) \le 2 \cdot d^{-1/10}\cdot (1-\varepsilon/16)^{\ell-1} \enspace.
\end{align}
 From \eqref{eq:SQRTPartTCas1A},
\eqref{eq:Case1bWVsBoundarySQR}, \eqref{eq:Case2aWVsBoundarySQR} and \eqref{eq:Case2bWVsBoundarySQR},
we conclude that when $B$ is a unicylic block, we have
\begin{align}\label{eq:CasesALLWVsBoundarySQRUniq}
\sqwsphere_{\partial T}(w,\ell) \le 2 \cdot d^{-1/10}\cdot (1-\varepsilon/16)^{\ell-1} \enspace.
\end{align}
Part (c) in \Cref{prop:SQSphereAllPathTree}, follows from \eqref{eq:CasesALLWVsBoundarySQRUniq}
and \eqref{eq:OneStepIndTHBoundaryFinBound}.
\hfill $\Box$

\spreadpoint
\newcommand{\blckGibbs}{\upmu}
\newcommand{\ShiftMT}{\red{\UpS}}
\newcommand{\NormMT}{\blue{\UpN}}
\newcommand{\RescaleMT}{\iris{\UpR}}
\newcommand{\MontDMatrix}{\UpD}

\newpage

\newcommand{\sqrBetaJSphere}{{\tt JBS}}

\section{Bounds on the m-LSI of $\mu_1$ -  Proof of \Cref{thm:BoundMLSIonH}} \label{sec:thm:BoundMLSIonH}

For each vertex $v\in V_n$ and  integer $\ell\geq 0$, we let $\sqrBetaJSphere(v,\ell)$ be defined by
\begin{align}\label{eq:DefOfSQW}
\sqrBetaJSphere(v,\ell)&=\sum\nolimits_{P=x_{0},\ldots, x_{\ell}} |\beta\cdot \InAct_B(x_{\ell-1},x_{\ell})|^2 \cdot \prod\nolimits_{0\leq j<\ell-1}|\tanh\left(\beta\cdot \InAct(x_j,x_{j+1})\right)|^2 \enspace,
\end{align}
where variable $P$ in the summation varies over all paths of length $\ell$ in $\G(n,d/n)$, emanating from $v$.

\begin{lemma}\label{lemma:Tail4HeavySQSphsere}
For any $\varepsilon \in (0,1)$ satisfying $1-\varepsilon/8\geq \upkappa$, there exists $d_0 = d_0(\varepsilon) > 0 $, such that for any 
$d \ge d_0$, $\beta \le \beta_c(d)$, and $0 \le \ell \le \maxEll$ the following is true:

With probability $1-\red{o(1)}$ over the instances of matrix $\InAct = \InAct(\G(n,d/n), \{\gauss_e\})$ defined as in \eqref{eq:InteractG(np)}, 
there exists an $\varepsilon$-block partition $\cB$ such that  
 for every multi-vertex block $B\in \cB$, for any $v\in V(B)$ we have 
$\sqrBetaJSphere(v,\ell)\leq 3\cdot \maxDelta\cdot\upkappa^\ell$, where $\maxDelta=\maxDeltaVal$ and the quantity $\upkappa$ is from \eqref{eq:defOfBc}. 
\end{lemma}

In light of the tail bound in \Cref{thm:GWSphereTailHeavy}, the proof of \Cref{lemma:Tail4HeavySQSphsere} is identical to the
part (a) of \Cref{lemma:BoundOnWSphereR}.  For this reason we omit it.

Let  $\cE$ be the event that
\begin{enumerate}[(a)]
\item graph $\G(n,d/n)$ admits an $\varepsilon$-block partition,
\item  for every multi-vertex block $B\in \cB$ and  $v\in V(B)$,  we have
$\sqrBetaJSphere(v,\ell)\leq 3\cdot \maxDelta\cdot\upkappa^\ell$, for $\ell \le \maxEll$, 
\item for any $u,w\in V(\G)$, we have $|\InAct(u,w)|\leq 10 \sqrt{\log n}$.
\end{enumerate}
In light of \Cref{theorem:HBlockExistence}, \Cref{lemma:Tail4HeavySQSphsere}, and standard tail bounds for the Gaussian distribution,  we have that $\Pr[\cE]=1-\red{o(1)}$. 
For the rest of  this proof, assume that the event $\cE$ holds.  

It is useful to mention that the construction of the $\varepsilon$-block partition $\cB$ guarantees that for any 
multi-vertex block $B\in \cB $ and any two vertices $u,w \in B$ there is no path that connects them with length $>\frac{\log n}{(\log d)^2}$. 

For each block $B \in \cB$  let $\blckGibbs_{\InAct_B,\beta,\Field}$ be  the 2-spin model on  $B$ with edge weights 
specified by matrix $\InAct$ and  inverse temperature $\beta \le \beta_c(d)$ with arbitrary external field $\Field\in \mathbb{R}^{V(B)}$.
We let $\GDmlogSob(\blckGibbs_{\InAct_B,\beta,\Field})$ be the modified log-Sobolev constant for the single-site Glauber dynamics on $\blckGibbs_{\InAct_B,\beta,\Field}$.

 \Cref{thm:BoundMLSIonH} follows by showing 
\begin{align}\label{eq:Target4thm:BoundMLSIonHA}
\GDmlogSob(\blckGibbs_{\InAct_B,\beta,\Field}) &\geq \textstyle  |V(B)|^{-1}\cdot n^{-\frac{820}{\sqrt{d}}}
&\forall B\in \cB \enspace.
\end{align}
Specifically, combining the above relation with the standard fact in \cite[Lemma 2.5]{goel2004modified},  we get
\begin{align}\nonumber
\GDmlogSob(\mu_{\InAct_H, \beta,\Field}) &\geq n^{-1}\cdot \min\nolimits_{B\in \cB} \left\{ \GDmlogSob(\blckGibbs_{\InAct_B,\beta,\Field}) \cdot |V(B)|\right\} \geq  
n^{-\left(1+\frac{820}{\sqrt{d}}\right)}  \enspace. 
\end{align}
We first show that when  $B$ is a treelike block in $\cB$,  or a Breadth First Search tree of a unicyclic block in $\cB$, then we have
\begin{align}\label{eq:Target4thm:BoundMLSIonHATree}
\GDmlogSob(\blckGibbs_{\InAct_B,\beta,\Field}) \geq |V(B)|^{-1}\cdot n^{-\frac{816}{\sqrt{d}}} \enspace.
\end{align}
We obtain a bound for $\GDmlogSob(\blckGibbs_{\InAct_B,\beta,\Field})$ when $B\in \cB$ is unicyclic from the following claim.

\begin{claim}
On the event $\cE$, for any unicyclic block $B\in \cB$,  we have
\begin{align}\label{eq:Target4thm:BoundMLSIonHAUnicyclic}
\GDmlogSob(\blckGibbs_{\InAct_B,\beta,\Field}) & \geq |V(B)|^{-1}\cdot n^{-\frac{816}{\sqrt{d}}-\frac{40\beta}{\sqrt{\log n}}}
 \enspace.
\end{align}
\end{claim}
\begin{proof}
Fix unicyclic block $B\in \cB$, In what follows, we abbreviate $\blckGibbs_{\InAct_B,\beta,\Field}$ to $\blckGibbs$.
Also,  let $T$ be a BFS tree of $B$, while let $\upnu=\upnu_{\InAct_T,\beta,\Field}$.  From \eqref{eq:Target4thm:BoundMLSIonHATree}
we have
\begin{align}\label{eq:Target4thm:BoundMLSIonHATreeAgain}
\GDmlogSob(\blckGibbs_{\InAct_T,\beta,\Field}) \geq |V(B)|^{-1}\cdot n^{-\frac{816}{\sqrt{d}}} \enspace.
\end{align}
Let $e$ be the extra edge that $B$ has compared to $T$. Since we assume the event $\cE$, the coupling at edge $e$
is at most $10\sqrt{\log n}$. This implies that for any configuration $\sigma\in \{\pm 1\}^{V(B)}$,  we have 
\begin{align}\label{eq:MLSIContPerturbation}
{\textstyle \exp(-20\cdot \beta\cdot\sqrt{\log n})\leq \frac{\upnu(\sigma)}{\blckGibbs(\sigma)} \leq \exp(20\cdot \beta\cdot \sqrt{\log n})}\enspace.
\end{align}
Then, \eqref{eq:Target4thm:BoundMLSIonHAUnicyclic} follows from \eqref{eq:Target4thm:BoundMLSIonHATreeAgain},  \eqref{eq:MLSIContPerturbation} and 
the well-known Holley-Stroock perturbation principle \cite{holley1986logarithmic} (see \Cref{Lemma:HolleyStroockPerturb}, in \Cref{sec:Lemma:HolleyStroockPerturb}).
\end{proof}

Clearly,  \eqref{eq:Target4thm:BoundMLSIonHA} follows from \eqref{eq:Target4thm:BoundMLSIonHAUnicyclic} and \eqref{eq:Target4thm:BoundMLSIonHATree}.

We now focus on showing that \eqref{eq:Target4thm:BoundMLSIonHATree} is true. 
Fix $B$ which  is either a treelike block in $\cB$,  or a Breadth First Search tree of a unicyclic block in $\cB$.
We utilise a stochastic localisation scheme for $\blckGibbs_{\InAct_B,\beta,\Field}$. 
That is,  we have the process $\{ \blckGibbs_t\}_{t\in [0,1]}$, where $\blckGibbs_0=\blckGibbs_{\InAct_B,\beta,\Field}$.
Before specifying  the control matrix $\CtrlMT$ for this scheme,  we  introduce some further matrices. 

We let the $V(B)\times V(B)$, diagonal  matrix  $\RescaleMT$ be such that 
\begin{align}\label{eq:RescaleMatrix}
\RescaleMT(w,w) &= \sqrt{\textstyle 1+\sum\nolimits_{u\sim w} \left( \beta\cdot \InAct_{B}(w,u)\right)^2}, 
&\forall w\in V(B) \enspace.
\end{align}
Also, let $\widehat{\InAct}=\beta\cdot \RescaleMT^{-1}\cdot \InAct \cdot \RescaleMT^{-1}$. Note that all entries in matrix $\widehat{\InAct}$ are in the interval $(-1,1)$. 
There is a collection of weights $\{\UpQ_{i,j}\}_{i,j \in V(B)}\in (-1,1) $ such that 
\begin{align}\label{eq:HatJVsQWeights}
\widehat{\InAct}(w,u)&= \frac{\UpQ(w,u)}{1-(\UpQ(w,u))^2},  & \forall u,w\in V(B) \enspace. 
\end{align}
The uniqueness of $\UpQ_{i,j}$'s,  follows from the observation that the range of function $f(x)=\frac{x}{1-x^{2}}$ for $x\in (-1,1)$, is $\mathbb{R}$, while $f$ is 
one to one and onto, i.e., since it is monotonically increasing. 

Also, let the $V(B)\times V(B)$, diagonal matrix $\MontDMatrix$ be such that  
\begin{align}
\MontDMatrix(w,w) &=\sum_{u \sim w}\frac{(\UpQ(w,u))^2}{1-(\UpQ(w,u))^2},  & \forall w\in V(B) \enspace,
\end{align}
In light of the above definitions,  the control matrix $\CtrlMT$ is given by
\begin{align}
\CtrlMT&=\left( \beta^{-1}\cdot \RescaleMT^2+ \beta^{-1}\cdot \RescaleMT \cdot \MontDMatrix\cdot \RescaleMT+\InAct_B\right)^{1/2}\enspace. 
\end{align}
Note that matrices $\RescaleMT^2$ and $\RescaleMT \cdot \MontDMatrix \cdot \RescaleMT$ are diagonal.  The following result implies that $\CtrlMT$ is well-defined. 
\begin{lemma}\label{prop:CmlsiMu1ShiftMatrix}
We have $\beta^{-1} \cdot \RescaleMT^2+ \beta^{-1}\cdot \RescaleMT \cdot \MontDMatrix\cdot \RescaleMT+\InAct_B\pdge 0$.
\end{lemma}
The proof of \Cref{prop:CmlsiMu1ShiftMatrix} appears in \Cref{sec:prop:CmlsiMu1ShiftMatrix}.

For the  stochastic localisation scheme we consider,  a.s.  $\blckGibbs_1$ is a product distribution,  and hence, we have
\begin{align}
\GDmlogSob(\blckGibbs_1)\geq  |V(B)|^{-1} \enspace. 
\end{align}
%\charis{Is it that \eqref{eq:Target4STMU1} all we need?}

In light of \Cref{thrm:LocBounMLSI}, and the above inequality, \eqref{eq:Target4thm:BoundMLSIonHATree} follows by showing that for 
$\beta \le \beta_c(d)$, any external field $\Field\in \mathbb{R}^{V(B)}$ and any  $t\in [0,1]$, we have
\begin{align} \label{eq:Target4STMU1}
\beta \cdot 
\opnorm{ \CtrlMT \cdot \Cov (\InAct_t,\beta, \Field) \cdot \CtrlMT }  \leq 816 \cdot L \enspace, 
\end{align}
where $L=\frac{\log n}{\sqrt{d}}$.

For a parameter $\delta>0$ which we specify later, we define the weight function $\badmormw: V(B) \to [0, +\infty)$ for each 
vertex of $B$ as follows:

For the root $r$ of block $B$, we define $\badmormw(r) = 1$. For every vertex $z$ with $\badmormw(z)$ being determined, we define 
$\badmormw(u) = \badmormw(z)\cdot(1+\delta) \cdot \Inf_{z,u}$ for all children $u$ of $z$. Recall that $\Inf_{z,u}=|\tanh(\beta\cdot \InAct_B(z,u))|$.

We let the $V(B) \times V(B)$,   diagonal  matrix $\NormMT$   be such that 
\begin{align}\nonumber 
\NormMT(w,w)&=\badmormw(w), &\forall w \in V(B)\enspace.
\end{align}
Abbreviating $\Cov (\InAct_t,\beta, \Field)$ to $\Cov_t$,  we have 
\begin{align}
%\lefteqn{
\opnorm{ \CtrlMT \cdot \Cov_t \cdot \CtrlMT}
%} \hspace{1cm} \nonumber \\ 
&\leq  \| (\CtrlMT \cdot \NormMT)^{-1}\cdot \CtrlMT \cdot \Cov_t \cdot \CtrlMT\cdot (\CtrlMT\cdot \NormMT) \|_{\infty} 
%\nonumber\\&
=\| \NormMT^{-1}\ \cdot \Cov_t \cdot \CtrlMT^2\cdot  \NormMT \|_{\infty} \nonumber\\
&\leq \| \NormMT^{-1}\ \cdot \Cov_t \cdot \InAct_B \cdot  \NormMT \|_{\infty} 
+\beta^{-1}\cdot \| \NormMT^{-1}\ \cdot \Cov_t\cdot (\RescaleMT^2+ \RescaleMT \cdot \MontDMatrix \cdot \RescaleMT) \cdot  \NormMT \|_{\infty} \enspace. 
\label{eq:Base4NormBound4mLSIMU1}
\end{align}
We establish \eqref{eq:Target4STMU1} by bounding appropriately each one of the two norms  in \eqref{eq:Base4NormBound4mLSIMU1}.

\begin{lemma}\label{lemma:NormBoundA4mLSI4MU1}
On event $\cE$, for $\updelta=\upkappa^{-3/4}-1$,  we have  $\| \NormMT^{-1}\ \cdot \Cov_t  \cdot \InAct_B \cdot  \NormMT \|_{\infty}  \leq 10\cdot \beta^{-1} 
\cdot (1-\upkappa^{1/4})^{-2}\cdot L$.
\end{lemma}

\begin{lemma}\label{lemma:NormBoundB4mLSI4MU1}
On event $\cE$, for $\updelta=\upkappa^{-3/4}-1$,  we  have  $\| \NormMT^{-1}\ \cdot \Cov_t  \cdot (\RescaleMT^2+ \RescaleMT \cdot \UpD\cdot \RescaleMT)\cdot  \NormMT \|_{\infty}   \leq 
60\cdot  (1-\upkappa^{1/4})^{-2}\cdot L $.
\end{lemma}

Recall from \eqref{eq:defOfBc} that $\upkappa = 1/4$.  Plugging the bounds from \Cref{lemma:NormBoundA4mLSI4MU1,lemma:NormBoundB4mLSI4MU1}, into \eqref{eq:Base4NormBound4mLSIMU1}
, 
we get \eqref{eq:Target4STMU1}. 

All the above conclude the proof of \Cref{thm:BoundMLSIonH}. \hfill $\Box$

\subsection{Proof of \Cref{prop:CmlsiMu1ShiftMatrix}}\label{sec:prop:CmlsiMu1ShiftMatrix}

We prove \Cref{prop:CmlsiMu1ShiftMatrix} by applying  \Cref{cor:PDBetheHessian}.
In particular, for $t\in \mathbb{R}$,  we let the  $V(B)\times V(B)$ matrix $\MontA_t$ be such that
$\MontA_t(w,u) =\frac{t\cdot \UpQ(w,u)}{1-(t\cdot \UpQ(w,u))^2}$, for all $u,w\in V(B) $. 
Similarly, we define the  $V(B)\times V(B)$, diagonal matrix $\UpD_t$  such that $\UpD_t(w,w)=\sum_{u}\frac{(t\cdot \UpQ(w,u))^2}{1-(t\cdot \UpQ(w,u))^2}$.
For both matrices $\MontA_t$ and $\UpD_t$, recall that   $\UpQ(w,u)$'s are the weights from \eqref{eq:HatJVsQWeights}.

Applying  \Cref{cor:PDBetheHessian} for matrices $\MontA_t, \UpD_{t}$, with $t=-1$, we get
$\Id+\UpD_{-1} - \MontA_{-1} \pdge 0$. But since $\UpD_{-1}=\MontDMatrix$ and $-\MontA_{-t}=\widehat{\InAct}$, we conclude that
$\Id+\UpD +\widehat{\InAct}\pdge 0$.

Note that we can apply \Cref{cor:PDBetheHessian} for $t=-1$ since we assume that $B$ is a tree, and hence
the corresponding weighted non-backtracking matrix that arise in the analysis is nilpotent.  

Recalling the definition of the diagonal matrix $\RescaleMT$, the above implies    that  $\RescaleMT\cdot \left(\Id+\UpD + \widehat{\InAct} \right)\cdot \RescaleMT  \pdge 0$.
Then,  recalling that $\beta\cdot \InAct_B=\RescaleMT\cdot \widehat{\InAct}\cdot \RescaleMT$, we have
\begin{align}\nonumber
0\pdle \beta^{-1}\cdot \RescaleMT\cdot \left(\Id+\UpD + \widehat{\InAct} \right)\cdot \RescaleMT  =\beta^{-1}\cdot \RescaleMT^2+ \beta^{-1}\cdot \RescaleMT \cdot \UpD\cdot \RescaleMT+\InAct_B \enspace. 
\end{align}
This concludes the proof of \Cref{prop:CmlsiMu1ShiftMatrix}. 
\hfill $\Box$

\subsection{Proof of \Cref{lemma:NormBoundA4mLSI4MU1,lemma:NormBoundB4mLSI4MU1}}

Before proving \Cref{lemma:NormBoundA4mLSI4MU1,lemma:NormBoundB4mLSI4MU1}, let us introduce 
some terminology which is common on both proofs. 

For integer $\ell \geq  0$,  we let $\UpZ_{\ell}$ be the $V(B)\times V(B)$ matrix such that for any $u,w\in V(B)$,
which are connected through path $P=x_0, \ldots, x_{q}$ in $B$,  we have
\begin{align}
\UpZ_{\ell}(w,u)&=\Ind\{\dist(u,w)=\ell\}\cdot \prod\nolimits_{0\leq i < q} | \tanh(\beta \cdot \InAct_B(x_i,x_{i+1}))| \enspace. 
\end{align}
For  $\ell=0$, note that we have $\UpZ_{\ell}=\Id$. 

 Let $\upnu=\upnu(\InAct_t,\beta, \Field)$ be the Gibbs distribution that 
corresponds to the covariance matrix $\Cov_t$.  
For  any two vertices $u,w\in V(B)$ different from each other,  connected through the path 
$P=x_0, \ldots, x_{q}$ in $B$,  we have
\begin{align}
| \Cov_t(w,u)| &= \left | \upnu_{w}(+1\ | \ \{u,+1\})- \upnu_{w}(+1) \right| \cdot \upnu_{u}(+1)\nonumber\\
&\leq   \left | \upnu_{w}(+1\ | \ \{u,+1\})- \upnu_{w}(+1) \right| 
%\nonumber \\& 
\leq \prod\nolimits_{0\leq i<q } \Inf(x_i,x_{i+1}) 
%\leq \prod\nolimits_{0\leq i< q } \left| \beta\cdot \InAct_B(x_i,x_{i+1}) \right| 
\enspace. 
\label{eq:CovTVsBetaJays}
\end{align}
where recall that $\Inf(x_i,x_{i+1})= |\tanh(\beta \cdot \InAct_B(x_i,x_{i+1}))|$. Also, recall that $B$ is a tree, 
hence there is only one path from $u$ to $w$. 
The last inequality in \eqref{eq:CovTVsBetaJays} follows from standard arguments, see e.g., in  \cite{KuiKuiSpinGlass2024}.

From the above we conclude that for any two  $u,w\in V(B)$, such that $\dist(u,w)=\ell$, we have
\begin{align}\label{eq:GeneralCovTVsUpZell}
|\Cov_t(w,u)| &\leq \UpZ_{\ell}(w,u)\enspace. 
\end{align}
Clearly \eqref{eq:GeneralCovTVsUpZell} is true also when $w=u$, i.e., since in this case $|\Cov_t(w,w)|\leq 1=\UpZ_0(w,w)$.

Furthermore, in the following proofs we make use of the following two standard inequalities 
\begin{align}
|\tanh(x)| & \leq |x| &\forall x\in \mathbb{R}, \label{eq:TanhVsX} \\
|\tanh(x)| & \leq 1 &\forall x\in \mathbb{R}. \label{eq:TanhVsOne}
\end{align}

\subsubsection{Proof of \Cref{lemma:NormBoundA4mLSI4MU1}}\label{sec:lemma:NormBoundA4mLSI4MU1}
For integer $\ell\geq  0$, we let the $V(B)\times V(B)$ matrix $\UpY_{\ell}$ be such that 
\begin{align}
\UpY_{\ell}(w,u) &=\Ind\{\dist(u,w)=\ell\}\cdot \left(  \Cov_t  \cdot \InAct_B  \right)(w,u),  &
\forall u,w\in V(B)\enspace.
\end{align}

We have
\begin{align}
\| \NormMT^{-1}\ \cdot \Cov_t \cdot \InAct_B \cdot  \NormMT \|_{\infty}  & \leq 
\sum\nolimits_{\ell\geq 0} \| \NormMT^{-1}\ \cdot  \UpY_{\ell}\cdot  \NormMT \|_{\infty} \enspace. \nonumber
\end{align}
In light of the above relation and elementary calculations, \Cref{lemma:NormBoundA4mLSI4MU1} follows by showing 
\begin{align}\label{eq:target4lemma:NormBoundA4mLSI4MU1}
\| \NormMT^{-1}\ \cdot  \UpY_{\ell}\cdot  \NormMT \|_{\infty} & \leq \frac{10\cdot \beta^{-1}\cdot L\cdot \upkappa^{\ell/4}}{1-\upkappa^{1/2}}, 
& \forall \ell\geq 0\enspace. 
\end{align}

Let us prove \eqref{eq:target4lemma:NormBoundA4mLSI4MU1} for $\ell=0$.  We have 
\begin{align}
\| \NormMT^{-1}\ \cdot  \UpY_{\ell}\cdot  \NormMT \|_{\infty} &\leq \max\nolimits_{w \in V(B)} \left\{ {  \sum\nolimits_{u\sim w} |\Cov_t(w,u)\cdot \InAct_B(u,w)|} \right\} 
\nonumber \\
& \leq  \beta^{-1}\cdot \max\nolimits_{w\in V(B)}\left\{{ \sum\nolimits_{u\sim w} |\beta\cdot \InAct_B(u,w)|^2 } \right\}  \nonumber \\
&\leq  (3/4)\cdot \beta^{-1} \cdot L  \enspace. \label{eq:WeightRowSumsBoundALZero}
\end{align}
The second inequality follows from the observation that  $|\Cov_t(w,u)|\leq |\beta \cdot \InAct_B(w,u)|$. 
In \eqref{eq:WeightRowSumsBoundALZero} we use the assumption that event $\cE$ holds.
Specifically, we use the bound on $\sqrBetaJSphere(w,1)$.  It is elementary to verify that
\eqref{eq:WeightRowSumsBoundALZero} implies that \eqref{eq:target4lemma:NormBoundA4mLSI4MU1}
is true for $\ell=0$.

We now consider $\ell>0$.
Fix  vertex $w \in V(B)$  and let $k_0=w, k_1, \ldots, k_{s-1}, k_s = r$ be the path from $w$ to the root $r$ of $B$, 
while  note that $s = \dist(w, r)$.

Fix $0\leq i \leq s$. Let $T_i$ be the subtree of $T$ rooted at $k_i$ and also containing its progeny. 

For $u\in V(T_i)$, let $z_0=k_i,  z_1, \ldots, z_{q-1}, z_q =u$ be the path in $T_i$ from the root of the subtree $T_i$ to vertex $u$.
For $q> 0$,  we have 
\begin{align}
\lefteqn{
\UpY_{\ell}(w,u)  =\Ind\{\dist(w,u)=\ell\}\cdot \sum\nolimits_{v\sim u}  \Cov_t(w,v) \cdot  \InAct_B(v,u) 
} \hspace{.5cm}\nonumber \\
&  \leq \Ind\{\dist(w,u)=\ell\}\cdot {\textstyle \sum\nolimits_{v\sim u}| \Cov_t(w,v)|\cdot | \InAct_B(v,u) | } \nonumber \\
&\leq \Ind\{\dist(w,u)=\ell\}\cdot \left( {\textstyle | \Cov_t(w, z_{q-1})|\cdot | \InAct_B(z_{q-1}, z_{q}) | +\sum_{v\sim u: v\neq z_{q-1}} | \Cov_t(w,v)|\cdot | \InAct_B(v,u) | } \right) \nonumber\\
&\leq  \UpZ_{\ell-1} (w,z_{q-1})\cdot  | \InAct_B(z_{q-1}, z_{q}) | +\beta^{-1}\cdot \UpZ_{\ell}(w,u)\cdot {\textstyle \sum_{v\sim u: v\neq z_{q-1}}  | \beta \cdot \InAct_B(v,u) |^2 } \label{eq:Note4EllMinus1}\\
&\leq \beta^{-1}\cdot \UpZ_{\ell} (w, u)\cdot \left( {\textstyle \frac{| \beta\cdot \InAct_B(z_{q-1}, z_{q}) |}{\tanh(|\beta \cdot \InAct_{B}(z_{q-1},z_{q})|)} +  \sum\nolimits_{v\sim u: v\neq z_{q-1}}  | \beta \cdot \InAct_B(v,u) |^2} \right) 
\nonumber \\
&\leq \beta^{-1}\cdot \UpZ_{\ell} (w, u)\cdot \left( 1+| \beta\cdot \InAct_B(z_{q-1}, z_{q}) | +  \sum\nolimits_{v\sim u: v\neq z_{q-1}}  | \beta \cdot \InAct_B(v,u) |^2 \right) 
\enspace.  \label{eq:PartAOfUpYL}
\end{align}
For  \eqref{eq:Note4EllMinus1} we use \eqref{eq:GeneralCovTVsUpZell}. %the observation that  for any $\ell>0$,  we have $\Cov_t(w, z_{q-1})\leq \UpZ_{\ell-1}(w,z_{q-1})$.
For the second summand in \eqref{eq:Note4EllMinus1}, for $v\sim u$, such that $v\neq z_{q-1}$, we have 
\begin{align}
|\Cov_t(w,v)|\leq \UpZ_{\ell+1}(w,v) =\UpZ_{\ell}(w,u)\cdot |\tanh(\beta\cdot \InAct_B(u,v))|\leq \UpZ_{\ell}(w,u)\cdot |\beta\cdot \InAct_B(u,v))|\enspace,
\end{align}
where in the last inequality we use \eqref{eq:TanhVsX}. For \eqref{eq:PartAOfUpYL}, we use the standard inequality $\frac{|x|}{\tanh(|x|)}\leq 1+|x|$.

Note that \eqref{eq:PartAOfUpYL} is also true for  $i>0$ and $q=0$, i.e., $u=k_i$.  In this case, 
we set $z_0=k_{i-1}$ and $z_{1}=k_i=u$. From this point on, the derivations are exactly 
the same as those we have above.  Finally, note when $i=0$ and $q=0$, we have $u=w$. In this 
case $Y_{\ell}(u,w)=0$ because  $\ell>0$ while $\dist(u,w)=0$.  

In the above derivations we often interchange between $u$ and  $z_q$ since it is the same vertex. This should not create any confusion. 
For brevity, let
\begin{align}\nonumber
\UpL(w,u)&=\left( {\textstyle 1+| \beta\cdot \InAct_B(z_{q-1}, z_{q}) | +  \sum_{v\sim u: v\neq z_{q-1}}  | \beta \cdot \InAct_B(z_q,v) |^2} \right)\enspace. 
\end{align}
Using \eqref{eq:PartAOfUpYL}, and recalling the definition of
$\badmormw(\cdot)$ for the same $w,u$ as above, we have
\begin{align}
\lefteqn{
\left | \left(\badmormw(w)\right)^{-1} \cdot \UpY_{\ell}(w, u) \cdot \badmormw(u) \right|
} \hspace{1cm}\nonumber \\
&\leq \beta^{-1}\cdot  \left( \badmormw(k_0) \right)^{-1}\cdot  \left(\prod\nolimits_{0< j \leq i} \Inf(k_{j}, k_{j-1}) \cdot  \UpZ_{\ell-i}(k_i,u)\right) \cdot \UpL(w,u) \cdot \badmormw(u) \enspace. 
\label{eq:FirstStepmlsiNormBoundAB}
\end{align}
Note that  \eqref{eq:FirstStepmlsiNormBoundAB} is true  since  for any vertex $u\in V(T_{i})$, we have $\UpZ_{\ell}(w, u)=\prod\nolimits_{0< j \leq i} \Inf(k_{j}, k_{j-1}) \cdot \UpZ_{\ell-i}(k_i,u)$. 
Furthermore, from the definition of $\badmormw(\cdot)$, we have
\begin{align} \nonumber
\badmormw(k_0)=\left((1+\delta)^{i} \cdot \prod\nolimits_{0< j\leq i}  \Inf(k_{j}, k_{j-1})  \right) \cdot \badmormw(k_i) \enspace.
\end{align}
Plugging the above into \eqref{eq:FirstStepmlsiNormBoundAB}, we obtain
\begin{align}
| \left(\badmormw(k_0)\right)^{-1} \cdot \UpY_{\ell}(k_0, u) \cdot \badmormw(u) |
&\leq \beta^{-1}\cdot 
 \left( \badmormw(k_i) \right)^{-1}\cdot (1+\delta)^{-i}\cdot 
\UpZ_{\ell-i}(k_i,u) \cdot \UpL(w,u) \cdot \badmormw(u)\enspace. \label{eq:Chi(w)VsChi(zi)New}
\end{align}
Recall that  $u = z_q$ is connected to $k_i$  through the path $z_0, \ldots, z_{q}$, while  $q=\ell-i$.  We have 
\begin{align}\nonumber
\badmormw(u)&=\left((1+\delta)^{\ell-i} \cdot \prod\nolimits_{0\leq j <q}  \Inf(z_j,z_{j+1})  \right) \cdot \badmormw(k_i) \leq  (1+\delta)^{\ell-i} \cdot \UpZ_{\ell-i}(k_i,u) \cdot \badmormw(k_i) \enspace,
\end{align}
Plugging the above into \eqref{eq:Chi(w)VsChi(zi)New}, we obtain
\begin{align}
| \left( \badmormw(w)\right)^{-1} \cdot \UpY_{\ell}(w,u) \cdot \badmormw(u)|
&\leq \beta^{-1}\cdot   (1+\delta)^{\ell-2i} \left(\UpZ_{\ell-i}(k_i,u)\right)^2 \cdot \UpL(w,u)
\enspace. \label{eq:RemoveChisFromNormB}
\end{align}
Note that the above is true even if $q=0$.  In this case we have $\ell-i=0$ and $u=k_i$, hence $\UpZ_{\ell-i}(k_i,u)=(\UpZ_{\ell-i}(k_i,u))^2=1$.

Furthermore, for $\ell>s$,  we have
\begin{align}
 \left(\UpZ_{\ell-i}(k_i,u)\right)^2 \cdot \UpL(w,u)
&\leq 
2\cdot | \beta\cdot \InAct_B(z_{q-1}, z_{q}) |^2 \cdot \left(\UpZ_{\ell-i-1}(k_i,z_{q-1})\right)^2 \nonumber \\
&\quad +\left(\UpZ_{\ell-i}(k_i,u)\right)^2 \cdot  \sum\nolimits_{v\sim u: v\neq z_{q-1}}  | \beta \cdot \InAct_B(v,u) |^2
\enspace. \label{eq:UpZSqrTimesLBase}
\end{align}
For the above, we use  the 
 inequalities $\Inf(z_j,z_{j+1})\leq |\beta\cdot \InAct_{B}(z_j, z_{j+1})|$ and  $\Inf(z_{q-1}, z_q)\leq 1$ to show that 
\begin{align}
| \beta\cdot \InAct_B(z_{q-1}, z_{q}) | \cdot (\UpZ_{\ell-i}(z_0,z_q))^2 
&\leq 
| \beta\cdot \InAct_B(z_{q-1}, z_{q}) | \cdot \Inf(z_{q-1}, z_{q})\cdot (\UpZ_{\ell-i-1}(z_0,z_{q-1}))^2 \nonumber\\
&\leq | \beta\cdot \InAct_B(z_{q-1}, z_{q}) |^2 \cdot \left(\UpZ_{\ell-i-1}(z_0,z_{q-1})\right)^2\enspace. \nonumber 
\end{align}
Working similarly, we get
\begin{align}\nonumber
\left(\UpZ_{\ell-i}(z_0,u)\right)^2 \leq  | \beta\cdot \InAct_B(z_{q-1}, z_{q}) |^2 \cdot \left(\UpZ_{\ell-i-1}(z_0,z_{q-1})\right)^2\enspace.
\end{align}
In the above derivations we sometimes interchange $k_i$ and $z_0$ as it is the same vertex. This should not create any
confusion. 
From \eqref{eq:UpZSqrTimesLBase} and \eqref{eq:RemoveChisFromNormB},  we get 
\begin{align}
\lefteqn{
\left|  \left( \badmormw(w)\right)^{-1} \cdot \UpY_{\ell}(w,u) \cdot \badmormw(u) \right|
 } \hspace{.2cm} \nonumber \\
 &\leq 
 \beta^{-1}\cdot   (1+\delta)^{\ell-2i}  
\left(\left(\UpZ_{\ell-i}(k_i,u)\right)^2\cdot \sum_{v\sim u: v\neq z_{q-1}}  | \beta \cdot \InAct_B(v,u) |^2
% \prod\nolimits_{0\leq j <q} \left( \beta\cdot \InAct_B(z_j,z_{j+1})\right)^2\left(2+  \sum\nolimits_{v\sim u: v\neq z_{q-1}}  | \beta \cdot \InAct_B(v,u) |^2 \right)
+ 2\cdot | \beta\cdot \InAct_B(z_{q-1}, z_{q}) |^2 \cdot \left(\UpZ_{\ell-i-1}(k_i,z_{q-1})\right)^2 \right)
 \enspace.
 \nonumber 
\end{align}
From the above, it is not hard to verify that 
\begin{align}
%\lefteqn{
\sum\nolimits_{u} \left| \left( \badmormw(w)\right)^{-1} \cdot \UpY_{\ell}(w,u) \cdot \badmormw(u)\right|
%} \hspace{1cm} \nonumber \\
&\leq \beta^{-1}\cdot \sum\nolimits_{0\le i \le s} (1+\delta)^{\ell-2i}
\left( 2\cdot \sqrBetaJSphere(k_i,\ell-i+1)+ \sqrBetaJSphere(k_i,\ell-i)\right)    \nonumber \\
&\leq 6\cdot \beta^{-1}\cdot L\cdot \sum\nolimits_{0\leq i\leq s} \upkappa^{\ell-i}\cdot (1+\delta)^{\ell-2i}
\enspace.  \label{eq:WeitghedRowSumWithEBound}
\end{align}
For \eqref{eq:WeitghedRowSumWithEBound}, our assumption that event $\cE$ holds implies the bounds for 
$\sqrBetaJSphere(k_i,\ell-i+1)$ and $\sqrBetaJSphere(k_i,\ell-i)$.

Since  $\delta=\upkappa^{-3/4}-1$,  we have $(1+\delta)\upkappa<1$, while 
$(1+\delta)^2\upkappa>1$. Then, we obtain the following 
\begin{align}
%\lefteqn{
\sum\nolimits_{u} \left| \left( \badmormw(w)\right)^{-1} \cdot \UpY_{\ell}(w,u) \cdot \badmormw(u)\right|
%} \hspace{1cm} \nonumber \\
&\leq 6\cdot \beta^{-1}\cdot L\cdot ((1+\updelta)\cdot \upkappa)^{\ell}\cdot  \sum\nolimits_{0\leq i\leq s} \left(\upkappa \cdot (1+\delta)^2\right)^{-i}\nonumber \\
&\leq \frac{6\cdot \beta^{-1}\cdot L\cdot \upkappa^{\ell/4}}{1-\upkappa^{1/2}}
\enspace. \label{eq:WeightRowSumsBoundALBiggerS}
\end{align}

Working similarly to the above, for $\ell\leq s$, we have 
\begin{align}
\sum_{u}\left| \left( \badmormw(w)\right)^{-1} \cdot \UpY_{\ell}(w,u) \cdot \badmormw(u) \right|
   &% \nonumber \\&
 \leq \frac{10\cdot \beta^{-1}\cdot L\cdot \upkappa^{\ell/4}}{1-\upkappa^{1/2}} \enspace. \label{eq:WeightRowSumsBoundALAtMostS}
\end{align}
Recall that 
$\| \NormMT^{-1}\ \cdot  \UpY_{\ell}\cdot  \NormMT \|_{\infty}=
\max\nolimits_{w\in V(B)} \left\{\sum\nolimits_{u\in V(B)} \left| (\badmormw(w))^{-1}\cdot \UpY_{\ell}(w,u) \cdot \badmormw(u) \right|
\right \}$.
Then,  \eqref{eq:WeightRowSumsBoundALZero}, \eqref{eq:WeightRowSumsBoundALBiggerS} and \eqref{eq:WeightRowSumsBoundALAtMostS} 
imply that  \eqref{eq:target4lemma:NormBoundA4mLSI4MU1} is true for $\ell\ge0$.  This concludes the proof of \Cref{lemma:NormBoundA4mLSI4MU1}. 
\hfill $\Box$

\subsubsection{Proof of \Cref{lemma:NormBoundB4mLSI4MU1}}\label{sec:lemma:NormBoundB4mLSI4MU1}
%%%
We have
\begin{align}
\| \NormMT^{-1}\ \cdot \Cov_t \cdot (\RescaleMT^2+ \RescaleMT \cdot \MontDMatrix \cdot \RescaleMT) \cdot  \NormMT \|_{\infty}  
&\leq 
\| \NormMT^{-1}\ \cdot \Cov_t \cdot \RescaleMT^2\cdot  \NormMT \|_{\infty}  
+\| \NormMT^{-1}\ \cdot \Cov_t \cdot \RescaleMT \cdot \MontDMatrix \cdot \RescaleMT \cdot  \NormMT \|_{\infty}  \nonumber 
\enspace.
\end{align}
\Cref{lemma:NormBoundB4mLSI4MU1} follows from the above inequality by showing
\begin{align}
\| \NormMT^{-1}\ \cdot \Cov_t \cdot \RescaleMT^2\cdot  \NormMT \|_{\infty}  &\leq  
30\cdot  (1-\upkappa^{1/4})^{-2}\cdot L
\enspace.
\label{eq:TargetA4lemma:NormBoundB4mLSI4MU1}\\
\| \NormMT^{-1}\ \cdot \Cov_t \cdot \RescaleMT \cdot \MontDMatrix\cdot \RescaleMT \cdot  \NormMT \|_{\infty} &\leq 
30\cdot  (1-\upkappa^{1/4})^{-2}\cdot L \enspace. 
 \label{eq:TargetB4lemma:NormBoundB4mLSI4MU1}
\end{align}
We first  focus on  \eqref{eq:TargetA4lemma:NormBoundB4mLSI4MU1}. 
Since $\NormMT$ and $\RescaleMT^2$ have non-negative entries, a simple application of the triangle inequality yields 
\begin{align}
\| \NormMT^{-1}\ \cdot \Cov_t \cdot \RescaleMT^2\cdot  \NormMT \|_{\infty}  &\leq  \| \NormMT^{-1}\ \cdot |\Cov_t| \cdot \RescaleMT^2 \cdot  \NormMT \|_{\infty} \enspace. 
\end{align}
For integer $\ell\geq 0$, we let the $V(B)\times V(B)$ matrix $\UpY_{\ell}$ be such that %for any $u,w\in V(B)$ we have
\begin{align}
\UpY_{\ell}(w,u) &=\Ind\{\dist(u,w)=\ell\}\cdot \left(  |\Cov_t | \cdot \RescaleMT^2 \right)(w,u),  & \forall u,w\in V(B) \enspace.
\end{align}
From the two above relations, we have
\begin{align}\label{eq:BreakCovTDiag2Yells}
\| \NormMT^{-1}\ \cdot  \Cov_t  \cdot \RescaleMT^2 \cdot  \NormMT \|_{\infty}  \leq 
\sum\nolimits_{\ell\geq 0} \| \NormMT^{-1}\ \cdot  \UpY_{\ell}\cdot  \NormMT \|_{\infty} \enspace. 
\end{align}
We obtain  \eqref{eq:TargetA4lemma:NormBoundB4mLSI4MU1} by  showing 
\begin{align}\label{eq:BasisTargetA4lemma:NormBoundB4mLSI4MU1}
\| \NormMT^{-1}\ \cdot  \UpY_{\ell}\cdot  \NormMT \|_{\infty} & \leq \frac{30 \cdot  L\cdot \upkappa^{\ell/4}}{1-\upkappa^{1/2}},  &\forall \ell\geq 0 \enspace.
\end{align}
For $\ell=0$, we have
\begin{align}
\| \NormMT^{-1}\ \cdot  \UpY_{\ell}\cdot  \NormMT \|_{\infty} &\leq \max\nolimits_{w\in V(B)} \{ |\Cov_t(w,w)|\cdot \left({\textstyle 1+\sum_{u\sim w} (\beta\cdot \InAct_B(w,u))^2}\right)\}\nonumber \\
&\leq 1+\max\nolimits_{w\in V(B)} \{{\textstyle \sum_{u\sim w} (\beta\cdot \InAct_B(w,u))^2} \} \leq 1+ (3/4)  \cdot L
\label{eq:UpYEllRSQRLZeroBound} \enspace,
\end{align}
where the second inequality is from the fact that $|\Cov_t(w,w)|\leq 1$, for each $w\in V(B)$. 
For \eqref{eq:UpYEllRSQRLZeroBound}, we use our assumption that event $\cE$ holds.
Specifically, we use the bound on  $\sqrBetaJSphere(w,1)$.
Clearly, \eqref{eq:UpYEllRSQRLZeroBound} implies that \eqref{eq:BasisTargetA4lemma:NormBoundB4mLSI4MU1} is true for
$\ell=0$.

Fix integer $\ell > 0$. Fix  vertex $w \in V(B)$  and let $k_0=w, k_1, \ldots, k_{s-1}, k_s = r$ be the path from $w$ to the root $r$ of $B$, 
while  note that $s = \dist(w, r)$.

Fix $0\leq i \leq s$. Let $T_i$ be the subtree of $T$ rooted at $k_i$ and also containing its progeny. 

For $u\in V(T_i)$, let $z_0=k_i,  z_1, \ldots, z_{q-1}, z_q =u$ be the path in $T_i$ from the root of the subtree $T_i$ to vertex $u$.
For $q> 0$,  we have 
\begin{align}
\UpY_{\ell}(w,u) &\leq \Ind\{\dist(u,w)=\ell\}\cdot | \Cov_t(w, u)|\cdot \RescaleMT^2 (u,u)   \nonumber \\
&=\Ind\{\dist(u,w)=\ell\}\cdot | \Cov_t(w, u)|\cdot \left(1+\sum\nolimits_{v\sim u}\left(\beta \cdot \InAct_B(v,u)\right)^2\right) \nonumber \\
&\leq \UpZ_{\ell}(w,u) \cdot \left(1+\sum\nolimits_{v\sim u}\left(\beta \cdot \InAct_B(v,u)\right)^2\right)  \nonumber \\
&\leq \UpZ_{\ell}(w,u) + \UpZ_{\ell}(w,u)\cdot \left(\beta \cdot \InAct_B(z_{q-1},z_{q})\right)^2+   \UpZ_{\ell}(w,u)\cdot \sum\nolimits_{v\sim u:v\neq z_{q-1}}\left(\beta \cdot \InAct_B(v,u)\right)^2  \enspace.  
\label{eq:YelllVSZelllCouplings2ndNorm}
\end{align}
The inequality in the third line follows from \eqref{eq:GeneralCovTVsUpZell}.

Assume that $\ell>s$. Using arguments very similar to those we use in the proof of \Cref{lemma:NormBoundA4mLSI4MU1}, we get
\begin{align}
| \left(\badmormw(w)\right)^{-1} \cdot  \UpZ_{\ell}(w, u)  \cdot \badmormw(u)|  &\leq (1+\delta)^{\ell-2i}\cdot 
\left(\UpZ_{\ell-i}(k_i,z_q)\right)^2 \label{eq:NormWeightsWithZLL2nd}\\ 
& \leq (1+\delta)^{\ell-2i}\cdot  \left(\UpZ_{\ell-i-1}(k_i,z_{q-1})\right)^2 \cdot \left(\beta\cdot \InAct_B(z_{q-1},z_{q}) \right)^2 \enspace. 
\label{eq:NormWeightsWithZLFinal2nd}
\end{align}
For the second inequality we use that $\UpZ_{\ell-i}(k_i,z_q)=\UpZ_{\ell-i-1}(k_i,z_{q-1})\cdot |\tanh(\beta \cdot \InAct_B(z_{q-1}, z_q))|$, 
and \eqref{eq:TanhVsX}. %the standard inequality $|\tanh(\beta \cdot  \InAct_B(z_{q-1}, z_q))|\leq (\beta \cdot \InAct_B(z_{q-1}, z_q)|$.

From \eqref{eq:NormWeightsWithZLL2nd}, we also get that 
\begin{align}
| \left(\badmormw(w)\right)^{-1} \cdot  \UpZ_{\ell}(w, u) \cdot \left( \beta \cdot \InAct_B(z_{q-1},z_{q})\right)^2 \cdot \badmormw(u) |
 &\leq (1+\delta)^{\ell-2i}\cdot  \left(\UpZ_{\ell-i}(k_i,z_q)\right)^2\cdot \left( \beta \cdot \InAct_B(z_{q-1},z_{q})\right)^2  \nonumber \\
&\leq (1+\delta)^{\ell-2i}\cdot  \left(\UpZ_{\ell-i-1}(k_i,z_{q-1} )\right)^2\cdot \left( \beta \cdot \InAct_B(z_{q-1},z_{q})\right)^2 \enspace. 
\label{eq:NormWeightsWithZLJZQminusFinal2nd}
\end{align}
For the second inequality we use that  $\UpZ_{\ell-i}(k_i,z_q)=\UpZ_{\ell-i-1}(k_i,z_{q-1})\cdot |\tanh(\beta \cdot \InAct_B(z_{q-1}, z_q))|$, 
and \eqref{eq:TanhVsOne}. %the standard inequality $|\tanh(\beta \cdot  \InAct_B(z_{q-1}, z_q))|\leq 1$.

Using \eqref{eq:NormWeightsWithZLL2nd} once more,  we have
\begin{align}
\lefteqn{
| \left(\badmormw(w)\right)^{-1} \cdot  \UpZ_{\ell}(w, u) \cdot   \left( \sum\nolimits_{v\sim u: v\neq z_{q-1}} \left(\beta \cdot \InAct_B(v,u)\right)^2
 \right) \cdot  \badmormw(u) |}   \hspace{4cm} \nonumber \\
&\leq (1+\delta)^{\ell-2i}\cdot  \left(\UpZ_{\ell-i}(k_i, z_q) \right)^2 \cdot   \sum\nolimits_{v\sim u: v\neq z_{q-1}} \left(\beta \cdot \InAct_B(v,z_q)\right)^2 \enspace. 
\label{eq:NormWeightsWithZLJZQPlusFinal2nd}
\end{align}
Using \eqref{eq:YelllVSZelllCouplings2ndNorm} and the bounds from \eqref{eq:NormWeightsWithZLFinal2nd}, \eqref{eq:NormWeightsWithZLJZQminusFinal2nd} and
\eqref{eq:NormWeightsWithZLJZQPlusFinal2nd},  we have
\begin{align}
| \left(\badmormw(w)\right)^{-1} \cdot  \UpY_{\ell}(w,u) \cdot \badmormw(u) |
&\leq | \left(\badmormw(w)\right)^{-1} \cdot  \UpZ_{\ell}(w, u)  \cdot \badmormw(u)|  \nonumber \\
&\quad + |\left(\badmormw(w)\right)^{-1} \cdot  \UpZ_{\ell}(w, u) \cdot \left(\beta \cdot \InAct_B(z_{q-1},u)\right)^2 \cdot \badmormw(u)| \nonumber \\
&\quad + | \left(\badmormw(w)\right)^{-1} \cdot  \UpZ_{\ell}(w, u) \cdot   \left( \sum\nolimits_{v\sim u: v\neq z_{q-1}} \left(\beta \cdot \InAct_B(v,u)\right)^2 \right) \cdot  \badmormw(u)|  \nonumber \\
&\leq 2\cdot (1+\delta)^{\ell-2i}\cdot  \left(\UpZ_{\ell-i-1}(k_i,z_{q-1})\right)^2 \cdot \left(\beta\cdot \InAct_B(z_{q-1},z_{q}) \right)^2 \nonumber \\
&\quad+  (1+\delta)^{\ell-2i}\cdot  \left(\UpZ_{\ell-i}(k_i, z_q) \right)^2 \cdot   \sum\nolimits_{v\sim u: v\neq z_{q-1}} \left( \beta \cdot \InAct_B(v,z_q) \right)^2 
\label{eq:Bound4YellwithXOneVertexB}
\enspace. 
\end{align}
%The last inequality, uses the bounds in \eqref{eq:NormWeightsWithZLFinal2nd}, \eqref{eq:NormWeightsWithZLJZQminusFinal2nd} 
%and \eqref{eq:NormWeightsWithZLJZQPlusFinal2nd}. 

Using \eqref{eq:Bound4YellwithXOneVertexB}, it is not hard to verify that 
\begin{align}
\sum\nolimits_{u} |\left(\badmormw(w)\right)^{-1} \cdot  \UpY_{\ell}(w,u) \cdot \badmormw(u)| &\leq 
\sum\nolimits_{0\leq i\leq s} (1+\delta)^{\ell-2i} \left(2\cdot  \sqrBetaJSphere(k_i,\ell-i)+  \sqrBetaJSphere(k_i,\ell-i+1)\right)\nonumber \\
&\leq  9 \cdot \maxDelta\cdot ((1+\delta)\cdot \upkappa)^\ell\cdot \sum\nolimits_{0\leq i\leq s} ((1+\delta)^2\cdot \upkappa)^{-i}\nonumber \\
&\leq \frac{9 \cdot  L\cdot \upkappa^{\ell/4}}{1-\upkappa^{1/2}} \label{eq:EllLargerSInequality4CovR2}
 \enspace. 
\end{align}
The second inequality is from our assumption that event $\cE$. This assumption provides the bounds for 
$\sqrBetaJSphere(k_i,\ell-i+1)$ and $\sqrBetaJSphere(k_i,\ell-i)$. The last inequality follows from
our assumption that $\delta=\upkappa^{-3/4}-1$ and derivations which are almost identical to those in
\eqref{eq:WeightRowSumsBoundALBiggerS}. 

With very similar arguments, we get that for $\ell\leq s$, we have
\begin{align}
\sum\nolimits_{u} |\left(\badmormw(w)\right)^{-1} \cdot  \UpY_{\ell}(w,u) \cdot \badmormw(u)| 
&\leq \frac{30 \cdot  L\cdot \upkappa^{\ell/4}}{1-\upkappa^{1/2}} \label{eq:EllSmallerSInequality4CovR2} \enspace. 
\end{align}
Recall that 
$\| \NormMT^{-1}\ \cdot  \UpY_{\ell}\cdot  \NormMT \|_{\infty}=
\max\nolimits_{w\in V(B)} \left\{\sum\nolimits_{u\in V(B)} \left| (\badmormw(w))^{-1}\cdot \UpY_{\ell}(w,u) \cdot \badmormw(u) \right|
\right \}$.
Then, from \eqref{eq:EllLargerSInequality4CovR2} and \eqref{eq:EllSmallerSInequality4CovR2}, it is immediate
that \eqref{eq:BasisTargetA4lemma:NormBoundB4mLSI4MU1} is true for any $\ell>0$. 

From all the above we conclude that  \eqref{eq:TargetA4lemma:NormBoundB4mLSI4MU1} is true. 
We proceed to show that \eqref{eq:TargetB4lemma:NormBoundB4mLSI4MU1} is also true.  To this end, we have the following result.

\begin{claim}\label{claim:Redaux4Dmatrix}
For any $w\in V(B)$, we have $0<(\RescaleMT \cdot \UpD\cdot \RescaleMT) (w,w) \leq (\RescaleMT^2)(w,w)$.
\end{claim}

Since $\RescaleMT \cdot \UpD\cdot \RescaleMT$ is diagonal with non-negative entries,  we have  
\begin{align}
\| \NormMT^{-1}\ \cdot \Cov_t \cdot \RescaleMT \cdot \UpD\cdot \RescaleMT \cdot  \NormMT \|_{\infty} 
&\leq  \| \NormMT^{-1}\ \cdot |\Cov_t | \cdot \RescaleMT \cdot \UpD\cdot \RescaleMT \cdot  \NormMT \|_{\infty} 
% \nonumber \\&
 \leq  \| \NormMT^{-1}\ \cdot |\Cov_t | \cdot  \RescaleMT^2  \cdot  \NormMT \|_{\infty} \nonumber \\
 &\leq  \sum\nolimits_{\ell\geq 0} \| \NormMT^{-1}\ \cdot  \UpY_{\ell}\cdot  \NormMT \|_{\infty} \enspace, \nonumber
\end{align}
where the last inequality follows from \eqref{eq:BreakCovTDiag2Yells}.
Then, \eqref{eq:TargetB4lemma:NormBoundB4mLSI4MU1} follows by using \eqref{eq:BasisTargetA4lemma:NormBoundB4mLSI4MU1}. 

All the above imply that both \eqref{eq:TargetA4lemma:NormBoundB4mLSI4MU1} and \eqref{eq:TargetB4lemma:NormBoundB4mLSI4MU1} are true. This
concludes the proof of \Cref{lemma:NormBoundB4mLSI4MU1}.
\hfill $\Box$

\begin{proof}[Proof of \Cref{claim:Redaux4Dmatrix}]
  It suffices to show that $\UpD(w,w) \le 1$ for every $w \in V(B)$.

Let the $V(B)\times V(B)$, diagonal matrix $\widehat{\UpD}$ be such that
\begin{align}
\widehat{\UpD}(w,w)&=\sum\nolimits_{u\sim w}\left( \textstyle \widehat{\UpJ}(w,u) \right)^2 & \forall w\in V(B)\enspace, \nonumber
\end{align}
where  matrix $\widehat{\UpJ}$ is defined in \eqref{eq:HatJVsQWeights}. 
It is not hard to verify that matrix $\widehat{\UpD}$ satisfies,    $\widehat{\UpD}(w,w)\leq 1$,  for each $w\in V(B)$. 
Specifically, recalling the definition of matrix $\widehat{\UpJ}$, 
for $w\in V(B)$, we have
\begin{align}
\widehat{\UpD}(w,w) &
\leq  \sum_{v\sim w}  { \frac{ \left( \beta\cdot \InAct_B(w,v)\right)^2}
{1+\sum\nolimits_{x \sim w} \left( \beta\cdot \InAct_{B}(w,x)\right)^2} } \leq 1\enspace.  \nonumber
\end{align}

From the definition of $\UpD$,  for any $w\in V(B)$, we have
\begin{align}
\UpD(w,w)&=\sum\nolimits_{u\sim w}\left(  \widehat{\UpJ}(w,u) \right)^2 (1-(\UpQ(w,u))^2) \leq \sum\nolimits_{u\sim w}\left(  \widehat{\UpJ}(w,u) \right)^2 = \widehat{\UpD}(w,w)\leq 1\enspace. \nonumber
\end{align}
\Cref{claim:Redaux4Dmatrix} follows.
\end{proof}

\spreadpoint

\newcommand{\rad}{\red{k}}

\section{Existence of Block-Partition - Proof of \Cref{theorem:HBlockExistence}}
\label{sec:theorem:HBlockExistence}

We prove \Cref{theorem:HBlockExistence}, by introducing a new notion of partition
of the set of vertices of $\G=\G(n,d/n)$. We call it $(\varepsilon,r)$-decomposition.
This notion is slightly different from what we describe in \Cref{def:BlockPartition}. 
The new decomposition is a stepping stone towards establishing the existence of the 
$\varepsilon$-block partition.

Consider distribution $\mu_{\InAct, \beta, \Field}$ on $\G=\G(n,d/n)$
as defined in \eqref{eq:InteractionMatrixAtBlockPartition} and \eqref{eq:DefOfGibbsii}
for $\beta$ and $\Field$ as specified in \Cref{theorem:HBlockExistence}.
Given parameter $\varepsilon\in (0,1)$, consider the $\varepsilon$-weights 
as defined in \eqref{def:VertexWeights} and \eqref{def:WightOfPath} with respect to 
the distribution $\mu_{\InAct, \beta, \Field}$.

For integer $\rad \geq 0$, we introduce the notion of $(\varepsilon,\rad)$-block vertex.
Specifically, every $u\in V_n$ is called $(\varepsilon, \rad)$-block vertex if the following holds:
\begin{enumerate}[(a)]
\item every path $P$ of length $\leq \log n$ emanating from $u$, 
satisfies $\cappedWA(P) < 1$\label{itm:blockA}, 
\item for every edge $e$ within distance $\rad$ from $u$, we have $\Inf_{e}\leq d^{-1/10}$,
\item for every vertex $v$ within distance $\rad+1$ from $u$, we have
$\sum_{w\sim v} (\InAct_{v,w})^2 \le (1+\varepsilon) d$.
\end{enumerate}
The above generalises the notion of $\varepsilon$-block vertex,
as this corresponds to the $(\varepsilon, \rad)$-block vertex for $\rad=0$.

\begin{definition}[$(\varepsilon,\rad)$-decomposition]\label{def:BlockDecomp}
For parameters $\varepsilon, \rad>0$, consider the $\varepsilon$-weights defined in \eqref{def:VertexWeights} and \eqref{def:WightOfPath}. 
The partition $\mathcal{B}=\{B_1, \ldots, B_{N}\}$ of the vertex set $V_n$ 
is called $(\varepsilon,\rad)$-decomposition if for every $B\in \cB$ the following is true: 
 \begin{enumerate}
\item the vertices in block $B$ induce in $\G(n,d/n)$ a tree with at most one extra edge, \label{itm:BPunicyclic}
\item if $B$ is a multi-vertex block, the following holds:
\begin{enumerate}
\item every $w\in \partial_{\rm out}B$ is a $(\varepsilon, \rad)$-block vertex with 
exactly one neighbour in $B$, \label{itm:BPSingleNeighBoundary}

\item if $B$ contains a cycle, this is a short one, i.e., its length is 
$\leq 4\frac{\log n}{(\log d)^4}$, \label{itm:BPShortCycle}
\item the distance of the short cycle in $B$ from $\partial_{\rm out}B$ is 
 $\geq \mincycdist$.
\label{itm:BuffCond}
\end{enumerate}
\item if $B=\{u\}$ is a single-vertex block, then $u$ is an $(\varepsilon,\rad)$-block vertex. 
\end{enumerate}
\end{definition}

An important difference between the above and \Cref{def:BlockPartition} is that it uses the
$(\varepsilon, \rad)$-block vertices. Furthermore, condition 2.(a) requires the block-vertices
to be in $\partial_{\rm out}B$ of every multi-vertex block $B$. For the $\varepsilon$-block
partition, condition 2.(a) requires the block-vertices to be in $\partial_{\rm in}B$. 
% 

% We use the following result for the $(\varepsilon,r)$-decomposition.

\begin{theorem}\label{lemma:DEpsilonBlockLemma}
For any $\varepsilon\in (0,1)$, integer \deeppink{$0\leq \rad \leq 100$}, $\Field\in \mathbb{R}^{V_n}$, there exists $d_0=d_0(\varepsilon)$ such that for any 
$d\ge d_0$ and $\beta< \beta_c(d)$, the following is true:

Let the Gibbs distribution $\mu_{\InAct,\beta, \Field}$, specified by $\G(n,d/n)$ and interaction matrix $\InAct$ defined as
in \eqref{eq:InteractionMatrixAtBlockPartition}. Then, with probability $1-o(1)$ over the instance of 
$\mu_{\InAct,\beta, \Field}$, the vertex set $V_n$ admits an $(\varepsilon, \rad)$-decomposition.
\end{theorem}
The proof of \Cref{lemma:DEpsilonBlockLemma} appears in \Cref{sec:lemma:DEpsilonBlockLemma}.

Let $\varepsilon, \exteps \in (0,1)$ be such that $\exteps=\varepsilon\cdot (1-10^{-4})<1$.
Let $\cD$ be a $(\varepsilon, 0)$-decomposition and $\cE$ be $(\exteps,\deeppink{100})$-decomposition.
Note that \Cref{lemma:DEpsilonBlockLemma} implies that we typically have both $\cD$ and $\cE$. 

We show (later in \Cref{prop:ProcGenBPart}) that there is a natural coupling such 
that for every multi-vertex block $A\in \cD$, there exists a multi-vertex block $B\in \cE$ 
such that $A\subset B$. Note that the containment is strict. Furthermore, 
the distance between $\partial_{\rm out}A$ and $\partial_{\rm out}B$ is at least \deeppink{90}.

\Cref{theorem:HBlockExistence} follows by showing that we can construct an $\varepsilon$-block 
partition $\cB$ by using the decompositions $\cD$ and $\cE$. 
Specifically, we show a construction for $\cB$ such that each multi-vertex block $K\in \cB$ is obtained 
as a sub-block of a multi-vertex block $B\in \cE$, i.e., $K\subseteq B$.

Fix $B\in \cE$ and let $\Upgamma_B$ be the set of multi-vertex blocks 
$A\in \cD$ such that $A\subset B$. Also, let set $\UpQ_{B}$ contain all 
vertices in $B$ which are within distance \deeppink{$90$} from $\partial_{\rm out}B$.

Let $L$ be the union of the vertex sets induced by the blocks in $\Upgamma_{B}$. 
Then, iteratively, set 
% % %
\begin{align}\label{eq:Iteration4K}
L \leftarrow L \cup (\partial_{\rm out} L \setminus \UpQ_B) \enspace. 
\end{align}
When no more vertices can be added into $L$, i.e., 
$\partial_{\rm out} L \setminus \UpQ_B=\emptyset$, we set 
\begin{align}\label{eq:AddBoundary2K}
K= L \cup \partial_{\rm out}L \enspace.
\end{align}
Repeating the above steps for all multi-vertex blocks in $\cE$, we obtain the multi-vertex
blocks in $\cB$. The vertices in $V_n$ which do not belong to any of the multi-vertex
blocks in $\cB$ are the single-vertex blocks of $\cB$.

\begin{proposition}\label{prop:ProcGenBPart}
In the above construction, $\cB$ is an $\varepsilon$-block partition. 
\end{proposition}

\Cref{theorem:HBlockExistence} follows as a corollary from \Cref{lemma:DEpsilonBlockLemma} and \Cref{prop:ProcGenBPart}. 
\hfill $\Box$

\subsection{Proof of \Cref{prop:ProcGenBPart}}

 Recall that $\cD$ is an $(\varepsilon,0)$-decomposition, while $\cE$
is an $(\exteps,\deeppink{100})$-decomposition, where $\exteps=\varepsilon(1-10^{-4})<1$. 

Firstly, we prove that we can have $\cD$ and $\cE$ such that for every multi-vertex block $B \in \cE$ and 
$A\in \cD$ satisfying $A\cap B\neq \emptyset$ we have
\begin{description}
\item[C.1] we have $A\subset B$,
\item[C.2]
the distance between $\partial_{\rm out}A$ and $\partial_{\rm out}B$ is at least \deeppink{90}. 
\end{description}

Consider the weights of the paths induced by the $\varepsilon$-weights and $\exteps$-weights, respectively. 
It is not hard to verify that for every path $P$, we have $\cappedWSA_{\varepsilon}(P) \leq \cappedWSA_{\exteps}(P)$.
This observation implies that every $(\exteps,\deeppink{100})$-block vertex is also an $(\varepsilon,0)$-block vertex.

For the multi-vertex block $B\in \cE$, recall that $\UpQ_B$ is the set of 
vertices inside $B$ such that $\dist(\UpQ_B,\partial_{\rm out}B)\leq \deeppink{90}$. 

\begin{claim}\label{claim:QBEBlockVertices}
For any multi-vertex block $B\in \cE$, all vertices in $\UpQ_B$ are $(\varepsilon,0)$-block vertices. 
\end{claim}

For tree-like block $B\in \cE$, consider each vertex $u\in B\setminus \UpQ_B$, which is not an $(\varepsilon,0)$-block vertex.
\Cref{claim:QBEBlockVertices} implies that for any path from $u$ to $\partial_{\rm out}B$, 
there is an $(\varepsilon,0)$-block vertex which is at distance $>\deeppink{90}$ from $\partial_{\rm out}B$. 
Hence, letting $A\in \cD$ be the multi-vertex block that $u$ belongs to, it is easy to check that 
$A$ and $B$ satisfy {\bf C.1} and ${\bf C.2}$.

Assume now that $B\in \cE$ contains the short cycle $C$. 
Let vertex $u\in B\setminus \UpQ_B$ which is not an $(\varepsilon,0)$-block vertex and does not
belong to $C$. Also, let block $A\in \cD$ be such that $u\in A$.
If block $A$ is tree, then, using similar arguments as those in the previous paragraph, 
we see that {\bf C.1} and ${\bf C.2}$ are satisfied for $A$ and $B$.

We additionally need to ensure that {\bf C.1} and ${\bf C.2}$ are satisfied 
when the short cycle $C$ in block $B\in \cE$ belongs to block $A\in \cD$. 
In this case, it might be that ${\bf C.1}$ and ${\bf C.2}$ are not satisfied because 
of the extra constraint we have for the distance of $C$ from the boundary of 
the corresponding block.
That is, $C$ is at distance $\geq \frac{150}{\exteps} \log \log n$ from $\partial_{\rm out}B$
and at distance $\geq \frac{150}{\varepsilon} \log \log n$ form $\partial_{\rm out}A$. 
This means that $\partial_{\rm out}A$ might need to use vertices
that are further away than those in $\UpQ_B$.

A simple calculation reveals that $\dist{\UpQ_{B}, C}\gg \mincycdist$. To see this, note that $\partial_{\rm in} B$ is at 
distance $\frac{150}{\exteps} \log \log n$ from $C$. Then, the distance between any 
$w\in \UpQ_B$ and cycle $C$ is at least
\begin{align}\label{eq:UPQBVsSCycleC}
\frac{150}{\exteps} \log \log n - \deeppink{90} > \frac{150}{\varepsilon} \log \log n \enspace.
\end{align}
The inequality \eqref{eq:UPQBVsSCycleC} holds for sufficiently large $n$. 

From all the above, it is not hard to deduce that in every path from cycle $C$ to 
$\partial_{\rm out}B$, there is at least one $(\varepsilon,0)$-block vertex which is
at distance \deeppink{$> 90$} from $\partial_{\rm out}B$ and at distance $\geq \mincycdist$ 
from $C$. Hence, block $A$ that contains cycle $C$ and block $B$ satisfy ${\bf C.1}$ and ${\bf C.2}$.

Hence, ${\bf C.1}$ and ${\bf C.2}$ always hold for $\cD$ and $\cE$. We now show that $\cB$, constructed 
as we described above, is an $\varepsilon$-block partition. 

The description of how we obtain $\cB$ from $\cE$ and $\cD$ makes it clear 
that $\cB$ specifies a partition of the set of vertices in $V_n$ with parts the blocks of $\cB$. 
It remains to show that $\cB$ satisfies all constraints in \Cref{def:BlockPartition}.
 
Fix multi-vertex blocks $K\in \cB$ and $B\in \cE$, such that $K$ is obtained from $B$
in the process described in \eqref{eq:Iteration4K} and \eqref{eq:AddBoundary2K}. 
It follows from the process that we always have $K\subseteq B$.
Hence, $K$ can only be a tree with an extra edge, i.e., like the blocks in $\cE$. 
This implies that $\cB$ satisfies conditions 1 and 2.b in \Cref{def:BlockPartition}.

From the construction of $K$, i.e., \eqref{eq:AddBoundary2K}, we have 
$\partial_{\rm in}K \subseteq \UpQ_B$.
Then, by \Cref{claim:QBEBlockVertices}, we have that $\partial_{\rm in}K$
are $(\varepsilon,0)$-block vertices, which is equivalent to
saying that they are $\varepsilon$-block vertices. 
Furthermore, it is not hard to verify that at step \eqref{eq:AddBoundary2K}, 
all vertices in $\partial_{\rm out}L$ have exactly one neighbour in $L$.
But then, since $\partial_{\rm in}K=\partial_{\rm out}L$, we have that
the number of neighbours of each vertex $u\in \partial_{\rm in }K$ inside 
block $K$ is 1. 
We conclude that $\cB$ satisfies condition 2.a in \Cref{def:BlockPartition}.

As argued in the previous paragraph, we have $\partial_{\rm in} K\subseteq \UpQ_{B}$. 
If $K$ has a short cycle $C$, then the distance of $C$ from $\partial_{\rm in}K$ is 
the same as that from $C$ to $\UpQ_{B}$. But then, we have shown in \eqref{eq:UPQBVsSCycleC}
that this distance is $\gg \frac{150}{\varepsilon}\log\log n$. This implies that
2.c in \Cref{def:BlockPartition} is also true for $\cB$.

Finally, note that every vertex $u\in V_n$, which is a single-vertex block in $\cB$, 
either comes from a single-vertex block in $\cE$ or it comes from a multi-vertex block $B\in \cE$. 
In the first case, we have that $u$ is an $(\exteps,\deeppink{100})$-block vertex
which implies that it is an $(\varepsilon, 0)$-block vertex, too.
In the second case, we have that $u\in \UpQ_B$. But according to 
\Cref{claim:QBEBlockVertices}, $u$ is an $(\varepsilon,0)$-block vertex.
We conclude that each single-vertex block in $\cB$ consists of an $(\varepsilon, 0)$-block vertex,
which is equivalent to saying it consists of an $\varepsilon$-block vertex.
Hence, $\cB$ satisfies condition 3 in \Cref{def:BlockPartition}.

This concludes the proof of \Cref{prop:ProcGenBPart}. \hfill $\Box$

\begin{proof}[Proof of \Cref{claim:QBEBlockVertices}]
Fix a multi-vertex block $B\in \cE$ and consider the set $\UpQ_{B}$ of the vertices inside $B$
such that $\dist(\UpQ_B, \partial_{\rm out}B)\leq \deeppink{90}$.
We show that any vertex $w\in \UpQ_{B}$ is an $(\varepsilon,0)$-block vertex. 

Since every vertex in $\partial_{\rm out}B$ is a $(\exteps, \deeppink{100})$-block vertex, 
then each $w\in \UpQ_B$ is such that 
\begin{enumerate}[(a)]
\item for every edge $e$, incident to vertex $w$, we have $\Inf_{e}\leq d^{-1/10}$,
\item we have $\sum_{z \sim w} (\InAct_{z,w})^2 \le (1+\varepsilon) d$.
\end{enumerate}
Hence, it suffices to prove that for any path $P$ inside $B$ that connects 
$w\in \UpQ_B$ with a vertex $u\in B$ which is not $(\varepsilon,0)$-block vertex, 
we have $\cappedWSA_{\varepsilon}(P) < 1$.

% Consider $u$ a heavy vertex in $B$ and the path $P$ inside $B$ from vertex $w$ to $u$. 
Letting $H$ be the set of heavy vertices in path $P$, we have
\begin{align}
\cappedWSA_{\varepsilon}(P) &= \prod\nolimits_{v\in P}\cappedWSA_{\varepsilon}(v) \ \leq \ (1-\varepsilon/4)^{|P|-|H|}\cdot \prod\nolimits_{v\in H}\cappedWSA_{\varepsilon}(v)
\nonumber\\
&\leq (1-\varepsilon/4)^{|P|-|H|}\cdot d^{|H|} \cdot \prod\nolimits_{v\in H}\WSA(v)
\enspace. 
\label{eq:MpathWeifhteps}
\end{align}
Since $w \in \UpQ_B$, while $\partial_{\rm out}B$ consists of $(\exteps, \deeppink{100})$-block vertices, we have 
\begin{align}\label{eq:Theta4ez}
\prod\nolimits_{x \in H}\WSA(x) \leq d^{-|H|} \cdot (1-\exteps/4)^{|H|-|P|-\deeppink{90}}\enspace.
\end{align}
Therefore, plugging \eqref{eq:Theta4ez} into \eqref{eq:MpathWeifhteps} gives
\begin{align}
\cappedWSA_{\varepsilon}(P) 
 & \leq \left(\frac{1-\varepsilon/4}{1-\exteps/4}\right)^{|P|-|H|} \cdot (1-\exteps/4)^{\deeppink{-90}}
 \leq \left(1-10^4\cdot\frac{\exteps}{4-\exteps}\right)^{|P|-|H|} \cdot (1-\exteps/4)^{\deeppink{-90}}\enspace, 
 \label{eq:MPBoundExtVsEps}
\end{align}
where the second inequality follows because we chose $\exteps$ such that
$\exteps=\varepsilon(1-10^{-4})$.

Notice that $(1-\exteps)^{\deeppink{-90}}$ is a fixed number independent of $d$, while 
$|P|-|H| \geq \frac{1}{2\exteps}\log d$. Hence, for any fixed number $\varrho<1$, we have $\cappedWSA_{\varepsilon}(P)<\varrho$, by choosing sufficiently large $d$. 
All the above imply that for sufficiently large $d$, $w$ is an $(\varepsilon,0)$-block vertex.
This concludes the proof of \Cref{claim:QBEBlockVertices}. 
\end{proof}

\spreadpoint

\section{Proof of \Cref{lemma:DEpsilonBlockLemma}}
\label{sec:lemma:DEpsilonBlockLemma}

\subsection{Existence of Decomposition}
Let $\varepsilon, \rad, d$ and $\beta$ be as specified in \Cref{lemma:DEpsilonBlockLemma}. Also, let $\G=\G(n,d/n)$, 
the collection of the i.i.d. standard gaussian random variables $\{\gauss_{e}\}_{e\in E(\G)}$ and the interaction 
matrix $\InAct$ obtained as described in \eqref{eq:InteractionMatrixAtBlockPartition}
with respect to $\G$ and $\{\gauss_{e}\}$.

To prove \Cref{lemma:DEpsilonBlockLemma} we describe an algorithm that generates an $(\varepsilon, \rad)$-decomposition for typical instances of $\InAct$. 
To be more precise, the output of the algorithm is either the $(\varepsilon, \rad)$-decomposition $\mathcal{B} = \{B_1, \ldots, B_N\}$ or the status ``fail". 
The algorithm returns ``fail" if it cannot find the desired block decomposition of $V_n$.

Recall that a cycle $C$ of $\G$ is called ``short" if its length is at most 
{$4\frac{\log n}{(\log d)^4}$}. We let $\cC$ be the set
of short cycles in $\G$. 
Firstly, the algorithm finds set $\cC$ and checks the validity of \Cref{eq:CHeck0Alg}.
\begin{condition}\label{eq:CHeck0Alg}
The distance of any two cycles in $\cC$ is at least $\textstyle 2\frac{\log n}{(\log d)^2}$.
\end{condition}

If \Cref{eq:CHeck0Alg} is false, then the algorithm terminates and outputs ``fail". 

Otherwise, i.e., \Cref{eq:CHeck0Alg} is true, the algorithm computes the set of $(\varepsilon, \rad)$-block vertices. With the set of these vertices available, we proceed with the creation of multi-vertex blocks in $\cB$. 
We start with the unicyclic blocks. 
To this end, we need the following condition to be true:

\begin{condition}\label{eq:CHeck1Alg}
For every $C \in \cC$, every path $P$ of length $\geq 2\cdot \maxPathL$ that does not intersect with $C$ contains 
at least one $(\varepsilon, \rad)$-block vertex $u$ such that $\dist(C,u)\geq \mincycdist$.
\end{condition}

If \Cref{eq:CHeck1Alg} is false, then the algorithm terminates and outputs ``fail".

Given that \Cref{eq:CHeck1Alg} is true, for each cycle $C\in \cC$, block $B_C$ is specified 
as follows: Let $\cQ_C$ be the set of vertices that are distance 
$ \geq \mincycdist$ from $C$. Then, $B_C$ consists of each vertex 
$V_n\setminus \cQ_C$ and every vertex in $ \cQ_C$ that is reachable by $C$ 
through a path $P$ such that $P\cap \cQ_C$ does not contain an $(\varepsilon, \rad)$-block. 
% 
% Due to \Cref{eq:CHeck1Alg}, 
Note that we have 
\begin{align}\label{eq:CHeck1AlgCor}
B_{C} \subseteq N\left(C, \textstyle{2\cdot \maxPathL} \right) , \; 
\text{ for all } C \in \cC \enspace,
\end{align}
where $N(S,r)$ is the set of vertices which is within graph distance $r$ from set $S$.

The above and \Cref{eq:CHeck0Alg}, \Cref{eq:CHeck1Alg} imply that, taking large $d$, there are no 
two cycles $C,C'\in \cC$ such that their corresponding blocks $B_C$ and $B_{C'}$ intersect.
This concludes the construction of the unicyclic blocks in $\cB$ and we proceed with
the tree-like ones. 

Let $U$ be the set that consists of each vertex $u$ which is not 
contained in any unicyclic block, while $u$ is not an $(\varepsilon, \rad)$-block vertex.

\begin{condition}\label{eq:CHeck2Alg}
For every $u \in U$, and for every vertex $w$ at distance $\geq 4\cdot \maxPathL$ from $u$, every path that connects $u$ to 
$w$ contains at least one $(\varepsilon, \rad)$-block vertex.
\end{condition}

If \Cref{eq:CHeck2Alg} is false, then the algorithm terminates and outputs ``fail".

Otherwise, i.e., if \Cref{eq:CHeck2Alg} is true, for each vertex $u\in U$ whose block has not been specified yet, the algorithm creates the block $B_u$ that consists of $u$, and 
all vertices $w$ that are reachable from $u$ through a path $(u=v_0,\ldots, v_{\ell}=w)$ such that none of the $v_i$'s is an $(\varepsilon, \rad)$-block vertex. 

Due to \Cref{eq:CHeck2Alg}, we have 
\begin{align}\label{eq:CHeck2AlgCor}
B_{u} \subseteq N\left(u, \textstyle{{4}\cdot \maxPathL }\right) , \; \text{ for all } u \in U \enspace.
\end{align}

%\charis{There is a discrepancy $2$ and $4$ in \red{\eqref{eq:CHeck2AlgCor} and \eqref{eq:CHeck1AlgCor}}}

All the above concludes the construction of the multi-vertex blocks in $\cB$. However, 
there may still be vertices $w$ whose block $B_w$ is not specified. 
For every such vertex $w$, we let $B_w=\{w\}$.

We use the following theorem to argue that Conditions \ref{eq:CHeck0Alg}, \ref{eq:CHeck1Alg} and \ref{eq:CHeck2Alg} are typically true.

\begin{theorem}\label{thm:algGivesPart}
For any $\varepsilon\in (0,1)$, \deeppink{$0 \leq \rad \leq 100$},  there exists $d_0=d_0(\varepsilon)$ such that for any 
$d\ge d_0$ and $\beta < \beta_c(d)$, the following is true:

Let the Gibbs distribution $\mu_{\InAct,\beta, \Field}$, specified by $\G(n,d/n)$ and interaction matrix $\InAct$ defined 
as in \eqref{eq:InteractionMatrixAtBlockPartition} and external field $\Field\in \mathbb{R}^{V_n}$. Also, let 
$\cE_{\rm All}$ be the event that ``Conditions \ref{eq:CHeck0Alg}, \ref{eq:CHeck1Alg} and \ref{eq:CHeck2Alg} are true". We have
\begin{align}\nonumber
\Pr[\cE_{\rm All} ]\ge \red{1-n^{-{2/3}}} \enspace,
\end{align}
where the probability term is with respect to the instances of $\mu_{\InAct,\beta, \Field}$.
\end{theorem}
The proof of \Cref{thm:algGivesPart} appears in \Cref{sec:thm:algGivesPart}.

%\charis{The argument below can be improved!}

We use \Cref{thm:algGivesPart}, to show is that when the 
algorithm does not fail $\cB$ is an $(\varepsilon, \rad)$-block decomposition. 

By construction all single-vertex blocks must be $(\varepsilon, \rad)$-block vertex, i.e. 
since all non-block vertices are considered in the first and second stages of the algorithm.
Hence, $\cB$ satisfies condition 3 in \Cref{def:BlockDecomp}.

If $B$ is a multi-vertex block then, by construction, every vertex of $\partial_{\rm out}B$ 
must be an $(\varepsilon, \rad)$-block vertex. 

Moreover, $B \cup \partial_{\rm out} B$ should contain the same number of cycles 
as $B$ does, implying that every vertex in $\partial_{\rm out}B$ has exactly one neighbour in 
$B$. 
% which implies \eqref{itm:BPSingleNeighBoundary}. 
% 
To see why this true, for the sake of contradiction assume the opposite.
That is, assume that there exists a multi-vertex block $B$ and $z\in \partial_{\rm out}B$ such that $z$ has two 
neighbours in $B$. Then, since our blocks are low diameter, our assumption implies that there is a short 
cycle in the vertices induced by $B\cup \partial_{\rm out}B$ that contains $z$. 
This is a contradiction. If $B$ is unicyclic, then the existence of this additional short cycle violates 
\Cref{eq:CHeck0Alg}. On the other hand, if $B$ is a tree, then 
this short cycle must have been considered at the first stage of the algorithm, so it cannot emerge at the second stage of the algorithm. Hence, $\cB$ satisfies condition \eqref{itm:BPSingleNeighBoundary} 
in \Cref{def:BlockDecomp}.

Note that above argument implies also that every block created in the first stage of the algorithm must be unicylic, while every block created in the second stage must be a tree. Therefore, we see that $\cB$
satisfies conditions \eqref{itm:BPunicyclic} and \eqref{itm:BPShortCycle} in \Cref{def:BlockDecomp}. 

Also, from the construction of the unicyclic block of $\cB$, it is easy to verify
that \eqref{itm:BuffCond} is true.

We now observe that no two blocks intersect. In particular, unicyclic blocks do not intersect 
 with each other due to \Cref{eq:CHeck0Alg} and \eqref{eq:CHeck1AlgCor}. 
The vertices in $U$ do not belong to any of the unicyclic blocks. 
Any two multi-vertex blocks $B_z$ and $B_x$ cannot intersect with 
each other as they are separated by using block vertices at their boundaries. 

Hence, we conclude that the set of blocks that is created by the algorithm is a 
$(\varepsilon, \rad)$-block decomposition.

\spreadpoint

\section{Proof of \Cref{thm:algGivesPart}}
\label{sec:thm:algGivesPart}

We first introduce a few notions. For integer $r\geq 0$, a path $P$ of $G$ is 
\emph{$r$-\pure} if $N_r(P)$, i.e., the set of all vertices at distance $\leq r$ from $P$, 
induces a tree in $G$. Also, for $\ell\geq 0$, we define the set of paths
\begin{align}\label{eq:DefOfPrlA}
\mathcal{P}_{r,\ell}(G) & = \left\{ P \ :\ \text{$P$ is of length $\ell$ and it is $r$-{\pure}} \right\} \enspace.
\end{align}

\begin{theorem}\label{thm:AllPathsGood}
For any $\varepsilon>0$ and integer \deeppink{$0\le \rad \le 100$} there exists $d_{0}=d_0(\varepsilon) \ge 1$ 
such that for any $d\ge d_0$, and for any $0<\beta \le \beta_c(d)$ the 
following is true:

Let the Gibbs distribution $\mu_{\InAct,\beta, \Field}$, specified by $\G(n,d/n)$ and interaction 
matrix $\InAct$ defined as in \eqref{eq:InteractionMatrixAtBlockPartition} and external field $\Field\in \mathbb{R}^{V_n}$. Then, for $r=\ell = \maxPathL$, 
let $\cE_\cP$ be the event that ``every path in $\cP_{r,\ell}(\G)$ contains an $(\varepsilon, \rad)$-block 
vertex." We have $\Pr\left[ \cE_\cP \right] \ge 1 - {n}^{-d^{1/8}}$. 
\end{theorem}

The proof of \Cref{thm:AllPathsGood} appears in \Cref{sec:thm:AllPathsGood}.

Let $\cE_S$ be the event that ``every set of vertices $S$ in $\G(n,d/n)$ with  $|S| \le 2\frac{\log n}{(\log d)^2}$, spans at most $|S|$ edges." 
\Cref{lem:SamllUnicyclicGnp} implies that   $\Pr[\cE_S]\geq 1-n^{-3/4}$.
Then,  a simple union bound implies 
\begin{align}
\Pr[\cE_\cP \cap \cE_S] \ge 1-n^{-2/3} \enspace.
\end{align}
\Cref{thm:algGivesPart} follows by showing on the events $\cE_\cP$ and $ \cE_S$ 
Conditions \ref{eq:CHeck0Alg}, \ref{eq:CHeck1Alg}, and \ref{eq:CHeck2Alg} are true.

We start by noticing that event $\cE_S$ implies \Cref{eq:CHeck0Alg}. For contradiction, 
assume that on event $\cE_S$ there exist short cycles $C , C^\prime$ within distance 
$\frac{\log n}{(\log d)^2}$ from each other. Let $P$ be a path of minimum length connecting $C$ 
to $C^\prime$. Then, for sufficiently large $d$, we have 
\begin{align*}
|C \cup C^\prime \cup P| \le \textstyle 2 \cdot 4 \frac{\log n}{\log^4 d} + \frac{\log n}{\log^2 d} < 2 \frac{\log n}{\log^2 d} \enspace,
\end{align*}
but $C \cup C^\prime \cup P$ induces a connected graph with at least two cycles. Hence, it contains 
at least $|C \cup C^\prime \cup P|+1$ edges. This is a contradiction because we have assumed $\cE_S$.

We now show the following result.

\begin{lemma}\label{lem:SmallDiamB}
On the events $\cE_\cP$ and $\cE_S$, Conditions \ref{eq:CHeck1Alg} and \ref{eq:CHeck2Alg} 
are true.
\end{lemma}

\begin{proof} 
We start by focusing on \Cref{eq:CHeck1Alg}. Assume that events $\cE_\cP$ and $\cE_S$ occur.
 
Let $B_C$ be a unicyclic block created from cycle $C \in \cC$. Let set $\hatC$ consist of the
vertices in $C$ and all other vertices which are within distance $\mincycdist$ from $C$. 
Also, let $P= (v_1,\ldots, v_{s})$ be a path that emanates from $\hatC$, 
while for $i=2,\ldots, s$, vertex $v_i$ is not in $\hat{C}$ and it is not an $(\varepsilon, \rad)$-block vertex. 

Condition \ref{eq:CHeck1Alg} follows by showing that the events $\cE_\cP$ and $\cE_S$ imply that $|P|< q$, where $q = \frac{3}{2}\cdot \maxPathL$.
We equivalently show that if $q \leq |P| \leq \frac{4}{3} q$, then $P$ must contain an $(\varepsilon,\rad)$-block vertex.

Event $\cE_S$ implies that path $P$ must 
contain a sub-path in $\cP_{r,\ell}(\G)$, where recall that $r=\ell = \maxPathL$. If no such subpath existed in $P$, then we could find another 
short cycle $C^\prime$ intersecting $P$. But then, this would imply that 
cycles $C$ and $C^{\prime}$ are so close that event $\cE_S$ is not true.
Furthermore, event $\cE_{\cP}$ implies that the sub-path of $P$
which belong to $\cP_{r,\ell}(\G)$ contains $(\varepsilon, \rad)$-block 
vertex. 
 Hence events $\cE_S$ and $\cE_{\cP}$, imply \Cref{eq:CHeck1Alg}.

We proceed with \Cref{eq:CHeck2Alg}. 
Let $W$ be the set of vertices that do not belong to any of the unicylic blocks in $\cB$. 
Also, let $U\subseteq W$ consists of each $u$ that it is not an $(\varepsilon, \rad)$-block vertex. 

Let $u\in U$ and $w\in W$, while let path $P$ be from $u$ to $w$ such that every vertex of $P$ is in set $W$. 
The events $\cE_{S}$ and $\cE_{\cP}$ imply that if $|P| \geq \frac{2q}{3}$, where recall that $q = \frac{3}{2}\cdot \maxPathL$,
then $P$ contains a sub-path $P'\in \cP_{r,\ell}(\G)$. Moreover, $P'$ has an $(\varepsilon, \rad)$-block vertex.
Hence, events $\cE_{S}$ and $\cE_{\cP}$ imply \Cref{eq:CHeck2Alg}.

 \Cref{lem:SmallDiamB} follows. 
\end{proof}

All the above conclude the proof of \Cref{thm:algGivesPart}.
\hfill{$\Box$}

\subsection{Proof of \Cref{thm:AllPathsGood}} \label{sec:thm:AllPathsGood}

Recall that the path $P$ of $\G=\G(n,d/n)$ is $r$-{\pure} if $N_r(P)$, i.e., the set of all 
vertices within distance $r$ from $P$, induces a tree in $\G$. Recall also set
\begin{align}\label{eq:DefOfPrlB}
\mathcal{P}_{r,\ell}(\G) = \left\{ P \ :\ \text{$P$ is of length $\ell$ and it is $r$-{\pure}} \right\} \enspace.
\end{align}
For what follows, we set $r=\ell =\maxPathL$.
Let the event 
\begin{align*}
\cS_A &= \{\text{every path in $\cP_{r,\ell}(\G)$ contains an $(\varepsilon, \rad)$-block vertex}\} \enspace.
\end{align*}
\Cref{thm:AllPathsGood} follows by showing $\Pr[\cS_A]\geq 1-n^{-d^{1/8}}$.
%Recall that, now, $r=\ell = \frac{\log n}{\sqrt{d}}$. 

For our proof we will also use the following refinement: for $q\le \log n$, we say that a vertex $u$ 
in $\G$ is $(\varepsilon, \rad)$-{\em block vertex of range} $q$, if the following holds
\begin{enumerate}[(a)]
\item every path $P$ of length $\leq q$ emanating from $u$, 
satisfies $\cappedWA(P) < 1$\label{itm:blockA}, 
\item for every edge $e$ within distance $\rad$ from $u$, we have $\Inf_{e}\leq d^{-1/10}$,
\item for every vertex $v$ within distance $\rad+1$ from $u$, we have
$\sum_{w\sim v} (\InAct_{v,w})^2 \le (1+\varepsilon) d$.
\end{enumerate}
Clearly, an $(\varepsilon, \rad)$-block vertex corresponds to an $(\varepsilon, \rad)$-block vertex of range $\log n$. 

We establish \Cref{thm:AllPathsGood} in two steps. 
First, we consider the event
\begin{align*}
\cS_B &= \{\text{every path in $\cP_{r,\ell}(\G)$ contains an 
$(\varepsilon, \rad)$-block vertex of range $r$}\}\enspace.
\end{align*}
Note that $\cS_B$ is different from $\cS_A$ in that it considers block vertices of range $r= \maxPathL$, rather than standard block vertices. Next, we consider the event
\begin{align*}
\cS_C &= \{\text{for every path $P$ in $\G$, with $r \le |P| \le \log n$, we have $\cappedWA(P) < 1$}\} \enspace. 
\end{align*}
Note that $\cS_C$ is about the weights of paths $P$ of length between $r$ and $\log n$.

Furthermore, in \Cref{sec:lem:ClacOfR} and \Cref{subsec:IndConstr} 
we prove the following results:

\begin{lemma}\label{lem:ClacOfR}
Under the hypotheses of \Cref{thm:AllPathsGood} for 
$\varepsilon, \rad, d$ and $\beta$, we have $\Pr[\cS_C] \ge 1 - n^{-d^{1/5}}$.
\end{lemma}

\begin{proposition}\label{thm:AllPathsGoodPart1}
Under the hypotheses of \Cref{thm:AllPathsGood} for $\varepsilon, \rad, d$, 
$\beta$, we have $\Pr[\cS_B] \ge {1 - n^{-d^{10/75}}}$.
\end{proposition}

% Note that $\cS_C\cap \cS_B$ implies $\cS_A$. 

% This follows from the observation that, under the event $\cS_C$ every block vertex of range $r$ is also a standard block vertex. 

Observing that $\cS_A \supseteq \cS_C\cap\cS_B$, and applying the union bound gives that:
\begin{align*}
\Pr[\cS_A] \ge \Pr[\cS_C \cap \cS_B] =1- \Pr\left[\ \overline{\cS_C} \cup \overline{\cS_B}\ \right]
\ge 1 - n^{-d^{1/5}} - n^{-d^{10/75}}
%\ge 1 - 2\cdot n^{-d^{10/75}} 
\ge 1 - n^{-d^{1/8}} \enspace.
\end{align*}

 \Cref{thm:AllPathsGood} follows. \hfill $\Box$

\newcommand{\GWT}{\pmb{\UpT}}
\newcommand{\ST}{\pmb{T}}

\subsection{Proof of \Cref{thm:AllPathsGoodPart1}}\label{subsec:IndConstr}
To ease the notation, for what follows consider the space induced by 
$(\G, \{\gauss_{e}\}, \beta)$, where $\G=\G(n,d/n)$, $\{\gauss_{e}\}_{e\in E(\G)}$ 
are i.i.d. standard Gaussians and $\beta\in \mathbb{R}_{>0}$. 
The triple $(\G, \{\gauss_{e}\}, \beta)$ encodes all the necessary information about the interaction matrix $\InAct$ and the Gibbs distribution $\mu_{\InAct,\beta,\Field}$ 
we need for the analysis since the external field $\Field$ does not have any effect. 

For what follows, we have $r=\ell=\maxPathL$. Recall that event $\cS_B$ is defined by
\begin{align*}
\cS_B &= \{\text{every path in $\cP_{r,\ell}(\G)$ contains an $\varepsilon$-block vertex of range $r$}\}\enspace,
\end{align*}
where set $\cP_{r,\ell}(\G)$ is defined in \eqref{eq:DefOfPrlA}.

For every $(\ell+1)$-tuple $P =(v_0, \ldots, v_{\ell})$ of vertices inducing a path in $\G$ we define the events
\begin{align*}
\cC_P &= \{ P \in \mathcal{P}_{r,\ell}(\G)\} &&\text{and}&&
\cE_P = \left\{P \text{ contains $<({4}/{10}) \ell$ many $(\varepsilon, \rad)$-block vertices of range $r$}\right\} \enspace.
\end{align*}
We define
\begin{align}\label{eq:DefOfBadPathBlock}
X &= \sum\nolimits_{P \in V^{(\ell+1)}_n} \Ind_{\cE_P} \times \Ind_{\cC_P} \enspace.
\end{align}
It is easy to verify that 
\begin{align}\label{eq:Step1Forthm:AllPathsGoodPart1}
\Pr[\overline{\cS}_B] \leq \Pr[X>0]\enspace,
\end{align}
where $\overline{\cS}_B$ is the complement of the event ${\cS}_B$.

For $P=(v_0, \ldots, v_{\ell})$ in ${V_n}^{\ell+1}$, we define the following 
random process generating the random tree $\GWT = \GWT(P,d)$. Starting from 
set $P$ we execute the following steps:
\begin{enumerate}
\item we (deterministically) add the edges $\{v_{i-1},v_{i}\}$, for every $i \in \{1,\ldots, \ell\}$, making $P$ a path,
\item from each $v_i$ in $P$, we hang a GW tree rooted at $v_i$, with offspring distribution {${\tt Binom}(n, d/n)$},
\item we truncate each GW tree at depth $r$. 
%\item for every edge $e$ of $\GWT$, we sample $\pmb{\UpJ}_e$ according to $\cN(0,1)$, independently.
\end{enumerate}

\noindent
Consider $(\GWT,\{\overline{\gauss}_{e}\}, \beta)$, where 
$\{\overline{\gauss}_{e}\}_{e\in E(\GWT)}$ are i.i.d. standard Gaussians.
We define $\varepsilon$-weights at the vertices and paths of $\GWT$ as described in \Cref{sec:WeightBlockPartition}, using $\{\overline{\gauss}_{e}\}$ and $\beta$. 
With  this weighting-scheme, we have the event 
\begin{align}\label{eq:DefineET}
\cE_{\GWT} &= \{P \text{ contains fewer than $({4}/{10}) \ell$ many $(\varepsilon, \rad)$-block vertices of range $r$ in } \GWT \} \enspace. 
\end{align}
\begin{lemma}\label{lem:couplingDomin}
For $P\in V^{\ell+1}_n$, let $\G_P$ be an instance of $\G(n,d/n)$ conditional on $P$ being a path. 
Also, let $\GWT=\GWT(P,d)$. There exists a coupling between $(\G_P,\{\gauss_{e}\}, \beta)$ and $(\GWT,\{\overline{\gauss}_{e}\}, \beta)$, such that 
$$\Ind_{\cE_P} \times \Ind_{\cC_P} \le \Ind_{\cE_\UpT}\enspace.$$
\end{lemma}

\begin{proof}

We consider the following ``BFS tree of the path $P$" in $\G_P$, which we denote with $\ST$.

Specifically, for $i = 0, 1 , \ldots, \ell$, we obtain a tree $\ST_i$ of depth $r$, hanging at 
vertex $v_i$ using the following recursive exploration procedure.

For $0 \le s \le r-1$, assume that the first $s$ levels of $\ST_i$ have been revealed, and let 
$\{x_1, \ldots, x_m\}$ be the vertices of its $s$-th level. Then, we obtain the level $(s+1)$ 
% of $\ST_i$ we proceed as follows. For $j = 1, \ldots, m$, we 
by revealing $x_j$'s neighbors in $\G_P$, excluding the connection with vertices 
that have been already explored, or belong to the path $P$, for $1\leq j\leq m$. 
That is, the children of $x_j$ in $\ST_i$ 
will be precisely its neighbors in $\G_P$ that do not belong to: (i) any tree $\ST_{l}$ for $l < i$, (ii) any 
of the vertices that have been revealed to belong in $\ST_i$, up to this point, (iii) the path $P$.
Tree $\ST$ is the union of the $\ST_i$'s and $P$.

The offspring of each vertex $u$ in $\ST$, is dominated by ${\tt Binom}(n, d/n)$. Therefore, we 
can couple the graph structures of $\G$ and $\GWT$, so that there
exists $\widehat{\GWT}$, subtree of $\GWT$ which is isomorphic to $\ST$.
Let the adjacency preserving bijection $h:V(\ST)\to V(\widehat{\GWT})$
between $\widehat{\GWT}$ and $\ST$.

For each pair of edges $e=\{u,w\}$ in $E(\ST)$ and $\hat{e}=\{h(u), h(w)\}$ in 
$E(\widehat{\GWT})$ we couple $\gauss_e$, $\overline{\gauss}_{\hat{e}}$ identically, i.e.,
we have $\overline{\gauss}_{\hat{e}}=\gauss_e$. For the remaining edges $e$ of $\GWT$, we draw 
$\overline{\gauss}_e\sim N(0,1)$, independently.

All the above, complete the coupling construction.
 
Consider the subgraph of $\G_P$ induced by $N_r(P)$, and note that this will be different from $\ST$ 
exactly when $N_r(P)$ contains at least one cycle. In that case we have $\Ind_{\cC_P} = 0$, and thus, 
$\Ind_{\cE_P} \times \Ind_{\cC_P} \le \Ind_{\cE_{\GWT}}$ holds trivially.

On the other hand, if $N_r(P)$ is a tree, then we have that $\Ind_{\cC_P} = 1$, and $N_r(P)$ coincides 
with $\ST$. Since the number of block vertices of $P$ decreases if we add more vertices or more edges 
in $N_r(P)$, and $\ST$ is a subtree of $\GWT$ we have that 
$\Ind_{\cE_P} \times \Ind_{\cC_P} = \Ind_{\cE_P} \le \Ind_{\cE_{\GWT}}$.

All the above conclude the proof of \Cref{lem:couplingDomin}.
\end{proof}

Furthermore, we have the following result.

\begin{proposition}\label{prop:ExpectedBadPaths}
For $\varepsilon>0$, for \deeppink{$1\le \rad \le 100 $}, there exists $d_0=d_0(\varepsilon)\ge 1$, such that for all $d\ge d_0$, and 
$\beta \le \beta_c(d)$, the following is true:

For $r=\ell = \maxPathL$, and $P\in V^{\ell+1}_n$, let 
$(\GWT,\{\overline{\gauss}_{e}\}, \beta)$ where $\GWT=\GWT(P,d)$. We have
\begin{align*}
\Pr[\cE_{\GWT}] &\le n^{-d^{1/7}} \enspace,
\end{align*}
where event $\cE_{\GWT}$ is defined by \eqref{eq:DefineET}.
\end{proposition}
The quantification of parameter $\rad$ in \Cref{prop:ExpectedBadPaths} is necessary to 
specify event $\cE_{\GWT}$.
The proof of \Cref{prop:ExpectedBadPaths} appears in \Cref{sec:AnalysisIndCon}.

 \Cref{lem:couplingDomin} and \Cref{prop:ExpectedBadPaths}, imply that 
\begin{align}\label{eq:ExpOfBadPathBlock}
\Pr[X>0] \le n^{-d^{10/75}} \enspace,
\end{align}
where the random variable $X$ is defined in \eqref{eq:DefOfBadPathBlock}. 
Specifically, we have 
\begin{align*}
\Exp[X] &= 
\sum_{P \in V^{(\ell+1)}} \Exp\left[ \Ind_{\cE_P} \times \Ind_{\cC_P} \right] 
 \le
n^{(\ell+1)} \cdot \left({d}/{n}\right)^{\ell} \cdot
 \Exp\left[ \Ind_{\cE_P} \times \Ind_{\cC_P} \mid P \text{ is a path} \right] 
 \enspace,
\end{align*}
which due to \Cref{lem:couplingDomin}, and \Cref{prop:ExpectedBadPaths}, we have
\begin{align*}
\Exp[X] &\le
n \cdot {d}^{\ell} \cdot \Exp\left[ \Ind_{\cE_{\GWT}} \right]
 \le n \cdot {n}^{\frac{\log d}{\lDenom} } \cdot n^{-d^{1/7}}
 \enspace.
\end{align*}
Therefore, there exists $d_0=d_0(\varepsilon)\ge 1$, such that 
for every $d \ge d_0$, \eqref{eq:ExpOfBadPathBlock} is true. 

Since $X$ is non-negative, Markov's inequality and \eqref{eq:ExpOfBadPathBlock} yield 
\begin{align*}
\Pr[X>0] \le \Exp[X] \le {n^{-d^{10/75}}} \enspace.
\end{align*}
\Cref{thm:AllPathsGoodPart1} follows by plugging the above into \eqref{eq:Step1Forthm:AllPathsGoodPart1}.
\hfill $\Box$

\spreadpoint

\section{Proof of \Cref{prop:ExpectedBadPaths}} \label{sec:AnalysisIndCon}

Recall that $r=\ell = \maxPathL$. 
Let $P=(v_0, \ldots, v_{\ell}) \in {V_n}^{\ell+1}$ and $\GWT =\GWT(P,d)$.
Consider an instance of $(\GWT,\{\overline{\gauss}_{e}\}, \beta)$ as in 
\Cref{subsec:IndConstr}. For $\varepsilon, \rad$ and $\beta$ as specified in the
statement of \Cref{prop:ExpectedBadPaths}, consider $\varepsilon$-weights
for the vertices and paths in $\GWT$.

\begin{definition}[Vertex Weight induced by Path]
For every $i\in \{0, \ldots, \ell\}$ we define the weight
\begin{align} \nonumber %\label{eq:DefineWB}
\WB_P(v_i) & = \max\nolimits_{Q} \left\{ \cappedWA(Q) \right\}\enspace, 
\end{align}
where $Q$ varies over all paths of length $\leq r$, that emanate from $v_i$ and do not intersect with $P$. %, i.e., they to do not share vertices. 
\end{definition}

\begin{definition}[Left/Right-Block Vertex]\label{def:W-Weights}
We say that vertex $v_j\in P$ is an
\begin{itemize}
 \setlength\itemsep{0.5em}
	\item {\em left-block vertex} with respect to $P$, if for all $ t\le j$, we have that
	$
	\prod^{j}_{k=t} \WB_{P}(v_k) <1
	$,
	\item {\em right-block vertex} with respect to $P$, if for all $ t \ge j$, we have that
	$
	\prod^{t}_{k=j} \WB_{P}(v_k) <1 
	$.
\end{itemize}
\end{definition}

Perhaps it is useful to remark that there is no assumption as to what is $\InAct_{e}$
for all edges within distance $\rad$ of a (right) left-block vertex $w$. 
% for any left-block vertex $w$ and its neighbour $z$. 
% Similarly for the right-block vertices. 
We have the following lemma.

\begin{lemma}\label{lemma:LRBlockVsBlockVertex} 
 If a vertex $v_i$ of $P$, is both left-block, and right-block vertex with respect to $P$, then every path $Q$ of length $\leq \log n$ that emanates from $v_i$, satisfies $\cappedWA(Q) < 1$.
\end{lemma}

\begin{proof}[Proof of \Cref{lemma:LRBlockVsBlockVertex}]
Let 
$Q$ be an arbitrary path of length $r$, starting at vertex $v_i$. We will show that $\cappedWA(Q) < 1$. Let $v_j$ be the last vertex of $P$ that appears on~$Q$, (note that this is well-defined since~$Q$ starts at $v_i$, and ${\GWT}$ is a tree). 

W.l.o.g. assume that $j\ge i$, using just the fact that $v_i$ is right-block, we will show that $\cappedWA(Q) < 1$, 
(if $j\le i$, we use the fact that $v$ is left-block). Let $Q_1$ be the sub-path of $Q$ staring at $v_i$ and arriving at $v_{j-1}$, 
and $Q_2$ be the sub-path of $Q$ staring at $v_j$ and arriving at the last vertex of $Q$, (so that $Q$ is the concatenation of $Q_1$ and $Q_2$). 
We have
\begin{align*} %\textstyle
 \cappedWA(Q) &= \cappedWA(Q_1) \cdot \cappedWA(Q_2) 
 =
 \left( \prod^{j-1}_{k=i} \cappedWA(v_{k}) \right)\cdot
 \cappedWA(Q_2)
 \le 
 \left( \prod^{j-1}_{k=i} \WB_{P}(v_{k}) \right)\cdot
 \WB_{P}(v_j)
 <1 \;\enspace,
\end{align*}
where the first inequality follows from the definition of $\WB_P$, and the last inequality is due to the fact that $v_i$ is a 
right-block vertex.
\end{proof}

We define the events
\begin{align}
{\cE^{\mathrm{left}}_{\GWT}} 
 &= \left\{P \text{ contains fewer than {$(1 -10^{-3})\ell$} vertices that are left-block in }{\GWT}\right\} \enspace,\\
{\cE^{\mathrm{right}}_{\GWT}} 
 &= \left\{P \text{ contains fewer than {$(1 - 10^{-3})\ell$} vertices that are right-block in }{\GWT}\right\} \enspace.
\end{align} 
We call a vertex $u\in P$ light, if 
\begin{itemize}
\item for every edge $e$ within distance $\rad$ from $u$, we have $\Inf_{e}\leq d^{-1/10}$,
\item for every vertex $v$ within distance $\rad+1$ from $u$, we have
$\sum_{w\sim v} (\InAct_{v,w})^2 \le (1+\varepsilon) d$.
\end{itemize}
 We let the event 
\begin{align}
{\cE^{\mathrm{light}}_{\GWT}} 
 &= \left\{P \text{ contains fewer than } ({99}/{100})\ell \text{ vertices in } \GWT \text{ that are light}\right\} \enspace.
\end{align}
Recall
\begin{equation*}
\cE_{\GWT}
 = \left\{P \text{ contains fewer than {$(1-2\cdot 10^{-3})\ell$} vertices that are $(\varepsilon, \rad)$-block of range $r$ in } {\GWT} \right\} \enspace.
\end{equation*}
Using \Cref{lemma:LRBlockVsBlockVertex} and a simple counting argument, it is not hard to verify that 
\begin{align}\label{eq:ETVsERightLeftLight}
\cE_{\GWT}\subseteq {\cE^{\mathrm{left}}_{\GWT}} \cap {\cE^{\mathrm{right}}_{\GWT}} \cap {\cE^{\mathrm{light}}_{\GWT}}\enspace.
\end{align}
We will prove the following result.

\begin{lemma}\label{lem:epsLight}
For any $\varepsilon>0$ and \deeppink{$0\leq \rad \leq \red{\sqrt{\log d}}$}, 
there exists $d_0=d_0(\varepsilon)\ge 1$, such that for all $d\ge d_0$, and 
$0<\beta \le \beta_c(d)$, the following is true:

For $r=\ell = \maxPathL$, and $P=(v_0,\ldots, v_{\ell}) \in V_n^{\ell+1}$, let $(\GWT,\{\overline{\gauss}_{e}\}, \beta)$ where $\GWT=\GWT(P,d)$. Then,
\begin{align} \label{eq:lem:epsLight}
\Pr\left[
{\cE^{\mathrm{light}}_{\GWT}} \right] &\le n^{-\frac{d^{1/5}}{4}}
\enspace.
\end{align}
\end{lemma}

\newcommand{\pbad}{{\rm P}_{\rm Bad}}

\begin{proof}[Proof of \Cref{lem:epsLight}]

Recall that $N_\rad(v,\GWT) = N_\rad(v)$ denotes the set of vertices in $\GWT$ that are within 
distance $\rad$ from $v$. For $v \in V(\GWT)$ we define the events 
$\cK_1(v) =\{\exists w, z \in N_{\rad}(v) \text{ with } \Inf_{w,z} > d^{-1/10} \}$, and 
$\cK_2(v) = \{ \exists w \in N_{\rad}(v) \text{ with } \sum_{z \sim w} (\InAct_{w,z})^2 > (1+\varepsilon) d\}$. 
We claim that 
\begin{align}\label{eq:OtherBNsh}
\Pr[\cK_i\text{ occurs for} < ({99.5}/{100})\ell \text{ of vertices in } P] \le n^{-\frac{d^{1/5}}{2}}, && i =1,2\enspace.
\end{align}
Below we establish \eqref{eq:OtherBNsh} only for $\cK_1$, as the corresponding inequality for $\cK_2$
is obtained by similar arguments. For each vertex $v \in V(\GWT)$, let 
\begin{align*}
\pbad(v) = \Pr \left[ \exists \; w,z\in N_\rad(v) \text{ with } \Inf_{\{w,z\}} > d^{-1/10} \right]
\enspace.
\end{align*}
We define the event $\cE_{v,\rad} = \{ \deg(w) \le d^2, \forall w \in N_\rad(v) \}$. From the law of total 
probability, we have
\begin{align}\label{eq:TotPrPBad}
\pbad(v) & \le \Pr\left[ \exists \; w,z\in N_\rad(v) \text{ with } \Inf_{\{w,z\}} > d^{-1/10} \mid \cE_{v,\rad} \right] + \Pr[\cE_{v,\rad}]\enspace. 
\end{align}
For an edge $e \in E(\GWT)$ to be heavy, i.e., $\Inf_{e} > d^{-1/10}$, we have
\begin{align}\label{eq:Gammas_i}
\Pr[\Inf_{e} > d^{-1/10}] &\le \Pr[\beta \cdot |\gauss_{e}|> d^{-1/10}] \le \sqrt{\frac{2}{\pi} \cdot \frac{\upkappa}{d}} \cdot d^{1/10}\cdot 
\exp\left(- \frac{d^{8/10}}{\upkappa }\right) 
\enspace, 
\end{align}
where the last inequality follows from the tail bound of the standard gaussian $\gauss_{e}$.

Since on the event $\cE_{v,\rad}$ we have that $N_{\rad}(v)$ spans at most $d^{2\rad}$ edges 
in $\GWT$, we have
\begin{align}\label{eq:bou2f}
\Pr\left[\exists \; w,z\in N_\rad(v) \text{ with } \Inf_{\{w,z\}} > d^{-1/10}
\mid \cE_{v,\rad}\right] \le d^{2\rad}\cdot\exp\left(- d^{8/10}/\upkappa \right)\enspace.
\end{align}
Moreover, using standard Chernoff bound, e.g., see \cite{boucheron2003concentration}, and the 
union bound we see
\begin{align}\label{eq:bou1f}
\Pr[\cE_{v,\rad}] \le d^{2\rad} \cdot
\Pr[\deg(v) > d^2] \le d^{2\rad}\cdot\exp(-d^2)\enspace. 
\end{align}

Therefore, plugging the bounds \eqref{eq:bou1f}, \eqref{eq:bou2f} into \eqref{eq:TotPrPBad} 
we get
\begin{align}
\pbad(v)\le 2 d^{2\rad}\cdot\exp\left(-4d^{8/10}\right)\enspace.
\end{align}
We split the vertices of $P$ into $\rad$ groups. In particular, let
\begin{align*}
V_j = \{ v_i \in V(P): i \equiv j\pmod{\rad} \} && \text{ for } j=1,2, \ldots, \rad
\enspace.
\end{align*}
 From the union bound, we get
\begin{align}\Pr[\cK_1 \text{ occurs for} < ({99.5}&/{100})\ell \text{ of vertices in } P]\nonumber\\
&\le \sum\nolimits_{ 1 \le s \le \rad}
\Pr\left[\cK_1 \text{ occurs for} < ({99.5}/{100})\ell \text{ of vertices in } V_s \right] 
\enspace. \label{eq:BreakV12}
\end{align}
Since the couplings on the edges of $\GWT$ are drawn independently, 
each of the vertices in $V_s$ is light independently of the others.
Thus, we see that the first term in the lhs \eqref{eq:BreakV12} is overestimated by
\begin{align}
S &=\Pr\left[\cK_1 \text{ occurs for} < ({99.5}/{100})\ell \text{ of vertices in } V_s \right] \nonumber \\
& = \sum\nolimits^{\ell}_{t = \lfloor 0.005 \ell\rfloor} \binom{\ell}{t} \cdot 
(1-\pbad(v))^{\ell-t}(\pbad(v))^t\enspace. 
\end{align}
Plugging $\binom{\ell}{t}\leq 2^{\ell}$ and $(1-\pbad(v))^{\ell-t}\leq 1$ into the above inequality, 
we have
\begin{align*} 
S&\le 2^\ell \sum\nolimits^{\ell}_{t = \lfloor 0.005 \ell\rfloor} (\pbad(v))^t 
\le 2^\ell\cdot 2^{0.005\ell}\cdot d^{0.01\cdot \rad \cdot \ell}\cdot\exp\left(-0.02\cdot d^{8/10}\cdot \ell \right)
\le n^{-d^{1/5}} \enspace,
\end{align*}
for $d$ sufficiently large. From symmetry, the same bound holds for every $V_j$, i.e., 
\begin{align}\label{eq:V2bound}
\Pr\left[\cK_1 \text{ occurs for} < ({99.5}/{100})\ell \text{ of vertices in } V_j \right] \le n^{-d^{1/5}}, &&
\text{ for all } j \in [\rad]
\enspace.
\end{align}
Plugging \eqref{eq:V2bound} into \eqref{eq:BreakV12} we get \eqref{eq:OtherBNsh} for $\cK_1$. Using similar 
arguments we establish \eqref{eq:OtherBNsh} for $\cK_2$. By the union bound 
\begin{align*}
\Pr\left[ {\cE^{\mathrm{light}}_{\GWT}} \right] &
\le \sum\nolimits_{i=1,2}\Pr\left[
\cK_i\text{ occurs for} < ({99.5}/{100})\ell \text{ of vertices in } P \right] \le n^{-\frac{d^{1/5}}{4}} \enspace,
\end{align*} 
concluding the proof of \Cref{lem:epsLight}. 
\end{proof}

\begin{proposition}\label{thrm:PrbOfManyLeftRight}
For any $\varepsilon>0$, there exists $d_0=d_0(\varepsilon)\ge 1$, such that for all $d\ge d_0$, and 
$0<\beta \le \beta_c(d)$, the following is true:

For $r=\ell = \maxPathL$, and $P=(v_0,\ldots, v_{\ell}) \in V_n^{\ell+1}$, let $(\GWT,\{\overline{\gauss}_{e}\}, \beta)$ where $\GWT=\GWT(P,d)$. Then,
\begin{align}\label{eq:BoundOnSmallNumOfLRB}
 \Pr\left[{\cE^{\mathrm{left}}_{\GWT}} \right] \le n^{-d^{1/6}}, && \text{ and }&& \Pr\left[{\cE^{\mathrm{right}}_{\GWT}} \right] \le n^{-d^{1/6}}\enspace,
\end{align}
where the probability is over the instances of $(\GWT,\{\overline{\gauss}_{e}\}, \beta)$.
\end{proposition}

Using \Cref{thrm:PrbOfManyLeftRight}, \Cref{lem:epsLight}, \eqref{eq:ETVsERightLeftLight} and the union bound, we further get
\begin{equation*}
\Pr\left[\cE_{\GWT} \right] 
\le 
\Pr\left[{\cE^{\mathrm{left}}_{\GWT}}\right] + \Pr\left[{\cE^{\mathrm{right}}_{\GWT}}\right] +
\Pr\left[{\cE^{\mathrm{light}}_{\GWT}}\right]
\le 4 \cdot n^{-d^{1/6}} 
\le n^{-d^{1/7}}
\enspace.
\end{equation*}
The above concludes \Cref{prop:ExpectedBadPaths}. 
\hfill{}$\Box$

\subsection{Proof of \Cref{thrm:PrbOfManyLeftRight}} \label{sec:thrm:PrbOfManyLeftRight} 
We prove the bound of \Cref{thrm:PrbOfManyLeftRight} only for $\Pr[{\cE^{\mathrm{left}}_{\GWT}}]$, as the derivations for $\Pr[{\cE^{\mathrm{right}}_{\GWT}}]$ are identical. 

We define the set of \emph{heavy} vertices of $P$ by , 
\begin{align*}
H = H(P)=\left\{i\in \{0, \ldots, \ell\}:\WB_{P}(v_i)\ge1\right\} 
\enspace .
\end{align*}
Recall the definition of $\WB_{P}(v_i)$ from \Cref{def:W-Weights}. 

For $i \in H$, let $t_i$ be the greatest index in $\{i, i+1, \ldots, \ell\}$ such that for all $ t \in \{i, i+1, \ldots, t_i\}$
\begin{align*}
\prod\nolimits_{i\leq k \leq t}\WB_{P}(v_k) \ge 1 \enspace,
\end{align*}
and let $L_i=\{i, i+1, \ldots, t_i\}$. For $i \notin H$, define $L_i =\emptyset$, and let $ L= L_0 \cup L_1 \cup \ldots \cup L_{\ell}$.

\begin{lemma}\label{lemma:LVsLeftblock}
For all $j \in \{0,1,\dots, \ell\} \setminus L$, vertex $v_j$ is a left-block vertex with respect to $P$.
\end{lemma}

\begin{proof}
We show the contrapositive, i.e., assuming $v_j$ is not left-block with respect to $P$, we show that $j \in L$. If $v_j$ is not left-block, then there must exist at least one index $i \in \{0,1, \ldots, j\}$ such that 
$
 \prod^{j}_{k=i}\; \WB_{P}(v_k) \ge 1 
$.

Let $i^*$ be the greatest such index; we claim that $j \in L_{i^*}$, i.e., for every $t \in \{i^*, i^*+1, \ldots, j\}$ we have $\prod^{t}_{k=i^*}\; \WB_{P}(v_k) \ge 1 $. Indeed, if there was a $t \in \{i^*, i^*+1, \ldots, j\}$ such that $\prod^{t}_{k=i^*}\; \WB_{P}(v_k) < 1 $, then we would have that $\prod^{j}_{k=t+1}\; \WB_{P}(v_k) \ge 1 $, which contradicts with the fact that $i^*$ is the greatest such index. 
\end{proof}

The result below shows that, typically, only a small fraction of vertices in $P$ to belong in $L$.

\begin{theorem}\label{prop:BoundingL}For any $\varepsilon>0$, there exists $d_0=d_0(\varepsilon)\ge 1$, such that for all $d\ge d_0$, and 
$0<\beta \le\beta_c(d)$, the following is true:

For $r=\ell = \maxPathL$, and $P=(v_0,\ldots, v_{\ell}) \in V_n^{\ell+1}$, let 
$(\GWT,\{\overline{\gauss}_{e}\}, \beta)$ where $\GWT=\GWT(P,d)$. Then,
\begin{align}\nonumber
 \Pr\left[|L|\geq \ell \right]\leq n^{-d^{1/6}} \enspace. 
\end{align}
\end{theorem}
In light of \Cref{lemma:LVsLeftblock}, and \Cref{prop:BoundingL}, \Cref{thrm:PrbOfManyLeftRight} follows.
Therefore, to finish proving \Cref{thrm:PrbOfManyLeftRight}, it only remains to prove \Cref{prop:BoundingL}.

\spreadpoint
\section{Proof of \Cref{prop:BoundingL}}
 \label{sec:prop:BoundingL}

Applying the union bound, we see that
\begin{align}\label{eq:UnionOnL}
|L|&=\left | \bigcup\nolimits_{i\in H} L_i \right| \leq \sum\nolimits_{i\in H}|L_i|\enspace. 
\end{align}
We use the following strategy to estimate each $|L_i|$ for $i\in H$. 
Each heavy vertex, $v_i$, introduces a weight, $\cappedWB_P(v_i)>1$. If $L_i = \{i, i+1\ldots,j\}$, then notice that $\{i, i+1\ldots,j+1\}$ is the first interval on the right of $i$ that ``absorbs'' weight $\WB_P(v_i)$, i.e., $j$ is the smallest index in $\{0, \ldots, i\}$ such that 
$$
{\prod\nolimits_{i\leq k \leq j+1} \WB_{P}(v_k) < 1} \enspace.
$$
Let $\theta = \varepsilon/4$. If $H \cap L_i=\emptyset$, we observe that $v_i$ requires $\lceil - \log_{(1-\theta)} \cappedWB_P(v_i) \rceil$ light vertices to get $\cappedWB_P(v_i)$ absorbed. 
Similarly, if $ H \cap L_i\neq \emptyset$, the number of light vertices needs to be
$$
\left\lceil - 
\sum\nolimits_{j \in H\cap L_i} \log_{(1-\theta)} \WB_P(v_j)\right\rceil 
\enspace.
$$
Due to \eqref{eq:UnionOnL}, the above yields 
the following corollary

\begin{corollary}\label{lem:BreakIntoPartsToBound}
Let $ \theta = \varepsilon/4$, then we have that
\begin{align}\label{eq:BreakIntoPartsToBound}
|L| \le |H| - \sum\nolimits_{i\in H} \log_{(1-\theta)} \WB_P(v_i)
\enspace.
 \end{align}
\end{corollary}

We also use the following technical lemma, which we prove in \Cref{sec:lem:MGFForcappedWB}.

\begin{lemma}\label{lem:MGFForcappedWB}
For any $\varepsilon>0$, there exists $d_0=d_0(\varepsilon)\ge 1$, such that for all $d\ge d_0$, and 
$0<\beta \le\beta_c(d)$, the following is true:

For $r=\ell = \maxPathL$, and $P=(v_0,\ldots, v_{\ell}) \in V_n^{\ell+1}$, 
let $(\GWT,\{\overline{\gauss}_{e}\}, \beta)$ where $\GWT=\GWT(P,d)$. Then, 
\begin{align}\nonumber
\Exp\left[( \WB_P(v_i) )^t\right] &\le \left(1-{\varepsilon}/{4}\right)^{q} 
& \textrm{for } 0\leq i\leq \ell \enspace,
\end{align}
where $t= d^{95/100}$ and $q= d^{93/100}$.
\end{lemma}

We also have the following result. 

\begin{lemma}\label{prop:BoundOnHevy} 
For any $\varepsilon>0$, there exists $d_0=d_0(\varepsilon)\ge 1$, such that for all $d\ge d_0$, 
and $0<\beta \le\beta_c(d)$, the following is true:

For $r=\ell = \maxPathL$, 
and $P=(v_0,\ldots, v_{\ell}) \in V_n^{\ell+1}$, let $(\GWT,\{\overline{\gauss}_{e}\}, \beta)$ where 
$\GWT=\GWT(P,d)$. Then, 
\begin{align}
\Pr\left[|H|\geq \ell \cdot d^{-1/10} \right ]\leq n^{-d^{1/4}} 
\enspace,
\end{align}
where $H$ is the set of heavy vertices defined above.
\end{lemma}

\begin{proof}[Proof of \Cref{prop:BoundOnHevy}]
Let $I_1$, $I_2$ ,be the sets of odd, even indices of $\{0, \ldots, \ell\}$, respectively. Writing 
$H_1 =H \cap I_1$, $H_2 =H \cap I_2$, and using the union bound, it is easy to see that
\begin{align}
\Pr\left[|H|\geq \ell \cdot d^{-1/10} \right ] 
\le 
\Pr\left[|H_1|\geq (1/2)\cdot \ell \cdot d^{-1/10} \right ]
+
\Pr\left[|H_2|\geq (1/2)\cdot \ell \cdot d^{-1/10} \right ]
\enspace.
\end{align}
$H_1$ and $H_2$ are identically distributed. Hence, it suffices to show that
\begin{align*}
\Pr\left[|H_1|\geq (1/2)\cdot \ell \cdot d^{-1/10} \right ] \le (1/2) \cdot n^{-d^{1/4}} 
\enspace.
\end{align*}
Notice that for any two indices $i, j \in I_1$, the corresponding random variables $\WB_{P}(v_i), \WB_{P}(v_j)$, 
are i.i.d., and thus, each index $i \in I_1$ belongs to $H_1$ with probability at most $\Pr[\WB_{P}(v_1) \ge 1]$. 
Specifically, per \Cref{lem:MGFForcappedWB}, and Markov's inequality, we have 
\begin{align*}
\Pr\left[\WB_{P}(v_1) \ge 1 \right] \le 
\Exp\left[\left(\WB_{P}(v_i)\right)^t\right] \le \left(1-{\varepsilon}/{4}\right)^{q}
\enspace,
\end{align*}
where $t= d^{95/100}$ and $q= d^{93/100}$.

$|H_1|$ is upper bounded by the number of successes of a binomial distribution with $\ell/2$ number 
of trials, and probability of success $\left(1-\frac{\varepsilon}{4}\right)^{q}$. Hence, for $m=\frac{1}{2}\frac{\ell}{d^{1/10}}$ 
and noting that $m\gg \frac{\ell}{2}(1-\varepsilon/4)^q$, we have
\begin{align*}
\Pr\left[|H_1|\ge m \right] 
&\le \sum\nolimits_{m \leq k \leq \ell/2} \binom{\ell/2}{k} \cdot \exp\left(\log {\left(1-{\varepsilon}/{4}\right)}\cdot q \cdot k \right)\\
&\le \frac{\ell}{2} \cdot \binom{\ell/2}{m} \cdot \exp\left(\log {\left(1-{\varepsilon}/{4}\right)}\cdot q \cdot m \right)\\
&\le \frac{\ell}{2}\cdot \left(\frac{e\cdot \ell}{2m}\right)^{m} \cdot \exp\left(\log {\left(1-{\varepsilon}/{4}\right)}\cdot q \cdot m \right)\\
&\le (\log n) \exp \left(\log n \left( d^{-1/2} +\log(1-\varepsilon/4) \cdot (1/2) \cdot \lBlockBound \cdot d^{\frac{33}{100}} \right) \right) \enspace.
\end{align*}
In the last inequality we used the inequalities $q\cdot m\leq 1/2\cdot \cdot d^{\frac{33}{100}-\lBlockBound}\cdot \log n$, 
$\frac{\ell}{2}\leq \log n$ and
\begin{align}
\left({\textstyle \frac{e \cdot \ell}{2m}}\right)^m&\leq \left( {\textstyle e\cdot d^{\frac{1}{10}}} \right)^m 
\leq \exp\left({\textstyle \frac{\log n}{\sqrt{d}}}\right)\enspace.
\end{align}
Hence, for {$d \ge d_0(\varepsilon)$}, we have 
\begin{align*}
\Pr\left[|H_1|\ge (1/2)\cdot \ell\cdot d^{-1/10}\right]= \Pr\left[|H_1|\ge m \right] \leq n^{-d^{1/4}} \le (1/2)\cdot n^{-d^{1/4}}
\enspace,
\end{align*}
as desired, concluding the proof of \Cref{prop:BoundOnHevy}.
\end{proof}

We have the following bound on the upper tail of $\sum_{i \in H} \log \WB_P(v_i)$.
\begin{theorem}\label{thm:BoundOnWBThm}
For any $\varepsilon>0$, there exists $d_0=d_0(\varepsilon) \ge 1$, such that for all $d\ge d_0$, and 
$0<\beta \le \beta_c(d)$, the following is true:

For $r=\ell = \maxPathL$, and $P=(v_0,\ldots, v_{\ell}) \in V_n^{\ell+1}$, let 
$(\GWT,\{\overline{\gauss}_{e}\}, \beta)$ where $\GWT=\GWT(P,d)$. Then, 
\begin{equation*} 
\Pr \left[ \sum\nolimits_{i \in H} \log \WB_P(v_i) \ge \ell \cdot d^{-1/50} \middle| 
|H|\le \ell\cdot d^{-{1}/{10}}\right] \le n^{-d^{1/3}}\enspace.
 \end{equation*}
\end{theorem}

Given \Cref{thm:BoundOnWBThm}, we establish \Cref{prop:BoundingL} as follows. 

\begin{proof}[Proof of \Cref{prop:BoundingL}]
Write $\theta = \varepsilon/4$, due to \Cref{lem:BreakIntoPartsToBound}, we have 
\begin{align}
\Pr\left[|L|\geq \ell \right]\leq \Pr\left[|H| + \theta^{-1}\cdot\sum\nolimits_{i\in H}\log\WB_P(v_i)\geq \ell \right] \enspace. 
\end{align}
From the law of total probability we have
\begin{multline}
\Pr\left[|H|+ \theta^{-1}\cdot\sum\nolimits_{i\in H}\log\WB_P(v_i)\geq \ell \right]\leq \\
\Pr\left[|H|\ge \ell\cdot d^{-1/10}\right] +
\Pr\left[|H| + \theta^{-1}\cdot\sum\nolimits_{i\in H}\log\WB_P(v_i)\geq \ell \middle||H|\le \ell \cdot d^{-1/10} \right] \enspace.
\label{eq:FirstApproxOFprobL}
\end{multline}
Moreover, for sufficiently large $d$, we have 
\begin{align*}
\Pr\left[|H| + \frac{1}{\theta}\cdot\sum_{i\in H}\log\WB_P(v_i)\geq \ell \middle||H|\le \frac{\ell}{d^{1/10}} \right]
\le
\Pr\left[\sum_{i\in H}\log\WB_P(v_i)\ge\frac{\ell}{d^{1/50}}\middle||H|\le \frac{\ell}{d^{1/10}}\right] \enspace.
\end{align*}
Plugging the above into \eqref{eq:FirstApproxOFprobL}, then from \Cref{thm:BoundOnWBThm}, 
and \Cref{prop:BoundOnHevy} we get
\begin{align*}
\Pr\left[|H| + \theta^{-1}\cdot\sum\nolimits_{i\in H}\log\WB_P(v_i)\geq \ell \right]\le n^{-d^{1/3}} + {n}^{-d^{1/4}} 
\le {n}^{-d^{1/6}}
\enspace.
\end{align*}
This concludes the proof of \Cref{prop:BoundingL}.
\end{proof}

Let us now prove \Cref{thm:BoundOnWBThm}.

\begin{proof}[Proof of \Cref{thm:BoundOnWBThm} ]
Similarly to the proof of \Cref{prop:BoundOnHevy}, we let $I_1$, and $I_2$ ,be the sets of odd, and even indices of $\{0, \ldots, \ell\}$, respectively. Writing $H_1 =H \cap I_1$, and $H_2 =H \cap I_2$, and using the union bound, we see it is sufficient to prove that
\begin{align*} 
2 \cdot \Pr \left[ \sum\nolimits_{i \in H_1} \log \WB_P(v_i) \ge \frac{1}{2}\cdot\frac{\ell}{d^{1/50}} \middle| |H_1|\leq \frac{\ell}{d^{{1}/{10}}}
	\right] \le
 n^{-d^{1/3}}\enspace.
 \end{align*}
From Markov's inequality we get that for every $t\ge0$
\begin{align}
\Pr \left[ \sum_{i \in H_1} \log \WB_P(v_i) \ge \frac{1}{2}\cdot\frac{\ell}{d^{1/50}} \middle| 
|H_1|\leq \frac{\ell}{d^{{1}/{10}}}
	\right] 
 &\le 
\frac{\Exp\left[\exp( t\cdot \sum_{i \in H_1}\cdot\log\WB_P(v_i)) 
\middle| |H_1|\leq \frac{\ell}{d^{{1}/{10}}} \right]}
	{\exp\left(\frac{t}{2\cdot d^{1/50}} \cdot \ell \right)}
 \nonumber
 \\
 &= 
\frac{
\left(\Exp\left[\exp(t\cdot\log\WB_P(v_1))\ | v_1 \in H_1\right]\right)
 ^{{\ell} \cdot d^{-1/10}}
 }
	{\exp\left(\frac{t}{2\cdot d^{1/50}} \cdot \ell \right)} \enspace, \label{eq:MarkovOnWB}
\end{align}
where the equality follows from the fact that the random variables $\WB_{P}(v_j)$'s where $j\in H_1$ are i.i.d.

Let us write $v$ instead of $v_1$, and notice that since $\WB_P(v)$ is non-negative, we also have
\begin{align} \label{eq:HoldForNoNegativeRatio}
\Exp\left[\exp(t \cdot \log \WB_P(v) )\ |\ v \in H_1\right] 
&\le \frac{\Exp\left[\left( \WB_P(v) \right)^t\right]}
 {\Pr[\WB_P(v)\ge 1]}\enspace. 
\end{align}
For the enumerator of \eqref{eq:HoldForNoNegativeRatio} we use the bounded provided by 
\Cref{lem:MGFForcappedWB}. Let us now bound the denominator of \eqref{eq:HoldForNoNegativeRatio}, i.e., 
$\Pr[\WB_P(v) \ge 1] $. Recalling the definitions of $\WSA, \cappedWSA$, and $\WB$, 
we have
\begin{align*}%\label{eq:lowerBoundWeightB}
 \Pr[\WB_P(v) \ge 1] 
 \ge {\Pr\left[\,\cappedWA(v) \ge 1\right]} 
 \ge {\Pr[\WSA(v) \ge 1]} 
 \ge \Pr[\WSA(v) \cdot \mathds{1}\{d< \degr(v) \le 2d\}\ge 1] \enspace.
\end{align*}
Using Baye's rule, we get
\begin{align}\label{eq:lowerBoundWeightB}
 \Pr[\WB_P(v) \ge 1] 
 \ge
 \Pr[{\WSA(v) \ge 1} \mid d \le \degr(v) \le 2d] \cdot \Pr[ d \le \degr(v) \le 2d]
 \enspace.
\end{align}
Let us now focus on the first factor in the rhs of \eqref{eq:lowerBoundWeightB}. 
Recall %Similarly to the proof of \Cref{theorem:TailBound4WA}, 
\begin{equation*}
\WSA(v) = \sum\nolimits_{1\leq i \leq \degr(v)} \left|\tanh\left({\beta} \gauss_i \right)\right|^2\enspace,
\end{equation*}
for i.i.d., standard gaussians $\gauss_i$'s. 
With that in mind, 
% we underestimate the probability $\Pr[\WSA(v) \ge 1 \mid d \le \degr(v) \le 2d]$ as follows
we have
\begin{align*}
\Pr[\WSA(v) \ge 1 \mid d \le \degr(v) \le 2d] 
&\ge 
\left(\Pr\left[\;\left|\tanh\left({\beta} \gauss_i \right)\right|^2 \ge 
 d^{-1}\;\right]\right)^{2d} \\
&\ge 
\left(2\cdot\Pr\left[\;\gauss_i\ge \beta^{-1}\cdot\mathrm{arctanh}\left(d^{-1/2}\right)\;\right]\right)^{2d}
\\
&\ge 
\left(2\cdot\Pr\left[\;\gauss_i\ge (2\beta)^{-1}\cdot\left(\frac{1+({1}/\sqrt{d})}
{1-({1}/\sqrt{d})}-1\right)\;\right]\right)^{2d}
\\
&\ge 
\left(2\cdot\Pr\left[\;\gauss_i\ge \frac{1}{\beta\cdot (\sqrt{d}-1)}\;\right]\right)^{2d}
\enspace.
\end{align*}
Since $\beta \le \beta_c(d)$, there exists $\lambda\leq 1/2$ such 
that $\beta = \frac{\lambda}{\sqrt{d}-1}$, and thus, we
rewrite the above as
\begin{align*}
\Pr[\WSA(v) \ge 1 \mid d \le \degr(v) \le 2d] \ge \left(2\cdot\Pr\left[\;\gauss_i\ge 
 \lambda^{-1}\right]\right)^{2d} \enspace.
\end{align*}
Since $\lambda^{-1}>0$, we have $\Pr\left[\;\gauss_i\ge \lambda^{-1}\right]<1/2$.
Hence, there exists constant $C_1>0$, such that 
\begin{align}\label{eq:finalUnderEstOfBuProb}
\Pr[\WSA(v) \ge 1 \mid d \le \degr(v) \le 2d] \ge \exp(-C_1\cdot d) \enspace.
\end{align}
Regarding the second factor in the rhs of \eqref{eq:lowerBoundWeightB}, we have
\begin{align}\label{eq:finalUnderEstOfD2Dprob}
\Pr[ d \leq \degr(v) \le 2d] 
&=\Pr[ d \leq \degr(v)]- \Pr[ \degr(v) > 2d] =\frac12-\Pr[ \degr(v) > 2d] \geq \frac{1}{4}. 
\end{align}
Where the last inequality follows from Chernoff bounds. 

% Therefore, for $d$ large enough, there exists a constant $C_2>0$, such that 
% \begin{align}\label{eq:finalUnderEstOfD2Dprob}
% \Pr[ d \le \degr(v) \le 2d] \ge \exp(-C_2\cdot d) \enspace.
% \end{align}
Plugging \eqref{eq:finalUnderEstOfBuProb} and \eqref{eq:finalUnderEstOfD2Dprob} into \eqref{eq:lowerBoundWeightB}, there exists constant $C>0$ such that
\begin{align}\label{eq:FinalLowerBoundinDenom}
 \Pr[\WB_P(v) \ge 1] 
 \ge \exp(-C\cdot d) \enspace.
\end{align}
Choosing $t=d^{95/100}$, while plugging into
\eqref{eq:HoldForNoNegativeRatio} the bound on $\Exp\left[( \WB_P(v_1) )^t\right]$ 
from \Cref{lem:MGFForcappedWB} and \eqref{eq:FinalLowerBoundinDenom}, we get 
$$
\Exp\left[\exp(t \cdot \log \WB_P(v) )\ |\ v \in H_1\right]\leq e^{d \cdot C} \cdot 
\left(1-\varepsilon/4\right)^{d^{\frac{93}{100}}} \leq \exp \left( \textstyle{d \cdot \frac{C}2} \right) \enspace. 
$$
Plugging the above into \eqref{eq:MarkovOnWB} for $t=d^{\frac{95}{100}}$, we have 
\begin{align*}
\nonumber 
\Pr \left[ \sum\nolimits_{i \in H_1} \log \WB_P(v_i) \ge (1/2) \cdot \ell \cdot d^{-1/50} \middle| 
|H_1|\leq \ell\cdot d^{-1/10} \right] 
&\le \exp \left(- d^{\frac{92}{100}}\cdot \ell \right) \le (1/2) \cdot n^{-d^{1/3}} \enspace,
\end{align*}
as desired. This concludes the proof of \Cref{thm:BoundOnWBThm}.
\end{proof}

\subsection{Proof of \Cref{lem:MGFForcappedWB}}\label{sec:lem:MGFForcappedWB}

Fix $i\in \{0,\ldots, \ell\}$. Let us write $v$ instead of $v_i$. For $s\ge 0$, and a path $Q =(w_0, \ldots, w_s)$, let us define 
 \begin{align*}
 \cappedWA^{\rm even}(Q) = \prod\nolimits_{0\leq i \leq \lfloor s/2\rfloor} \cappedWA(w_{2i}) \enspace,
 \end{align*}
 that is, the weight of $Q$ contributed only by vertices at even distance from $w_0$. Similarly, define $\cappedWA^{\rm odd}(Q)$ to be the weight of $Q$ contributed only by vertices at odd distance from $w_0$. Let us also define
 also define $\WB_{P}^{\rm even}(v) = \max_Q\{ \cappedWA^{\rm even}(Q)\}$, and $\WB_{P}^{\rm odd}(v) = \max_Q\{ \cappedWA^{\rm odd}(Q)\}$ where the maximisation is over all paths $Q$ of length at most $r$, that emanate from $v$ and do not intersect with $P$, i.e., they to do not share vertices. Since $\WB_P(v) \le \WB_{P}^{\rm odd}(v) \cdot \WB_{P}^{\rm even}(v)$, the union bound yields
 \begin{equation*}
 \Pr \left[\WB_P(v) \ge 1\right] \le \Pr[\WB_{P}^{\rm odd}(v) \ge 1] + \Pr[\WB_{P}^{\rm even}(v) \ge 1] \enspace,
 \end{equation*}
 and thus, it suffices to show that
 \begin{equation*}
 \Pr[\WB_{P}^{\rm even}(v) \ge 1]\le \left(1-{\varepsilon}/{4}\right)^{q}/2 \enspace,
 \end{equation*}

 Writing $\mathrm{path}_v(r)$ to denote all paths of length $r$ in ${\GWT}$, that emanate from $v$, and do not intersect with $P$, and $\mathrm{path}_v(\le r) = \cup_{k\le r} \mathrm{path}_v(k)$, we see that for any $t>0$ we have that
\begin{align}
\label{eq:OnlySumTrickStepToExP}
\Exp\left[\left( \WB_P^{\rm even}(v) \right)^t\right] 
&= \Exp\left[\left( \max_{Q \in \mathrm{path}_v(\le r)}\{\cappedWA^{\rm even}(Q)\}\right)^t\right] 
\le \Exp\left[ \sum_{Q \in \mathrm{path}_v\left(\le r\right)} \left(\cappedWA^{\rm even}(Q)\right)^t\right] \enspace.
\end{align}
Note that we cannot pull the sum out of the expectation in the rhs of the inequality above, as the set $\mathrm{path}_v\left(\le r\right)$ is a random variable. Therefore, we think in the following way. For $k = 0 \ldots r$, there are at most $n^k$ potential paths of length $k$ emanating from $v$, and each potential path of has probability $(d/n)^k$ to be present in ${\GWT}$. Denoting with $ (w_0, \ldots , w_k)$ an arbitrary such potential path emanating from $v$, i.e., $w_0=v$, we see that
\begin{align}
\Exp\left[\left( \WB_P^{\rm even}(v) \right)^t\right] 
&\le \sum_{k=0}^r n^k \cdot \left(\frac{d}{n}\right)^{k} \cdot 
\Exp\left[
\left(
\prod_{i=0}^{\lfloor k/2 \rfloor} \cappedWA\left(w_{2i}\right) 
\right)^t\right]
\le
\sum_{k=0}^r d^k \cdot \left(\Exp\left[{\cappedWA^t\left(v\right)}\right] \right)^{\lfloor k/2 \rfloor+1}
\enspace,
\label{eq:ExpNeedsCalc}
 \end{align}
where the last inequality follows from the independence of the weights $\cappedWA$ corresponding to same parity vertices along $(w_0, \ldots, w_\ell)$.
To upper bound $\Exp\left[{\cappedWA^t\left(v\right)} \right]$ we consider two regimes in terms of the degree of vertex $v$:
 \begin{align}
\Exp\left[{\cappedWA^t\left(v\right)} \right] 
&= 
\Exp\left[{\cappedWA^t\left(v\right)} \cdot\Ind\{\degr(v) \le 3d\} \right]
+
\Exp\left[{\cappedWA^t\left(v\right)} \cdot\Ind\{\degr(v) > 3d\} \right]
\label{eq:ExpectationBigSmallDeg}
 \end{align}
Let us now focus on the first term of \eqref{eq:ExpectationBigSmallDeg}. Writing $g(x)$ for the 
pdf of $\cappedWA(v)\cdot\Ind\{\degr(v) \le 3d\}$, and noticing that $\cappedWA(v)$ is at most 
$3d^2$, we see that
\begin{align}
 \Exp\left[{\cappedWA^t\left(v\right)} \cdot\Ind\{\degr(v) \le 3d\} \right]
&=
\int_{0}^{1-\varepsilon/4} x^t \cdot g(x) \; dx
+
\int_{1-\varepsilon/4}^{3d^2} x^t \cdot g(x) \; dx
%\label{eq:SumOfIntToBeBounded}
\nonumber\\
&\le
 \left(1-{\varepsilon}/{4}\right)^t \cdot \Pr\left[\WSA(v) \le 1-{\varepsilon}/{2}\right]
+
(3d^2)^{t} \cdot \int_{1-\varepsilon/2}^{3d^2} g(x) \; dx
\nonumber\\
&\le
\left(1-{\varepsilon}/{4}\right)^t
 +
(3d^2)^{t} \cdot \Pr\left[\WSA(v) \ge {1-{\varepsilon}/{2}}\right]
\nonumber\\ \label{eq:FromOurBoundWA}
&\le
\left(1-{\varepsilon}/{4}\right)^t
+
(3d^2)^{t} \cdot 
\exp\left(-d\cdot{\varepsilon^2 }/{(18\upkappa^2)}\right)
\enspace, 
\end{align}
where \eqref{eq:FromOurBoundWA} follows from \Cref{theorem:TailBound4WA}.
We now focus the large-degree case. Accounting only for the randomness on the degree of $v$, and overestimating each term of $\cappedWA(v)$ to be equal to $d$, we get 
\begin{align}
\nonumber
\Exp\left[{\cappedWA^t\left(v\right)} \cdot\Ind\{\degr(v) > 3d\} \right]
 &\le
\sum_{k=3d}^{n} (d\cdot k)^t \cdot\Pr[\degr(v) = k] 
% \\
% &
 \le
 d^t \cdot\sum_{k=3d}^{n} k^t \cdot \binom{n}{k} \cdot \left(\frac{d}{n}\right)^k
 \\
 &\le
 d^t \cdot
 \sum_{k=3d}^{n} k^t \cdot \left(\frac{ne}{k}\right)^k \cdot \left(\frac{d}{n}\right)^k
 % \\
 % &
 \le d^t \cdot
 \sum_{k=3d}^{n} k^t \cdot \left(\frac{de}{k}\right)^k \enspace.
 \label{eq:finalToBound}
\end{align}
Substituting \eqref{eq:FromOurBoundWA} and \eqref{eq:finalToBound} to \eqref{eq:ExpectationBigSmallDeg}, gives
\begin{equation*}
 \Exp\left[{\cappedWA^t\left(v\right)} \right] \le
 \left(1-\frac{\varepsilon}{4}\right)^t
+
(3d^2)^{t} \cdot \exp\left(-\frac{\varepsilon^2 }{18\upkappa^2}\cdot d\right)
+
 d^t \cdot \sum_{k=3d}^{n} k^t \cdot \left(\frac{de}{k}\right)^k \enspace.
\end{equation*}
Since $ed/k < 1$, for every $k >3d$, choosing $t = d^{{95}/{100}}$, we see that there exist $d_0(\varepsilon)\ge 1$, such that for every $d \ge d_0$, we have that
$\Exp\left[{\cappedWA^t\left(v\right)} \right] \le (1-\frac{\varepsilon}{4})^{t^\prime}$, where $t^\prime = d^{{94}/{100}}$. We now bound \eqref{eq:ExpNeedsCalc} as
\begin{align*}
 \Exp\left[\left( \WB^{\rm even}_P(v) \right)^{{t}}\right] \le \frac{(1-\frac{\varepsilon}{4})^{{t^\prime}}}{1-d^2\cdot(1-\frac{\varepsilon}{4})^{{t^\prime}}}
 {\le
 {\left(1-\frac{\varepsilon}{4}\right)^{{q}}}
 }
 \enspace,
\end{align*}
where we can take $q = d^{93/100}$.

\spreadpoint

\section{Remaining Proofs}

\subsection{Proof of \Cref{theorem:TailBound4WA}}\label{sec:theorem:TailBound4WA}

First, recall that $ \WSA(v)= \sum_{{z \sim v }} \Inf^2_{\{v,z\}}$, with
$\Inf_{e} = \left|\tanh\left(\beta \cdot\gauss_e \right) \right|$
where $\gauss_e$ follows the standard Gaussian distribution. 
Since $|\tanh(x)| \le |x|$, for every real $x$, we see that
$
\Inf_{e} \le {\beta} \cdot |\gauss_e| 
%\enspace,
$.
 Hence, bounding from above the upper-tail of $\sum_e {\beta}^2\cdot |\gauss_e|^2$, provides an upper bound to the corresponding tail of $\sum_e \Inf_{e}^2$, i.e., for every $x\in \mathbb{R}$ we have
\begin{equation}\label{eq:BboundedBySumOfHalfNorm}
\textstyle
 \Pr\left[\WSA(v) \ge x\right] \le 
 \Pr \left[\sum_{e} {\beta}^2 \cdot |\gauss_e|^2 \ge x\right] \enspace.
\end{equation}
Note also that using Wald's lemma we get
\begin{align*}
\Exp[\WSA(v)] = d\cdot \Exp[|\tanh(\beta \cdot \gauss)|^2]\le
d\cdot \Exp[|\beta \cdot\gauss|^2] =d \cdot \beta^2
\enspace.
\end{align*}
Let the event $\cE:=\{\degr(v) \le \!\left(1+\frac{\updelta}{2}\right)d\}$.
From the total probability law we have
\begin{align}\label{eq:breakPrCondDeg}
\Pr\left[\WSA(v) \geq d\beta^2 + {\updelta}/{2} \right] \! \le\!
\Pr\left[\WSA(v) \geq d \beta^2 \!+\! {\updelta}/{2} \!\mid\! \cE \right]\!
\!
+ \Pr\left[ \cE\right]\enspace. 
\end{align}

Let us now focus on the first term in the rhs of \eqref{eq:breakPrCondDeg}. Applying the tail bounds for chi-square distribution 
we see that the first term is upper bounded by
\begin{align}
\Pr\left[\sum_{i=1}^{(1+\updelta/2)d}\gauss^2_i \geq d + \frac{\updelta}{2\beta^2} \right]
&\le
\exp\left(-\frac{(1+\updelta/2)d}{2}
\cdot\left[\frac{1-\updelta/2}{\beta^2(1+\updelta/2)d} -1 -\log \left(\frac{1-\updelta/2}{\beta^2(1+\updelta/2)d}\right)\right]\right) \nonumber \enspace.
\end{align}
Since $0<\beta<\sqrt{{1}/{d}}$, we have
\begin{align}\label{eq:PrBSmallIfDegSmallSpecificEtaTheta}
\Pr\left[\WSA(v) \geq d\cdot \beta^2 +
\frac{\updelta}{2}\mid \cE \right] 
\le
\exp\left(-\frac{1+\updelta/2}{2}d\left[\frac{2\updelta^2}{(1 + \delta)}\right] \right)
\le
\exp\left(-\frac{ \updelta^2}{ (1+\delta/2)}\cdot d \right)
\enspace. 
\end{align}

Let us now turn to the second term of \eqref{eq:breakPrCondDeg}. Since $\degr(v)$ is a sum of independent Bernoulli random variables, applying the Chernoff tail bound imply 
\begin{align}\label{eq:PrDegBeLarge}
\Pr[\cE]=\Pr\left[\degr(v) > (1+\delta/2) d\right] \le \exp\left(-\frac{\delta^2 }{8 + 2\delta}d\right)
\enspace.
\end{align}
Substituting \eqref{eq:PrBSmallIfDegSmallSpecificEtaTheta} and \eqref{eq:PrDegBeLarge}, in \eqref{eq:breakPrCondDeg} gives \eqref{eq:WeightAVertexTailBound}. 
%
% \begin{align*}%\label{eq:TotalPrBSmallSubs}
% \Pr\left[\WSA(v) \geq d\cdot \beta^2 +
% \frac{\updelta}{2} \right]
% \le \exp\left(-\frac{ \updelta^2}{ 2+\updelta}\cdot d \right)
% \enspace,
% \end{align*}
% which for sufficiently large $d$ gives the result.
\hfill $\Box$

\subsection{Proof of Lemma \ref{lem:ClacOfR}}\label{sec:lem:ClacOfR}
We show that for an arbitrary path $P$ in $\G$, with $r\le|P|\le \log n$, we have 
\begin{align}\label{eq:SmallTuples}
\Pr \left[\cappedWA(P) \ge 1\right] \le {n^{-d^{1/4}}} \enspace.
\end{align}
In light of \eqref{eq:SmallTuples}, applying the union bound over all such paths gives us the result
\begin{align*}
\Pr \left[\cS_1\right] 
\ge 1 - \sum_{k = r}^{\log n} \binom{n}{k +1}\left(\frac{d}{n}\right)^k \cdot n^{-d^{1/4}} 
 \ge 1 - \sum_{k = r}^{\log n} {n}\cdot{d}^k \cdot n^{-d^{1/4}}
 % \\
 \ge 1 - n^{-d^{1/5}},
\end{align*}
where the last inequality holds for large $d$ and $n$. Therefore, we now focus on proving ~\eqref{eq:SmallTuples}. 

Let $P= (v_0, \ldots, v_k)$, with $r \le k \le \log n$. We split the vertices of $P$ into two sets, $\In(P)$, and $\Out(P)$. The set $\In(P)$ is comprised by all vertices $v_i$ of $P$ that are adjacent to a vertex in $V(P) \setminus \{v_{i-1}, v_{i+1}\}$, and $\Out(P) = V(P) \setminus \In(P)$, so that 
\begin{align} \label{eq:BreakLamInTwo}
\cappedWA(P) = 
\cappedWA\left(\In(P)\right)\cdot 
\cappedWA(\Out(P))= \prod_{v\in \In(P)} \cappedWA(v) \cdot \prod_{w\in \Out(P)} \cappedWA(w) \enspace.
\end{align}
 We next bound separately $\cappedWA\left(\In(P)\right)$, and $\cappedWA\left(\Out(P)\right)$, so that their product is less than $1$, w.h.p..

Let us start by bounding $\cappedWA\left(\In(P)\right)$. We first notice the following tail bound for $|\In(P)|$:
\begin{align}\label{eq:InSetbound}
\Pr\left[|\In(P)| \ge 2\sqrt{d}\right]
\le 
\sum_{s = \sqrt{d}}^{k}
\binom{k^2}{s}
\cdot
\left( \frac{d}{n}\right)^s
\le 
\sum_{s = \sqrt{d}}^{k}
\left( \frac{d\cdot k^2}{n}\right)^s
\le
2\cdot\left( \frac{d\cdot k^2}{n}\right)^{\sqrt{d}}
\le
 n^{-{\sqrt{d}}/{2}} \enspace,
\end{align}
where the last two inequalities hold for large $d$ and $n$. 

For a vertex $v_i$ in $P$, let $\degr_{\rm in}(v_i)$ be the number of neighbors
that $v_i$ has in $P$, while let $\degr_{\rm out}(v_i)$ be the number of 
neighbors in $V\setminus P$.

For a vertex $w\in \In(P)$ we have that $\degr_{\rm out}(w)$ is dominated by 
${\tt Binom}(n, d/n)$. Then, we obtain 
\begin{align}\label{eq:InDegBound}
\Pr\left[\degr_{\rm out}(w)\ge d\cdot\log n, \;\text{ for at least one } w \in \In(P)\ \mid\ |\In(P)| <2\sqrt{d}\right] 
\le n^{-d^{{9}/{10}}}
\enspace .
\end{align}
Indeed, the Chernoff bound gives
\begin{align*}
\Pr\left[\ \degr_{\rm out}(w)\ge d\cdot\log n\ |\ w\in \In(P)\right] 
\le \exp\left(-\frac{(\log n - 2)^2}{\log n} \cdot d\right)
\le \exp\left(-\frac{\log n}{2} \cdot d\right)
\le n^{-d^{95/100}}
\enspace , 
\end{align*}
where the last two inequalities hold for large enough $d$ and $n$. 
The above and a  union-bound imply \eqref{eq:InDegBound}.

Let $B_1$ be the event that $|\In(P)| < 2\sqrt{d}$, while let 
$B_2$ be the event that for all $w\in \In(P)$ we have $\degr_{\rm out}(w)<d\log n$.

On the events $B_1$ and $B_2$, we have that for all $w\in \In(P)$
\begin{enumerate} 
\item $\degr(w)<d\log n+2+2\sqrt{d}$,
\item $\cappedWA(w)\leq 2d^2\log n$.
\end{enumerate}

The first item follows by noticing that $\degr(w)=\degr_{\rm in}(w)+\degr_{\rm out}(w)$, 
and that $\degr_{\rm in}(w)\leq 2+|\In(P)|$, for all $w\in \In(P)$.

The second item follows from the first item, and by noticing that 
for any vertex $v$ in the graph we have $\cappedWA(v) \leq d\cdot \degr(v)$. 
Hence, we have that 
\begin{align}
\Pr\left[\cappedWA(\In(P)) < (2d^2 \log n)^{2\sqrt{d}}\right] &\ge \Pr[B_1\cap B_2] 
% \nonumber\\ &\geq 
\geq 1-\Pr[\overline{B}_1]-\Pr[\overline{B}_2] 
%& \mbox{[union bound]} \nonumber \\ &
\geq 1- n^{-d^{1/3}}\enspace , 
%\mbox{[from \cref{eq:InDegBound} \& \cref{eq:InSetbound}]} 
\label{eq:FinalBoundForInP}
\end{align}
% where the second inequality is from the union bound, while 
for the last inequality we use
\cref{eq:InDegBound} \& \cref{eq:InSetbound}.

To bound $\cappedWA(\Out(P))$, we follow the proof strategy of Lemma \ref{lem:MGFForcappedWB}. In particular, we partition $\Out(P)$ into two sets: the set of vertices with even index in $P$, and the set of vertices with odd index in $P$. Moreover, we let $\cappedWA_{\rm even}(\Out(P))$ be the product of weights over the set of even vertices in $\Out(P)$, and $\cappedWA_{\rm odd}(\Out(P))$ similarly. We claim that
\begin{align}\label{eq:CappedEvenBoundOut}
\Pr \left[\cappedWA_{\rm even}(\Out(P)) \ge (\log n)^{-d}\right] \le n ^{-d^{2/5}} \enspace.
\end{align}
Indeed, Markov's inequality yields that for every $t>0$
\begin{align}
\Pr \left[\cappedWA_{\rm even}(\Out(P)) \ge (\log n)^{-d}\right]
\le 
{\Exp\left[\cappedWA^t_{\rm even}(\Out(P))\right]}\cdot(\log n)^{t\cdot d}
\enspace.
\end{align}
In the proof of Lemma \ref{lem:MGFForcappedWB}, we have that for $t= d^{95/100}$, and $s=d^{94/100}$ and arbitrary vertex $v$, we have that $\Exp\left[{\cappedWA^t\left(v\right)} \right] \le (1-\frac{\varepsilon}{4})^{s}$, (notice that although we prove this for the random tree construction, the arguments work precisely for $\G(n, d/n)$ as well). Acknowledging the fact that $\cappedWA_{\rm even}(\Out(P))$ is a product of i.i.d. random variables, we get that
\begin{align}
\Pr \left[\cappedWA_{\rm even}(\Out(P)) \ge (\log n)^{-d}\right] 
\le 
\frac
{\left(1-\varepsilon/4\right)^{ks/2}}
{(\log n)^{-t\cdot d}}
\le 
{(\log n)^{d^2}}
\cdot
{n^{-d^{42/100}}}
\le 
n^{-d^{2/5}}\enspace,
\end{align}
Notice that, by symmetry, 
\eqref{eq:CappedEvenBoundOut} holds for $\cappedWA_{\rm odd}(\Out(P))$ as well. Moreover, since $\cappedWA(\Out(P)) = \cappedWA_{\rm odd}(\Out(P))\cdot\cappedWA_{\rm even}(\Out(P))$, using the union bound further yields
\begin{align}\label{eq:FinalBoundForOutP}
\Pr \left[\cappedWA(\Out(P)) < (\log n)^{-2d}\right] \ge 1- 2\cdot n ^{-d^{2/5}} \ge 1-n^{-d^{1/3}} \enspace.
\end{align}
From the union bound, equations \eqref{eq:FinalBoundForInP} and \eqref{eq:FinalBoundForOutP} , and the fact that
$\cappedWA(P) = \cappedWA(\In(P))\cdot\cappedWA(\Out(P))$,
we have that
\begin{align*}
\Pr \left[\cappedWA(P) < 1\right] \ge 1- 2\cdot n ^{-d^{1/3}} \ge 1-n^{-d^{1/4}} \enspace,
\end{align*}
concluding the proof of Lemma \ref{lem:ClacOfR}.

\spreadpoint

\section{Tail Bounds}

\subsection{A Tail bound for $\sqrBetaJSphere(v,\ell)$ \& $\sqwsphere(v,\ell)$}

We consider the 2-spin model on a Galton-Watson tree.
For $d>0$, we let $\TT = \TT(d)$ be a Galton-Watson  tree rooted at vertex $v$, with offspring distribution $\Poisson(d)$. 
For integer $\ell \ge 1$, we write $\TT_\ell$ for the tree obtained from $\TT$ by deleting all vertices at distance $>\ell$ from the root. 
Also consider a set of i.i.d. standard Gaussians $\{\gauss_e: e \in E(\TT_\ell)\}$  and define the 
$V(\TT_\ell) \times V(\TT_\ell)$ interaction matrix of the 2-spin model by 
\begin{align}
 \InAct (u,w) = \begin{cases}
 \gauss_{\{u,w\}} &\text{if } u\sim w \\
 0 &\text{otherwise}
 \end{cases}
 \enspace.
\end{align}
%Let $\upnu_{\TT_{\ell}}$ be the 2-spin model on $\TT_\ell$ at inverse temperature $\beta > 0$ and external field $\Field\in \mathbb{R}^{V(\TT_\ell)}$ is defined 
%in that natural way. %similarly to \eqref{eq:DefOfGibbs}. 

For any $\beta>0$, and the root $v$ of $\TT_{\ell}$,   we let $\sqrBetaJSphere(v,\ell)$ be defined by
\begin{align}\label{eq:DefOfSQW}
\sqrBetaJSphere_{\TT_{\ell}}(v,\ell)&=\sum\nolimits_{P=x_{0},\ldots, x_{\ell}} |\beta\cdot \InAct_B(x_{\ell-1},x_{\ell})|^2 \cdot \prod\nolimits_{0\leq j<\ell-1}|\tanh\left(\beta\cdot \InAct(x_j,x_{j+1})\right)|^2 \enspace,
\end{align}
where variable $P$ in the summation varies over all paths of length $\ell$ emanating from $v$.

Similarly to the above, for any $\beta>0$, and the root $v$ of $\TT_{\ell}$,   we let $\sqwsphere_{\TT_{\ell}}(v,\ell)$ be defined by
\begin{align}\nonumber
\sqwsphere(v,\ell) & =\sum\nolimits_{P=x_{0},\ldots, x_{\ell}}  \prod\nolimits_{0\leq j<\ell}|\tanh\left(\beta\cdot \InAct(x_j,x_{j+1})\right)|^2
\enspace,
\end{align}
where the sum is over all paths $P$ of length $\ell$ emanating from the vertex $v$. 

In what follows, we provide tail bounds for the two quantities $\sqrBetaJSphere_{\TT_{\ell}}(v,\ell)$ and 
$\sqwsphere(v,\ell) $. 

\subsubsection{A tail bound for $\sqrBetaJSphere_{\TT_{\ell}}(v,\ell)$}

\begin{lemma} \label{thm:GWSphereTailHeavy}
For $d>0$,  let $r\ge1$ and let $\TT_r$ be a Galton-Watson tree with offspring distribution 
$\Poisson(d)$, rooted at a vertex $v$, and consider the 2-spin model on $\TT_r$ at inverse temperature $\beta \leq \beta_c(d)$. 
 For any $C>100$, and for any $t \in [0, \frac{d}{2\upkappa})$, we have
 \begin{align}\nonumber
 \Pr\left[ \sqrBetaJSphere_{\TT}(v,r) > C \cdot \upkappa^r\right]\le \exp\left((1-C)\cdot t\right) \enspace,
 \end{align}
 where $\upkappa$ is from \eqref{eq:defOfBc}.
\end{lemma}

\begin{proof}[Proof of \Cref{thm:GWSphereTailHeavy}]
%Let us write $\sqrBetaJSphere(r) = \sqrBetaJSphere_{\TT}(v,r)$. 
For $t \in \left [0, \frac{d}{2\upkappa}\right)$ we  define
\begin{align}
\norwsphere(r) = \frac{\sqrBetaJSphere(r)}{{\upkappa}^r}, &&\text{ and }&& \itfun_r(t) 
 = \Exp[ \exp\left(t \cdot \norwsphere(r) \right)]
\enspace.
\end{align}
Let  $\fun=\fun_d : \mathbb{R} \to \mathbb{R}$ be the function  $\fun(x) = e^{d \cdot (x-1)}$.  
Notice that the  moment generating function (m.g.f.) of  the $\Poisson(d)$ distribution is given by 
$M(t) = f(e^t)$.  For $\gauss \sim \Gaussian (0,1)$, we define
$$\itfun_0(t) = \Exp[\exp(t \cdot |\tanh(\beta \cdot \gauss)|^2)]\enspace.$$

By induction on $r \ge 1$, we prove that for any $t \in \left[0, \frac{d}{2\upkappa}\right)$ 
we have
\begin{align}\label{eq:IndBoundOnItFun}
\itfun_r(t) &\le e^t \enspace.
\end{align}
To establish the base case, i.e., to prove \eqref{eq:IndBoundOnItFun} for $r=1$, let 
$v_1, \ldots, v_{\bK}$ be the children of the root $v$, with corresponding couplings 
$\gauss_1, \ldots, \gauss_{\bK}$, where $\bK\sim \Poisson(d)$. We have
\begin{align}\label{eq:RecR1}
\itfun_1(t) &= \Exp \left[\exp \left( \left( {t}/{\upkappa}\right) \cdot  \sum\nolimits_{0\leq i \leq \bK} |\beta \cdot \gauss_i|^2\right)
\right] \nonumber \\
&= \Exp_{\bK} \left[ \prod\nolimits_{0\leq i \leq \bK}  \Exp_{ \gauss_i}\left[\exp\left( \left( {t}/{\upkappa}\right)  \cdot |\beta \cdot \gauss_i|^2\right)\right]
\right] \nonumber \\
&=\Exp_{\bK}\left[\left(\itfun_0\left( {t}/{\upkappa}\right)\right)^{\bK}\right] =\fun\left(\itfun_0\left({t}/{\upkappa}\right)\right)
\enspace.
\end{align}

Using the mgf for the chi-square distribution for argument $< 1/2$, we bound $\itfun_0(t/\upkappa)$ as follows:
\begin{align}\label{eq:BoundONG0}
 \itfun_0\left({t}/{\upkappa}\right)   &\le \left(1-2\cdot {\beta^2 \cdot t}/{\upkappa}\right)^{-1/2} 
\le \left(1-2\cdot {t}/{d}\right)^{-1/2}  \le 1+ {t}/{d} \enspace.
\end{align}
Combining \eqref{eq:RecR1} and \eqref{eq:BoundONG0} gives
\begin{align}
\itfun_{1}(t) = \fun\left(\itfun_0\left({t}/{\upkappa}\right)\right) 
\le \fun\left(1 + {t}/{d}\right) = e^t
\enspace,
\end{align}
where the inequality follows from the fact that $\fun$ is increasing.

Assume now that \eqref{eq:IndBoundOnItFun} holds for any 
$t \in \left[0, \frac{d}{2\upkappa} \right)$, and any depth $\le r$. We establish \eqref{eq:IndBoundOnItFun} for depth $r+1$.

Let $v_1, \ldots, v_{\bK}$ be the children of the root $v$, with corresponding couplings $\gauss_1, \ldots, \gauss_{\bK}$, and let 
 $\TT_i$ be the subtree of $\TT$ containing $v_i$ and its progeny. For $i= 1, \ldots, \bK$, and $\ell \ge0$.
We have 
\begin{align}
\itfun_{r+1}(t) 
&= \Exp \left[  \exp\left( \left( {t}/{\upkappa}\right) \cdot \sum\nolimits_{0\leq i \leq \bK} \norwsphere_i(r) \cdot  |\tanh(\beta \cdot \gauss_i)|^2\right) \right] \nonumber \\
&=  \Exp_{\bK} \left[  \prod\nolimits_{0\leq i \leq \bK}   \Exp_{\norwsphere_i, \gauss_i}\left[\exp\left( \left( {t}/{\upkappa}\right) \cdot \norwsphere_i(r)\cdot |\tanh(\beta \cdot \gauss_i)|^2\right)\right] \right] \nonumber \\
&=  \Exp_{\bK} \left[ \left(  \Exp_{\gauss} \left[ \Exp_{\norwsphere} \left[ \exp\left( \left( {t}/{\upkappa}\right)  \cdot \norwsphere(r)\cdot  |\tanh(\beta \cdot \gauss)|^2\right)  \right]\right] \right)^{\bK} \right] \nonumber \\
&= \fun \left( \Exp_{\gauss} \left[ \Exp_{\norwsphere} \left[ \exp\left( \left( {t}/{\upkappa}\right)  \cdot \norwsphere(r)\cdot |\tanh(\beta \cdot \gauss)|^2\right) \right] \right] \right)  \nonumber \\
\label{eq:IndOnlyF}
&= \fun \left(\Exp_{\gauss} \left[ \itfun_{r} \left(  \left( {t}/{\upkappa}\right) \cdot |\tanh(\beta \cdot \gauss)|^2 \right) \right] \right) \enspace.
\end{align}
Since $|\tanh(\beta\cdot \gauss)|^2 \le 1$, we apply the inductive hypothesis for $\itfun_r$ in \eqref{eq:IndOnlyF}, and get
\begin{align}\nonumber 
\itfun_{r+1}(t) 
&\le 
\fun \left( \Exp_{\gauss} \left[ \exp\left( ({t}/{\upkappa}) \cdot |\tanh(\beta \cdot \gauss)|^2\right) \right] \right)  
\le  \fun \left( \Exp_{\gauss} \left[  \exp\left(({t}/{\upkappa}) \cdot |\beta\cdot \gauss| ^2\right)
\right] \right) \enspace.  
\end{align}
Using the mgf for chi-square as above, we further see that 
\begin{align}\nonumber
\itfun_{r+1}(t) 
&
\le \fun\left( \left(1-2\cdot {t\cdot \upkappa}/{d}\right)^{-1/2} \right)
\le f\left( 1+ {t}/{d}\right) \le e^t 
\enspace,  
\end{align}
where the last inequality holds for sufficiently large $d$.

Using Markov's inequality, for any $C>100$, for every $r\ge 1$, and any 
{$t \in \left[0, \frac{d}{2\upkappa} \right)$,} 
we have 
\begin{align}\nonumber %\label{eq:WsphMarkV}
\Pr[\sqrBetaJSphere_{\TT}(v,r)(v,r) > C \cdot \upkappa^r]
\le  \frac {\itfun_r(t)}{\exp\left(C t \right)}
\le \exp\left((1-C)\cdot t\right)
\enspace,
\end{align}
as desired.
\end{proof}

\subsubsection{A Tail bound for $\sqwsphere(v,\ell)$}

\begin{lemma} \label{thm:GWSphereTail}
For $d>0$,  let $r\ge1$ and let $\TT_r$ be a Galton-Watson tree with offspring distribution 
$\Poisson(d)$, rooted at a vertex $v$, and consider the 2-spin model on $\TT_r$ at inverse temperature $\beta \leq \beta_c(d)$. 
 For any $C>100$, and for any $t \in [0, \frac{d}{2\upkappa})$, we have
 \begin{align}\nonumber
 \Pr\left[ \sqwsphere_{\TT}(v,r) > C \cdot \upkappa^r\right]\le \exp\left((1-C)\cdot t\right) \enspace,
 \end{align}
 where $\upkappa$ is from \eqref{eq:defOfBc}.
\end{lemma}
\begin{proof}
In light of \Cref{thm:GWSphereTailHeavy} the proof of \Cref{thm:GWSphereTail} is straightforward. Specifically, it suffices to notice that
with probability $1$, for any $r\geq 0$, we have that
\begin{align}
\sqwsphere_{\TT}(v,r)&\leq \sqrBetaJSphere_{\TT}(v,r)\enspace. \nonumber
\end{align}
The above is true by noting that
\begin{align}\nonumber
\sqwsphere(v,\ell) & =\sum\nolimits_{P=x_{0},\ldots, x_{\ell}}  \prod\nolimits_{0\leq j<\ell}|\tanh\left(\beta\cdot \InAct(x_j,x_{j+1})\right)|^2 \nonumber\\
&\leq \sum\nolimits_{P=x_{0},\ldots, x_{\ell}} |\beta\cdot \InAct_B(x_{\ell-1},x_{\ell})|^2 \cdot \prod\nolimits_{0\leq j<\ell-1}|\tanh\left(\beta\cdot \InAct(x_j,x_{j+1})\right)|^2 &=
\sqrBetaJSphere_{\TT_{\ell}}(v,\ell)\enspace.\nonumber
\end{align}
In the inequality above we use the standard inequality $|\tanh(\beta\cdot \InAct_B(x_{\ell-1},x_{\ell}))|^2 \leq |\beta\cdot \InAct_B(x_{\ell-1},x_{\ell})|^2$. 
\end{proof}

\spreadpoint

\bibliographystyle{plainurl}
\bibliography{PapersSamplingGlassses-new}

\begin{thebibliography}{10}

\bibitem{OptasOghlan08}
Dimitris Achlioptas and Amin Coja{-}Oghlan.
\newblock Algorithmic barriers from phase transitions.
\newblock In {\em 49th Annual {IEEE} Symposium on Foundations of Computer
  Science, {FOCS} 2008}, pages 793--802. {IEEE} Computer Society, 2008.
\newblock URL: \url{https://doi.org/10.1109/FOCS.2008.11}.

\bibitem{alaoui2020algorithmic}
Ahmed~El Alaoui and Andrea Montanari.
\newblock Algorithmic thresholds in mean field spin glasses.
\newblock {\em arXiv preprint arXiv:2009.11481}, 2020.

\bibitem{OptElAlaoui}
Ahmed~El Alaoui, Andrea Montanari, and Mark Sellke.
\newblock {Optimization of mean-field spin glasses}.
\newblock {\em The Annals of Probability}, 49(6):2922 -- 2960, 2021.
\newblock URL: \url{https://doi.org/10.1214/21-AOP1519}.

\bibitem{AnariSTOC24}
Nima Anari, Frederic Koehler, and Thuy{-}Duong Vuong.
\newblock Trickle-down in localization schemes and applications.
\newblock In Bojan Mohar, Igor Shinkar, and Ryan O'Donnell, editors, {\em
  Proceedings of the 56th Annual {ACM} Symposium on Theory of Computing, {STOC}
  2024, Vancouver, BC, Canada, June 24-28, 2024}, pages 1094--1105. {ACM},
  2024.
\newblock \href {https://doi.org/10.1145/3618260.3649622}
  {\path{doi:10.1145/3618260.3649622}}.

\bibitem{OptMCMCIS}
Nima Anari, Kuikui Liu, and Shayan~Oveis Gharan.
\newblock Spectral independence in high-dimensional expanders and applications
  to the hardcore model.
\newblock In {\em 61st {IEEE} Annual Symposium on Foundations of Computer
  Science, {FOCS} 2020, Durham, NC, USA, November 16-19, 2020}, pages
  1319--1330. {IEEE}, 2020.

\bibitem{BezakovaGGS22}
Ivona Bez{\'{a}}kov{\'{a}}, Andreas Galanis, Leslie~Ann Goldberg, and Daniel
  Stefankovic.
\newblock Fast sampling via spectral independence beyond bounded-degree graphs.
\newblock In {\em 49th International Colloquium on Automata, Languages, and
  Programming, {ICALP} 2022, July 4-8, 2022, Paris, France}, volume 229, pages
  21:1--21:16, 2022.

\bibitem{bobkov2006modified}
Sergey~G Bobkov and Prasad Tetali.
\newblock Modified logarithmic sobolev inequalities in discrete settings.
\newblock {\em Journal of Theoretical Probability}, 19:289--336, 2006.

\bibitem{boucheron2003concentration}
St{\'e}phane Boucheron, G{\'a}bor Lugosi, and Olivier Bousquet.
\newblock Concentration inequalities.
\newblock In {\em Summer school on machine learning}, pages 208--240. Springer,
  2003.

\bibitem{bubley1997path}
Russ Bubley and Martin Dyer.
\newblock Path coupling: A technique for proving rapid mixing in markov chains.
\newblock In {\em Proceedings 38th Annual Symposium on Foundations of Computer
  Science}, pages 223--231. IEEE, 1997.

\bibitem{chen2021almost}
Yuansi Chen.
\newblock An almost constant lower bound of the isoperimetric coefficient in
  the kls conjecture.
\newblock {\em Geometric and Functional Analysis}, 31:34--61, 2021.

\bibitem{chen2022localization}
Yuansi Chen and Ronen Eldan.
\newblock Localization schemes: A framework for proving mixing bounds for
  markov chains.
\newblock In {\em 2022 IEEE 63rd Annual Symposium on Foundations of Computer
  Science (FOCS)}, pages 110--122. IEEE, 2022.

\bibitem{VigodaSpectralInd}
Zongchen Chen, Kuikui Liu, and Eric Vigoda.
\newblock {Rapid mixing of Glauber dynamics up to uniqueness via contraction}.
\newblock In {\em 2020 IEEE 61st Annual Symposium on Foundations of Computer
  Science (FOCS)}, pages 1307--1318. IEEE, 2020.

\bibitem{ChManMo23}
Zongchen Chen, Nitya Mani, and Ankur Moitra.
\newblock From algorithms to connectivity and back: Finding a giant component
  in random \emph{k}-sat.
\newblock In {\em Proceedings of the 2023 {ACM-SIAM} Symposium on Discrete
  Algorithms, {SODA} 2023}, pages 3437--3470. {SIAM}, 2023.
\newblock URL: \url{https://doi.org/10.1137/1.9781611977554.ch132}.

\bibitem{COghlanEfth11}
Amin Coja{-}Oghlan and Charilaos Efthymiou.
\newblock On independent sets in random graphs.
\newblock In {\em Proceedings of the Twenty-Second Annual {ACM-SIAM} Symposium
  on Discrete Algorithms, {SODA} 2011}, pages 136--144. {SIAM}, 2011.
\newblock URL: \url{https://doi.org/10.1137/1.9781611973082.12}.

\bibitem{CoEfJKKCMI}
Amin Coja-Oghlan, Charilaos Efthymiou, Nor Jaafari, Mihyun Kang, and Tobias
  Kapetanopoulos.
\newblock {Charting the replica symmetric phase}.
\newblock {\em Communications in Mathematical Physics}, 359:603--698, 2018.

\bibitem{ding2023new}
Jian Ding, Jian Song, and Rongfeng Sun.
\newblock A new correlation inequality for ising models with external fields.
\newblock {\em Probability Theory and Related Fields}, 186(1):477--492, 2023.

\bibitem{dyer2006randomly}
Martin Dyer, Abraham~D Flaxman, Alan~M Frieze, and Eric Vigoda.
\newblock Randomly coloring sparse random graphs with fewer colors than the
  maximum degree.
\newblock {\em Random Structures \& Algorithms}, 29(4):450--465, 2006.

\bibitem{DyerFriez10}
Martin~E. Dyer and Alan~M. Frieze.
\newblock Randomly coloring random graphs.
\newblock {\em Random Struct. Algorithms}, 36(3):251--272, 2010.
\newblock URL: \url{https://doi.org/10.1002/rsa.20286}.

\bibitem{EfthICALP22}
Charilaos Efthymiou.
\newblock {On Sampling Symmetric Gibbs Distributions on Sparse Random Graphs
  and Hypergraphs}.
\newblock In {\em 49th International Colloquium on Automata, Languages, and
  Programming, {ICALP} 2022, July 4-8, 2022, Paris, France}, volume 229 of {\em
  LIPIcs}, pages 57:1--57:16. Schloss Dagstuhl - Leibniz-Zentrum f{\"{u}}r
  Informatik, 2022.

\bibitem{EftFeng23}
Charilaos Efthymiou and Weiming Feng.
\newblock On the mixing time of glauber dynamics for the hard-core and related
  models on g(n, d/n).
\newblock {\em CoRR}, abs/2302.06172, 2023.
\newblock \href {https://arxiv.org/abs/2302.06172} {\path{arXiv:2302.06172}},
  \href {https://doi.org/10.48550/arXiv.2302.06172}
  {\path{doi:10.48550/arXiv.2302.06172}}.

\bibitem{EfthymiouHSV18}
Charilaos Efthymiou, Thomas~P. Hayes, Daniel Stefankovic, and Eric Vigoda.
\newblock Sampling random colorings of sparse random graphs.
\newblock In {\em Proceedings of the Twenty-Ninth Annual {ACM-SIAM} Symposium
  on Discrete Algorithms, {SODA} 2018, New Orleans, LA, USA, January 7-10,
  2018}, pages 1759--1771. {SIAM}, 2018.

\bibitem{EfthZampICALP23}
Charilaos Efthymiou and Kostas Zampetakis.
\newblock Broadcasting with random matrices.
\newblock In {\em 50th International Colloquium on Automata, Languages, and
  Programming, {ICALP} 2023, July 10-14, 2023, Paderborn, Germany}, volume 261,
  pages 55:1--55:14. Schloss Dagstuhl - Leibniz-Zentrum f{\"{u}}r Informatik,
  2023.
\newblock URL: \url{https://doi.org/10.4230/LIPIcs.ICALP.2023.55}, \href
  {https://doi.org/10.4230/LIPICS.ICALP.2023.55}
  {\path{doi:10.4230/LIPICS.ICALP.2023.55}}.

\bibitem{EfthZamoCLT24}
Charilaos Efthymiou and Kostas Zampetakis.
\newblock On sampling diluted spin-glasses using glauber dynamics.
\newblock In {\em The Thirty Seventh Annual Conference on Learning Theory},
  pages 1501--1515. PMLR, 2024.

\bibitem{eldan2013thin}
Ronen Eldan.
\newblock Thin shell implies spectral gap up to polylog via a stochastic
  localization scheme.
\newblock {\em Geometric and Functional Analysis}, 23(2):532--569, 2013.

\bibitem{eldan2022spectral}
Ronen Eldan, Frederic Koehler, and Ofer Zeitouni.
\newblock A spectral condition for spectral gap: fast mixing in
  high-temperature ising models.
\newblock {\em Probability theory and related fields}, 182(3-4):1035--1051,
  2022.

\bibitem{fan2017well}
Zhou Fan and Andrea Montanari.
\newblock How well do local algorithms solve semidefinite programs?
\newblock In {\em Proceedings of the 49th Annual ACM SIGACT Symposium on Theory
  of Computing}, pages 604--614, 2017.

\bibitem{franz2001exact}
Silvio Franz, Michele Leone, Federico Ricci-Tersenghi, and Riccardo Zecchina.
\newblock {Exact solutions for diluted spin glasses and optimization problems}.
\newblock {\em Physical review letters}, 87(12):127209, 2001.

\bibitem{GalStefVigJACM15}
Andreas Galanis, Daniel Stefankovic, and Eric Vigoda.
\newblock Inapproximability for antiferromagnetic spin systems in the tree
  nonuniqueness region.
\newblock {\em J. {ACM}}, 62(6):50:1--50:60, 2015.

\bibitem{GamarnikJWFOCS20}
David Gamarnik, Aukosh Jagannath, and Alexander~S. Wein.
\newblock Low-degree hardness of random optimization problems.
\newblock In {\em 61st {IEEE} Annual Symposium on Foundations of Computer
  Science, {FOCS} 2020, Durham, NC, USA, November 16-19, 2020}, pages 131--140.
  {IEEE}, 2020.
\newblock \href {https://doi.org/10.1109/FOCS46700.2020.00021}
  {\path{doi:10.1109/FOCS46700.2020.00021}}.

\bibitem{GamSudan17}
David Gamarnik and Madhu Sudan.
\newblock Performance of sequential local algorithms for the random {NAE-K-SAT}
  problem.
\newblock {\em {SIAM} J. Comput.}, 46(2):590--619, 2017.

\bibitem{goel2004modified}
Sharad Goel.
\newblock Modified logarithmic sobolev inequalities for some models of random
  walk.
\newblock {\em Stochastic processes and their applications}, 114(1):51--79,
  2004.

\bibitem{guerra2004high}
Francesco Guerra and Fabio~Lucio Toninelli.
\newblock {The high temperature region of the Viana-Bray diluted spin glass
  model}.
\newblock {\em Journal of statistical physics}, 115:531--555, 2004.

\bibitem{holley1986logarithmic}
Richard Holley and Daniel~W Stroock.
\newblock Logarithmic sobolev inequalities and stochastic ising models.
\newblock 1986.

\bibitem{KoehLRColt22}
Frederic Koehler, Holden Lee, and Andrej Risteski.
\newblock Sampling approximately low-rank ising models: {MCMC} meets
  variational methods.
\newblock In {\em Conference on Learning Theory, 2022}, volume 178 of {\em
  Proceedings of Machine Learning Research}, pages 4945--4988. {PMLR}, 2022.

\bibitem{KuiKuiSpinGlass2024}
Kuikui Liu, Sidhanth Mohanty, Amit Rajaraman, and David~X. Wu.
\newblock Fast mixing in sparse random ising models.
\newblock In {\em 65th {IEEE} Annual Symposium on Foundations of Computer
  Science, {FOCS} 2024, Chicago, IL, USA, October 27-30, 2024}, pages 120--128.
  {IEEE}, 2024.
\newblock \href {https://doi.org/10.1109/FOCS61266.2024.00018}
  {\path{doi:10.1109/FOCS61266.2024.00018}}.

\bibitem{mezard1990spin}
M.~M\'ezard, G.~Parisi, and M.~Virasoro.
\newblock {\em Spin glass theory and beyond}.
\newblock World Scientific, 1987.

\bibitem{mossel2010gibbs}
Elchanan Mossel and Allan Sly.
\newblock Gibbs rapidly samples colorings of g (n, d/n).
\newblock {\em Probability theory and related fields}, 148(1-2):37--69, 2010.

\bibitem{mossel2013exact}
Elchanan Mossel and Allan Sly.
\newblock Exact thresholds for ising--gibbs samplers on general graphs.
\newblock {\em The Annals of Probability}, 41(1):294--328, 2013.

\bibitem{panchenko2013parisi}
Dmitry Panchenko.
\newblock The parisi ultrametricity conjecture.
\newblock {\em Annals of Mathematics}, pages 383--393, 2013.

\bibitem{RSBParisi}
Giorgio Parisi.
\newblock {Infinite number of order parameters for spin-glasses}.
\newblock {\em Physical Review Letters}, 43(23):1754, 1979.

\bibitem{SKModel}
David Sherrington and Scott Kirkpatrick.
\newblock {Solvable model of a spin-glass}.
\newblock {\em Physical review letters}, 35(26):1792, 1975.

\bibitem{SlySun12}
Allan Sly and Nike Sun.
\newblock The computational hardness of counting in two-spin models on
  d-regular graphs.
\newblock In {\em 53rd Annual {IEEE} Symposium on Foundations of Computer
  Science, {FOCS} 2012}, pages 361--369. {IEEE} Computer Society, 2012.
\newblock URL: \url{https://doi.org/10.1109/FOCS.2012.56}.

\bibitem{SteinNewmanSpinGlassBook}
Daniel~L Stein and Charles~M Newman.
\newblock {\em {Spin glasses and complexity}}, volume~4.
\newblock Princeton University Press, 2013.

\bibitem{stephan2022non}
Ludovic Stephan and Laurent Massouli{\'e}.
\newblock Non-backtracking spectra of weighted inhomogeneous random graphs.
\newblock {\em Mathematical Statistics and Learning}, 5(3):201--271, 2022.

\bibitem{TalagrandAnnals}
Michel Talagrand.
\newblock The {Parisi} formula.
\newblock {\em Ann. Math. (2)}, 163(1):221--263, 2006.
\newblock \href {https://doi.org/10.4007/annals.2006.163.221}
  {\path{doi:10.4007/annals.2006.163.221}}.

\end{thebibliography}

\newpage

\appendix

\section{Some Standard Proofs}

\subsection{Holley-Stroock perturbation principle}\label{sec:Lemma:HolleyStroockPerturb}

We utilise the well-known Holley-Stroock perturbation principle \cite{holley1986logarithmic}, 
in the proof of \Cref{thm:BoundMLSIonH}, \Cref{sec:thm:BoundMLSIonH}.

\begin{lemma}\label{Lemma:HolleyStroockPerturb}
 Let $\mu$ and $\nu$ be two probability measures over $\{\pm 1\}^{V_n}$, 
 while   suppose that for some $c > 1$  (possibly depending on $n$),
 $\frac{1}{c}\leq \frac{\nu(\sigma)}{\mu(\sigma)}< c$,   uniformly over all
$\sigma\in \{\pm 1\}^{V_n}$.

 Suppose Glauber dynamics for $\mu$  satisfies the modified log-Sobolev inequality with constant 
 $\GDmlogSob(\mu)$. Then, the corresponding Glauber dynamics on $\nu$ satisfies the modified log-Sobolev inequality with constant
 $\GDmlogSob(\nu)\geq c^{-2}\cdot \GDmlogSob(\mu)$.

\end{lemma}

\subsection{Proof of Lemma \ref{lem:SamllUnicyclicGnp}}\label{sec:lem:SamllUnicyclicGnp}
Write $\cE_S$ for the event that every set of vertices $S \subseteq V(\G)$, with cardinality at most $2\frac{\log n}{\log^2 d}$, spans at most $|S|$ edges in $\G$. We have

\begin{eqnarray}
\Pr\Big[ \; \overline{\cE_S} \;\Big]
&\leq 
&\sum_{k=1}^{2\frac{\log n}{ (\log d)^2}} \binom{n}{k} \binom{\binom{k}{2}}{k+1}\left(\frac{d}{n}\right)^{k+1}
%%%
\leq \sum_{k=1}^{2\frac{\log n}{ (\log d)^2}} \left(\frac{ne}{k}\right)^k\left(\frac{k^2e}{2(k+1)}\right)^{k+1}\left(\frac{d}{n}\right)^{k+1}
 \nonumber\\
&\leq& \frac{1}{n}\sum_{k=1}^{2\frac{\log n}{ (\log d)^2}}\left(\frac{ekd}{2}\right)\left(\frac{e^2d}{2}\right)^{k}
\leq \frac{ed}{ (\log d)^2}\ \frac{\log n}{n}\sum_{k=1}^{2\frac{\log n}{ (\log d)^2}}\left(\frac{e^2d}{2}\right)^{k}
\nonumber \\
&\leq& n^{-9/10}\left({e^2d}/{2}\right)^{2\frac{\log n}{ (\log d)^2}}
\leq n^{-3/4} \enspace, \nonumber
\end{eqnarray}
where the last two inequalities hold for sufficiently large $d$.

\subsection{Standard Linear Algebra}
\begin{claim}\label{claim:FromPDOrder2NormBound4Cov}
For  $N\times N$ symmetric matrices $\UpD,\UpY$, such that $\UpD$ is non-singular, 
we have $\|\UpD\cdot \UpY\cdot \UpD\|\leq \|\UpY\|\cdot \|\UpD^2\|$. 
\end{claim}
\begin{proof}[Proof of \Cref{claim:FromPDOrder2NormBound4Cov}]
We have 
\begin{align}\label{eq:FirstStepclaim:FromPDOrder2NormBound4Cov}
\| \UpD \cdot \UpY \cdot \UpD \| &=
\| \UpD \cdot \UpY\cdot \UpD
\cdot \overline{\UpD} \cdot \overline{\UpY} \cdot \overline{\UpD} 
\|^{1/2}\enspace,
\end{align}
where for matrix $\UpA$, we let $\overline{A}$ denote its transpose. 
For any symmetric $N\times N$ matrix $\UpB$ we have
 $$\|\UpB\|\leq \| \UpD^{-1}\cdot \UpB\cdot \UpD\|.$$ 
Recall that $\UpD$ is non-singular, hence the above is well-defined. 
 
Since the matrix inside the norm on the r.h.s. of \eqref{eq:FirstStepclaim:FromPDOrder2NormBound4Cov}
is symmetric, we  get 
\begin{align}
\| \UpD \cdot \UpY \cdot \UpD \| & \leq 
\| \UpY \cdot \UpD 
\cdot \overline{\UpD} \cdot \overline{\UpY} \cdot \overline{\UpD} \cdot \UpD
\|^{1/2} 
%\nonumber \\  &
\leq \| \UpY\|^{1/2} \cdot \| \UpD 
\cdot \overline{\UpD} \|^{1/2} \cdot \| \overline{ \UpY} \|^{1/2}\cdot \| \overline{\UpD}\cdot \UpD 
\|^{1/2} \leq \| \UpY\| \cdot \| \UpD^2 \| \enspace. \nonumber 
\end{align}
%For the second inequality we use the sub-multiplicativity property of the matrix norm.
For the last inequality we use that our matrices are symmetric. 
The claim follows. 
\end{proof}

\subsection{Tail Bound for sums of half-normal}\label{sec:HalfNormalTails}

We say that $Y$ follows the \emph{half-normal} distribution with parameter $\sigma$, if $Y=|X|$, where $X$ follows the Gaussian distribution $\mathcal{N}(0,\sigma^2)$. For $N>0$ integer, and $\sigma\ge 0$, let $X_1, \ldots, X_N \sim \mathcal{N}(0,\sigma^2)$ be i.i.d standard Gaussians, and write 
%\begin{equation*}
$X= \sum\nolimits_{1\leq i\leq N} |X_i| $. %\enspace.
%\end{equation*}

We show the following concentration bound for $X$
\begin{theorem}\label{thm:UpperTailBoundSumOfHalfNormRelative}
For every $\delta \ge 0$, we have 
%\begin{equation}\label{eq:thm:UpperTailBoundSumOfHalfNormRelative}
$\Pr\left[\,X > (1+\delta)\cdot\Exp[X]\,\right]\leq\exp\left(-N\cdot {\delta^{2}}/{\pi}\right)$. % \enspace.
%\end{equation}
\end{theorem}

\begin{proof}
For every real $t \ge 0$, consider the moment generating function $\Exp[\exp(t\cdot X)]$, as a non-negative random variable. Applying Markov's inequality on $\Exp[\exp(t\cdot X)]$, we get that for every $a \in \mathbb{R}$
\begin{equation}\label{eq:markov}
\Pr\left[X > a\right] = \Pr\left[\exp(t\cdot X) > \exp(t\cdot a)\right] \le \Exp[\exp(t\cdot X)]\cdot \exp(-t\cdot a) \enspace.
\end{equation}
To calculate $\Exp[\exp(t\cdot X)]$, we first observe that since $X_i$'s are i.i.d., so that
\begin{align}\label{eq:ExpOfProdMGF}
\Exp[\exp(t\cdot X)] = \Exp\left[\exp\left(t\cdot \sum\nolimits_{1\leq i \leq N} |X_i|\right)\right] = \left(\Exp\left[\exp\left(t\cdot|X_1|\right)\right] \right)^N \enspace.
\end{align}
We now focus on $\Exp\left[\exp\left(t\cdot|X_1|\right)\right] $. We have 
\begin{align}
\Exp\left[\exp\left(t\cdot|X_1|\right)\right]
%	&= 
%\int_{-\infty}^{+\infty}\exp\left(t\cdot|x|\right) \frac{1}{\sigma\sqrt{2\pi}} \exp\left(-\frac{x^2}{2\sigma^2}\right) dx\nonumber \\
	&=2
\int_{0}^{+\infty}\frac{1}{\sigma\sqrt{2\pi}}\exp\left(tx-\frac{x^2}{2\sigma^2}\right) dx\nonumber\\
%	&=2\cdot\exp\left(\frac{t^2\sigma^2}{2}\right) \cdot
%\int_{0}^{+\infty}\frac{1}{\sigma\sqrt{2\pi}}\exp\left(-\frac{t^2\sigma^2}{2}+tx-\frac{x^2}{2\sigma^2}\right) dx\nonumber\\
	&=2\cdot\exp\left(\frac{t^2\sigma^2}{2}\right) \cdot
\int_{0}^{+\infty}\frac{1}{\sigma\sqrt{2\pi}}\exp\left(-\frac{1}{2}\left(\frac{x}{\sigma}-\sigma t\right)^2\right) dx\nonumber\\
	&=2\cdot\exp\left({t^2\sigma^2}/{2}\right) \cdot 
\mathrm{\Phi}\left({\sigma t}\right) \label{eq:thatsMGF}
\enspace.
\end{align}
Hence, from \eqref{eq:ExpOfProdMGF} and \eqref{eq:markov}, we further have that
\begin{align}\label{eq:LastEquality}
\Pr\left[X > a\right] \le \frac{\Exp[\exp(t\cdot X)]}{\exp(t\cdot a)} = \exp\left(N\cdot\frac{t^2\sigma^2}{2} - t\cdot a+N\cdot\log
\left[2\cdot\mathrm{\Phi}\left(\sigma t\right)\right] \right)\enspace.
\end{align}
As we prove in \red{Lemma~\ref{lem:BoundPhi}},  below, for every $x\geq 0$, we have
$
\mathrm{\Phi}(x)\leq \frac{1}{2} +\frac{x}{\sqrt{2 \pi}} $.
Therefore, 
\begin{align}
\Pr\left[X > a\right]
	&\le 
 \exp\left(N\cdot t^2\cdot \sigma^2/2 - t\cdot a+N\cdot\log
\left[1+\sigma \cdot t\cdot\sqrt{{2}/{\pi}}\right] \right)\nonumber\\
	&\le 
	\exp\left(N\cdot{t^2\cdot \sigma^2}/{2} - t\cdot a+N\cdot\sigma \cdot t\cdot\sqrt{{2}/{\pi}}\right)\nonumber\\
	&=
	\exp\left(N\left[t^2\cdot {\sigma^2}/{2} + t\cdot\left(\sigma \cdot\sqrt{{2}/{\pi}}-{a}/{N}\right)\right]\right) 
 \label{eq:noMinimize}
\enspace. 
\end{align}
Minimizing the exponent in the rhs of \eqref{eq:noMinimize} with respect to $t \geq 0$, yields that for every 
\begin{equation}\label{eq:Alphacond}
\textstyle
a \geq 
N \sigma \cdot \sqrt{{2}/{\pi}} \enspace,
\end{equation}
we have
\begin{equation}\label{eq:MMRwrite}
\Pr\left[X > a\right]\leq\exp\left(N\left[-\frac{ \left(a \sqrt{\pi}-\sigma N\sqrt{2} \right)^{2}}{2 \pi N^{2} \sigma^{2}}\right]\right) \enspace.
%-\frac{\left(-1+\varepsilon +2 a\right)^{2}}{\pi \left(-1+\varepsilon \right)^{2}}
\end{equation}
Considering the derivative at $t=0$ of the m.g.f.  of $|X_1|$ in \eqref{eq:thatsMGF}, it is easy to check that 
\begin{equation*}
\Exp[X] 
	= \Exp\left[\sum\nolimits_{i=1}^N|X_i| \right]
	= N\cdot\Exp\left[|X_1| \right]
	= N \cdot \sigma \cdot \sqrt{{2}/{\pi}}\enspace.
\end{equation*}
So that condition \eqref{eq:Alphacond}, can be written as $a\ge \Exp[X]$, or equivalently, $a=(1+\delta)\cdot\Exp[X]$ for some $\delta \ge 0$. 
\Cref{thm:UpperTailBoundSumOfHalfNormRelative} follows by substituting $a=(1+\delta)\cdot\Exp[X]$ in \eqref{eq:MMRwrite}.
%, gives \eqref{eq:thm:UpperTailBoundSumOfHalfNormRelative}, concluding theproof.
\end{proof}
 \begin{lemma}\label{lem:BoundPhi}
Let $\mathrm{\Phi}: \mathbb{R} \to [0,1]$ be the CDF of the Standard Gaussian distribution. Then, for every $x \ge 0$
%\begin{equation*}
$\mathrm{\Phi}(x)\leq \frac{1}{2} +\frac{x}{\sqrt{2 \pi}}$. % \enspace.
%\end{equation*}
\end{lemma}

\begin{proof}
Define $H: [0,+\infty) \to \mathbb{R}$ by $H(x) = \frac{1}{2} +\frac{x}{\sqrt{2 \pi}} -\mathrm{\Phi}(x)$.
Differentiating gives
%\begin{align*}
$H^\prime(x) = \frac{1}{\sqrt{2 \pi}} -\frac{\exp(-\frac{x^2}{2})}{\sqrt{2\pi}} \ge 0,$
%\end{align*}
so that $H$ is increasing. Noticing also that $H(0) = 0$, yields $H \ge 0$, as desired.
\end{proof}

\subsection{Proof of \Cref{claim:FromSpinGlass2IsingBlock}}\label{sec:claim:FromSpinGlass2IsingBlock}

\begin{proof}
Let $\mu^{(1)}$ be the Gibbs distribution on $\{\pm 1\}^{V_n}$, induced
by the interaction matrix $\InAct_{\hat{B}}$, inverse temperature $\beta$ and
the external field $\Field$. Let $\Cov_1$ and $\infmatrix_1$ be the covariance
and the influence matrices of $\mu^{(1)}$, respectively. 

Similarly, define $\mu^{(2)}$ to be the distribution induced by 
$(|\InAct_{\hat{B}}|, \beta, \Field)$, with covariance and influence matrices
$\Cov_2$ and $\infmatrix_2$, respectively.

Also, define $\mu^{(3)}$ to be the distribution induced by 
$(|\InAct_{\hat{B}}|, \beta, 0)$, with covariance and influence matrices
$\Cov_3$ and $\infmatrix_3$, respectively. 

Both $\mu^{(2)}$ and $\mu^{(3)}$ correspond to (non homogeneous) 
ferromagnetic Ising models.
The difference in the two distributions is that in $\mu^{(2)}$ there 
is an external field, whereas in $\mu^{(3)}$ there is not. 

\Cref{eq:claim:FromSpinGlass2IsingBlock} follows by showing that for any two vertices $w,v\in V_n$, we have
\begin{equation} \label{eq:Target4FromSpinGlass2IsingBlock}
0\leq |\Cov_1(w,v)| \leq \Cov_3(w,v) \enspace.
\end{equation}
%Then, it is standard that the above implies \eqref{eq:claim:FromSpinGlass2IsingBlock}. 

From \Cref{prop:Inf2TreeRedaux} and \Cref{lem:InfToTanh}, we have 
\begin{align}\label{eq:InflVsPhis}
\infmatrix_{i}(w,v) &=\sum\nolimits_{P}\prod\nolimits_{e\in P}\infweight_i( e)\enspace,
\end{align}
where the variable $P$ runs over the paths in the tree of self-avoiding walks 
$\Tsaw(\hat{B}, w)$ from the root to the copies of vertex $v$ in the tree. 
Also, function $\infweight_i( e)$ is from \eqref{def:OfInfluenceWeights}, 
defined now with respect to  $\mu_i$. 

For $i\in \{2,3\}$, it is elementary to verify that 
\begin{align}\label{eq:PositivePhi2and3}
\infweight_i( e) &\geq 0\enspace.
\end{align}
Clearly, $\infweight_1( e)$ is not always non-negative.
%To see the above consider the properties of function $\dlogtrecur_e$ \eqref{eq:DerivOfLogRatio} as the sign of $\InAct_e$ changes.

For any path $P$ in 
$\Tsaw(\hat{B}, w)$ from the root to a copy of vertex $v$ in the tree, we have
\begin{align}\label{eq:Phi1VsPh2}
\prod\nolimits_{e\in P} | \infweight_1( e)| &=
\prod\nolimits_{e\in P} \infweight_2( e)\enspace. 
\end{align}
Then, we obtain
\begin{align}
|\infmatrix_{1}(w,v) | &\leq \sum\nolimits_{P}
\prod\nolimits_{e\in P} |\infweight_1(e)| \ =\ \prod\nolimits_{e\in P} \infweight_2( e) = \infmatrix_{2}(w,v)
\end{align}
The first inequality above is from \eqref{eq:InflVsPhis} and the triangle inequality. 
The equality in the middle is from \eqref{eq:Phi1VsPh2}, while the final 
equality is from \eqref{eq:InflVsPhis}. 

Furthermore, it is standard to show, e.g., see \cite{ding2023new}, that 
\begin{align}
\infmatrix_{2}(w,v) \leq \infmatrix_{3}(w,v)\enspace. 
\end{align}
The two above inequalities imply that for any $v,w\in V_n$, we have
\begin{align}\label{eq:Infl1VsInfl3}
| \infmatrix_{1}(w,v)| & \leq \infmatrix_{3}(w,v)\enspace. 
\end{align}
For any $i\in \{1,2,3\}$ and any two vertices $w,v\in V_n$, 
we have 
\begin{align}\label{eq:CovVsInfAllI}
\Cov_{i}(w,v) &= \infmatrix_{i}(w,v) \cdot \mu^{(i)}_w(+1) \cdot (1-\mu^{(i)}_w(+1))\enspace. 
\end{align}
The above is standard and can be obtained from the definition of $\Cov_{i}$ and $\infmatrix_{i}$. 
We then have 
\begin{align}
|\Cov_{1}(w,v)| &= |\infmatrix_{1}(w,v)| \cdot \mu^{(1)}_w(+1) \cdot (1-\mu^{(1)}_w(+1)) 
& \mbox{[from \eqref{eq:CovVsInfAllI}]} \nonumber\\
&\leq \infmatrix_{3}(w,v) \cdot \mu^{(1)}_w(+1) \cdot (1-\mu^{(1)}_w(+1)) 
& \mbox{[from \eqref{eq:Infl1VsInfl3}]} \nonumber\\
&\leq \Cov_{3}(w,v) \cdot 4 \cdot \mu^{(1)}_w(+1) \cdot (1-\mu^{(1)}_w(+1)) 
&\mbox{[from \eqref{eq:CovVsInfAllI}]} \nonumber \nonumber \\
&\leq \Cov_{3}(w,v) \nonumber \enspace. 
\end{align}
For the third line we use the observation that $\mu^{(3)}_w(+1)=1/2$. For the
last inequality we use the observation that 
$0\leq \mu^{(1)}_w(+1) \cdot (1-\mu^{(1)}_w(+1)) \leq 1/4$.

The above implies \eqref{eq:Target4FromSpinGlass2IsingBlock} and concludes
the proof of \Cref{claim:FromSpinGlass2IsingBlock}. 
\end{proof}

\end{document}